\documentclass[a4paper,10pt]{article}
\usepackage{amsthm,amsfonts,amssymb,euscript}
\usepackage{latexsym, multicol, fancybox}
\usepackage{graphicx}
\usepackage{color}
\usepackage{amsmath, amsthm, amssymb,bm}
\usepackage{epstopdf}
\usepackage{caption}
\usepackage{psfrag}
\usepackage{comment}

\usepackage{mathrsfs}
\usepackage{xcolor}
\usepackage[citebordercolor={green}]{hyperref}
\usepackage{tikz}
\usepackage{bbm}
\usepackage{dsfont}

\numberwithin{equation}{section}

\newtheorem{theorem}{Theorem}[section]
\newtheorem{lemma}[theorem]{Lemma}
\newtheorem{proposition}[theorem]{Proposition}
\newtheorem{corollary}[theorem]{Corollary}
\newtheorem{definition}[theorem]{Definition}
\newtheorem{remark}[theorem]{Remark}

\newcommand{\Red}{\textcolor{red}}

\newcommand{\bea}{\begin{eqnarray}}
    \newcommand{\eea}{\end{eqnarray}}
\newcommand{\beq}{\begin{equation}}
    \newcommand{\eeq}{\end{equation}}
\def\beaa{\begin{eqnarray*}}
    \def\eeaa{\end{eqnarray*}}
\def\ba{\begin{array}}
    \def\ea{\end{array}}

\DeclareFontFamily{U}{mathx}{\hyphenchar\font45}
\DeclareFontShape{U}{mathx}{m}{n}{
    <5> <6> <7> <8> <9> <10>
    <10.95> <12> <14.4> <17.28> <20.74> <24.88>
    mathx10
}{}
\DeclareSymbolFont{mathx}{U}{mathx}{m}{n}
\DeclareFontSubstitution{U}{mathx}{m}{n}
\DeclareMathAccent{\widecheck}{0}{mathx}{"71}

\newcommand{\Romanupper}[1]
{\MakeUppercase{\romannumeral #1}}

\def\a{{\alpha}}
\def\al{{\alpha}}
\def\b{{\beta}}
\def\be{{\beta}}
\def\ga{\gamma}

\def\de{\delta}
\def\De{\Delta}
\def\ep{\epsilon}

\def\la{\lambda}

\def\Si{\Sigma}
\def\om{\omega}

\def\Th{\Theta}
\def\Ththree{\Theta_3}

\def\vphi{\varphi}

\def\th{\theta}

\def\nab{\nabla}

\def\pr{{\partial}}
\def\les{\lesssim}
\def\c{\cdot}

\def\AA{{\mathcal A}}
\def\BB{{\mathcal B}}
\def\CC{{\mathcal C}}

\def\NN{{\mathcal N}}

\def\LL{{\mathcal L}}
\def\II{{\mathcal I}}

\def\FF{{\mathcal F}}
\def\EE{{\mathcal E}}
\def\HH{{\mathcal H}}
\def\GG{{\mathcal G}}
\def\TT{{\mathcal T}}
\def\WW{{\mathcal W}}
\def\VV{{\mathcal V}}
\def\VVt{\widetilde{\VV}}
\def\OO{{\mathcal O}}
\def\SS{{\mathcal S}}
\def\UU{{\mathcal U}}

\def\KK{{\mathcal K}}
\def\Lie{{\mathcal L}}

\def\DD{{\mathcal D}}
\def\PP{{\mathcal P}}
\def\RR{{\mathcal R}}
\def\QQ{{\mathcal Q}}
\def\YY{{\mathcal Y}}
\def\AA{{\mathcal A}}

\def\HH{{\mathcal H}}

\def\Lie{{\mathcal L}}

\def\lap{{\triangle}}

\def\D{{\bf D}}

\def\M{{\bf M}}

\def\X{{\bf X}}
\def\g{{\bf g}}
\def\e{{\bf e}}
\def\k{{\bold k}}

\def\w{{\bf w}}

\def\T{T}
\def\Z{Z}

\def\SSS{{\Bbb S}}

\def\f12{{\frac 1 2}}

\def\Pt{\widetilde{P}}
\def\It{\widetilde{I}}
\def\Jt{\widetilde{J}}
\def\JtP {\, ^{(P)} \mkern-3.5mu\Jt}

\def\Kt{\widetilde{K}}
\def\KKt{\widetilde{\KK}}
\def\DDt{\widetilde{\DD}}

\def\Bk{{\mathfrac{B}}}

\def\Bk{\mathfrak{B}}

\def\Tfrak{\mathfrak{T}}
\def\Ik{ {\protect{\mathfrak{I}}}}

\def\Jk{{\mathfrak{J}}}
\def\Kk{\mathfrak{K}}

\def\dk{\mathfrak{d}}

\def\Div{{\mbox{ \bf Div} \hspace{- 0.2pt}}}

\def\piX{\, ^{(X)}\pi}

\def\err{{\mbox{Err}}}
\def\ov{\overline}

\def\f12{\frac 1 2}
\def\lab{\label}

\def\bsplit{\begin{split}}

    \def\That{{{ \widehat T}}}
    \def\Rhat{{{\widehat R}}}
    \def\Zhat{{{\widehat Z}}}
    \def\rhat{{\widehat{r}}}

\def\Phat{\widehat{P}}

\newcommand{\echeck}{\widecheck{e}}
\newcommand{\Bcheck}{\widecheck{B}}

    \def\aund{{\underline{a}}}
    \def\bund{{\underline{b}}}
    \def\cund{{\underline{c}}}
    \def\dund{{\underline{d}}}

    \def\SS{\mathcal{S}}
    \def\phit{{\tilde{\phi}}}
    \def\gt{\tilde{\g}}

    \def\RRtp{\GG'}
    
    \def\GGtp{\GG'}
    
    \def\UUwtp{\widetilde{\UU}}

    \def\aund{{\underline{a}}}
    \def\bund{{\underline{b}}}
    \def\cund{{\underline{c}}}

    \def\Sa{S_{\underline{a}}}

    \def\SSa{\SS_{\underline{a}}}

    \def\psia{\psi_{\aund}}
    \def\psib{\psi_{\bund}}
    
    \def\phia{\phi_{\aund}}

    \def\FFab{\FF^{\aund\bund}}

    \def\RRa {\GG^{\underline{a}}}

    \def\gadot{\dot{\ga}}
    \def\lz{{\bf \ell_z}}

    \def\ntrap{trap\mkern-18 mu\big/\,}
    
    \def\Dtrap{\,\DD_{trap}}

    \def\AAt{\widetilde{\AA}}
    \def\PPt{\widetilde{\PP}}

    \def\ges{\gtrsim}

    \def\Tht{\widetilde{\Th}}

    \def\ntrap{trap\mkern-18 mu\big/\,}

    \def\AAt{\widetilde{\AA}}
    \def\PPt{\widetilde{\PP}}

    \def\ges{\gtrsim}

    \def\BEF{B\hspace{-2.5pt}E \hspace{-2.5pt} F}
    \def\BE{B\hspace{-2.5pt} E}
    \def\EF{E \hspace{-2.5pt} F}
    \def\BF{B\hspace{-2.5pt} F}
    \def\BEFdot{\dot{\BEF}\hspace{-2.5pt}}
    \def\BEFddot{\ddot{\BEF}\hspace{-2.5pt}}
    
    \def\BPEF{B\hspace{-2.5pt} P  \hspace{-2.5pt}      E \hspace{-2.5pt} F}
    \def\BP{B\hspace{-2.5pt} P  \hspace{-2.5pt}\,}  
\def\BPF{B\hspace{-2.5pt} P  \hspace{-2.5pt} F \hspace{-2.5pt}\,}

    \def\LE{L \hspace{-2.5pt} E}

    \def\BEFdeg{\,^{(deg)} \BEF}
    \def\EFdeg{\,^{(deg)} \EF}
    \def\BEdeg{\,^{(deg)} \BE}
    \def\BFdeg{\,^{(deg)} \BF}
    \def\BPdeg{\,^{(deg)} \BP}
    \def\BPFdeg{\,^{(deg)}\BPF}
    \def\BPEFdeg{\,^{(deg)} \BPEF}

    \def\psit{\widetilde{\psi}}

    \def\Mor{\mbox{Mor}}
    
    \def \Bdeg{\,^{(deg)} B}
    
    \def \Edeg{\,^{(deg)} E}
    \def\Fdeg{\,^{(deg)} F}
    \def\ahat{\hat{a}}
    
    \def\Shat{\widehat S}
    \def\SShat{\widehat{\SS}}
    
    \def\GGhat{\widehat{\GG}}
    \def\GGhatp{\widehat{\GG}'}

    \def\grad{\mbox{grad}}
    \def\gradf{\,^{(f)}\grad}
    \def\gradt{\,^{(t)}\grad}
    \def\gradr{\,^{(r)}\grad}
    \def\gradtau{\,^{(\tau)}\grad}
    \def\Tt{\widetilde{T}}
    \def\Rt{\widetilde{R}}
    \def\rt{{\tilde{r}}}

    \def\gt{{\widetilde{\g}}}

    \def\Nt{\widetilde{N}}
    \def \Tplus{{T_{+}}}

    \def\ezero{\,^{(0)} e}
    \def\Dezero{\,^{(0)}\De}
    
    \def\nabzero{\,^{(0)}\nab}
    \def\dkzero{\,^{(0)}\dk}
    \def\Tring{\mathring{T}}

    \def\Jtgood{\, ^{(g)}\hspace{-2.5pt}\Jt}
    \def\Jtbad{\,^{(b)}\hspace{-2.5pt}\Jt}
    \def\Mgood{\,^{(g)}\hspace{-2.5pt}M}
    \def\vgood{\,^{(g)}\hspace{-2pt}v}

    \def\eSS{\,^{(\SSS)} \hspace{-1.8 pt}e}

\def\nabSS{\,^{(\SSS)}\hspace{-1.6 pt}\nab}

\def\lapSS{\,^{(\SSS)}\hspace{-2 pt}\lap}

 \def\ezero{\,^{(0)} e}
\def\Dezero{\,^{(0)}\De}

\def\nabzero{\,^{(0)}\nab}
\def\dkzero{\,^{(0)}\dk}
\def\Tring{\mathring{T}}

\def\rdot{\dot{r}}

\def\yhat{\widehat{y}}
\def\Ta{\,^{(a)} T}

    \def\Bdegg{\,^{(deg)}\overline{B}}

\def\DDtrap{ \DD_{trap}}

\def\EEt{\widetilde{\EE}}

\def\piY{\, ^{(Y)}\pi}
\def\pithree{\, ^{(e_3)}\pi}
\def\piring{\mathring{\pi}}

\begin{document}

        \title{ A Physical space derivation of Morawetz-Energy estimates in Kerr spacetimes with large angular momentum.}
        \author{Lili He and Sergiu Klainerman}

        \maketitle

        \begin{abstract}

    We revisit the derivation of Morawetz-energy estimates for scalar wave
equations in the domain of outer communication of a Kerr spacetime
$\KK(a,m)$. Our goal is to develop robust physical space methods which
are well suited for extension to realistic perturbations of Kerr.

The proof rests on several ingredients. First, we derive conditional
Morawetz estimates which extend the physical space techniques initiated by
Andersson and Blue \cite{AB}, and later adapted in \cite{GKS} to
perturbations of slowly rotating Kerr, by exploiting a physical-space
characterization of the full $r$-range of trapped null geodesics. Second,
we use an idea introduced by Stogin \cite{St} in the axially symmetric case
to handle the low-frequency difficulties in the Morawetz estimates. In the
general case, the control of the lower order terms also requires making
full use of the principal trapping term in the Morawetz bulk norm, together
with a new use of Hardy-type inequalities.

A crucial new ingredient is the control of the boundary terms generated by the
Morawetz estimates.  We use as input the results of our companion paper
\cite{He-K2}, which provides a frequency independent estimate for the flux
based on a  purely  physical space  version   of the seminal Whiting transform
\cite{W}.

Finally, a continuity
argument yields an unconditional global-in-time Morawetz estimate, while a
new energy estimate is obtained from the construction of a causal vectorfield
which is Killing on the trapping set. In this new version we have added  a new, much  simpler,  proof which covers the full subextremal case.

The results proved here are restricted to scalar wave equations, corresponding
to spin $0$, in the range $|a|/m\leq 0.75$. We expect this restriction
to be technical, and the methods developed in this paper to extend to the
Teukolsky equation.
        \end{abstract}

        \tableofcontents


        \section{Introduction}

\subsection{General Motivation}

This is the first in a projected sequence of papers whose principal goal is
to revisit the derivation of Morawetz-energy estimates for wave equations
on Kerr spacetimes $\KK(a,m)$, in the full subextremal range
$|a|<m$, using robust physical space methods. Since such estimates are
principally motivated by the Kerr black-hole stability problem, we begin by
recalling the present status of the conjecture.\footnote{For a more complete
discussion, see the introductions to \cite{KS:Kerr}, \cite{GKS}, and the
survey \cite{Chr-Survey}.}

The first complete nonlinear stability result for the Kerr family was
obtained, for $|a|/m\ll1$, in the sequence of works
\cite{KS-GCM1}, \cite{KS-GCM2}, \cite{KS:Kerr}, \cite{GKS}, and
\cite{Shen}. As observed already in \cite{KS:Kerr}, the smallness assumption on the
angular momentum enters only in \cite{GKS}, mainly through the derivation of
decay estimates for the Teukolsky equations satisfied by the extreme
curvature components. The general strategy for such estimates combines
Morawetz-energy estimates with $r^p$-weighted estimates.

For the exact Kerr spacetime $\KK(a,m)$, Morawetz-energy estimates in the
full subextremal range $|a|<m$ are already known: see \cite{DRS} for the
scalar wave equation, and \cite{S-Rita1}, \cite{S-Rita2}, \cite{Millet} for
the Teukolsky equation. These results rely in an essential way on frequency
analysis and on Whiting's deep mode stability theorem \cite{W}. For this
reason, they appeared difficult to adapt directly to realistic perturbations
of Kerr, of the type expected in a proof of nonlinear stability. This
difficulty has recently been overcome by Ma and Szeftel in \cite{MaS} and
\cite{MaS2}, where they derive full Morawetz-energy estimates for both the
scalar wave equation and the Teukolsky equation on realistic perturbations of
subextremal Kerr, using the results of \cite{DRS} and \cite{Millet} as black
boxes together with delicate microlocal techniques in the spirit of
\cite{TT} and \cite{LT}.

It now appears that, by combining these advances with the methods developed
for the nonlinear stability of slowly rotating Kerr, especially the gauge
constructions in \cite{KS:Kerr}, \cite{KS-GCM1}, \cite{KS-GCM2}, and
\cite{Shen}, the remaining\footnote{In addition to the treatment of the Teukolsky equations, needed for   decay,  one also needs  
to  derive bounds  for the top derivative curvature estimates. This step, implemented in  the last part of  \cite{GKS}  for small angular momentum,   is  expected  to be less demanding.    } conceptual obstacles to the full subextremal Kerr
stability conjecture have largely been removed. Although we can no longer
aspire to be the first to provide a proof for the missing piece, we believe that the
physical space methods developed here and in our forthcoming paper
\cite{He-K2} remain valuable. They offer a way to simplify, unify, and
clarify aspects of how to derive decay of waves in both Kerr and perturbations of Kerr, which
otherwise remains technically involved and dependent on several sophisticated
black boxes. We also expect these methods to be useful in other problems in
general relativity, and possibly beyond.

As we were finishing this paper, we became aware of a recent work of Hintz
\cite{Hintz1} which also claims to settle the nonlinear stability of the full
subextremal Kerr family. Though there are major differences
between the two approaches, one can make a few comparisons between the results. Hintz works with an
unconventional class of initial data, consisting of a structured,
polyhomogeneous, more slowly decaying part, together with a general remainder
term with $O(r^{-3-\ep})$ decay. In the traditional framework of weighted
$C^k$ or $H^s$ spaces, only the $O(r^{-3-\ep})$ component is generic;
in this sense the result is more restrictive than the
$O(r^{-3/2-\ep})$ assumptions used in \cite{KS:Kerr}, \cite{GKS}, and
\cite{MaS2}.\footnote{This is also the initial data decay used in the
original Christodoulou--Klainerman stability theorem for Minkowski space.
Major improvements on the required decay were obtained by Bieri
\cite{Bieri} and by Shen \cite{Shen23}, \cite{Shen24}; the latter proves
stability of Minkowski space with only $O(r^{-\de})$ decay.}

In the language of the null-frame curvature formalism, Hintz's result for
the core $O(r^{-3-\ep})$ component is consistent with full peeling,
namely $O(r^{-5})$ and $O(r^{-4})$ decay for the $\alpha$ and
$\beta$ curvature components, rather than the weaker
$O(r^{-7/2-\ep})$ behavior. The structured, more slowly decaying part of
the data, however, leads to serious violations of peeling and is responsible
for some of the main difficulties addressed in that work.\footnote{It is
known that full peeling is inconsistent with many important physical
applications; see, for example, \cite{Ch-memory}, \cite{Ch02}, \cite{BD86},
and \cite{Kehr1}.}

On the technical side, Hintz's proof relies on microlocal methods developed
in his long-term collaboration with A. Vasy. One striking  feature of  his proof  is the reduction of the gauge problem, which is intrinsically
infinite-dimensional, to a finite-dimensional  one.\footnote{This reduction appears to be crucial  in making the Hintz-Vasy analytic  strategy work.}   The reduction appears to be intimately connected  
with  the special, highly structured choice  of initial conditions;  see Remark 1.3 in \cite{Hintz1}. Another   feature of  Hintz's result is the loss of derivatives  in the linear estimates,  treated by 
  a Nash-Moser  argument.

 Finally, it is
important to note that the approaches of \cite{MaS2} and \cite{Hintz1},
despite being very different, rely on Whiting's mode stability  result
as an indispensable input.

            \subsection{Main features of our results and methods}
                For simplicity of the exposition we limit ourselves in this paper to the case of the standard wave equation $\square_{a, m}\psi=N$. Moreover, for technical reasons, our results are restricted to the case $|a|/ m\le 0.75$.

          To start with, our method of deriving Morawetz
           estimates
            relies on the commutation properties of the relevant geometric wave equations in $\KK(a,m)$ with the second order operators generated by the Killing vectorfields $T$, $Z$ and the
        second Carter operator $\OO$. The method, to which we refer broadly, as the Andersson-Blue (AB) invariant operators method, was first introduced in \cite{AB} for solutions of the scalar wave equation $\square_{a,m} \psi=0$, for slowly rotating Kerr $\KK(a,m)$, $|a|/m\ll 1$. It was extended in \cite{GKS} to derive decay estimates for the extreme Teukolsky variables (part II of \cite{GKS}) in perturbations of Kerr as well as to provide control of the top derivatives of the curvature components (part III of \cite{GKS}). The assumption of small rotation $|a|/m\ll 1$ was crucially used in \cite{GKS} both in the derivation of the conditional Morawetz estimates (chapters 6-9) and the derivation of the energy estimates (section 7.3. of \cite{GKS}).

Morawetz-energy estimates in the full subextremal range $|a|<m$ are known:
see \cite{DRS} for the scalar wave equation, and
\cite{S-Rita1}--\cite{S-Rita2}, \cite{Millet} for the Teukolsky equation.
The proof in \cite{DRS}, in particular, relies on three ingredients closely
tied to Carter's mode separation for the wave equation on Kerr
\cite{Carter1}--\cite{Carter2}: the quantitative mode stability theorem of
Shlapentokh-Rothman \cite{Yacov}, based on Whiting's transformation method;
the strict separation between superradiant and trapped frequencies in the
full subextremal range; and the freedom to choose vectorfields adapted to
individual mode solutions.

         In our work we rely instead on the following new ingredients:
        \begin{enumerate}
         \item A far reaching extension of the Andersson-Blue (AB) invariant operators method \cite{AB} which takes full advantage of the precise $r$-range of trapped null geodesics (discussed in section \ref{section:geodesics-Kerr}).
          \item
           The innovative technique introduced by Stogin in his PhD thesis \cite{St}, in the context of axially symmetric solutions, to deal with the low frequency difficulties of the Morawetz estimates. The control of the lower order Morawetz bulk term for general solutions requires, in addition, making full use of the principal trapping term\footnote{Typically, in the trapping set, one only uses the non-negativity of the trapping term.} of the bulk Morawetz norm, as well as a new way to make use of Hardy type inequalities.
         
         \item
 Control of the  flux  along the  horizon,  based on a  new physical  space  adaptation of     Whiting's   transformation technique\footnote{To  overcome  the lack of coerciveness of the $\T$-energy,   in the context of mode solutions,   Whiting makes a  
		sequence of  surprising  transformations,   to find  a new metric  for  which $\T$   is  causal.},  designed      to  prove  mode stability in \cite{W},  for  which we refer to our  companion  paper \cite{He-K2}.

    \item The estimates of Theorem \ref{Thm:Flux-estimate} only provides bounds for $\T^2 \psi$. To derive a similar bound for $\psi$ we need the additional result of Proposition \ref{Prop.1.6-Intro}.

         \item A continuity type argument\footnote{ Similar to one used in \cite{DRS}, in a micro-localized setting.}, which ties together the preceding conditional estimates above (i.e. conditional on controlling the energy restricted to a compact set) to derive an unconditional global in time Morawetz estimate.
    \item A new energy type estimate based on the construction of a causal vectorfield, in the range $|a|/m\leq 0.9$, which is Killing in the trapping set. See also Remark \ref{remark:new-energyproof}.
        \end{enumerate}

        \subsection{Statement of the Main Theorem}
        We rely on the notations and definitions made in the preliminary section \ref{section:preliminaries} to state
        our main result. Most importantly we work with Boyer-Lindquist coordinates $(t,r,\th, \phi)$, see \eqref{eq:coordsBL}, and the principal null pair $(e_3, e_4, e_1, e_2) $ given by \eqref{eq:null-pair-in},
        \eqref{eq:canonicalHorizBasisKerr} and the two Killing vectorfields $T=\pr_t, Z=\pr_\phi$. In addition the Kerr metric has a
        second order symmetry connected to the Carter tensor defined by Lemma \ref{lemma:inversemetricexpressioninKerr}. Thus the conformal wave operator
         $|q|^2 \square$ commutes not only with the first order operators $\T$ and $\Z$ but also the second order operator $\OO$
         given by Proposition \ref{prop:decompose-square}.

         Our estimates for solutions of \eqref{eq:scalarwave} are stated in
        the domain of outer communication $\DD=\{ r\ge r_+\}$
        where
          $r= r_+= m+\sqrt{m^2-a^2}$ corresponds to the event horizon $\HH$.
         More precisely we restrict to the causal domains $\DD(\tau_1, \tau_2)=\DD\cap\{ \tau_1\le \tau\le \tau_2\}$, see Definition \ref{def:DDtau}, where $\tau$ is a strict time function in $\DD$.
        We denote by $\DD_{trap}$ the trapping region $
\DD_{trap}=\big\{ \rhat_1\le r\le \rhat_2\big\} $, with $\rhat_1, \rhat_2$ characterized in Lemma \ref{Lemma:rangeofrfortrappednullgeodesics}. The precise values for $\rhat_1, \rhat_2$, given in Lemma \ref{Lemma:rangeofrfortrappednullgeodesics} appear in \cite{Teo}, see also \cite{BPT} as well as\footnote{The latter two papers show that all trapped null geodesics are orbital. } \cite{Dyatlov} and \cite{CeJa}. For the convenience of the reader, we include a straightforward proof in section \ref{section:geodesics-Kerr} based on the characteristic polynomial
\beq
\lab{eq:charact-polynomial}
\Th(r):=\TT_{-a}^2- 4 a^2 r^2 \De, \qquad \De:=r^2-2mr+a^2,\quad\TT_{-a}:=r\big(r^2-  3mr+2a^2\big).
\eeq
We reserve $\Th$ for this sixth-degree polynomial and denote its cubic factor by
\[
\Ththree(r):=r^{-3}\Th(r)=r^3-6mr^2+9m^2r-4ma^2,
\qquad \Th(r)=r^3\Ththree(r).
\]
We show in fact, see Proposition \ref{Prop:r-trapped.region}, that $\Th\le 0$ on the trapping set and $\Th>0$ on its complement.
As it turns out, this plays an important role in our work.

We rely also on the renormalized
        vectorfield $\That=\T+\frac{a}{r^2+a^2}\Z$, which is causal in the entire domain of outer communication (DOC), $r\ge r_+$,
         and $\Rhat=\frac{\De}{r^2+a^2}\pr_r$, regular in DOC.

As mentioned above, we restrict\footnote{Our methods can be extended to higher spin equations. Full sub-extremal results for such wave equations were derived in \cite{S-Rita1,S-Rita2} based on mode decompositions. See also \cite{Millet} for a simpler proof based on microlocal analysis and spectral theory.}
 our analysis to solutions of the inhomogeneous scalar wave equation
  \beq\lab{eq:scalarwave}
        \square_{a,m}\psi=N
        \eeq
           in the exterior region $\DD$ of $\KK(a,m)$.
            To state our main results below we refer to
        section \ref{section:BEFquantities} where the bulk, energy and flux quantities $B, P,  E, F$, the combined quantities
         $\BEF$,   $\BPEF$,  $\BPF$, $\BP$,     $\EF$, $\BF$, as well as the bulk quantity $\NN$ for the inhomogeneous term $N$, are given.
           Higher derivative versions are denoted by $B^s, P^s, E^s, F^s$, etc. Each of these quantities is accompanied by a degenerate version denoted $\Bdeg^s, \Edeg^s$ etc. It is common that one first proves a degenerate version of a particular integral estimate, in which certain terms vanish at the horizon, and then passes to the non-degenerate form by an additional red shift inequality\footnote{Such estimates go back to \cite{DR1}. Since degeneracies in our estimates occur only in the $e_3$ direction at the horizon, we use a partial red shift inequality, which appears first in \cite{MMTT}.}. We also make use of the $r^p$ weighted version of all these quantities\footnote{ We note that $r^p$-weighted estimates, which go back to \cite{Mor}, were first used in the Schwarzschild context in \cite{DR2}. We refer the reader to chapter 10.3 in \cite{GKS} for a detailed treatment relevant to our work.}.

                  Note the precise characterization of the trapping set $[\rhat_1,\rhat_2]$ which appears in the definition of the Morawetz bulk quantities $B, P$,
          \beq
          \bsplit
B[\psi](\tau_1, \tau_2)&=\int_{\DD(\tau_1, \tau_2) }  r^{-2} m | \Rhat  \psi|^2  +mr^{-4} |\psi|^2+ \int_{\DD_{\ntrap}(\tau_1, \tau_2)} \left(   m r^{-2} |(e_3, e_4) \psi|^2        + r^{-1}  |\nab  \psi|^2\right)\\
P[\psi](\tau_1, \tau_2)&=\int_{\DD_{trap}(\tau_1, \tau_2) }\left|\Tfrak[\psi]\right|^2, \qquad 
\BP[\psi](\tau_1, \tau_2)=B[\psi](\tau_1, \tau_2)+P[\psi](\tau_1, \tau_2)
\end{split}
\eeq
where $\nab$ refers to the horizontal derivative in the span of $e_1, e_2$ and   $\Tfrak$ is 
\emph{trapping-frequency operator}\footnote{The symbol of $\Tfrak$
corresponds precisely to the expression $\TT_{\e,\ell_z}$ in Lemma
\ref{lemma:calculationPiDe}, which vanishes along orbital null geodesics.}
\beq
\lab{eq:trapping-frequency-operator}
\Tfrak[\psi]:=\big(\TT_{-a}T-a(r-m)\Zhat\big)\psi, \qquad \Zhat:=Z+ a T.
\eeq

        \begin{theorem}[Main Theorem]
            \lab{Thm:MainMorawetz}
            Consider a Kerr spacetime $\KK(a,m)$ with $|a|/m\leq 0.75$ and general solutions\footnote{For which $E^8_p[\psi](\tau_1) + \NN_p^8[N] (\tau_1, \infty) <\infty$.} to the scalar wave equation $ \square_{a,m} \psi=N$. Then, for any integer $s\ge 8$, $1+\de <p< 2-\de$, and any $\tau_1<\tau_2$, we have
            \beq\lab{eq:Thm-MainMorawetz-BEF}
            \BEF^s_p [\psi](\tau_1, \tau_2) \les  E_p^s[\psi](\tau_1) + \NN_p^s[\psi, N](\tau_1, \tau_2).
            \eeq
A more precise  statement  includes also trapping  information  encoded  by  $P[\psi]$;   using the notation of section 2.4,
  \beq\lab{eq:Thm-MainMorawetz}
            \BPEF^s_p [\psi](\tau_1, \tau_2) \les  E_p^s[\psi](\tau_1) + \NN_p^s[\psi, N](\tau_1, \tau_2).
            \eeq
                \end{theorem}

                The proof of the Main Theorem \ref{Thm:MainMorawetz} is based on four steps,
                which we state as theorems of independent interest.

                  Theorem \ref{Thm:Morawetz1}, which establishes a Morawetz type estimate conditional on the control of the degenerate energy
                 flux through the horizon and a future spacelike hypersurface, is the core of this paper.

                \begin{theorem}[Morawetz]
                \lab{Thm:Morawetz1}
                 Under the same assumptions, with $s\ge 2$, we have
                    \beq
     \lab{eq:Thm-Morawetz1}
\BPEF_p^s[\psi](\tau_1, \tau_2)\les \EFdeg[(e_3, e_4,r\nab)^{\leq s}\psi](\tau_1, \tau_2)+ E_p^s[\psi](\tau_1)+ \NN_p^s[\psi,N](\tau_1, \tau_2).
\eeq
\end{theorem}
In a typical situation, when $\T$ is both Killing and timelike in $\DD$, the $\EF$ term on the right hand side can be easily controlled by $\T$-energy estimates. This, of course, is not the case of Kerr due to the presence of the ergoregion. The case $|a|/m\ll 1 $ can however be treated, as in \cite{GKS}, by replacing $\T$ with $\That=\T+\frac{a}{r^2+a^2} \Z$, which is timelike in $\DD$ and, though not Killing, has a small deformation tensor proportional to $|a|/m$.

The case of large $|a|/m$ requires a completely different approach. To start with, one can easily check, using the coerciveness of the standard $\T$-energy inequality, that the $\Edeg$ part of the $\EFdeg$ norm of the right hand side of \eqref{eq:Thm-Morawetz1} can be replaced by a term of the form $\int_{\Si(\tau_2)\cap\{r\le r_0\}}|\dk^s\Z\psi|^2$ supported on a compact set and $\Fdeg$. As one expects that $\int_{\Si(\tau_2)\cap\{r\le r_0\}}|\dk^s\Z\psi|^2$ decays to zero at $\tau_2\to \infty$ one expects to be able to deal with it by a continuity argument.

  It thus remains to control the flux $\Fdeg[\psi]$ through the horizon.
  In \cite{DRS} the authors rely on a quantitative version of Whiting's mode stability result, established in \cite{W}, to control the flux in the restricted case of bounded frequencies. To treat large superradiant frequencies\footnote{For the non-superradiant frequencies, standard $T$-energy estimate gives control of the flux.} they make use of the separation between superradiant and trapping frequencies, mentioned above.

  In our work we   derive   the full flux estimate,  with no  restriction on the frequencies,  conditional  on   controlling  the compact energy plus  an additional   small multiple  of the  $\BPEF_p^s$  term, which can  be later  absorbed. To  achieve such an estimate   we rely in a fundamental way 
  on a novel, physical  space,  extension of Whiting's  transformation, which we introduce in \cite{He-K2}.
   The idea is to transform solutions $\psi$ of   $\square_\g \psi=N$, with $\g$ the Kerr metric, into  solutions $\psit$  verifying $\square_\gt \psit=\Nt$ for a   metric $\gt$, defined on a new  manifold  $\KKt=\KKt(a,m)$.  This manifold has  two  asymptotically flat, stationary,  ends $\KKt=\DDt_{-}\cup \DDt_+$,  and as in \cite{W},  preserves the symmetries of the Kerr metric.\footnote{Moreover,   for a certain  subdomain $\DDt\subset\DDt_-$,  which can be identified as   the exterior  of a black hole region in $\KKt$,  the metric $\gt$ is Lorentzian   and the time  translation $\T  $ is timelike. }  Using the  geometric properties  of the  two  ends  one  can then deduce,  by   classical    asymptotic Fourier  analysis techniques,      the desired  bound  for  the    flux of  $\psi$ at the future event horizon.

\begin{theorem}[Conditional Flux estimate]
\lab{Thm:Flux-estimate}
The following estimate holds true for solutions of $\square_{a,m} \psi =N$, and for $p>0$, for $r_0$ sufficiently large and $\ep_0$ sufficiently small
\beq
\lab{eq:Flux-estimate}
\bsplit
 \Fdeg[ T^2\psi](\tau_0, \tau) & \les
\int_{\tau}^{\tau+1} E^2_{\leq r_0}[\psi](s) ds +\ep_0 \BPEF[\psi](\tau_0, \tau) \\
 &  +E_p^2[\psi](\tau_0)+\NN_p^2[\psi, N](\tau_0, \tau).
 \end{split}
\eeq
\end{theorem}
We postpone the proof of Theorem \ref{Thm:Flux-estimate} to \cite{He-K2} where we also give a full account of our transformation technique.

Note that the estimate \eqref{eq:Flux-estimate} does not provide any information in the stationary case, i.e. when $\T \psi=0$. To correct for this we need to make use of the following.

\begin{proposition}
\lab{Prop.1.6-Intro}
The following estimate holds true for solutions of $\square_{a,m} \psi =N$ with $s\geq 2$, $1+\de<p<2-\de$ and
$r_0$ sufficiently large
\beq
\lab{eq:Prop.1.6-Intro}
\bsplit
\EFdeg^s_p[\psi](\tau_0, \tau) &\les   \Edeg^s_{\le r_0}[\psi](\tau_0) + r_0^{1-p} B_p[\psi](\tau_0, \tau) + r_0^{3-p} \Bdeg_p^{s-1}[\T\psi] (\tau_0, \tau)\quad\\
&+ \NN_p^s[\psi,N](\tau_0, \tau).
\end{split}
\eeq
\end{proposition}

We rely on Theorems \ref{Thm:Flux-estimate} and Proposition \ref{Prop.1.6-Intro} to deduce the following\footnote{This step also requires making use of the commutations of $\square_{a,m}$ with $\T, \Z, \OO$ to upgrade \eqref{eq:Flux-estimate}, as well as the norm equivalence lemmas of section \ref{section:normEqivLemmas}. }

\begin{theorem}[Conditional Morawetz-Energy]
\lab{Thm:Morawetz2}
                Under the same assumptions, with $s\ge 4$, for some fixed $r_0> 2m$, sufficiently large
\beq
\lab{eq:Thm-Morawetz2}
\bsplit
\BPEF_p^s[\psi](\tau_0,\tau_1)&\les \int_{\Si(\tau_1)\cap\{r\le 	 r_0\}}|\dk^s\Z\psi|^2+ \int_{\tau_1}^{\tau_1+1}\int_{\Si(s)\cap\{r\le 	 r_0\}}|\dk^s\Z\psi|^2\,ds\\
&+ E_p^s[\psi](\tau_0) + \NN^s_p[\psi, N](\tau_0, \tau_1).
    \end{split}
\eeq
\end{theorem}
\begin{remark}To    derive  the  desired estimate   in \eqref{eq:Thm-Morawetz2} we need to  make use of   the smallness of  $r_0^{1-p} $   in \eqref{eq:Prop.1.6-Intro}.
The proof of Proposition \ref{Prop.1.6-Intro}, in section \ref{section:Proof-Prop.1.6-Intro}, is based on another energy type estimate generated by an appropriate, non-Killing, deformation of $\T$.
\end{remark}

Based on the expectation that the first two local energy terms $\LE$ on the right hand side of \eqref{eq:Thm-Morawetz2}, converge to zero as $\tau\to \infty$,  we  derive the following unconditional control of the global $\BF$ norms.

\begin{theorem}[Global Morawetz]
\lab{Thm:Morawetz3}

Under the same assumptions as in the Main Theorem, with $s\ge 8$, we have
\beq
\lab{eq:Thm-Morawetz3}
\BPF_p^{s}[\psi] (\tau_0, \infty) \les  E^s_{p}[\psi](\tau_0)+ \NN_p^s[\psi,N](\tau_0, \infty).
\eeq
\end{theorem}
  The proof of Theorem \ref{Thm:Morawetz3}, presented in section \ref{section:Proof of Thm-Morawetz3}, uses a continuity argument\footnote{A similar argument was used in \cite{DRS}, in the microlocalized setting of that paper.} with respect to $a$.
     Roughly, we assume that the expectation mentioned above, i.e. $\LE\to 0 $ as $\tau\to \infty$, holds true for $a\le a_0$, note that it is unconditionally true\footnote{Based on the small $|a|/ m$ results of \cite{GKS}. The closeness requirement in the continuity argument is straightforward.} for small $a_0$, and show that it must also hold
    true for a small neighborhood of $a_0$, for which the result\footnote{In fact a simple corollary of it, see Proposition \ref{Prop:Main-quantitative-estim}.} of Theorem \ref{Thm:Morawetz2} is valid. We note that
this is the only part of the paper which requires some mild frequency decomposition (in fact only decomposition into $Z$-frequency modes). As in \cite{DRS} the continuity argument makes use of the orthogonality properties of the $\Z$-frequency modes.

To finish the proof of estimate \eqref{eq:Thm-MainMorawetz} in the Main Theorem we
 still need to control the energy. This is done by the following.

\begin{theorem}[Conditional Energy]
\lab{Thm:Morawetz4}
Under the assumption $|a|/m <0.9$, $s\ge 0$, $\de<p< 2-\de$, we have
\beq\lab{eq:Thm-Morawetz4}
E^s_{p}[\psi](\tau_2)\les  E^s_{p}[\psi](\tau_1)+ B_p^s[\psi] (\tau_1, \tau_2) +\NN_p^s[\psi,N](\tau_1, \tau_2)
\eeq
\end{theorem}
The proof of Theorem \ref{Thm:Morawetz4} is based on a choice of a vectorfield multiplier which is precisely Killing and timelike on the trapping set $[\rhat_1, \rhat_2]$.
 More precisely,
    due to our precise characterization of the trapping set, in section \ref{section-range.trapped.ngeod},
       we construct, under the assumption that $|a|/m\le 0.9$, a causal vectorfield $\Tring$ which is Killing and timelike in the trapping region $[\rhat_1, \rhat_2]$, equal to
       $\That=T+\frac{a}{r^2+a^2}Z$ before the trapping, and equal to $\T$ everywhere else.
      The corresponding energy-flux estimate generates a bulk term, supported away from the trapping region, which can be controlled in view of Theorem \ref{Thm:Morawetz3}.
      Note that the range for which $\Tring $ can be constructed far exceeds the range $|a|/m<\sqrt{2}/2$, for which the trapping set is disjoint from the ergoregion, in physical space.
      
     \begin{remark}
     \lab{remark:new-energyproof}
     In section \ref{subsec:new-unrestricted-energy} we sketch a much shorter,
     unrestricted proof of Theorem \ref{Thm:Morawetz4}, which was signaled to
     us by S.~Ma.  The main new ingredient in the proof is   based on the  control the  $P$ term. 
     \end{remark}

\begin{remark}
The results mentioned above remain true if we work with a time function $\tau$, which is spacelike, asymptotically null, instead of asymptotically flat.
The only modification required in that case is to adapt the definition of the energy norms $E_p^s[\psi](\tau)$, see section 6.1.5 in \cite{GKS}.
 \end{remark}

 \subsection{Comments   on  the proof of  Theorem  \ref{Thm:Morawetz1}}

 The proof of Theorem \ref{Thm:Morawetz1} is itself the consequence of four steps.
 The first one is an adaptation of a result proved by Stogin in his PhD thesis \cite{St}, in the context of axially symmetric solutions\footnote{Note that in that case it is straightforward to derive full control of the energy-flux and as such the
    result reduces to deriving a Morawetz type estimate.} of the wave equation in the sub-extremal case. For simplicity, in this section, we ignore the contribution of the inhomogeneous term $N$ in \eqref{eq:scalarwave}, which can be easily added to the final result.

 \begin{theorem}[$\Zhat$-conditional Morawetz]
            \lab{Thm:Moraw2}
            The following estimate holds true for solutions of \eqref{eq:scalarwave} in the full sub-extremal case, with
            $\Zhat=Z+aT$,
            \beq
            \bsplit
            \lab{eq:Thm-Moraw2}
            \Mor[\psi](\tau_1,\tau_2) +\EF[\psi](\tau_1,\tau_2)&\les  \EFdeg[\psi](\tau_1,\tau_2) +   E[\psi](\tau_1)          +\int_{\DD(\tau_1,\tau_2)} \frac{a^2}{mr^4}|\Zhat\psi|^2
             \end{split}
            \eeq
            where\footnote{This polynomial, introduced in \eqref{eq:charact-polynomial},
    is closely related to the determination of the $r$-range of trapped null geodesics in section
\ref{section-range.trapped.ngeod}}
            \begin{align*}
            \Mor[\psi](\tau_1, \tau_2) &:=  \Bdegg[\psi ](\tau_1, \tau_2)+        \int_{\DD(\tau_1, \tau_2) }
    \frac{\TT_{-a}^2}{r^6} \left(\frac{m}{r^2} |(e_3, e_4) \psi|^2 + \frac{1}{r}|\nab\psi|^2\right),\\
             \Bdegg[\psi](\tau_1, \tau_2)&:=\int_{\DD(\tau_1, \tau_2) } m r^{-2}  | \Rhat  \psi|^2  +mr^{-4} |\psi|^2.
    \end{align*}
    Here $\TT_{-a}=r\big( r^2 - 3mr+2a^2\big)$ and $\nab$ refers to the horizontal derivative in the span of $e_1, e_2$. 
            \end{theorem}

 The second result, which holds true outside the $r$-range of the trapping set $[\rhat_1,\rhat_2]$, is conditional
  on controlling the degenerate quantity $ \Bdegg[\psi](\tau_1, \tau_2)$.

        \begin{theorem}[Non-trapping, conditional, Morawetz]
            \lab{Thm:Moraw3}
            Under the same assumptions as in Theorem \ref{Thm:Moraw2} we have the following conditional result
            \beq\lab{eq:Thm-Moraw3}
            \bsplit
                \int_{\DD_{\ntrap}(\tau_1, \tau_2)}  \left(\frac{m}{r^2} |\That \psi|^2 + r^{-1}|\nab\psi|^2\right)
                \quad&\les
                \Bdegg[\psi](\tau_1, \tau_2) +\EFdeg[\psi](\tau_1,\tau_2)\\
                &+     \Edeg[\psi](\tau_1)
            \end{split}.
            \eeq
        \end{theorem}
        The following provides the main step in the proof of Theorem \ref{Thm:Morawetz1}. The crucial feature of this step is that
        it relies on the above non-trapping conditional Morawetz estimate \eqref{eq:Thm-Moraw3} and the commutation properties of the second order operators
        $\SS_1=\T^2, \, \SS_2=\T \Z,\,  \SS_3=\Z^2,\,  \SS_4=\OO$. More precisely  if  $\psi$ verifies  $\square_{a,m}\psi =0$  so  do   $\psi_\aund=\SS_\aund\psi$, $\aund=1,2,3,4$.
 \begin{theorem}[$\SS$-Morawetz]
            \lab{Thm:Moraw4}
            The following estimate holds true for solutions of \eqref{eq:scalarwave} in the range\footnote{See Remark \ref{remark:restriction-am} to track down where this purely technical restriction is necessary in our proof.} $|a|/m\leq 0.75$, \beq
            \lab{eq:Thm-Moraw4}
            \sum_{\aund=1}^4 \BP[\psi_\aund](\tau_1, \tau_2) \les\sum_{\aund=1}^4\Big(\EFdeg^2[\psia](\tau_1,\tau_2)+E[\psia](\tau_1)\Big).
            \eeq
        \end{theorem}

 Combining Theorem \ref{Thm:Moraw2} (applied to $(T,\Zhat)^{\leq1}\psi$) with Theorem \ref{Thm:Moraw4}, together with the Lemma below, gives the estimate \eqref{eq:Thm-Morawetz1} of Theorem \ref{Thm:Morawetz1} for $s=2$.

  \begin{lemma}
        \lab{Lemma:calculuslemmaintro}
            The following estimate holds true,
            \beq
            \BEF^2[\psi](\tau_1,\tau_2)\les \sum_{\aund=1}^4  \BEF[\psi_\aund](\tau_1,\tau_2) +\BEF[(T, Z)^{\leq 1}\psi] (\tau_1, \tau_2).
            \eeq
        \end{lemma}
 The passage from $s=2$ to $s>2$ in Theorem \ref{Thm:Morawetz1} follows from Lemma \ref{lem:higherordercalculuslemma}. In short, the idea is to apply the result for $s=2$ to $(\T, \Z) \psi$ and then use elliptic theory to estimate all other derivatives.
    The $r^p$ version of the estimates in Theorem \ref{Thm:Morawetz1}
can then be derived precisely as in section 10.3 of \cite{GKS}.


    \subsection{Main ideas in the proof of Theorems \ref{Thm:Moraw2}--\ref{Thm:Moraw4}}

    The standard way to prove a Morawetz estimate for solutions of $\square_{a, m}\psi=N$ is to apply the divergence theorem
    in $\DD(\tau_1, \tau_2)$ to an appropriate current $\PP$ obtained, as in \eqref{eq:Gen-current}, based on the energy-momentum tensor $\QQ$ of $\square_{a,m}$, a vectorfield $X$
    and scalar $w$. Thus, denoting as in \eqref{definition-EE-gen1} (with $M=0$),
$\EE[\psi][X, w]:= \D^\mu \PP_\mu[X, w] - \left(X(\psi)+\frac 1 2 w \psi\right)\c \square_{a, m} \psi$
one derives an identity of the form, see \eqref{identity:prop.Morawetz1},
\beq
\lab{eq:identityPP-intro}
|q|^2\EE[X, w]  =\AA |\pr_r\psi|^2 +P+\VV |\psi|^2, \qquad P= \UU^{\a\b}\D_\a \psi \D_\b \psi.
    \eeq
    The goal then is to choose an appropriate vectorfield $X$, proportional to $\pr_r$ and scalar $ w$ so that, modulo other divergence terms, $\AA, \VV$ are strictly positive and $P $ non-negative.
    The expectation is that $P$ degenerates on the trapping set. Another important consideration is that the boundary terms, generated by the divergence theorem, be, at least in principle, commensurate with
    the energy flux generated by the causal vector-field $\That$. This imposes serious restrictions on the asymptotic behavior of $X$ and $w$ as $r\to \infty$.
    For convenience, we refer to these two properties as necessary to achieve \textit{partial coerciveness}.

    We typically choose $X= -zfh \pr_r$ and $w=-z\pr_r(fh)$ in which case, see section \ref{sect:Basicscalarident},
    \beq
\lab{eq:coeeficientsUUAAVV-intro}
\bsplit
\AA&=-z^{1/2}\Delta^{3/2} \partial_r(h \frac{ z^{1/2}  f }{\Delta^{1/2}} ),   \\
\UU^{\a\b}&=   \frac{ 1}{2}  h f \pr_r( \frac z \De\GG^{\a\b}),\\
\VV&=   \frac 1 4\pr_r\Big(\De \pr_r \big(
z \pr_r ( h f )  \big)  \Big)
\end{split}
\eeq
where see Lemma \ref{lemma:inversemetricexpressioninKerr},
\beq
\lab{eq:GG-intro}
\GG^{\a\b}=\GG^{\aund}S_{\aund}^{\a\b}= -(r^2+a^2) ^2\That^\a \That^\b  +\De O^{\a\b}.
\eeq
with $
\GG^1=-(r^2+a^2)^2, \quad \GG^2 = -2a(r^2+a^2), \quad \GG^3 =-a^2, \quad \GG^4=\De.
$

It can be shown, even in the far simpler case when $a=0$, or in the case of axially symmetric solutions, that partial coerciveness cannot be achieved, unless one either gives up on the smoothness of $X$
    and $w$, i.e. allow them to blow up at $r=r_+$, or one relies on a modified $\PP=\PP[X, w, M]$, depending on an additional $1$-form $M$, as in \eqref{eq:Gen-current},
     which is used to rebalance the integrals of $I=\AA |\pr_r\psi|^2$ and $ J= \VV |\psi|^2$.
     Indeed choosing $M = v(r,\th) \pr_r$, for some function $v=v(r,\th)$
     the new divergence identity takes the form, see \eqref{expression-Div-M-I},
     \beq\lab{expression-Div-M-I-intro}
          \bsplit
     &|q|^2\EE[X, w, M] \\
     &\quad =\AA |\pr_r\psi|^2 +P+\VV |\psi|^2+\frac 1 4 |q|^2\Big( 2 v(r,\th)\psi\c \nab_r \psi + (\pr_r v+ \frac{2r}{|q|^2} v) |\psi|^2 \Big).
     \end{split}
     \eeq
     One can then look for a function $v$ which achieves the balancing. We refer to this choice of $v$ as
      a \textit{Hardy type} procedure.

     In the general case, even for small $|a|/m$, it is also known (see \cite{A}) that no choice of $f, z, h$ or $v$ can achieve the positivity of $P$. It is for this reason that Theorem \ref{Thm:Moraw2} can only be conditional.

    \subsubsection{The case of  axial symmetry}
    \lab{sec:intro-axisym}
      In the special case when $\psi$ is axially symmetric, or $a=0$, when the trapped set reduces to the surface $\TT=0$, with
      $\TT= r^3 - 3 mr^2 +ra^2 + m a^2$, one can choose\footnote{ Following Stogin, we call $z$ a ``trapped geodesic potential" in view of the fact that $\pr_rz$ vanishes along \textit{some} trapped null geodesic.} $z=\frac{\De}{(r^2+a^2)^2}$ and $ f=\pr_r z=-\frac{2\TT}{ (r^2+a^2)^3}$
    such that $P\ge 0$, vanishing only on the trapped set. Indeed, for that choice of $z$ and $f$
    \beq
 \lab{eq:P-firstMorawetz}
P = \frac 1 2 h f  \left(  -\frac{2\TT}{ (r^2+a^2)^3} O^{\a\b} \D_\a \psi\c\D_\b\psi+ \frac{4ar}{(r^2+a^2)^2} \That\psi \c\Z \psi\right).
 \eeq
and, since $\Z\psi=0$ and $O^{\a\b}\D_\a\psi \D_\b\psi$ is positive definite (see Proposition \ref{prop:decompose-square}\,), one can choose $h$  positive  to derive the  non-negativity  of $P$.

     That choice, together with an appropriate choice of $h$, also works to ensure the positivity of $\AA$. The difficulty then is to also ensure, for some choice of $v$, the positivity of the quadratic form
    \[
    \AA |\pr_r\psi|^2 + \VV |\psi|^2+\frac 1 4 |q|^2\Big( 2 v(r,\th)\psi\c \nab_r \psi + (\pr_r v+ \frac{2r}{|q|^2} v) |\psi|^2 \Big).
    \]
    This is already non-trivial for $|a|/ m\ll 1$, see chapter 7 in \cite{GKS}, but much harder in the full sub-extremal case. The main cause of the difficulty is the complicated form of the expression $ \VV= \frac 1 4\pr_r\Big(\De \pr_r \big(
z \pr_r ( h f ) \big) \Big)$ which involves up to 4 derivatives of $z$.

     In his thesis Stogin introduced a different approach based on the idea of ensuring first the non-negativity of $\VV$, by choosing
     $u=h f$ appropriately so that $z\pr_r{u}$ is constant near the event horizon. The choice ensures the positivity of $\AA$ and non-negativity of $\VV$, at the price that $u$ blows up logarithmically near $r=r_+$. To also control the crucial term $P$, we need to choose $u$ such that it has the same sign as $-\TT$ and this can easily be arranged by choosing an appropriate integration constant.
      Remarkably Stogin's method works for the entire sub-extremal case and does not require the Hardy procedure\footnote{However a simpler local Hardy estimate is needed to fix the logarithmic divergence of the vectorfield $X$ at the horizon, see Lemma \ref{lem:localHardy}. We note that \cite{MMTT} also works with such vectorfields.}.

     \subsubsection{Comments on the proof of Theorem \ref{Thm:Moraw2}}
    \lab{section-Thm:Moraw2-intro}

In our proof of Theorem \ref{Thm:Moraw2} we follow Stogin's procedure but make a different choice
of the function $z$. Choosing $z=\frac{\De}{r^4}$ and $f=\pr_r z=-2\frac{\TT_{-a}}{ r^6}$, $P$ takes the form
\beq
\lab{eq:P-introd}
P=hf\Big( -\frac{2\TT_{-a}}{ r^6}   O^{\a\b}\D_\a\psi\D_\b \psi + \frac{4a}{r^3} \frac{r^2+a^2}{r^2}\That \psi\c \Zhat \psi\Big).
\eeq
To derive the estimate \eqref{eq:Thm-Moraw2} one needs to modify the choice of $w$ in the definition of the current $\PP$ to allow the presence of $|\That \psi|^2 $ in the above formula for $P$, that is by making an appropriate \textit{Lagrangian modification}\footnote{ Modifications of $\PP$ by a multiple of the Lagrangian does not affect the trapped set. Indeed, in the high frequency approximation null Lagrangians correspond to \textit{null} geodesics. } of the original $\PP$.
One can then apply Cauchy--Schwarz for the integral of the term in $\That \psi\c \Zhat \psi$
and absorb the term in $\That \psi$ to the left.
\begin{remark}
The reason we prefer the choice $z=\frac{\De}{r^4}$, rather than the old choice,
is intimately tied to the polynomial $\TT_{-a}$ which appears both in the derivative of $z$ and the characteristic polynomial $\Th$, see \eqref{eq:charact-polynomial}.

\end{remark}

     \subsubsection{Comments on the Proof of Theorem \ref{Thm:Moraw3}}

     The positivity of the quadratic term in $\That\psi, \nab\psi$ in the Morawetz bulk on the left hand side of the estimate \eqref{eq:Thm-Moraw3} corresponds exactly to the positivity of $\Theta$ (see \eqref{eq:charact-polynomial})
away from the trapping set. The proof, given in section \ref{section:Largemodified}, requires an appropriate Lagrangian modification of the current $\PP$, to introduce $| \T\psi|^2$
 in the formula for $P$ and a \textit{completing the square} argument
 for the corresponding quadratic form, based on the positivity of $\Th $ in the complement of the trapping set.
  See remark \ref{rem:choice-a} for further details.

    \subsubsection{Comments on the proof of Theorem \ref{Thm:Moraw4} (high frequency)}

    The principal conceptual difficulty for deriving a coercive bulk Morawetz estimate based on the approach discussed in
     the proof of Theorem \ref{Thm:Moraw2} is due to the form \eqref{eq:P-introd} of $P$ which fails to have a favorable sign in the trapping set $[\rhat_1, \rhat_2]$.

      The crucial observation in the approach initiated in \cite{AB} is that, modulo a divergence term, $P$ can be rewritten in the form
     \[
     P=- h f r^{-3} \Big(  -\frac{2\TT_{-a}}{ r^3} \OO\psi+ 4 a \frac{r^2+a^2}{r^2}\That \Zhat \psi\Big)\psi= h f r^{-3}\Psi \c \psi
     \]
     with
     \beq
     \lab{eq:defPsi-intro}
     \Psi=\Big( -\frac{2\TT_{-a}}{ r^3} \OO+ 4 a \frac{r^2+a^2}{r^2}\That \Zhat \Big)\psi
     \eeq
     The symbol of $\Psi$ vanishes  on  the trapping set\footnote{The symbols of all three terms  in \eqref{eq:Psi-Tfrak}  vanish  on the trapping set.}, in fact  a simple calculation shows that 
     \bea
     \lab{eq:Psi-Tfrak}
\Psi=-\frac{2}{r\De}
\big((r^2+a^2)T+aZ\big)\Tfrak[\psi]
-\frac{2\TT_{-a}}{r^3}|q|^2\square_\g\psi
+\frac{2\TT_{-a}}{r^3}\partial_r(\De\partial_r\psi).
\eea
  where $ \Tfrak$ is  precisely the 
trapping-frequency operator introduced in \eqref{eq:trapping-frequency-operator}.

     The above form of $P$ suggests the possibility that a coercive bulk estimate can be derived for the second order derivatives
     $\psi_\aund=\SS_\aund \psi$, in such a way that the new expression  is quadratic in $\pr \Psi$ rather than $\pr \psi$, and non-negative.
     To achieve that one makes use of the generalized current $\PP_\mu$, associated with the $4$-tuple $(\psi_1, \psi_2, \psi_3, \psi_4)$, described in section \ref{section:SS-valuedidentity}, a family of vector-fields
       $X^{\aund\bund}=\FF^{\aund\bund}\pr_r$ and scalars $w^{\aund\bund} =- z\pr_r (h f^{\aund\bund})$, with $z=\frac{\De}{r^4}$ as before. The identities
       \eqref{eq:identityPP-intro} and \eqref{eq:coeeficientsUUAAVV-intro} are then replaced by
       \beq
       \lab{eq:EE[X, w]-intro}
|q|^2\EE[X, w]   =\AA^{\aund\bund}  \pr_r\psia\c \pr_r\psib + \UU^{\a\b\aund\bund} \, \D_\a \psia \c \D_\b \psib  +\VV^{\aund\bund} \psia\c\psib
\eeq
where
\begin{align*}
\AA^{\aund\bund}&=-z^{1/2}\Delta^{3/2} \partial_r(h \frac{ z^{1/2}  f^{\aund\bund} }{\Delta^{1/2}}  ),\\
\UU^{\a\b\aund\bund}&=   \frac{ 1}{2}  h f^{\aund\bund} \pr_r( \frac z \De\GG^{\a\b})=\frac{ 1}{2}  h f^{\aund\bund} \pr_r( \frac z \De\GG^\cund)S_\cund^{\a\b},\\
\VV^{\aund\bund}&=\frac 1 4 \pr_r \Big(\De \pr_r \big( z \pr_r( h  f^{\aund\bund} ) \big)\Big).
\end{align*}
Because of the term $\pr_r\left(\frac{z}{\Delta}\GG^{\cund}\right)
$ in $\UU^{\a\b\aund\bund}$, we choose $ f^{\aund\bund} $, roughly of the form,
\[
f^{\underline{a}\underline{b}}= \UU^{(\underline{a}} \LL^{\underline{b})}= \frac 1 2 \big(\UU^{\aund}\LL^{\bund}+\UU^{\bund}\LL^{\aund} \big),\quad \UU^{\aund} =\pr_r\left(\frac{z}{\Delta}\GG^{\aund}\right)
\]
where $\LL^{\aund} $ are non-negative constants, such as $\LL=(m^2, 0, 1, 0)$. We note that $\UU^\aund\psi_\aund=\Psi$ where $\Psi$ is precisely the trapping quantity in \eqref{eq:defPsi-intro}. Therefore, the above choice of $f^{\aund\bund}$ (which is proportional to $\UU^\aund$) ensures that, modulo a divergence term,
\[
 \UU^{\a\b\aund\bund} \, \D_\a \psia \c \D_\b \psib=\frac12h\UU^\aund\LL^\bund\UU^\cund S_\cund^{\a\b}\D_\a\psia\c\D_\b\psib=\frac12h\LL^\bund S_\bund^{\a\b}\D_\a(\UU^\aund\psi_\aund)\c\D_\b(\UU^\cund\psi_\cund)
 \]
is indeed a positive quadratic form in $\pr\Psi$ (provided that $h$ is positive).
The above choice of $f^{\aund\bund}$ also ensures that the computation of the coefficients of $\AA^{\aund\bund}, \VV^{\aund\bund}$
follows the same pattern as in the scalar case, with $hf$ replaced by $h\UU^\aund$.

We still have freedom in choosing the positive function $h$. We can choose it appropriately (e.g. $h=r^5$) such that $\AA^{\aund\bund}\pr_r\psia\c\pr_r\psib$ is positive. More specifically
\[
\AA^{\aund\bund}\pr_r\psia\c\pr_r\psib=\AA^\aund\pr_r\psia\c\LL^\bund\psib,\quad \AA^\aund\psia=\frac{\De^2}{r^4}(\frac{8a^2}{r}\Zhat\Zhat\psi+\frac{2(3mr-4a^2)}{r}\OO\psi).
\]
Finally, we note that the actual proof in section
\ref{section-proof-mor-2}, and more precisely Proposition
\ref{Prop:positivityP}, uses the modified trapping quantities
$\Psi_1,\Psi_2$ defined in \eqref{eq:DefPsi-P}, together with the integral
trapping norm introduced in section \ref{section:BEFquantities}.  Their
relation to $\Psi$, formula \eqref{eq:Psi-Tfrak}, and the geometry of the
trapping set is explained in Remark \ref{remark-Psi-andP}.

\subsubsection{Comments on the proof of Theorem \ref{Thm:Moraw4} (low frequency)}
\lab{subsection:lowfreq-intro}

We are left with the most technical part of the proof, which consists of estimating the lower order term
 $\VV^{\aund\bund} \psia\c\psib$. Unfortunately, with the above choice $f^{\aund\bund}, h$, the term $\VV^{\aund\bund} \psia\c\psib$ is not necessarily positive. In fact, we have\footnote{From here it is easy to derive control of the very low (near stationary) time frequencies.}
\[
\VV^{\aund\bund} \psia\c\psib=\VV^\aund\psia\c\LL^\bund\psib,\quad \VV^\aund\psia=\VV^2T\Zhat\psi+\VV^4\OO\psi.
\]
Due to the presence of the term $\VV^2T\Zhat\psi$, we see that $\VV^{\aund\bund} \psia\c\psib$ has an indefinite sign (even when $\VV^4$ is positive). To handle the term $\VV^2T\Zhat\psi$ we make use of the form of the expression of $\Psi$ (see \eqref{eq:defPsi-intro})
\[
\Psi=\frac{4a}{r^3}T\Zhat\psi+\frac{4a^2}{r^5}\Zhat\Zhat\psi+\frac{-2\TT_{-a}}{r^6}\OO\psi,
\]
which  can be controlled, in view of  the positivity   $\UU^{\a\b\aund\bund} \, \D_\a \psia \c \D_\b \psib$.
Remarkably, the indefinite expression $\Psi$ provides important information even in the trapping set\footnote{Typically one uses only the non-negativity of the trapping term, while here we make full use of its structure.}.
 Accordingly we can rewrite $\VV^2T\Zhat\psi$ in terms of $\Psi$ and the more manageable terms $ \Zhat\Zhat\psi, \OO\psi$, see Proposition \ref{Prop:refinedP}.

  With these observations, multiple integration by parts and a lot of straightforward algebra, we can reduce to a situation similar to that of the scalar case discussed above around \eqref{expression-Div-M-I-intro}.
   More precisely,
\[
I=\AA|\pr_r\widetilde{\psi}|^2,\quad J=\VV|\widetilde{\psi}|^2,\quad \widetilde{\psi}=m\T\Zhat\psi, \ \Zhat\Zhat\psi, \ m\T|q|\nab\psi, \ \Zhat|q|\nab\psi
\]
where $\AA\geq 0$ and $\VV$ is only negative in a compact region. To handle the negativity of $\VV$, we adopt the \textit{Hardy type} procedure discussed in the scalar case. That is, we rely on a modified $\PP=\PP[X, w, M]$, depending on an additional $1$-form $M=v^{\aund\bund}\pr_r$, as in \eqref{eq:generalizedcurrent}, to rebalance the integrals of $I=\AA |\pr_r\widetilde{\psi}|^2$ and $ J= \VV |\widetilde{\psi}|^2$. Finally, to recover control of the component $TT\psi$, we need a Lagrangian correction to \eqref{eq:EE[X, w]-intro}.

The detailed proof of Theorem \ref{Thm:Moraw4}, given in sections \ref{section-proof-mor-2} and \ref{sec:proofofSSMor}, is technically more involved. This is in fact the only place in the paper we
rely on Mathematica to deal with the algebraic complexity of the calculations involved.


     \subsection{Acknowledgements}

     Both authors are currently funded by the NSF grant DMS-2453843. S. Klainerman
     was previously funded by the NSF grant DMS-220103. L. He was supported by the Shiing-Shen Chern Foundation for Mathematical Research Fund during her stay at IAS in the 2024-25 academic year. Both authors are grateful for the hospitality of IHES in France
     and Tsinghua University in Beijing. Finally, we thank S.~Ma for signaling
     to us the much shorter, unrestricted proof of Theorem
     \ref{Thm:Morawetz4}, which we present in section
     \ref{subsec:new-unrestricted-energy}.
\subsection{Note added in the proof }

As mentioned earlier   the limitation  to $|a|/m\le 0.75$  is only due to  the treatment of the lower order term. Recently we identified a new method for removing this restriction, taking advantage of the flexibility
  of adapting the choice of  multipliers in the  control of $B$ to the  size of spherical modes of the solution.  For high  spherical modes we can take advantage of    the Poincar\'e inequality   to vastly simplify  and improve   our   Hardy procedure, while for bounded spherical modes
    we make a better choice of the Morawetz multiplier.  


    \section{Preliminaries}\lab{section:preliminaries}


\subsection{Basic facts about the Kerr metric}
\lab{section:Kerr-metric}
The Kerr metric $\g=\g_{a,m}$ in Boyer-Lindquist coordinates
is given by
\beq
\lab{eq:coordsBL}
\g =    -\frac{\left(\Delta-a^2\sin^2\theta\right)}{|q|^2}dt^2-
\frac{ 4 amr }{|q|^2}   \sin^2\theta     dt  d\phi+\frac{|q|^2}{\Delta}dr^2+ |q|^2 d\theta^2+
\frac{ \Si^2}{|q|^2}\sin^2\theta
d\phi^2
\eeq
where $q=r+ i a \cos\th,$
and
\[
\begin{cases}
    \Delta &= r^2-2mr+a^2,\\
    |q|^2 &= r^2+a^2(\cos\theta)^2,\\
    \Sigma^2 &= (r^2+a^2)|q|^2+2mra^2(\sin\theta)^2=(r^2+a^2)^2-a^2(\sin\theta)^2\Delta.
\end{cases}
\]
Note that $\g_{tt} \g_{\phi \phi}- \g_{t\phi}^2 =-\De \sin^2 \th$ and that
the non-vanishing components of the inverse metric are given by
\beq
\lab{eq:inversemetric-Kerr}
\bsplit
\g^{00}&=-\frac{\Si^2}{|q|^2 \De}, \qquad
\g^{0\phi}=-\frac{2 a mr }{|q|^2\De},\qquad
\g^{\phi\phi} =\frac{\De- a^2 \sin^2\th}{|q|^2 \De \sin^2 \th}, \\
\g^{rr}&=\frac{\De}{|q|^2}, \qquad \g^{\th\th}=\frac{1}{|q|^2}.
\end{split}
\eeq
The volume element $d\mu$ of $\g$ is given by
\[
d\mu=|q|^2 \sin\th dt dr d\th d\phi, \qquad \sqrt{|g|} =|q|^2 \sin\th.
\]
The roots of the quadratic polynomial $\De $ are $r_{-}< r_+$ with $r=r_+=m+\sqrt{m^2-a^2}$
denoting the event horizon $\HH$.

The Killing vectorfield $T=\pr_t$ is causal in the region for which $ \frac{\left(\Delta-a^2\sin^2\theta\right)}{|q|^2}\ge 0$
and spacelike in its complement called ergoregion. Following Chapter 3 in \cite{GKS}, we also introduce
the vectorfields
\beq
\lab{define:That-Rhat}
\That=\T+\frac{a}{r^2+a^2} \Z, \qquad \Rhat=\frac{\De}{r^2+a^2}\pr_r.
\eeq
Note that $\Rhat$ is regular at the horizon, as opposed to $\partial_r$. Unlike $\T$, which is spacelike in the ergoregion,
the vectorfield $\That$ is null along the horizon and time-like in the domain of outer communication. Along the horizon we have, with $\om_{+}=\frac{a}{r_+^2+a^2}$,
\beq
\lab{def:T_{om_+}}
\That|_{\HH} = T_+:=T +\om_{+} Z
\eeq

Besides the Killing vectorfields $T=\pr_t$ and $Z=\pr_\phi$ the Kerr metric
also admits a higher order symmetry connected to the Carter tensor. This can be best seen from the form of the inverse metric.

\begin{lemma}
\lab{lemma:inversemetricexpressioninKerr}
The inverse Kerr metric can be written in the form
\beq\lab{inverse-metric}
|q|^2 \g^{\a\b}=\De \pr_r^\a \pr_r^\b +\frac{1}{\De}\GG^{\a\b}
\eeq
with
\beq
\lab{eq:expressionRR-O}
\bsplit
\GG^{\a\b}&= -(r^2+a^2) ^2\pr_t^\a\pr_t^\b- 2a(r^2+a^2) \pr_t^{(\a} \pr_\phi^{\b)}- a^2 \pr_\phi^\a \pr_\phi^\b  +\De O^{\a\b},\\
&= -(r^2+a^2) ^2\That^\a \That^\b  +\De O^{\a\b}
\end{split}
\eeq
where
\beq
\lab{equation:def-O}
O^{\a\b}= \pr_\th^\a \pr_\th^\b +\frac{1}{\sin^2\th} \pr_\phi^\a \pr_\phi^\b +  2 a \pr_t^{(\a} \pr_\phi^{\b)}+ a^2\sin^2 \th \pr_t^\a \pr_t^\b.
\eeq
Thus, also
\beq
\lab{inverse-metric-vfs}
|q|^2 \g^{\a\b}=\frac{(r^2+a^2)^2}{\De} \big( -\That^\a\That^\b +      \Rhat^\a \Rhat^\b\big) + O^{\a\b}.
\eeq
\end{lemma}

\begin{proof}
Follows immediately from \eqref{eq:inversemetric-Kerr} and the definitions of $\That, \Rhat$ and $O$.
\end{proof}

\begin{proposition}
\lab{prop:decompose-square}
The following statements hold true.
\begin{enumerate}
\item
In BL coordinates the scalar wave operator in $\KK(a, m)$ takes the form
\beq
\lab{eq:decompose-square}
|q|^2 \square =\RR+\OO
\eeq
where $\OO, \RR$ are the second order operators
\beq
\lab{eq:RR-OO}
\bsplit
\RR&=
\pr_r\big(\De \pr_r \big) - \frac{1}{\De}\Big((r^2+a^2)^2 \pr_t^2 +2  a(r^2+a^2)  \pr_t\pr_\phi+a^2 \pr_\phi^2\Big),\\
\OO&=  a^2 \sin^2 \th \pr_t^2+ 2 a \pr_t\pr _\phi +\frac{1}{\sin^2\th} \pr_\phi^2 + \frac{1}{\sin\th}  \pr_\th \big(\sin\th \pr_\th \big)\\
&=\frac{1}{\sin \th}\pr_\a(\sin \th O^{\a\b}\pr_\b).
\end{split}
\eeq
\item $\OO$ commutes with $\RR$ and therefore also with $|q|^2 \square$.
\item Denoting by $\D$ the covariant derivative operator in $\KK(a,m)$,
\beq
\lab{eq:OO-tensorial}
\OO\psi= |q|^2 \D_\a\big( |q|^{-2} O^{\a\b} \D_\b \psi\big).
\eeq
\end{enumerate}
\end{proposition}
\begin{proof}
The proof is a straightforward verification. We give details in appendix \ref{appendix:decompose-square}.
\end{proof}

The importance of the decomposition \eqref{inverse-metric} is due to the fact that $\GG^{\a\b}$ can be written in terms of the tensors $\partial_t^\a \partial_t^\b$, $\partial_t^{(\a} \partial_\phi^{\b)}$, $\partial_\phi^\a \partial_\phi^\b$ and $O^{\a\b}$.

\begin{definition}\lab{definition:tensors-S}
We define the following symmetric spacetime 2-tensors
\beq
\lab{def:S_aund}
\bsplit
S_1^{\a\b}:= \T^\a \T^\b,\quad
S_2^{\a\b}:=  \T^{(\a} \Z^{\b)}, \quad
S_3^{\a\b}:= \Z^\a \Z^\b,\quad
S_4^{\a\b}:=O^{\a\b}.
\end{split}
\eeq
We denote the set of the above tensors as $S_\aund$, for $\aund=1,2,3,4$. We also define the associated second order operators,
\beq
\lab{def:SS_aund}
\SS_\aund\psi:=|q|^2 \D_\mu \big(|q|^{-2}  S^{\mu\nu}_\aund \D_\nu\psi\big)
\eeq
Clearly $\big[|q|^2\square_{a,m}, \SS_\aund\big]=0$ and $ [\SS_\aund, \SS_\bund]=0$.
Moreover, for $\aund=1,2,3$, we have $\SS_\aund\psi= S^{\mu\nu}_\aund\pr_\mu \pr_\nu \psi. $
\end{definition}

With the above definition, from \eqref{eq:expressionRR-O} we write
\[
\GG^{\a\b}= -(r^2+a^2) ^2S_1^{\a\b}- 2a(r^2+a^2)S_2^{\a\b}- a^2S_3^{\a\b}  +\De S_4^{\a\b}.
\]
More compactly, using the repetition in $\aund$ to signify summation over $\aund=1,2,3,4$,
\[
\GG^{\a\b} =\GG^\aund  \Sa^{\a\b},
\]
with $\GG^\aund$, $\aund=1,2,3,4$, given by
\beq\lab{components-RR-aund}
\GG^1=-(r^2+a^2)^2, \quad \GG^2 = -2a(r^2+a^2), \quad \GG^3 =-a^2, \quad \GG^4=\De.
\eeq
Thus, in view of \eqref{eq:expressionRR-O}
\beq
\lab{eq:RR-Sa}
\GG^{\a\b}=\GG^\aund  \Sa^{\a\b}= -(r^2+a^2) ^2\That^\a \That^\b  +\De O^{\a\b}.
\eeq
\begin{remark}
Note that $\SS_1, \SS_2, \SS_3$ commute with $|q|^2$.
\end{remark}

\subsection{Ingoing PN frame}
\lab{section:ingoingPNframe}


The ingoing PN frame (with $\D_3 e_3=0$), regular towards the future for all $r>0$, is given by (see section 3.3 in \cite{GKS})
\beq
\lab{eq:null-pair-in}
e_4=\frac{r^2+a^2}{|q|^2} \pr_t +\frac{\De}{|q|^2} \pr_r +\frac{a}{|q|^2} \pr_\phi, \qquad
e_3   =\frac{r^2+a^2}{\De} \pr_t -\pr_r +\frac{a}{\De} \pr_\phi.
\eeq
We complete the PN frame with the following specific choice
of horizontal frames $e_1, e_2$,
\beq
\lab{eq:canonicalHorizBasisKerr}
e_1=\frac{1}{|q|}\pr_\th,\qquad e_2=\frac{a\sin\th}{|q|}\pr_t+\frac{1}{|q|\sin\th}\pr_\phi.
\eeq
We denote by $\nab$ the induced horizontal derivative operator.

Note that the operators $\That,\Rhat$ introduced in \eqref{define:That-Rhat} verify
\bea
\lab{eq:formulasThat-Rhat}
\That =\frac 1 2 \left( \frac{|q|^2}{r^2+a^2} e_4+\frac{\De}{r^2+a^2}  e_3\right), \quad
\Rhat = \frac 1 2 \left( \frac{|q|^2}{r^2+a^2} e_4-\frac{\De}{r^2+a^2}  e_3\right).
\eea

The Killing vectorfields $\T, \Z$ take the form (see section 4.3 in \cite{GKS})
\beq
\lab{eq:TZ-frame}
\bsplit
\T &:= \frac{1}{2}\left(e_4+\frac{\Delta}{|q|^2}e_3 -2a\Re(\Jk)^be_b\right)\\
\Z &:= \frac 1 2 \left(2(r^2+a^2)\Re(\Jk)^be_b -a(\sin\th)^2 e_4 -\frac{a(\sin\th)^2\De}{ |q|^2} e_3\right)\\
&= (r^2+a^2)\Re(\Jk)^be_b-\frac 1 2  a(\sin\th)^2 \big( e_4+\frac{\De}{|q|^2} e_3\big).
\end{split}
\eeq
where $\Jk$ is the complex $1$ form (see section 3.4.2 in \cite{GKS}) whose horizontal components in $e_1, e_2$ are given by
$\Jk_1=\frac{i\sin\th}{|q|}, \,  \Jk_2=\frac{\sin\th}{|q|}$. Also
\beaa
Z&=& -a\sin^2\th T+|q|\sin\th e_2.
\eeaa
and
\beq
\lab{eq:O-frame}
O^{\a\b}= |q|^2 \big( e_1^\a e_1^\b+ e_2^\a e_2^\b\big).
\eeq
As a corollary we derive
\begin{lemma}
The following identity holds true
\beq
\lab{eq:identityZhat}
\Zhat:=Z+aT=a\cos^2\th \T+ |q|\sin\th  e_2.
\eeq
In particular $\Zhat=Z+aT=|q|e_2=re_2$ when $\sin\th=1$
\end{lemma}
\begin{proof}
Straightforward verification.
\end{proof}
\subsection{Regular  time function  $\tau$ and adapted coordinates}

\begin{definition}
Given a function $f$ we define its spacetime gradient $\gradf= \g^{\mu\nu} \pr_\mu f \pr_\nu$.
In particular,
\beq
\lab{eq:gtadt-gradr}
\bsplit
\gradt&=\g^{tt}  T+\g^{t\phi} Z=-\frac{\Si^2}{|q|^2 \De} \big( T+\frac{2amr}{\Si^2} Z\big) \\
\gradr&= g^{34} e_3(r)  e_4+ g^{43} e_4(r)  e_3=\frac 1 2 \big( e_4 -\frac{\De}{|q|^2} e_3\big)
\end{split}
\eeq
\end{definition}

\begin{lemma}
\lab{lemma:TtRtandZ}
The following statements hold true.
\begin{enumerate}
\item
The vectorfields $\big\{\Tt:=-\frac{|q|^2\De}{\Si^2}\gradt, \Rt =-\gradr, e_1, Z\big\}$ form an orthogonal frame in $\DD$ with
\[
\g(\Tt,\Tt)=-\frac{\Delta|q|^2}{\Sigma^2}, \quad   \g(\Rt, \Rt)= \frac{\De}{|q|^2}, \quad \g(\Z, \Z)=\frac{\Si^2 \sin^2\th}{|q|^2}.
\]
\item
Relative to this frame the Lagrangian $\LL[\psi]=\g^{\a\b}\pr_\a\psi\pr_\b\psi$ of a scalar function $\psi$ takes the form
\beq
\lab{eq:Lagr-psi}
\bsplit
\g^{\a\b}\pr_\a\psi\pr_\b\psi&=\frac{1}{\g(\widetilde{T}, \widetilde{T})}|\widetilde{T}\psi|^2+\frac{1}{\g(\widetilde{R}, \widetilde{R})}|\widetilde{R}\psi|^2+|e_1\psi|^2+\frac{1}{\g(Z, Z)}|Z\psi|^2\\
&=-\frac{\Sigma^2}{\Delta|q|^2}|\widetilde{\T}\psi|^2+\frac{|q|^2}{\Delta}|\widetilde{R}\psi|^2+|e_1\psi|^2+\frac{|q|^2}{\Sigma^2\sin^2\th}|Z\psi|^2.
\end{split}
\eeq
\end{enumerate}
\end{lemma}
\begin{proof}
Note that  $\g\big(\, ^{(f_1)} \mbox{grad}, \, ^{(f_1)} \mbox{grad}\big)=\g^{\mu\nu}\pr_\mu f_1\pr_\nu f_1$
and
$ \g(\gradt, \Z)= Z( t) =0.
$
The identity \eqref{eq:Lagr-psi} is an immediate consequence of the orthogonality statement.
\end{proof}
\subsubsection{Time function $\tau$}
\lab{subs:taufunction}


Given the degenerate behavior of the function $t$ near the horizon we introduce the function $\tau$ as follows
\beq
\lab{def:tau}
\tau=t+ f(r), \qquad f'(r)= \big(\frac{r^2+a^2}{\Delta}-\frac{m^2}{r^2} \big) \chi
\eeq
with $\chi\ge 0$ a smooth cut-off function of $r$ defined as follows:
\beq
\lab{def-cutoff}
\chi=\begin{cases}  1, \qquad r \le r_0 \\
0, \qquad  r\ge  2r_0
\end{cases}
\eeq
for any fixed $r_0 \gg r_+$.

\begin{lemma}
\lab{lemma:N_Si}
The function $\tau$ as defined in \eqref{def:tau} and \eqref{def-cutoff} is timelike, coincides with $t$ for $r\ge 2r_0$ and verifies the following properties.
\begin{enumerate}
\item The spacetime gradient of $\tau$ given by $\gradtau$ is given by the formula
\beq
\lab{eq:N_Si}
\gradtau=   - \frac{\Si^2}{|q|^2 \De}\Tt -    f'(r) \Rt.
\eeq
\item We have
\beq
\lab{eq:g(N_Si,N_Si)}
\g(\gradtau,\gradtau)=-\frac{\Si^2}{|q|^2 \De}+\frac{\De}{|q|^2} f'(r) f'(r) \sim -1.
\eeq
\item Also\footnote{Here $O_+(1)$ denotes a $O(1)$ positive quantity.}
\beq
\lab{eq:N_Si-frame}
\gradtau=\begin{cases} -\frac 12\big( \frac{m^2}{r^2}  e_4+ O_+(1)  e_3\big)+  \frac{a\sin\th}{|q|} e_2, \quad  r\le r_0\\
- O_+(1)\big(e_3+e_4) +O(ar^{-1})\nab,  \quad  r\ge 2r_0
\end{cases}.
\eeq

\end{enumerate}

\end{lemma}
\begin{proof}
See appendix \ref{Appendix:Lemma:N_Si}.
\end{proof}


\subsubsection{Adapted coordinates}
\lab{section:adapted-taucoordinates}

We perform the
change of coordinates $(t, r, \th, \phi) \to (\tau, \tilde{r}, \tilde{\th}, \tilde{\phi})$
where
\[
\tau = t+f(r), \quad  \tilde r = r, \quad \tilde \th=\th, \quad \tilde \phi=\phi+g(r),\quad g'(r)=\frac{a}{\Delta}\chi(r).
\]
The inverse metric $\gt$ in the new coordinates is given by $\gt^{\a\b}=\g^{\mu\nu}\frac{\pr \tilde{x}^\a}{\pr x^\mu} \frac{\pr \tilde{x}^\b}{\pr x^\nu}$.
Therefore
\begin{align*}
\gt^{\tau\tau} &=\g^{tt} +\g^{rr}   (f'(r))^2=-\frac{1}{|q|^2}\Big(\frac{(r^2+a^2)^2}{\Delta}-a^2\sin^2\th-\Delta(f'(r))^2\Big)\sim -1 ,\\
 \gt^{\tau \tilde  r}&=\g^{rr} \pr_r \tau \pr_ r r=\g^{rr} f'(r)=\frac{\De}{|q|^2} f'(r)= \frac{\De}{|q|^2}\big(\frac{r^2+a^2}{\De} -\frac{m^2}{r^2} \big)\chi,\\
\gt^{\tau\tilde\phi}&=\g^{t\phi}+\g^{rr}f'(r)g'(r)=-\frac{1}{|q|^2}\Big(\frac{2amr}{\Delta}-\Delta f'(r)g'(r)\Big)=\begin{cases}
\frac{a}{|q|^2}(1-\frac{m^2}{r^2})    ,\quad &r\leq r_0\\
-\frac{2amr}{|q|^2\De},\quad &r\geq 2r_0
\end{cases} ,\\
\gt^{\tilde\phi\tilde\phi}&=\g^{\phi\phi}+\g^{rr}(g'(r))^2=\frac{1}{|q|^2}\Big(\frac{\Delta-a^2\sin^2\th}{\Delta\sin^2\th}+\Delta(g'(r))^2\Big)=\begin{cases}
\frac{1}{|q|^2\sin^2\th},\quad &r\leq r_0\\
\frac{\Delta-a^2\sin^2\th}{|q|^2\Delta\sin^2\th},\quad &r\geq 2r_0
\end{cases} ,\\
\gt^{\tilde r\tilde\phi}&=\g^{rr}\pr_r\tilde\phi\pr_r\tilde r=\g^{rr}g'(r)=\frac{\De}{|q|^2}\frac{a}{\De}\chi=\frac{a}{|q|^2}\chi.
\end{align*}
All other components remain unchanged, that is,
\[
\gt^{\tilde r\tilde r}=\g^{rr}, \quad \gt^{\tilde\th\tilde\th}=\g^{\th\th}.
\]
Note also that $\pr_\tau=\pr_t=T, \quad \pr_{\tilde r}= \pr_r - f'(r) \pr_t-g'(r)\pr_\phi, \quad \pr_{\tilde \th}=\pr_\th,\quad \pr_{\phit}=\pr_\phi=Z$.
The wave operator $\square=\square_{a,m}$ takes the form
\beq
\lab{eq:square-tau}
\begin{split}
\square &=\Big( \gt^{\tau\tau}  \pr^2_\tau   + 2 \gt^{\tau\tilde\phi}\pr_\tau\pr_\phi +  \gt^{\tau \tilde r} \pr_\tau \pr_{\tilde r} +  \gt^{ \tilde r\tilde\phi} \pr_\phi \pr_{\tilde r}+\gt^{\tilde\phi\tilde\phi }\pr_\phi^2 \Big)  \psi +\frac{1}{|q|^2 } \pr_{\tilde r}    \big(  |q|^2  \gt^{\tilde r\tau } \pr_\tau  \psi\big)\\
&+ \frac{1}{|q|^2 } \pr_{\tilde r}    \big(  |q|^2  \gt^{\tilde r\tilde\phi } \pr_\phi  \psi\big)+  \frac{1}{|q|^2 } \Big(\pr_{\tilde r}\big(\De \pr_{\tilde r} \psi\big)+ \frac{1}{\sin \th} \pr_\th \big(\sin\th \pr_\th \psi\big)\Big).
\end{split}
\eeq
Using \eqref{eq:square-tau}, we derive the following useful lemma.
\begin{lemma}
\lab{lemma:square-taucoords}
 In the new adapted coordinates the scalar wave operator $\square=\square_{a,m}$ in $\KK(a, m)$ takes the form
\beq
\lab{eq:square-taucoords}
\bsplit
&\square \psi=\frac{1}{|q|^2} \Big(\pr_{\rt}\big(\De \pr_{\rt}  \psi\big) +2 (r^2+a^2)\That\pr_{\rt}\psi+\frac{1}{ \sin \th} \pr_\th \big(  \sin \th \pr_\th \psi\big)
+\frac{1}{\sin^2\th}  \pr_{\phi}^2\psi  \Big)
+\err_1[\psi](\T, \Z),\\
&\err_1[\psi](\T, \Z)\les  \Big( \big| (T, r^{-2}\De\pr_{\rt}, r^{-1}Z)(\T, r^{-1} \Z)   \psi \big|
+ r^{-1}\big| (\T, r^{-1} \Z)\psi \big| \Big).
\end{split}
\eeq
Alternatively, we write
\beq
\lab{eq:square-taucoords2}
\bsplit
&\frac{\De}{r^2+a^2}\square \psi=\frac{1}{|q|^2(r^2+a^2)} \Big((r^2+a^2)^2\Rhat^2\psi+ \frac{\De}{ \sin \th} \pr_\th \big(  \sin \th \pr_\th \psi\big)
+\frac{\De}{\sin^2\th}  \pr_{\phi}^2\psi  \Big)
+\err_2[\psi],\\
&\err_2[\psi]\les  \Big( \big| (T, r^{-1}Z)^2 \psi \big|
+ r^{-1}\big| \Rhat\psi \big| \Big).
\end{split}
\eeq
\end{lemma}
\begin{proof}
The proof of \eqref{eq:square-taucoords} follows from \eqref{eq:square-tau} and the behavior of the metric components $\gt^{\a\b}$
\begin{align*}
\gt^{\tau\tau}\sim -1,\quad \gt^{\tau\phit}=O(r^{-3}),\quad  \gt^{\phit\phit}=\frac{1}{    |q|^2\sin^2\th}+O(r^{-4}),
\end{align*}
and $\gt^{\tau\tilde r}, \gt^{\tilde r\phit}$ are compactly supported. The proof of \eqref{eq:square-taucoords2} follows from \eqref{eq:decompose-square}, \eqref{eq:RR-OO} and \eqref{define:That-Rhat}.
\end{proof}
\subsection{Main Integral quantities}

\lab{section:BEFquantities}
\subsubsection{Domains of interest }
\begin{definition}
\lab{def:DDtau}
Let $\de_{trap}$ denote a sufficiently small positive number.
We define the following regions in $\KK(a,m)$:
\begin{itemize}
\item $\DD$ is the region in $\KK(a,m)$ for which $r\ge r_+$
with null boundary
\[
\HH=\big\{ r=r_+\big\}.
\]
\item $\DD_{trap}$ denotes the trapping region
\[
\DD_{trap}=\big\{ \rhat_1\le r\le \rhat_2\big\}.
\]
with $\rhat_1, \rhat_2$ the two roots in $(r_+,\infty)$ of the polynomial
$\Theta(r):=-Q_r(1)=(\TT+a^2(r-m))^2-4a^2r^2\Delta$, equivalently the two
exterior roots of $\Ththree$, as defined in \eqref{eq:PolynTheta}, whose
values were explicitly given in Lemma \ref{Lemma:rangeofrfortrappednullgeodesics}.
\item $\DD_{\ntrap}$ is the complement of the trapped region $\DDtrap$. More precisely, for\footnote{This constant appears
in the estimates of Theorem \ref{Thm:Moraw3}.} $\de_{trap}>0$,
\[
\DD_{\ntrap}= \big\{ r< \rhat_1-\de_{trap}\big\} \cup \big\{  r> \rhat_2+\de_{trap}\big\}.
\]
\end{itemize}
Given $\tau$ a time function\footnote{Whose level surfaces are strictly spacelike everywhere.} in $\DD$ we define the causal domains
$\DD(\tau_1, \tau_2) =\DD\cap \big\{ \tau_1\le \tau \le \tau_2\big\}$ with boundaries $\pr\DD(\tau_1,\tau_2)=\HH(\tau_1,\tau_2) \cup \Si(\tau_1) \cup \Si(\tau_2)$
where $\Si(\tau_1), \Si(\tau_2)$ are level surfaces of $\tau$.
\end{definition}

\begin{remark}

A specific time function $\tau$ was defined in section \ref{subs:taufunction}.
\end{remark}

\subsubsection{Main degenerate quantities}
We introduce below the main integral quantities used in our estimates.
\begin{definition}[Degenerate quantities]
Define
\beq
\lab{def:Bdegg}
\Bdegg[\psi](\tau_1, \tau_2):=\int_{\DD(\tau_1, \tau_2) }  r^{-2}  | \Rhat  \psi|^2  +r^{-4} |\psi|^2
\eeq
and
\beq
\lab{eq:Bdegnorm}
\bsplit
\Bdeg[\psi](\tau_1, \tau_2)&:=\Bdegg[\psi](\tau_1, \tau_2) + \int_{\DD_{\ntrap}(\tau_1, \tau_2)} \left(  r^{-2} |\That \psi|^2        + r^{-1}  |\nab  \psi|^2\right)\\
\Fdeg[\psi](\tau_1,\tau_2)&:=\int_{\HH(\tau_1,\tau_2)} |e_4\psi|^2  \\
\Edeg[\psi](\tau)&:= \int_{\Si(\tau)} | e_4\psi|^2+\frac{|\De|}{r^2}   |e_3 \psi|^2 + |\nab\psi|^2+r^{-2}|\psi| ^2 .
\end{split}
\eeq
Also
\bea
P[\psi](\tau_1, \tau_2)=\int_{\DD_{trap}(\tau_1, \tau_2) }\left|\Tfrak[\psi]\right|^2, \qquad \Tfrak[\psi]:=\big(\TT_{-a}T-a(r-m)\Zhat\big)\psi.
\eea

\end{definition}

\subsubsection{Weighted  version of the main  quantities}
\begin{definition}[Weighted  quantities]
For a given $r_0$ sufficiently large we define
\beq
\lab{eq:Bnorm-p-deg}
\bsplit
\Bdeg_p[\psi](\tau_1, \tau_2)&:=\Bdeg[\psi](\tau_1, \tau_2)+\int_{\DD(\tau_1, \tau_2)\cap\{r\ge r_0\}} r^{p-3}\big(|\dk\psi|^2 + |\psi|^2\big)
\\
\Edeg_p [\psi](\tau)&:=    \Edeg[\psi](\tau)+       \int_{\Si(\tau)\cap\{ r\ge r_0\}}   r^p   \big(  r^{-1} |\nab_4 (r\psi)|^2+ |\nab\psi|^2 +     r^{-2}|\psi|^2 \big)\\
\Fdeg_p [\psi](\tau)&:= \Fdeg[\psi].
\end{split}
\eeq
where, as in \cite{GKS}, $\dk\psi =\big( r(e_4, \nab)\psi, e_3\psi\big)$.
\end{definition}

\subsubsection{Main non-degenerate quantities}
\lab{section:main-nondeg}

\begin{definition}[Non-degenerate   quantities]\lab{def:nondegeneratequantities}
Define\footnote{Note that $\Rhat$ degenerates at the horizon but the contribution from the non-trapped integral is non-degenerate.}
\beq
\lab{eq:Bnorm}
\bsplit
B[\psi](\tau_1, \tau_2)&:=\int_{\DD(\tau_1, \tau_2) }  r^{-2}  | \Rhat  \psi|^2  +r^{-4} |\psi|^2+ \int_{\DD_{\ntrap}(\tau_1, \tau_2)} \left(  r^{-2} |e_{3, 4}\psi|^2        + r^{-1}  |\nab  \psi|^2\right)\\
F[\psi](\tau_1,\tau_2)&:=\int_{\HH(\tau_1,\tau_2)} |e_4\psi|^2
+|\nab\psi|^2+|\psi|^2\\
E[\psi](\tau)&:= \int_{\Si(\tau)} | e_4\psi|^2+ |e_3 \psi|^2 + |\nab\psi|^2+r^{-2}|\psi|^2.
\end{split}
\eeq
and their $r^p $ weighted version
\beq
\lab{eq:Bnorm-p}
\bsplit
B_p[\psi](\tau_1, \tau_2)&:=B[\psi](\tau_1, \tau_2)+\int_{\DD(\tau_1, \tau_2)\cap\{r\ge r_0\}} r^{p-3}\big(|\dk\psi|^2 + |\psi|^2\big)
\\
E_p [\psi](\tau)&:=    E[\psi](\tau)+       \int_{\Si(\tau)\cap\{ r\ge r_0\}}   r^p   \big(  r^{-2} |\nab_4 (r\psi)|^2+ |\nab\psi|^2 +     r^{-2}|\psi|^2 \big)\\
F_p [\psi](\tau)&:= F[\psi].
\end{split}
\eeq
\end{definition}

\subsubsection{Combined quantities}
We define the combined quantities
\beq
\bsplit
\EFdeg[\psi](\tau_1, \tau_2)&=\Fdeg[\psi](\tau_1, \tau_2) + \Edeg[\psi](\tau_1)+\Edeg[\psi](\tau_2)\\
\BEFdeg[\psi](\tau_1, \tau_2)&=  \Bdeg[\psi](\tau_1, \tau_2)+ \EFdeg[\psi](\tau_1, \tau_2)\\
\BPdeg[\psi](\tau_1, \tau_2)&=  \Bdeg[\psi](\tau_1, \tau_2)+ P[\psi](\tau_1, \tau_2), \\
\BPFdeg[\psi](\tau_1, \tau_2)&=\BPdeg[\psi](\tau_1, \tau_2)+  F[\psi](\tau_1, \tau_2),\\
 \BPEFdeg[\psi](\tau_1, \tau_2)&=\BPdeg[\psi](\tau_1, \tau_2)+  \EFdeg[\psi](\tau_1, \tau_2).
\end{split}
\eeq

The corresponding non-degenerate quantities are defined in the same manner.

For completeness, we also define the $r^p$-weighted combined quantities by
\beq
\bsplit
\EF_p[\psi](\tau_1,\tau_2)
&=F_p[\psi](\tau_1,\tau_2)+E_p[\psi](\tau_1)+E_p[\psi](\tau_2),\\
\BEF_p[\psi](\tau_1,\tau_2)
&=B_p[\psi](\tau_1,\tau_2)+\EF_p[\psi](\tau_1,\tau_2),\\
\BP_p[\psi](\tau_1,\tau_2)
&=B_p[\psi](\tau_1,\tau_2)+P[\psi](\tau_1,\tau_2),\\
\BPF_p[\psi](\tau_1, \tau_2)&=\BP_p[\psi](\tau_1, \tau_2)+  F_p[\psi](\tau_1, \tau_2),\\
\BPEF_p[\psi](\tau_1,\tau_2)
&=\BP_p[\psi](\tau_1,\tau_2)+\EF_p[\psi](\tau_1,\tau_2).
\end{split}
\eeq
The degenerate quantities $\EFdeg_p$, $\BEFdeg_p$, $\BPdeg_p$, and
$\BPEFdeg_p$ are defined by the same formulas with $B_p,E_p,F_p$ replaced
by $\Bdeg_p,\Edeg_p,\Fdeg_p$.  Notice that $P$ is not given an additional
$r^p$ weight: its defining integral is supported in the fixed compact region
$\DD_{trap}$.

\subsubsection{Higher order quantities}

\begin{definition}
We define the higher derivative quantities
\beq
Q^s[\psi]=\sum_{k\le s} Q[\dk^k\psi],\qquad
Q_p^s[\psi]=\sum_{k\le s} Q_p[\dk^k\psi].
\eeq
where, as in \cite{GKS}, $\dk\psi =\big( r(e_4, \nab)\psi, e_3\psi\big)$,
and $Q$ is any of the quantities $B,E,F$ and $\Bdeg, \Edeg, \Fdeg$;
$Q_p$ denotes the corresponding $r^p$-weighted quantity.  In addition, we
define
\beq
P^s[\psi](\tau_1,\tau_2)
:=\sum_{k\le s}P[\dk^k\psi](\tau_1,\tau_2).
\eeq
Thus
\beq
\bsplit
\EF_p^s[\psi](\tau_1,\tau_2)
&=F_p^s[\psi](\tau_1,\tau_2)+E_p^s[\psi](\tau_1)+E_p^s[\psi](\tau_2),\\
\BEF_p^s[\psi](\tau_1,\tau_2)
&=B_p^s[\psi](\tau_1,\tau_2)+\EF_p^s[\psi](\tau_1,\tau_2),\\
\BP_p^s[\psi](\tau_1,\tau_2)
&=B_p^s[\psi](\tau_1,\tau_2)+P^s[\psi](\tau_1,\tau_2),\\
\BPF_p^s[\psi](\tau_1,\tau_2)&=\BP_p^s[\psi](\tau_1,\tau_2)+F_p^s[\psi](\tau_1,\tau_2),\\
\BPEF_p^s[\psi](\tau_1,\tau_2)
&=\BP_p^s[\psi](\tau_1,\tau_2)+\EF_p^s[\psi](\tau_1,\tau_2).
\end{split}
\eeq
The unweighted versions are obtained by dropping the subscript $p$, and the
degenerate versions are defined by replacing $B,E,F$ with
$\Bdeg,\Edeg,\Fdeg$, while retaining the same $P^s$ term. 
\end{definition}


\subsubsection{Restricted quantities}
\lab{section:restrictedQ}
Given any quantity $Q$ as above we denote by $Q_I$ the restriction of the integral defining $Q$ to an
interval $I$ in $r$.

\subsubsection{Quantities  for the inhomogeneous terms}
\lab{section:NN-quantities}
In the case of the inhomogeneous wave equation $\square_{a,m} \psi=N$ we also need the following quantities.
\begin{definition}[Basic $\NN$ quantity] Define
\beq
\lab{eq:NN-quantity}
\bsplit
\NN^s[ N](\tau_1, \tau_2) :=\sum_{k\leq s}\int_{\DD(\tau_1, \tau_2)}|\dk^kN|^2 .
\end{split}
\eeq
and
\beq
\lab{eq:NN-psi-quantity}
\bsplit
\NN^s[\psi,  N](\tau_1, \tau_2) =&\NN^s[N](\tau_1,\tau_2)+\sum_{k\leq s}\left|\int_{\Dtrap{(\tau_1, \tau_2)}} (\T, \Z)\dk^k\psi \c\dk^kN\right|.
\end{split}
\eeq
\end{definition}

\begin{definition}[Weighted $\NN$ norm]
For $0<p<2$ and $r_0$ sufficiently large, we define
\beq
\lab{eq:NN-weighted-source}
\NN^s_p[N](\tau_1,\tau_2)
:=\NN^s[N](\tau_1,\tau_2)
+\sum_{k\leq s}\int_{\DD_{r\geq r_0}(\tau_1,\tau_2)}
r^{p+1}|\dk^kN|^2.
\eeq
The corresponding quantity including the mixed spacetime error is
\beq
\lab{eq:NN-weighted-psi-source}
\NN^s_p[\psi,N](\tau_1,\tau_2)
:=\NN^s_p[N](\tau_1,\tau_2)
+\sum_{k\leq s}\left|
\int_{\Dtrap{(\tau_1,\tau_2)}}
(\T,\Z)\dk^k\psi\c\dk^kN
\right|.
\eeq
\end{definition}

\subsection{Partial  red shift estimate}
\lab{section:PRS}
We rely on a partial red shift estimate, weaker than the standard such estimates in the literature.\footnote{See for example \cite{DR1} or Proposition 9.4.1 in \cite{GKS}. This partial version of the red shift appears first in \cite{MMTT}.}
\begin{proposition}\lab{prop:partialredshiftestimate}
Let $\psi$ {be} a solution to the equation $\square_{a,m}\psi=N$. Then, for $|a|<m$, there exists a small enough constant $0<\de_{red}\ll1-\frac{|a|}{m}$ such that the following estimate holds true.
\begin{align*}
&\int_{\DD_{< r_+(1+\de_{red})}(\tau_1, \tau_2)} |\nab_3 \psi|^2 +   \int_{\Si (\tau_2){< r_+(1+\de_{red})} }   \big(|e_3 \psi|^2  +|\nab \psi|^2 \big)+  \int_{\HH(\tau_1,\tau_2)}|\nab \psi|^2\\
     &\les  \int_{\DD_{r_+(1+2\de_{red})}(\tau_1, \tau_2)}  \frac{\De^2}{ r^4} |\pr_r \psi|^2+ \int_{\DD_{r_+(1+\de_{red} )\le r \le r_+(1+2\de_{red})}(\tau_1, \tau_2)} |\nab\psi|^2+ E_{r_+(1+2\de_{red})}(\tau_1)\\
     &+\int_{\DD_{< r_+(1+2\de_{red})}(\tau_1, \tau_2)}|N|^2.
\end{align*}

\end{proposition}
\begin{proof}
The proof is based on the construction of a vectorfield $Y$ proportional to $e_3$ of the form
\[
Y=d(r)e_3=d(r)(-\pr_r+\frac{r^2+a^2}{\De}T+\frac{a}{\De}Z)
\]
 with $d(r)\geq0$ defined as follows
\beq\lab{eq:choicedr}
d(r)=\chi(\frac{r-r_+}{r_+\de_{red}})\Big(1+\frac{r-r_+}{r_+}\Big)
\eeq
where $\chi(s)=1$ for $s\leq 1$ and $\chi(s)=0$ for $s\geq 2$.
The proof is then an immediate consequence of the following Lemma.
\end{proof}
\begin{lemma}
\lab{lemma:PRS}
Let $Y$ be defined as above. Then if $\de_{red}$ is sufficiently small, we have
\beq
|q|^2\EE[Y,0,0]\geq\frac18(r-m)d(r)|e_3\psi|^2+\chi'\frac{1}{r_+\de_{red}}\big||q|\nab\psi\big|^2-\frac{40d(r)}{r^2(r-m)}|\De\pr_r\psi|^2.
\eeq
Also
\beq
\begin{split}
\QQ(Y, N_{\Si})&\gtrsim d(r)(|e_3 \psi|^2+|\nab \psi|^2 ),\\
\QQ(Y, N_{\HH})&=\QQ(e_3, \frac12e_4)=\frac12|\nab\psi|^2.\end{split}
\eeq
where $\QQ$ is the energy momentum tensor of $\square_{a,m}$ and, with the notation of
 section \ref{sect:Basicscalarident}, $\EE[Y,0,0]=\Div\big(\QQ \c Y)- Y(\psi) \square \psi$.
\end{lemma}
\begin{proof}
We delay the proof of the Lemma to the appendix \ref{appendix:Lemma-PRS}.
\end{proof}

\subsection{Norm equivalence lemmas}
\lab{section:normEqivLemmas}


In this section we make use of Lemma \ref{lemma:square-taucoords} to establish several norm equivalence Lemmas.
\begin{lemma}
\lab{Lemma:calculuslemmadeg}
For the inhomogeneous wave equation $\square_{a,m} \psi=N$, the following estimates hold true.
\begin{enumerate}
\item We have
\beq\lab{eq:seconddegenergy}
\bsplit
\Edeg[(r^{-2}\De e_3, e_4, r\nab)^{\leq 2}\psi](\tau)&\les  \Edeg[\OO\psi](\tau)+\Edeg[(T, Z)^{\leq 2}\psi](\tau)\\
&+\NN[(r^{-2}\De e_3, e_4, r\nab)^{\leq 2}N](\tau_1,\tau_2) .
\end{split}
\eeq
\item We also have
\beq\lab{eq:seconddegbulk}
\bsplit
&\Bdeg[(r^{-2}\De e_3, e_4, r\nab)^{\leq 2}\psi](\tau_1, \tau_2)\\
&\quad\les \Bdeg[\OO\psi](\tau_1,\tau_2)+\Bdeg[(T,Z)^{\leq 2}\psi](\tau_1,\tau_2)\\
&\quad+\Edeg[(r^{-2}\De e_3, e_4, r\nab)^{\leq 2}\psi](\tau_1, \tau_2)+\NN[(r^{-2}\De e_3, e_4, r\nab)^{\leq 1}N](\tau_1,\tau_2).
\end{split}
\eeq
\end{enumerate}
Therefore
\beq\lab{eq:seconddegbulkenergy}
\bsplit
&\BEdeg[(r^{-2}\De e_3, e_4, r\nab)^{\leq 2}\psi](\tau_1,\tau_2)\\
&\quad\les\BEdeg[\OO\psi](\tau_1,\tau_2)+\BEdeg[(T, Z)^{\leq 2}\psi](\tau_1,\tau_2)\\
&\quad+\NN[(r^{-2}\De e_3, e_4, r\nab)^{\leq 2}N](\tau_1,\tau_2).
\end{split}
\eeq
\end{lemma}
\begin{proof}
See appendix \ref{appendix:Lemma-calculuslemmadeg}.
\end{proof}
We then derive the following non-degenerate estimates.
\begin{lemma}
\lab{Lemma:calculuslemma}
For the inhomogeneous wave equation $\square_{a,m} \psi=N$, the following estimate holds true.
\beq
\bsplit
&\BEF[(e_3,e_4,r\nab)^{\leq2}\psi](\tau_1,\tau_2)\\
&\quad\les  \BEF[\OO\psi](\tau_1,\tau_2)+\BEF[(T, Z)^{\leq 2}\psi](\tau_1,\tau_2)\\
&\quad+\NN[(e_3,e_4,r\nab)^{\leq2}N](\tau_1,\tau_2)+E[(e_3,e_4,r\nab)^{\leq2}\psi](\tau_1).
\end{split}
\eeq
\end{lemma}
\begin{proof}
See appendix \ref{appendix:Lemma-calculuslemma}.
\end{proof}
We also have the following norm equivalence lemma.
\begin{lemma}
\lab{lem:higherordercalculuslemma}
For the inhomogeneous wave equation $\square_{a,m} \psi=N$ and $s\geq 1$, the following estimate holds true.
\beq\lab{eq:normequiv}
\bsplit
&\BEF[(e_3,e_4,\nab)^{\leq s+1}\psi](\tau_1,\tau_2)\\
&\quad\les  \BEF[(e_3,e_4, \nab)^{\leq s}(T,Z)^{\leq1}\psi]+E[(e_3, e_4,\nab)^{\leq s+1}\psi](\tau_1)\\
&\quad+\NN[(e_3, e_4, \nab)^{\leq s+1}N](\tau_1,\tau_2).
\end{split}
\eeq
Moreover, the $r^p$ version of the norm equivalence is given as follows.
\beq\lab{eq:rpnormequiv}
\BEF_p^{s+1}[\psi](\tau_1,\tau_2)\les  \BEF^s_p[(T,Z)^{\leq1}\psi]+E_p^{s+1}[\psi](\tau_1)+\NN_p^{s+1}[\psi,N](\tau_1,\tau_2).
\eeq
\end{lemma}
\begin{proof}
See appendix \ref{appendix:higherordercalculuslemma}.
\end{proof}


\subsection{Identities for the energy-momentum  tensor $\QQ$}
\lab{section:identities-QQ}

Let $\QQ$ be the energy-momentum tensor for $\square=\square_{a,m}$. We apply the divergence Lemma
in $\DD(\tau_1, \tau_2) $ with $N_\HH=-\Rt=\frac 12 e_4$ and $N_\Si=c_{\Si} \Nt=-c_\Si \gradtau $ where
$c_\Si=\frac{1}{e_3(\tau)} $ so that $\g(N_\HH, e_3)=\g(N_\Si, e_3)=-1$. In what follows
$\Tplus=T+\frac{a}{r_+^2+a^2}Z$, see \eqref{def:T_{om_+}}.
\begin{lemma}
    \lab{lemma:Energy1}
    The following statements hold true.
    \begin{itemize}
        \item Along the horizon $r=r_{+}$, the normal $N_{\HH}=-\Rt=\frac 1 2 e_4= \frac{r^2+a^2}{|q|^2}T_+$ and
        \beq
        \lab{eq:Festimates}
        \bsplit
        \QQ(T, N_\HH)&=   \frac{r^2+a^2}{|q|^2}T\psi \c  T_+\psi.
        \\
        \QQ(T_+, N_\HH)&=  \frac{r^2+a^2}{|q|^2}|T_+\psi|^2.
    \end{split}
    \eeq
    \item Along $\Si(\tau)$, for $r\ge r_+$, some small $c_0>0$ and large $C$,
    \bea
\lab{eq:EFestimates1}
\QQ(\T, N_\Si)\ge c_0\Big( | e_4\psi|^2+\frac{|\De|}{r^2} |e_3 \psi|^2 + |\nab\psi|^2\Big) - C\frac{a^2m^2 }{r^6 }  |\Z\psi|^2
\eea
\end{itemize}
\end{lemma}
\begin{proof}
See appendix \ref{appendix:identities-QQ}.
\end{proof}





\section{Null geodesics in Kerr}\lab{section:geodesics-Kerr}


In this section, we analyze the behavior of the trapped null geodesics in $\KK(a,m)$.
The characterization of the $r$-range of  the  trapped set  given in Proposition \ref{Prop:r-trapped.region}\,
plays an essential role in the proof of Theorem \ref{Thm:Morawetz1}.


\subsection{The constants of motion for geodesics}


Let $\ga(\la)$ be a null geodesic in Kerr. Using the expression for the inverse of the metric given by \eqref{inverse-metric}, along $\ga(\la)$, since $\g(\gadot, \gadot)=0$ we have, with $\gadot_r=\pr_r^\a \gadot_\a$, $\gadot_t=\pr_t^\a \gadot_\a$, $\gadot_\vphi=\pr_\vphi^\a \gadot_\a$
\beq
\lab{eq:identity-null}
0= |q|^2 \g^{\a\b} \gadot_\a \gadot_\b=\left( \De \pr_r^\a \pr_r^\b +\frac{1}{\De}\GG^{\a\b}\right)\gadot_\a \gadot_\b=
\De\gadot_r \gadot_r+ \frac{1}{\De} \GG^{\a\b}\gadot_\a \gadot_\b
\eeq
with
\[
\GG^{\a\b}\gadot_\a \gadot_\b= -(r^2+a^2) ^2\gadot_t \gadot_t- 2a(r^2+a^2) \gadot_t \gadot_\vphi- a^2  \gadot_\vphi \gadot_\vphi  +\De O^{\a\b} \gadot_\a \gadot_\b.
\]

Since $\pr_t=T$ and $\pr_\vphi=Z$ are Killing vectorfields, we deduce that $\gadot_t=\g(\gadot, T)$ and $\gadot_\vphi= \g(\gadot, Z)$ are constants of the motion, i.e. constants along $\ga$, and respectively called the energy and the azimuthal angular momentum. We write,
\[
\e:=-\g(\gadot, T) , \qquad \lz :=-\g(\gadot, Z).
\]
We also define\footnote{Observe that $\k^2$ is a positive constant of motion.}
\[
\k^2:=K^{\a\b} \gadot_\a\gadot_\b
\]
for the Carter tensor $K$ in Kerr. Since $K$ is Killing, $\k^2$ is also a constant of motion. Indeed, we have
\[
\frac{d}{d\lambda} \k^2\big(\ga(\lambda) \big)=\D_\la K_{\a\b}\,\gadot^\la \gadot^\a\gadot^\b=0.
\]

Since $K= -(a^2\cos^2\th) \g +O$ (see section 3.7 in \cite{GKS}) and {since}
$\ga$ is null we deduce, with $\gadot_a=\g(\gadot, e_a)$
\[
\k^2= O^{\a\b} \gadot_\a \gadot_\b=|q|^2  \big( e_1^\a e_1^\b+ e_2^\a e_2^\b\big) \gadot_\a \gadot_\b =|q|^2 \big(|\gadot_1|^2+|\gadot_2|^2 \big).
\]

We summarize the result in the following.
\begin{proposition}
The quantities
\[
\e=-\g(\gadot, T) , \qquad \lz=-\g(\gadot, Z), \qquad \k^2=K^{\a\b} \gadot_\a\gadot_\b
\]
are constants along null geodesics. Moreover
\[
\k^2=|q|^2 \big(|\gadot_1|^2+|\gadot_2|^2 \big).
\]
\end{proposition}

With these constants we have, relative to the BL coordinates $(t, r, \th, \vphi)$,
\begin{align*}
\GG^{\a\b}\gadot_\a \gadot_\b&= -(r^2+a^2) ^2\gadot_t \gadot_t- 2a(r^2+a^2) \gadot_t \gadot_\vphi- a^2  \gadot_\vphi \gadot_\vphi  +\De O^{\a\b} \gadot_\a \gadot_\b\\
&= -(r^2+a^2) ^2 \e^2- 2a(r^2+a^2) \e \, \lz - a^2  \lz^2  +\De \k^2
\end{align*}
which is only a function of $r$ along any fixed null geodesic.
We introduce the notation
\[
G(r; a, m, \e,\lz, \k):=  -(r^2+a^2)^2 \e^2- 2a(r^2+a^2) \e \, \lz - a^2  \lz^2  +\De \k^2.
\]
Note that we can rewrite
\beq
\lab{identity:RR(r,a,m,e,lz,q}
-G(r; a, m, \e,\lz, \k)=\big(( r^2+a^2)\e + a \lz\big)^2-\De\k^2.
\eeq
In view of \eqref{eq:identity-null}, we infer that
\[
0=
\De\gadot_r \gadot_r+ \frac{1}{\De} \GG^{\a\b}\gadot_\a \gadot_\b =\De\gadot_r \gadot_r+ \frac{1}{\De}G(r; a, m, \e,\lz, \k).
\]
Differentiating $ r= r(\ga^\a(\la)), \th=\th(\ga^\a(\la))$ with respect to $\la$,
\begin{align*}
\frac{dr}{d\la} &=\gadot^r=\g^{rr}\gadot_r=\frac{\De}{|q|^2}\gadot_r\\
\frac{d\th}{d\la}&=\gadot^\th=\g^{\th\th}\gadot_\th=\frac{1}{|q|^2}\gadot_\th.
\end{align*}
Since, see \eqref{eq:canonicalHorizBasisKerr},
\[\
e_1=\frac{1}{|q|}\pr_\th,\qquad e_2=\frac{a\sin\th}{|q|}\pr_t+\frac{1}{|q|\sin\th}\pr_\phi.
\]
we also deduce
\[
\k^2=|q|^2 \big(\g(\gadot, e_1)^2 +\g(\gadot, e_2)^2 \big)= \gadot_\th^2+\frac{1}{\sin^2\th}(\lz+a\sin^2\th\e)^2.
\]
We finally obtain
\beq
\lab{eq:eqnfornullgeor}
\bsplit
|q|^4\Big(\frac{dr}{d\la} \Big)^2&=\De^2\gadot_r^2=-G(r; a, m, \e,\lz, \k)\\
|q|^4\Big(\frac{d\th}{d\la}\Big)^2&=\gadot_\th^2=\k^2-\frac{1}{\sin^2\th}(\lz+a\sin^2\th\e)^2.
\end{split}
\eeq
which are the equations for a null geodesic with constants of motion $\e$, $\lz$, $\k$.


\subsection{Trapped null geodesics}
\lab{section:trapped-nullg}


\begin{definition}
Null geodesics $\ga(\la)$ which are restricted to a bounded region $[r_1, r_2]$ of $r$ with $r_+<r_1<r_2<+\infty$ for all values of $\la$
 are called trapped. Trapped null geodesics for which $r$ remains constant are called orbital.
\end{definition}

It is known that all trapped null geodesics are in fact orbital, see for example Proposition 2 in \cite{CeJa}. Thus, from now on, we do not distinguish between trapped and orbital null geodesics.

If $r$ is constant we deduce from the first equation in \eqref{eq:eqnfornullgeor},
\[
-\partial_rG(r; a, m, \e,\lz, \k)=\partial_r(|q|^4)\Big(\frac{dr}{d\la} \Big)^2+2 |q|^4 \frac{dr}{d\la} \partial_r \Big(\frac{dr}{d\la} \Big)=0.
\]
The $r$ values for which such solutions are possible must then verify the equations
\[
G(r; a, m, \e,\lz, \k)=\pr_rG(r; a, m, \e,\lz, \k)=0.
\]
Thus, introducing
\beq
\lab{def;Pi-ell_z}
\Pi:=( r^2+a^2)\e + a \lz,
\eeq
we write from \eqref{identity:RR(r,a,m,e,lz,q}
\beq
\lab{eq:trappe-nullgeod}
\bsplit
-G(r; a, m, \e,\lz, \k)&=\Pi^2-\De \k^2=0\\
-\pr_r G(r; a, m, \e,\lz, \k)&=2\Pi (\pr_r \Pi) -(\pr_r \De)\k^2=0.
\end{split}
\eeq
From the second equation, we deduce
\beq
\lab{eq:orbittingnullgeod0}
\k^2=2\Pi \frac{\pr_r\Pi }{\pr_r \De}.
\eeq
Multiplying the first equation by $\pr_r \De$ and the second by $\De$ and subtracting we derive
\[
\Pi\big(\pr_r \De \Pi- 2 \De \pr_r \Pi\big)=0,
\]
or, assuming\footnote{If $\Pi=0$, the first equation in
\eqref{eq:trappe-nullgeod} and $\De>0$ give
$\k^2=\gadot_\th^2+\sin^{-2}\th(\lz+a\sin^2\th\e)^2=0$.
Thus $\gadot_\th=0$ and $\lz+a\sin^2\th\e=0$.  Combining the latter
identity with $\Pi=(r^2+a^2)\e+a\lz=0$ yields
$|q|^2\e=(r^2+a^2\cos^2\th)\e=0$, and hence $\e=\lz=0$.
Together with $\gadot_r=\gadot_\th=0$, this forces the tangent vector to
vanish, so it does not define a nontrivial null geodesic.} $\Pi\neq0$,
\beq
\lab{eq:orbittingnullgeod1}
\Pi\pr_r \De - 2 \De \pr_r \Pi=0.
\eeq

We make use of the following calculation.
\begin{lemma}
\lab{lemma:calculationPiDe}
We have the identity\footnote{i.e., for $\Pi\neq 0$, $ \pr_r\Big(\frac{\De}{\Pi^2}\Big)=-2 \frac{ \TT_{\e, \lz}}{\Pi^3}$.}
\beq\label{-2T-Pi-De}
\Pi(\pr_r\De)-2 (\pr_r \Pi) \De = - 2 \TT_{\e, \lz}
\eeq
where
\[
\TT_{\e, \lz}:= \big( r^3-3mr^2 + r a^2+ma^2\big)\e-  (r-m) a\lz.
\]
\end{lemma}

\begin{proof}
We have
\begin{align*}
(\pr_r \De) \Pi- 2\De (\pr_r\Pi) &= 2(r-m) \big( (r^2+a^2) \e + a \lz\big)- 4 r  \big( r^2+ a^2- 2rm\big) \e\\
&=2\Big(  (r-m) \big(( r^2+a^2)\e+  a \lz\big) - 2 r  \big( r^2+ a^2- 2rm\big)\e\Big)\\
&=2\Big( \big(-r^3+3mr^2 - r a^2-ma^2\big) \e+  (r-m) a\lz\Big)\\
&= - 2\TT_{\e, \lz}
\end{align*}
as stated.
\end{proof}
As a consequence, we deduce that all nontrivial orbital null geodesics are given by the equation
\beq
\TT_{\e, \lz}= \big( r^3-3mr^2 + r a^2+ma^2\big)\e-  (r-m) a\lz=0.
\eeq

\begin{corollary}
\lab{Cor:Trapping}
The following hold true.
\begin{enumerate}
\item There are no trapped null geodesics perpendicular to $T=\pr_t$ (i.e. $\e=0$) in the exterior of a sub-extremal Kerr.

\item The values of $r$ for which trapped null geodesics exist depend on the ratio $\lz/ \e$. More precisely, along trapped null geodesics, we have
\[
\frac{r^3-3mr^2 + r a^2+ma^2}{r-m}=\frac{ a\lz}{\e}.
\]
In particular, for $\lz=0$, the trapped null geodesics are given by the equation
\beq\lab{definition-TT}
\TT:= r^3- 3m r^2+ a^2r+a^2m=0.
\eeq
\item Along a trapped null geodesic we have
\beq
\lab{eq:k^2-trapped}
\k^2 =\frac{\Pi^2}{\De}.
\eeq

\end{enumerate}
\end{corollary}
\begin{proof}
For a nontrivial orbital null geodesic, $\Pi\neq0$ by the discussion
above.  If $\e=0$, then $\Pi=a\lz$ and
\[
\TT_{0,\lz}=-(r-m)a\lz=-(r-m)\Pi.
\]
Since $r>r_+>m$ in the exterior of a sub-extremal Kerr spacetime, the
trapping equation $\TT_{0,\lz}=0$ would imply $\Pi=0$, a contradiction.
The other statements are immediate consequences of what was discussed above.
\end{proof}

\begin{remark}
\lab{remark:TT_lz-TT_a}
We will also make use of the notation
\beq
\lab{defin:TT_{lz}}
\TT_{\lz}:=\TT_{e=1,\lz} = \TT-  (r-m) a\lz=2r\De-(r^2+a^2+a\lz)(r-m).
\eeq
Note that, in the particular case when $\e=1$ and $\lz=-a$ the characteristic polynomial $\TT_{\e,\lz}$ takes the form
\beq
\lab{eq:defineTT_a}
\TT_{-a}:=\TT_{e=1,\lz=-a}=  \TT+  (r-m) a^2= r\big(r^2-  3mr+2a^2\big).
\eeq
\end{remark}

\subsubsection{Instability of trapped null geodesics}

For completeness we repeat below the well-known instability argument for the trapped null geodesics, i.e. that\footnote{Note that this instability is related to the control of the radial derivative in a Morawetz estimate. } $\pr_r^2 G(r; a, m, \e,\lz, \k) \leq 0$.
We have
\[
\Pi=( r^2+a^2) \e + a \lz, \quad
\pr_r \Pi =2r \e,\quad
\pr_r^2 \Pi=2 \e,
\]
and using \eqref{eq:orbittingnullgeod0} to write $\k^2=4r \frac{\Pi }{\pr_r \De} \e$, we have
\begin{align*}
-\pr^2_r G(r; a, m, \e,\lz, \k)&=2(\pr_r \Pi)^2+2\Pi (\pr^2_r \Pi) -2\k^2=8r^2 \e^2+4\Pi \e-8r \frac{\Pi }{\pr_r \De} \e\\
&=\frac{4}{\partial_r \De} \left( 2r^2 \pr_r\De \e^2+(\pr_r\De) \Pi \e-2r \Pi  \e\right)
=\frac{4\e}{\partial_r \De} \left( 4r^2 (r-m) \e-2m \Pi \right).
\end{align*}
Using \eqref{eq:orbittingnullgeod1} to write $\Pi=\frac{4r\De \e }{\pr_r\De}$
we deduce
\begin{align*}
-\pr^2_r G(r; a, m, \e,\lz, \k)&=\frac{4\e}{\partial_r \De} \left( 4r^2 (r-m) \e-2m\frac{4r\De   \e }{\pr_r\De} \right)=\frac{32r\e^2}{(\partial_r \De)^2} \left( r (r-m)^2 -m \De  \right)\\
&= \frac{8r}{(r-m)^2}\e^2\Big( r(r-m)^2 -m (r^2+a^2- 2mr)\Big)\\
&= \frac{8r}{(r-m)^2}\e^2\Big((r-m)^3 +m( m^2-a^2)\Big),
\end{align*}
which is positive since $r> m$ and $|a|< m$.


\subsection{The $r$-range of trapped null geodesics}
\lab{section-range.trapped.ngeod}

To specify the values\footnote{Following the analysis in \cite{Teo}.} of $r$ for which trapped null geodesics exist we analyze the second equation in \eqref{eq:eqnfornullgeor}, i.e.
\[
|q|^4\Big(\frac{d\th}{d\la}\Big)^2=\k^2-\frac{1}{\sin^2\th}(\lz+a\sin^2\th\e)^2
\]
By the previous analysis on the radial equation, we find that at the trapped null geodesics
\[
\frac{\lz}{\e}=\frac{\TT}{a(r-m)}=\frac{r^3-3mr^2 + r a^2+ma^2}{a(r-m)}
=\frac{2r\Delta-(r-m)(r^2+a^2)}{a(r-m)}.
\]
Then
\begin{align*}
\k^2&=\frac{\Pi^2}{\Delta}=\frac{\Big((r^2+a^2)\e+a\lz\Big)^2}{\Delta}=\frac{\e^2}{\Delta}\Big((r^2+a^2)+\frac{2r\Delta-(r^2+a^2)(r-m)}{r-m}\Big)^2\\
&=\e^2\frac{4r^2\Delta}{(r-m)^2}.
\end{align*}
Therefore, at the trapped null geodesics, we have
\begin{align*}
|q|^4\Big(\frac{d\th}{d\la}\Big)^2&=\k^2-\frac{1}{\sin^2\th}(\lz+a\sin^2\th\e)^2=\e^2\Big(\frac{\k^2}{\e^2}-\frac{1}{\sin^2\th}(\frac{\lz}{\e}+a\sin^2\th)^2\Big)\\
&=\e^2\Big(\frac{4r^2\Delta}{(r-m)^2}-\frac{1}{\sin^2\th}\big(\frac{\TT}{a(r-m)}+a\sin^2\th\big)^2\Big)\\
&=\frac{\e^2}{a^2(r-m)^2}\Big(4a^2r^2\Delta-\frac{1}{\sin^2\th}\big(\TT+a^2(r-m)\sin^2\th\big)^2\Big).
\end{align*}
We rewrite in the form,
\[
\frac{a^2(r-m)^2}{\e^2}|q|^4\Big(\frac{d\th}{d\la}\Big)^2\sin^2\th= 4a^2r^2\Delta\sin^2\th-\big(\TT+a^2(r-m)\sin^2\th\big)^2.
\]
Setting $u=\cos\th$, and then with $x=\sin^2 \th$ we have\footnote{ In fact
$\frac{a^2(r-m)^2}{\e^2}|q|^4\Big(\frac{du}{d\la}\Big)^2=Q_{r(\la)}(x\big(\la)\big)$.}
\beq
\lab{eq:Q_r(x)}
\bsplit
&\frac{a^2(r-m)^2}{\e^2}|q|^4\Big(\frac{du}{d\la}\Big)^2=Q_r(x), \qquad
Q_r(x):= 4a^2r^2\Delta x -\big(\TT+a^2(r-m) x \big)^2.
\end{split}
\eeq
Therefore, along the geodesic $r(\la)$, we must have $Q_{r(\la)}(x\big(\la)\big)\ge 0$. In particular for every $ r$ in the range of $r(\la)$, $Q_r(x)\ge 0$ for some $x\in[0, 1]$.
We show below, see Lemma \ref{lemma:Q_r-increases}, that for every $r>r_{+}$ the polynomial $Q_r(x)$ is increasing in the interval $x\in [0,1]$. Therefore for every $ r$ in the range of $r(\la)$ we must have $Q_r(1)\ge 0$.
Consider the polynomial
\beq
\lab{eq:PolynTheta}
\Theta(r):=-Q_r(1)=(\TT+a^2(r-m))^2-4a^2r^2\Delta.
\eeq

\begin{remark}
\lab{remark-Theta}
The sixth-degree polynomial $\Theta(r):=(\TT+a^2(r-m))^2-4a^2r^2\Delta$
factorizes as
\[
\Theta(r)=r^3\Ththree(r),
\qquad
\Ththree(r)=r^3-6mr^2+9m^2r-4ma^2.
\]
\end{remark}
We need to determine the $r$-interval, $r\ge r_+$, for which $\Th(r)\le 0 $.
\begin{proposition}
\lab{Prop:r-trapped.region}
Let $r_\pm$ be the roots of $\De=0$, with $r=r_+$ the event horizon, and $r_1<r_2 $ the non-vanishing roots of the polynomial
\[
\TT+a^2(r-m)=r\big( r^2 - 3mr +2a^2\big)= r(r-r_1)(r-r_2)=0.
\]
Then, for all $0<|a|<m$,
\begin{enumerate}
\item We have the sequence of inequalities
$-\infty< r_{-}< r_1<r_{+}< r_2< \infty$.
\item
The polynomial $\Tht(r)=r^{-2}\Th(r)=r\Ththree(r)$ verifies
$\Tht(\pm \infty)=\infty,\,\, \Tht(0)=0$, $(r^{-1} \Tht)(0)<0$, as well as
\begin{align*}
\Tht(r_-)&=\frac{(\TT+a^2(r-m))^2}{r^2}(r_-)= (r_{-}-r_1)^2(r_{-}-r_2)^2>0\\
\Tht(r_1)&=- 4a^2\Delta(r_1)=-4a^2 (r_1- r_{-}) (r_1- r_{+})>0,\\
\Tht(r_+)&=\frac{(\TT+a^2(r-m))^2}{r^2}(r_+)= (r_{+}-r_1)^2(r_{+}-r_2)^2>0,\\
\Tht(r_2)&=- 4a^2\Delta(r_2)=-4 a ^2(r_2- r_{-}) (r_2- r_{+}) <0.
\end{align*}
\item The polynomial $\Th(r)$ has precisely two roots in $(r_+,\infty)$,
equivalently the two exterior roots $\rhat_1<\rhat_2$ of $\Ththree(r)$, such that
\[
r_{+} <\rhat_1<r_2 <\rhat_2<\infty.
\]
Moreover $\Th(r)\le 0$ in the interval $[\rhat_1, \rhat_2]$.
\end{enumerate}
Therefore the range\footnote{In view of Lemma \ref{Lemma:rangeofrfortrappednullgeodesics}, $r=\rhat_1, \rhat_2$ are in fact the equatorial trapped null geodesics (i.e. the null geodesics on the equator $\frac{d\th}{d\la}=0,\th=\frac{\pi}{2}$.)} of values of $r$ for which null trapped geodesics exist\footnote{In the case $a=0$, excluded above, we have $\rhat_1=\rhat_2=3m$.} is included in the interval $[\rhat_1, \rhat_2]$, i.e.
$r_{trap}\in [ \rhat_1, \rhat_2]$.
\end{proposition}
\begin{proof}
To prove the first part of the proposition we check the sequence of inequalities
\[
r_{-}\, \le r_1\,<  r_{+}\,  \le 2m <\, r_2
\]
where
\[
r_{\pm}=m\pm \sqrt{m^2-a^2}, \qquad r_{1,2}=\frac{3m\pm \sqrt{9m^2-8a^2}}{2}.
\]
The first inequality $r_{-} \le r_1$ is equivalent to, with $ \hat{a}=\frac{a}{m}\in[0,1)$, $2\sqrt{1-\hat{a}^2} +1 \ge \sqrt{9- 8 \hat{a}^2},$
which, after squaring both sides, reduces to $\sqrt{1-\hat{a}^2} \ge 1-\hat{a}^2$.
The second inequality $r_1\le r_+ $ becomes, with $ \hat{a}=\frac{a}{m}$, $1- 2 \sqrt{1-\hat{a}^2} < \sqrt{9-8\hat{a}^2}
$
which after squaring both sides becomes $\hat{a}^2-1< \sqrt{1-\hat{a}^2}$. The other two inequalities are immediate.
The second part of the Lemma requires only an immediate verification.

The last part is an immediate consequence of the second part. Note that in fact $\Tht$ has four roots $0=\rhat_{-2}< \rhat_{-1} <\rhat_1<\rhat_2$
with $0=\rhat_{-2}< \rhat_{-1} < r_{-}$, of no interest to us,
and
\[
r_{+} <\rhat_1<r_2 <\rhat_2<\infty.
\]
\end{proof}

\begin{lemma}\lab{Lemma:rangeofrfortrappednullgeodesics}
In general one can show, see \cite{Teo},
\beq
\bsplit
\rhat_1 &:=2m\left(1+\cos\left(\frac{2}{3}\arccos\left(-\frac{|a|}{m}\right)\right)\right),\\
\rhat_2 &:=2m\left(1+\cos\left(\frac{2}{3}\arccos\left(\frac{|a|}{m}\right)\right)\right).
\end{split}
\eeq
\end{lemma}
\begin{proof}
See appendix \ref{sec:rangeofrfortrappednullgeodesics}.
\end{proof}

In section \ref{section:EnergyEstim} we make use of the following calculations for the first and second derivatives of $\rhat_1, \rhat_2$
with respect to $\mu=\frac{a^2}{m^2}$.
\begin{lemma}
\lab{Lemma:Deriv.yhat}
The following identities  hold true for $\yhat_1= m^{-1} \rhat_1, \, \yhat_2= m^{-1} \rhat_2$.
\beq
\lab{eq:1^stDeriv-yhat}
\bsplit
\frac{d}{d\mu}\yhat_1&=-  \frac 1 3 \frac{1}{\sqrt{\mu(1-\mu)}}  \sqrt{\yhat_1(4-\yhat_1)},\\
\frac{d}{d\mu}\yhat_2&= \frac  1 3   \frac{1}{\sqrt{\mu(1-\mu)}} \sqrt{\yhat_2(4-\yhat_2)}.
\end{split}
\eeq
and
\beq
\lab{ eq:2^ndDeriv-yhat}
\bsplit
\frac{d}{d\mu}\Big(\sqrt{\mu(1-\mu)} \yhat_1' \Big)&= \frac{1}{3}  \frac{1}{\sqrt{\mu(1-\mu) }} \big(\yhat_1-2\big), \\
\frac{d}{d\mu}\Big(\sqrt{\mu(1-\mu)} \yhat_2' \Big)&=-\frac{1}{3}  \frac{1}{\sqrt{\mu(1-\mu) }} \big(\yhat_2-2\big).
\end{split}
\eeq
\end{lemma}
\begin{proof}
See appendix \ref{sec:Deriv.yhat}.
\end{proof}


\subsection{Monotonicity of $Q_r(x)$}


\begin{lemma}
\lab{lemma:Q_r-increases}
For every $r>r_{+}$ the polynomial\footnote{Recall that $Q_r(1)=-\Th(r)$.} $Q_r(x)$ is increasing in $x\in [0,1]$.
\end{lemma}
\begin{proof}
We rewrite
\begin{align*}
Q_r(x)&= 4a^2r^2\Delta x -\big(\TT+a^2(r-m) x \big)^2\\
&=-a^4(r-m)^2 x^2 +\big(4a^2r^2\Delta-2a^2(r-m)\TT\big) x -\TT^2\\
&=-a^4(r-m)^2 x^2+ 2a^2\big(2 r^2\Delta-(r-m)\TT\big) x -\TT^2.
\end{align*}
In view of \eqref{defin:TT_{lz}}, i.e. $\TT= 2 r \De-(r-m)(r^2+a^2)$,
\[
2 r^2\Delta-(r-m)\TT= 2 r^2\Delta-(r-m) \big( 2 r \De-(r-m)(r^2+a^2)\big)= 2mr \De+ (r-m)^2 (r^2+a^2).
\]
Therefore,
\[
Q_r(x)=-a^4(r-m)^2 x^2 +2a^2\big(2mr\Delta+(r-m)^2(r^2+a^2)\big) x -\TT^2
\]
and consequently,
\begin{align*}
Q_r'(x) &= -2 a^4(r-m)^2 x  +2a^2\big(2mr\Delta+(r-m)^2(r^2+a^2)\big)\\
&= 2a^4(r-m)^2\Big(-x+\frac{\big(2mr\Delta+(r-m)^2(r^2+a^2)\big)}{a^2(r-m)^2}\Big)
\end{align*}
Note that for $r>r_+$, i.e. $\De>0$,
\[
\frac{\big(2mr\Delta+(r-m)^2(r^2+a^2)\big)}{a^2(r-m)^2}>\frac{r^2+a^2}{a^2}>1,
\]
Thus $Q_r'(x)> 2a^4(r-m)^2(1-x) $
which implies, as stated, that $Q_r(x)$ is increasing on $x\in[0,1]$ for all $r>r_+$.

\end{proof}


\subsection{Separation  results  for  the trapped  set}
\lab{section:Separation}
For the sake of completeness we show below that i) the ergoregion and trapping region are strictly separated in physical space
if $\frac{|a|}{m} <\frac{\sqrt{2}}{2}$ and that ii) superradiant frequencies are not trapped, even though we
 will not make use of any of these two statements in this paper. Similar less precise statements appear first in Lemma 4.8 of \cite{FKSY}
  and in \cite{DRS}, where it played an essential role in the derivation of their Morawetz estimate.

\subsubsection{Separation of the trapped set from the ergoregion}

\begin{lemma}
\lab{Le:separation}
If $\frac{|a|}{m} <\frac{\sqrt{2}}{2}$ then $2m< \rhat_1$, i.e. the ergoregion and trapping set are separated.
\end{lemma}

\begin{proof}
Note first that the maximal value of $r$ in the ergoregion $ \frac{\left(\Delta-a^2\sin^2\theta\right)}{|q|^2}<0$ is given by $r=2m$.
It thus suffices to show that the $ \Th(2m)>0$ if $\frac{|a|}{m} <\frac{\sqrt{2}}{2}$.

Recall that $\TT+a^2(r-m)=r\big( r^2 - 3mr +2a^2\big)$ and hence
\[
\frac{\Th(r)}{r^2}=\frac{\big(\TT+a^2(r-m)\big)^2- 4 a^2r^2\Delta}{r^2}=\big(r^2-3mr+2a^2\big)^2- 4 a^2\Delta.
\]
To have $\rhat_1>2m$, it suffices to have $(r^{-2}\Th)(2m)>0$.
We thus calculate
\[
(r^{-2}\Th)(2m)=(2a^2-2m^2)^2-  4    a^2 a^2 =4m^2(m^2-2a^2)
\]
from which the desired statement follows.
\end{proof}

 \subsubsection{Superradiant frequencies  are  not trapped}

\begin{definition}
We define the flux quantity, with $\Tplus=\T+\frac{a}{r_+^2+a^2} \Z= \T+\frac{a}{2m r_+} \Z $ tangent to the generators of the horizon $\HH$
\beq
\FF[\psi](\tau_1, \tau_2)=\int_{\HH(\tau_1, \tau_2)}\Re\big( T\psi \ov{\Tplus\psi} \big)
\eeq
\end{definition}
\begin{definition}
We say that    a solution   $\psi$ of $\square \psi=0$ is a $ (\om, \ell) $ mode solution if  $T\psi= -i\omega\psi, \, Z\psi=  i\ell\psi$.
\end{definition}
Given such mode solution $\psi$ we have
\begin{align*}
\FF[\psi](\tau_1, \tau_2)&= \int_{\HH(\tau_1, \tau_2)} \Re\Big(-i\om\psi \c \ov{\big( -i \om+i\ell\frac{a}{2m r_+} \big)\psi}\Big)\\
&=\int _{\HH(\tau_1, \tau_2)}\om^2 \big(1 -\frac{\ell}{\om}\frac{a}{2mr_+} \big)|\psi|^2
\end{align*}
\begin{definition}
We say that the frequencies $(\om, \ell)$ are superradiant if $\frac{a\ell}{\om} > 2m r_+$. In that case, for any $\tau_1<\tau_2$, we have $ \FF[\psi](\tau_1, \tau_2)<0$.

\end{definition}
 \begin{proposition}
  Given a superradiant $(\om, \ell)$-frequency, the null geodesic with $\e=\om$ and $\lz=-\ell$ is not trapped.
 \end{proposition}
 \begin{proof}

 Recall that the values of $r$ for which trapped null geodesics exist depend on the ratio $\lz/ \e$. More precisely, along trapped null geodesics, we have
\beq\lab{eq:trappingfrequency}
\frac{\TT}{r-m}=\frac{r^3-3mr^2 + r a^2+ma^2}{r-m}=\frac{ a\lz}{\e}.
\eeq
Recall that the range of values of $r$ for which null trapped geodesics exist is included in the interval $[\rhat_1, \rhat_2]$ where $r_+<\rhat_1< \rhat_2$ are the two roots in $(r_+,\infty)$ of the cubic polynomial
\beq\Ththree(r)=r^3-6mr^2+9m^2r-4ma^2.
\eeq
Using \eqref{eq:trappingfrequency} and the information on $\rhat_1, \rhat_2$, we further determine the range of values of $\frac{a\lz}{\e}$ for which null trapped geodesics exist. Since for $r\in [r_+, 4m]\supset [\rhat_1, \rhat_2]$,
\begin{align*}
\frac{d}{dr}\Big(\frac{\TT}{r-m}\Big)&=\frac{d}{dr}\Big(\frac{r^3-3mr^2+ra^2+ma^2}{r-m}\Big)= \frac{2(r^3-3mr^2+3m^2r-ma^2)}{(r-m)^2}\\
&=2 \frac{(r-m)^3 + m(m^2-a^2)}{(r-m)^2}\ge 0,
\end{align*}
we infer that $\frac{\TT}{r-m}$ is increasing for $r\ge r_+$.
According to Proposition \ref{Prop:r-trapped.region}, $\rhat_1< r_2 <\rhat _2$
where $r_2>r_+$ is the larger root of the polynomial $ \TT+a^2(r-m)=r\big( r^2 - 3mr +2a^2\big)= r(r-r_1)(r-r_2)$. Thus $\frac{\TT(r_2)}{r_2-m}=- a^2$ and therefore,
\[
\frac{ \TT(\rhat_1)}{\rhat_1-m}< \frac{\TT(r_2)}{r_2-m}=-a^2<0
\]
On the other hand using Lemma \ref{Lemma:rangeofrfortrappednullgeodesics} we check\footnote{See also Remark \ref{rem:monotonicity-rhat}.} that $\rhat_1< 3 m < \rhat_2$. Therefore
\[
\frac{ \TT(\rhat_2)}{\rhat_2-m}> \frac{\TT(3m)}{3m-m}=\frac{4ma^2}{2m}=2 a^2.
\]
Hence
\[
\frac{\TT(\rhat_1)}{(\rhat_1-m)} < 0< \frac{\TT(\rhat_2)}{(\rhat_2-m)}
\]
and therefore for any trapped null geodesic with energy $\e$ and angular momentum $\lz$ we have
\beq
\lab{eq:trappingfreq}
\frac{a\lz}{\e}  \in\Big(\frac{\TT(\rhat_1)}{(\rhat_1-m)}, \frac{\TT(\rhat_2)}{(\rhat_2-m)}\Big).
\eeq
We now make use of the following lemma.

\begin{lemma}\lab{lemma:trapping-superradiant}
    We have
    \beq
    -2mr_+\leq \frac{\TT(\rhat_1)}{(\rhat_1-m)}
    \eeq
    where the equality holds true if and only if $a=m$.
\end{lemma}
\begin{proof}
Observe that
\beaa
\frac{d}{da}\Big(\frac{\TT(\rhat_1)}{(\rhat_1-m)}\Big)=\left[\frac{d}{dr}\Big|_{r=\rhat_1}\Big(\frac{\TT(r)}{r-m}\Big)\right]\c\frac{d\rhat_1}{da}<0.
\eeaa
Indeed, in view of Lemma \ref{Lemma:Deriv.yhat}, $\rhat_1$ is decreasing in $a$, i.e. $\frac{d\rhat_1}{da}< 0$
 and $\frac{d}{dr}\Big(\frac{\TT}{r-m}\Big)>0$. We deduce
 \beaa
 -2mr_+\leq -2mr_+|_{a=m}=-2m=\frac{\TT(\rhat_1)}{(\rhat_1-m)}\Big|_{a=m}\leq \frac{\TT(\rhat_1)}{(\rhat_1-m)}.
 \eeaa
as stated.
\end{proof}
As a consequence of the Lemma and \eqref{eq:trappingfreq} we infer that
any trapped null geodesic with energy $\e=\om$ and angular
momentum $\lz=-\ell$ is such that $- 2 m r_+ <- \frac{a\ell}{\om}$, i.e.
\[
 \frac{a\ell}{\om}< 2m r_+
\]
i.e. the frequency pair $(\om, \ell)$ is not superradiant.
This ends the proof of the proposition.
\end{proof}



\section{Proof of the conditional Morawetz estimate in the scalar case.}
\lab{sect:section5}

\subsection{Basic scalar spacetime identity}
\lab{sect:Basicscalarident}

We recall the following general identity\footnote{This holds true for an arbitrary Lorentzian manifold $(\DD, \g)$, see \cite{GKS}.}
\beq
\lab{eq:Gen-identityXwM}
\bsplit
\D^\mu  \PP_\mu[X, w, M] &= \frac 1 2 \QQ  \c\piX 
+\frac 12  w \LL[\psi] -\frac 1 4|\psi|^2   \square_\g  w  + \frac 1 4  \Div(|\psi|^2 M)\\
&+ \left(X( \psi )+\frac 1 2   w \psi\right)\c \square \psi.
\end{split}
\eeq
where $\PP_\mu$ is the generalized current
\beq
\lab{eq:Gen-current}
\PP_\mu[X, w, M]:=\QQ_{\mu\nu} X^\nu +\frac 1 2  w \psi \c \D_\mu \psi -\frac 1 4|\psi|^2   \pr_\mu w +\frac 1 4 |\psi|^2 M_\mu.
\eeq
Here
\[
\QQ_{\mu\nu}= \D_\mu\psi \D_\nu \psi-\frac 12 \g_{\mu\nu} \LL[\psi], \qquad \LL[\psi]=\g^{\a\b}\D_\a\psi \D_\b\psi
,
\]
$ X$ is an arbitrary vectorfield with deformation tensor
\beq\label{eq:definitionofdeformationtensor}
\piX_{\mu\nu}:=\D_\mu X_\nu+\D_\nu X_\mu,
\eeq
$w$ is a scalar function and $M$ is an arbitrary one form.
We denote
\beq \lab{definition-EE-gen1}
\EE[\psi][X, w, M] := \D^\mu  \PP_\mu[X, w, M] -  \left(X(\psi)+\frac 1 2   w \psi\right)\c  \square \psi
\eeq
We write $\EE[\psi][X, w] =\EE[\psi][X, w, 0] $ in the case $M=0$.
We apply the general identity \eqref{eq:Gen-identityXwM} to the case of the Kerr metric in the BL coordinates
and a vectorfield of the form $X=\FF \pr_r$. The result is a slightly simplified version of 
Proposition 7.1.5 in \cite{GKS}.
\begin{proposition}
\lab{proposition:Morawetz1}
Let $\FF, w_{red}$ {be} given functions of $r$.
With the choice of vectorfield $X=\FF \pr_r$ and scalar function $w= |q|^2 \D_\a\big( |q|^{-2} X^\a\big)-w_{red}$, the generalized current $\EE[X, w, M]=\EE[\psi][X, w, M]$ defined in \eqref{definition-EE-gen1}
verifies
\beq
\lab{identity:prop.Morawetz1}
\bsplit
|q|^2\EE[X, w, M]  &=\AA |\pr_r\psi|^2 + \UU^{\a\b}\D_\a \psi \D_\b \psi +\VV |\psi|^2
+\frac 1 4 |q|^2  \D^\mu (|\psi|^2 M_\mu)
\end{split}
\eeq
where
\begin{align*}
\AA&= \De \pr_r \FF- \frac 1 2 \FF\pr_r \De-\frac 1 2 \De w_{red},\\
\UU^{\a\b}&=  -\frac 1 2  \FF\pr_r \left(\frac 1 \De\GG^{\a\b}\right)-\frac 1 2   w_{red}\frac 1 \De \GG^{\a\b},\\
\VV&= -\frac 14 |q|^2 \square_\g  w .
\end{align*}
If in addition we choose, for fixed functions $z, f, h$ depending on $r$,
\beq
\lab{def-w-red-in-fun-FF-00-wave}
\FF=-z h f, \qquad               w_{red}=  \FF  z^{-1}\partial_r z = - (\partial_r z ) h  f , \qquad w =- z \pr_r \big( h  f  \big),
\eeq
then
\beq
\lab{eq:coeeficientsUUAAVV}
\bsplit
\AA&=-z^{1/2}\Delta^{3/2} \partial_r\left(h \frac{ z^{1/2}  f }{\Delta^{1/2}}  \right),   \\
\UU^{\a\b}&=   \frac{ 1}{2}  h f \pr_r\left( \frac z \De\GG^{\a\b}\right),\\
\VV&=   \frac 1 4\pr_r\Bigg(\De \pr_r \Big(
z \pr_r \big( h f \big)  \Big)  \Bigg)
\end{split}
\eeq
Moreover if $M = v(r,\th) \pr_r$, for some function $v=v(r,\th)$, we have
\beq\lab{expression-Div-M-I}
\frac 1 4 |q|^2 \Div(|\psi|^2 M)= \frac 1 4 |q|^2\left( 2 v(r,\th)\psi\c \nab_r \psi + \left(\pr_r v+ \frac{2r}{|q|^2} v\right) |\psi|^2 \right).
\eeq
\end{proposition}
\begin{proof}

According to Lemma 7.1.3 in \cite{GKS}\footnote{We repeat the proof of the first part of Proposition 7.1.5 for the convenience of the reader.},
\begin{align*}
|q|^2  (\QQ  \c\piX +w\LL[\psi] &=  \big(2 \De \pr_r \FF- \FF\pr_r \De\big)|\pr_r\psi|^2
- \FF\pr_r\left(\frac 1 \De\GG^{\a\b}\right)  \D_\a \psi\D_\b \psi \\
&+\Big(X\big( |q|^2\big)-|q|^2\Div_\g X+|q|^2 w\Big)\LL[\psi] 
\end{align*}
Setting $ w_{red}= |q|^2 \D_\a\big( |q|^{-2} X^\a\big)- w= \Div_\g X-|q|^{-2} X(|q|^2) -w$,
\begin{align*}
|q|^2\EE[X, w, M]      &=  \left( \De \pr_r \FF- \frac 1 2 \FF\pr_r \De\right)|\pr_r\psi|^2
-\frac 1 2  \FF\pr_r\left(\frac 1 \De\GG^{\a\b}\right)  \D_\a \psi \D_\b \psi \\
&-\frac 1 2 |q|^2w_{red} \LL[\psi] -\frac 1 4 (|q|^2 \square   w )|\psi|^2+ \frac 1 4 |q|^2 \Div(|\psi|^2 M).
\end{align*}
Also, writing,
\begin{align*}
|q|^2  \LL[\psi]&=|q|^2 \g^{\a\b}\D_\a\psi\D_\b\psi=\left( \De \pr_r^\a \pr_r^\b +\frac{1}{\De}\GG^{\a\b}\right) \D_\a\psi\D_\b\psi\\
&=\De | \pr_r\psi|^2+\frac 1 \De  \GG^{\a\b}\D_\a\psi\D_\b\psi
\end{align*}
we derive
\begin{align*}
|q|^2\EE[X, w, M]      &=  \left( \De \pr_r \FF- \frac 1 2 \FF\pr_r \De-\frac 1 2 \De w_{red}\right)|\pr_r\psi|^2\\
&-\frac 1 2 \left( \FF \pr_r\left(\frac 1 \De \GG^{\a\b}\right) + w_{red} \frac{1}{\De} \GG^{\a\b} \right)  \D_\a \psi \D_\b \psi\\
&-\frac14(|q|^2 \square  w )|\psi|^2+ \frac 1 4 |q|^2 \Div(|\psi|^2 M).
\end{align*}
as stated.
Finally we refer the reader to Proposition 7.1.5 in \cite{GKS} for the proofs of \eqref{eq:coeeficientsUUAAVV}, \eqref{expression-Div-M-I}.
\end{proof}

\subsection{Choice of $z,f$ and $h$}
\label{section:choicezhf}
In this section, we present choices for the functions $z, h, f$. The ultimate goal is to choose them so that the generalized current $\EE[X,w,M]$ as given in Proposition \ref{proposition:Morawetz1} is positive definite. According to \eqref{identity:prop.Morawetz1}, we have
\[
|q|^2\EE[X,w,M]=\AA|\pr_r\psi|^2+\UU^{\a\b}\D_\a\psi\D_\b\psi+\VV|\psi|^2+\frac14|q|^2\D^\mu(|\psi|^2M_\mu).
\]
We start by looking at the \textit{principal term}, quadratic in derivatives of $\psi$,
\[
P=\UU^{\a\b}\D_\a\psi\D_\b\psi=\frac12 hf\RRtp^{\a\b}\D_\a\psi\D_\b\psi
\]
where, recalling \eqref{eq:RR-Sa}, \eqref{components-RR-aund},
\[
\RRtp^{\a\b}=\pr_r\left(\frac{z}{\Delta}\GG^{\a\b}\right)=\pr_r\left(\frac{z}{\Delta}\GG^{\aund}S_{\aund}^{\a\b}\right)=\pr_r\left(\frac{z}{\Delta}\GG^{\aund}\right)S_{\aund}^{\a\b}.
\]
We define
\beq\lab{eq:defofRRtp}
\RRtp^\aund[z]:=\pr_r\left(\frac{z}{\Delta}\GG^{\aund}\right).
\eeq
Then we write
\beq
\label{eq:principaltermP}
P:=\UU^{\a\b}\D_\a\psi\D_\b\psi=\frac12 hf\RRtp^{\aund}[z]S_{\aund}^{\a\b}\D_\a\psi\D_\b\psi.
\eeq
We make the crucial choice for $z$ (see Lemma \ref{lemma:calculationPiDe} for motivation)
\beq
\lab{eq:z_lz}
z:=\frac{\De}{\Pi^2}, \qquad
\Pi=( r^2+a^2)  + a \lz.
\eeq


\begin{lemma}
\lab{Lemma:compute-RRtp-z_lz}
The following hold true for $z$ defined in \eqref{eq:z_lz}.
\begin{enumerate}
\item We have, with $\TT=r^3-3mr^2+a^2r+ma^2$ and $\TT_{\lz}=\TT-a\lz(r-m)$
\beq
\lab{eq:compute-RRtp-z_lz}
\bsplit
\RRtp^{\underline{1}}[z]&=\frac{-4a r(r^2+a^2)\lz}{\Pi^3},\quad
\RRtp^2[z]=\frac{4ar(r^2+a^2)-4a^2r\lz}{\Pi^3},\quad
\RRtp^3[z]=\frac{4a^2r}{\Pi^3}\\
\RRtp^4[z]&=\frac{-2\TT_{\lz}}{\Pi^3}.
\end{split}
\eeq
\item Notice the identity
\bea
\lab{eq:identityRRtp}
\RRtp^1[z]+\lz\RRtp^2[z]+\lz^2\RRtp^3[z]=0.
\eea
\item Setting $Z_{\lz}:=Z-\lz T$,
\beq
\RRtp^{\cund}[z]S_{\cund}^{\a\b}
=\frac{4ar(r^2+a^2)}{\Pi^3}\That^{(\a}
Z_{\lz}^{\b)}-\frac{2\TT_{\lz}}{\Pi^3}O^{\a\b}.
\eeq
\item We have
\beq
\lab{eq:Pterm}
P=  \frac12hf\bigg(\frac{4ar(r^2+a^2)}{\Pi^3}\That\psi
\c Z_{\lz}\psi-\frac{2\TT_{\lz}}{\Pi^3}\OO^{\a\b}\D_\a\psi\D_\b\psi\bigg).
\eeq

\end{enumerate}

\end{lemma}

\begin{proof}
The first three relations in \eqref{eq:compute-RRtp-z_lz} are immediate and so is also the identity \eqref{eq:identityRRtp}.
To check the last we compute, using \eqref{components-RR-aund},
\begin{align*}
\RRtp^4[z]&=\pr_r(\frac{z}{\Delta}\GG^{\underline{4}})=\pr_r(\frac{\Delta}{\Pi^2})=\frac{\pr_r\Delta\Pi^2-4r\Pi\Delta}{\Pi^4}\\
&=\frac{2(r-m)\big((r^2+a^2)+a\lz\big)-4r(r^2-2mr+a^2)}{\Pi^3}\\
&=\frac{(-2r^3+6mr^2-2a^2r-2ma^2)+2a(r-m)\lz}{\Pi^3}\\
&=\frac{-2\big(\TT-a\lz(r-m)\big)}{\Pi^3}=\frac{-2\TT_{\lz}}{\Pi^3}.
\end{align*}
Using the identity $\RRtp^1[z]+\lz\RRtp^2[z]+\lz^2\RRtp^3[z]=0$, we further compute
\begin{align*}
\RRtp^{\cund}[z]S_{\cund}^{\a\b}
&=\RRtp^1[z]T^\a T^\b+\RRtp^2[z]T^{(\a}Z^{\b)}+\RRtp^3[z]Z^\a Z^\b+\RRtp^4[z]O^{\a\b}\\
&=\Big(\RRtp^1[z]+\lz\RRtp^2[z]+\lz^2\RRtp^3[z]\Big)T^\a T^\b+\RRtp^2[z](T^{(\a}Z^{\b)}-\lz T^\a T^\b)\\
&\quad+\RRtp^3[z](Z^\a Z^\b-\lz^2 T^\a T^\b)+\RRtp^4[z]O^{\a\b}\\
&=\RRtp^2[z]T^{(\a} Z_{\lz}^{\b)}+\RRtp^3[z](Z+\lz T)^{(\a} Z_{\lz}^{\b)}+\RRtp^4[z]O^{\a\b}\\
&=\RRtp^2[z]T^{(\a} Z_{\lz}^{\b)}+\RRtp^3[z](Z_{\lz}+2\lz T)^{(\a}Z_{\lz}^{\b)}+\RRtp^4[z]O^{\a\b}\\
&=(\RRtp^2[z]+2\lz\RRtp^3[z])T^{(\a} Z_{\lz}^{\b)}+\RRtp^3[z]Z_{\lz}^\a Z_{\lz}^\b+\RRtp^4[z]O^{\a\b}
\end{align*}
Making use of the definition of $\Pi=r^2+a^2 +a\lz$, $Z_{\lz} =Z-\lz T$,
$\That= T+\frac{a}{r^2+a^2} Z $, we continue as follows
\beq\label{eq:coefficientsRRtp23}
\bsplit
\RRtp^{\cund}[z]S_{\cund}^{\a\b}
&=\Big(\frac{\RRtp^2[z]+2\lz\RRtp^3[z]}{\Pi}\Big)\big(\Pi T+aZ_{\lz}\big)^{(\a} Z_{\lz}^{\b)}\\
&+\Big(\RRtp^3[z]-\frac{a}{\Pi}(\RRtp^2[z]+2\lz\RRtp^3[z])\Big)Z_{\lz}^\a Z_{\lz}^\b+\RRtp^4[z]O^{\a\b}\\
&=\Big(\RRtp^2[z]+2\lz\RRtp^3[z]\Big)\frac{r^2+a^2}{\Pi}\That^{(\a}Z_{\lz}^{\b)}\\
&+\Big(\RRtp^3[z]-\frac{a}{\Pi}(\RRtp^2[z]+2\lz\RRtp^3[z])\Big)Z_{\lz}^\a Z_{\lz}^\b+\RRtp^4[z]O^{\a\b}.
\end{split}
\eeq
According to the formulas for $\RRtp^1[z], \RRtp^2[z], \RRtp^3[z], \RRtp^4[z]$, we finally obtain
\[
\RRtp^{\cund}[z]S_{\cund}^{\a\b}
=\frac{4ar(r^2+a^2)}{\Pi^3}\That^{(\a}
Z_{\lz}^{\b)}-\frac{2\TT_{\lz}}{\Pi^3}O^{\a\b}.
\]
Therefore, back to \eqref{eq:Pterm},
\[
P=  \frac12hf\bigg(\frac{4ar(r^2+a^2)}{\Pi^3}\nab_{\That}\psi
\nab_{Z_{\lz}}\psi-\frac{2\TT_{\lz}}{\Pi^3}O^{\a\b}\D_\a\psi\D_\b\psi\bigg)
\]
as stated.
\end{proof}
In order to make, at least, the coefficient of the term $\OO^{\a\b}\D_\a\psi\D_\b\psi$ in \eqref{eq:Pterm} positive, we need to choose $hf$ such that it has the same sign as $-\TT_{\lz}$.
\subsubsection{The case $\lz=-a$}
\lab{section:lz=-a}


From now on we set $\lz=-a$ in which case
\[
\Pi=r^2,\qquad z=\frac{\Delta}{r^4},\qquad \TT_{-a}=r(r^2-3mr+2a^2).
\]
We define
\beq
u:=hf.
\eeq
  Let $r_2=\frac{3m+\sqrt{9m^2-8a^2}}{2}$ be the largest root of $\TT_{-a}=r(r^2-3mr+2a^2)=0$ and $r_3=2m+\sqrt{4m^2-3a^2}> r_2>r_+$ be the largest root\footnote{The largest zero of the polynomial $r^2-4mr+3a^2$ is chosen to ensure
   the $C^1$ property of $w=-z\pr_ru$, see \eqref{eq:w-is-C^1}.} of $r^2-4mr+3a^2=0$. Following Stogin's construction \cite{St}, we choose
\beq\label{eq:choiceofderifh}
\pr_ru=\begin{cases}
-\frac{r_3z(r_3)}{z(r)},\quad &r\leq r_3\\
-r,\quad &r\geq r_3
\end{cases}
\eeq
and
\beq\lab{eq:choiceoffh}
u=\int_{r_2}^r\pr_ru\,dr=\begin{cases}
\int_{r_2}^r-\frac{r_3z(r_3)}{z(r)}\,dr\quad &r\leq r_3\\
-\frac12r^2+\frac12r_3^2+u(r_3)\quad &r\geq r_3
\end{cases}.
\eeq
Therefore, $u\geq 0$ for $r\leq r_2$ and $u\leq 0$ for $r\geq r_2$. In particular, $u\sim -m^2\ln\frac{r-r_+}{m}$ near $r=r_+$.

We also have
    \beq
    \lab{eq:w-is-C^1}
    \pr_r(z\pr_ru)=\begin{cases}
        0\quad &r\leq r_3\\
        -\pr_r\big(rz\big)=\frac{r^2-4mr+3a^2}{r^4}\quad &r\geq r_3
    \end{cases}
    \eeq
    which implies that $w:=-z\pr_ru$ is $C^1$ and piecewise smooth on $[r_+,\infty)$.

    Then we compute
\beq\label{eq:expressionsfoAA-a}
\bsplit
\AA&=-z^{1/2}\Delta^{3/2}\pr_r\left(\frac{z^{1/2}}{\Delta^{1/2}}u\right)=-\frac{\De^2}{r^2}\pr_r\left(\frac{u}{r^2}\right)\\
&=\begin{cases}
\frac{\De^2}{r^4}\big(\frac{2u}{r}+\frac{r_3z(r_3)}{z(r)}\big)\quad &r\leq r_3\\
\frac{\De^2}{r^5}(r_3^2+2u(r_3))\quad &r\geq r_3
\end{cases}
\end{split}.
\eeq
Since $u\geq0$ for $r\leq r_2$, it is easy to check that $\AA>0$ for $r\in(r_+, r_2]$. Note also that $\AA\sim r^{-1}\De$ near $r=r_+$. Then it remains to check the positivity of $\AA$ for $r_2<r\leq r_3$ (as the positivity of $\AA$ for $r\geq r_3$ follows from its positivity at $r=r_3$). We write for $r_2<r\leq r_3$,
\begin{align*}
\frac{2u}{r}&= \frac{2}{r}\int_{r_2}^{r}-\frac{r_3z(r_3)}{z(s)}\,ds\geq -\frac{2r_3z(r_3)}{rz(r)}\int_{r_2}^{r}1\,ds=-\frac{r_3z(r_3)}{z(r)}(2-\frac{2r_2}{r}) \\
&\geq -\frac{r_3z(r_3)}{z(r)}\frac{2(r_3-r_2)}{r_3}\geq -\frac{r_3z(r_3)}{z(r)}\end{align*}
where we use the facts that $z$ is decreasing on $[r_2,r_3]$ in the second inequality and $r_3\leq 2r_2$ in the last one. This proves that $\AA>0$ for $r>r_+$.

Finally, we compute
\beq\label{eq:expressionsfoVV-a}
\bsplit
\VV&= \frac 1 4\pr_r\Big(\De \pr_r \big(
z \pr_ru  \big)  \Big) =\begin{cases}
0\quad &r\leq r_3\\
-\frac 1 4\pr_r\Big(\De \pr_r \big(rz \big)  \Big)\quad &r\geq r_3\end{cases}\\
&=\begin{cases}
0 \qquad & r\leq r_3\\
\frac14\pr_r\Big((1-\frac{2m}{r}+\frac{a^2}{r^2})(1-\frac{4m}{r}+\frac{3a^2}{r^2})\Big)\qquad &r\geq r_3
\end{cases}
\end{split}
\eeq
Since for $r\geq r_3=2m+\sqrt{4m^2-3a^2}$,
\[
1-\frac{2m}{r}+\frac{a^2}{r^2},\quad \pr_r(1-\frac{2m}{r}+\frac{a^2}{r^2}),\quad \pr_r(1-\frac{4m}{r}+\frac{3a^2}{r^2})\ >0,\quad 1-\frac{4m}{r}+\frac{3a^2}{r^2}\ \geq 0,
\]
we conclude that $\VV>0$ for $r\geq r_3$ and $\VV=0$ for $r\leq r_3$.
\begin{remark}
We note that the choice made in both \cite{AB} and \cite{GKS} was $\lz=0$. In that case $\Pi=r^2+a^2, z=\frac{\Delta}{(r^2+a^2)^2}$. Moreover, the choices of $u$ in both \cite{AB} and \cite{GKS} are regular at $r=r_+$ while the choice of $u$ in this work, based on Stogin's method, has a logarithmic singularity\footnote{ To our knowledge the use of singular vectorfields to derive Morawetz estimates was first considered in \cite{MMTT}.}.
\end{remark}

We summarize the results above in the following.
\begin{proposition}
\lab{prop:Choice-zhf}
The coefficients of the current $ |q|^2\EE[X, w] =|q|^2\EE[X, w, M=0]$, see \eqref{identity:prop.Morawetz1} of Proposition \ref{proposition:Morawetz1},
\[
|q|^2\EE[X, w, 0]  =\AA |\pr_r\psi|^2 + \UU^{\a\b}\D_\a \psi \D_\b \psi +\VV |\psi|^2
\]
with the choices
\[
X=-zu\pr_r, \quad
w =- z \pr_ru, \quad  z=\frac{\De}{r^4},
\]
 and $u$ defined as in \eqref{eq:choiceofderifh} and \eqref{eq:choiceoffh} (with $u\sim -m^2\ln(r-r_+)$ near $r=r_+$), satisfies the following
\begin{enumerate}

\item The principal term $ P=\UU^{\a\b}\D_\a\psi \D_\b \psi$ is given by
\beq\lab{eq:principal-term}
P=\frac12u\left( - \frac{2\TT_{-a}}{ r^6}   O^{\a\b}\D_\a\psi\D_\b \psi + \frac{4ar(r^2+a^2)}{r^6} \That \psi\c \Zhat \psi\right).
\eeq
where $\TT_{-a}=r(r^2-3mr+2a^2)$.

\item The coefficient $ \AA $ in equation \eqref{eq:coeeficientsUUAAVV} is given in \eqref{eq:expressionsfoAA-a}. In particular $\AA>0$ for $r>r_+$ for $0\le \frac{|a|}{m} < 1$ and $\AA\sim r^{-1}\De$ near $r=r_+$.
\item The coefficient $\VV$ in equation \eqref{eq:coeeficientsUUAAVV} is given in \eqref{eq:expressionsfoVV-a}. Moreover, $\VV>0$ for $r\geq r_3$ and $\VV=0$ for $r\leq r_3$ where $r_3=2m+\sqrt{4m^2-3a^2}$.
\end{enumerate}
\end{proposition}

\subsubsection{Main difficulties in the proof of Theorem \ref{Thm:Moraw2}}
\lab{section:maindifficulties}

As described in the introduction there are various problems with the choices made above
in Proposition \ref{prop:Choice-zhf}.
\begin{enumerate}
\item The vectorfield $X$ behaves like $\ln(r-r_+)\Rhat $ near the event horizon $r=r_+$. To deal with this singular behavior we modify
$X, w$ (as in \cite{MMTT} and \cite{St}) in section \ref{section:temperMorawetz}, by introducing the modified family
 $X_{\ep}, w_{\ep}$ with $\ep >0$ sufficiently small. This introduces
 a non-positive error term of the form $\VV_{{error}}|\psi|^2$ to the current $|q|^2\EE$.
Note however that the positivity of $\AA$ and the sign of $P$ are not affected by this procedure.

\item The coefficient $\VV$ of $\psi^2$ vanishes for $r_+\leq r\leq r_3$. In fact, given the error term $\VV_{error}$ introduced by the modification of $X,w$ above, it is even slightly negative, proportional to the small parameter $\ep$. This issue is dealt with by using a local Hardy type estimate in section \ref{sec:controloflowerterm}.
\item The control of $|\That\psi|^2$ is missing from the form of $P$. This, as explained in the introduction, is easily fixed by adding a Lagrangian term in the Morawetz estimate (see section \ref{sec:Smalltrapmodified}).
\end{enumerate}

\subsection{The tempered pair $(X_{\ep}, w_\ep)$}
\lab{section:temperMorawetz}
In this section, we address the first issue discussed in section \ref{section:maindifficulties}.

We define $X_{\ep}$ and $w_{\ep}$ by changing $u$ so that it is bounded near the event horizon. Let $F: \mathbb{R} \to \mathbb{R}$ be a smooth function satisfying $F(x) = x$ for $x \leq 1$ and $F(x) = 2$ for $x \geq 3$. We use $F(x)$ to temper $u=hf$ by defining
\beq\lab{eq:defofhftemper}
u_{\ep} = \ep^{-1} F(\ep u),
\eeq
where $\ep> 0$ is a small constant to be chosen later. Then we define
\beq\lab{eq:defofXwtemper}
X_{\ep}= -zu_{\ep}\partial_r,\quad w_{\ep}=-z\pr_ru_{\ep}.
\eeq
We note that $X_{\ep}, w_{\ep}$ agree with $X, w$ in the region where $u<\ep^{-1}$. So $X_{\ep}, w_{\ep}$ only differ from $X, w$ in the region where $u>\ep^{-1}$, which is a small neighborhood of the event horizon.

Since $F(x) \le 2$, it follows that $u_{\ep}\leq 2\ep^{-1}$, which means that for all $\ep> 0$, the vectorfield $X_{\ep}$ is  regular up to the event horizon. Furthermore, since $u_{\ep}=2\ep^{-1}$ near the event horizon, the function $w_{\ep}$ is  zero in a neighborhood of the event horizon, which implies that $w_{\ep}$ has no contribution to the boundary term at the event horizon (when integrating \eqref{eq:EE_ep}\,).

Recall that
\beq
\lab{eq:EE_ep}
    |q|^2\EE[X_{\ep}, w_{\ep},0] =   \AA_{\ep} |\pr_r\psi|^2 + \UU^{\a\b}_{\ep}\D_\a \psi \D_\b \psi +\VV_{\ep} |\psi|^2
  \eeq
    where
   \begin{align*}
\AA_{\ep}&=-z^{1/2}\Delta^{3/2} \partial_r\left(\frac{ z^{1/2}  u_{\ep} }{\Delta^{1/2}}  \right),   \\
\UU^{\a\b}_{\ep}&=   \frac{ 1}{2}  u_{\ep} \pr_r\left( \frac z \De\GG^{\a\b}\right),\\
\VV_{\ep}&=   \frac 1 4\pr_r\Big(\De \pr_r \big(
z \pr_r u_{\ep}  \big)  \Big).
\end{align*}
First, $u_{\ep}$ also satisfies the property that $u_{\ep}\geq 0$ for $r_+\leq r\leq r_2$ and $u_{\ep}\leq 0$ for $r\geq r_2$.

Second,
\beq\lab{eq:AAepsilontemper}
\AA_{\ep}=\begin{cases}
\AA\quad &u\leq \ep^{-1}\\
\frac{\De^2}{r^4}\Big(\frac{2u_{\ep}}{r}-\pr_ruF'(\ep u)\Big)\geq0\quad &u\geq \ep^{-1}
\end{cases}.
\eeq
This proves the positivity of $\AA_{\ep}$ for $r>r_+$. In particular, $\AA_{\ep}\sim \ep^{-1}r^{-5}\De^2$ near the event horizon $r=r_+$.

Finally, we have
\[
\VV_{\ep}=\begin{cases}
\VV\quad &u\leq \ep^{-1}\\
\frac 1 4\pr_r\Big(\De \pr_r \big(
z \pr_r u_{\ep} \big)  \Big)\quad &u\geq \ep^{-1}
\end{cases}.
\]
A straightforward calculation yields for $u\geq \ep^{-1}$
\begin{align*}
\frac 1 4\pr_r\Big(\De \pr_r \big(
z \pr_ru_{\ep}  \big)  \Big)
&=\frac 1 4\pr_r\Big(\De \pr_r \big(
F'(\ep u)z\pr_ru  \big)  \Big)\\
&=\frac14\pr_r\Big(\De \ep\pr_ru \c F''(\ep u)\c z\pr_ru\Big)\\
&=\frac14\epsilon^2z\De(\pr_ru)^3F'''+\frac14\ep\pr_r\Big(\De \pr_ru\c z\pr_ru\Big)F''
\end{align*}
which implies that in the region where $u\geq \epsilon^{-1}$
\[
|\VV_{\ep}|\les \ep^2 r^4(r-r_+)^{-1}\chi_{\ep}(r)+\ep r\chi_{\ep}(r)
\]
where $\chi_{\ep}$ is some bounded function supported where $u\in [\ep^{-1}, 3\ep^{-1}]$. Since $u \sim -m^2\ln (r-r_+)$, $\chi_{\ep}$ is supported in an interval of the form $(r-r_+)\in[e^{-\frac{c_2}{m^2\ep}}, e^{-\frac{c_1}{m^2\ep}}]$. Thus,
\[
   \int_{u\geq \epsilon^{-1}} |\VV_{\ep}|\,dr\les \epsilon(e^{-c_1/\epsilon} + c_2 - c_1)\les \ep.
\]
Therefore
\beq\lab{eq:VVepsilontemper}
\VV_{\ep}=\begin{cases}
\VV\quad &u\leq \ep^{-1}\\
\VV_{error}\quad &u\geq \ep^{-1}
\end{cases}\quad \text{with}\quad \lVert\VV_{{error}}\rVert_{L^1(r)}\les \ep.
\eeq
We summarize the results above in the following.
\begin{proposition}
\lab{prop:Choice-zhftemper}
The coefficients of the current
\[
|q|^2\EE[X_{\ep}, w_{\ep}, 0]  =\AA_{\ep} |\pr_r\psi|^2 + \UU_{\ep}^{\a\b}\D_\a \psi \D_\b \psi +\VV_{\ep} |\psi|^2
\]
with the choices
\[
X_{\ep}=-zu_{\ep} \pr_r 
,  \quad
w_{\ep} =- z \pr_r  u_{\ep} 
\]
and $z=\frac{\De}{r^4}$ and $u_{\ep}$ defined as in \eqref{eq:defofhftemper}, satisfies the following
\begin{enumerate}

\item The principal term $ P=\UU_{\ep}^{\a\b}\D_\a\psi \D_\b \psi$ is given by
\beq\lab{eq:principal-termtemper-redshift}
P=\frac12u_{\ep}\left(  -\frac{2\TT_{-a}}{ r^6}   O^{\a\b}\D_\a\psi\D_\b \psi +\frac{4ar(r^2+a^2)}{r^6} \That \psi\c \Zhat \psi\right).
\eeq
where $\TT_{-a}=r(r^2-3mr+2a^2)$ and $u_{\ep}\geq 0$ for $r_+\leq r\leq r_2$ and $u_{\ep}\leq 0$ for $r\geq r_2$.

\item The coefficient $ \AA_{\ep} $ is given in \eqref{eq:AAepsilontemper}. In particular $\AA_{\ep}>0$ for $r>r_+$ for $0\le \frac{|a|}{m} < 1$ and $\AA_{\ep}\sim \epsilon^{-1}r^{-5}\De^2$ near $r=r_+$.
\item The coefficient $\VV_{\ep}$ is given in \eqref{eq:VVepsilontemper}. Moreover, $\VV_{\ep}=\VV+\VV_{error}$ where $\VV_{{error}}$ is supported in a small neighborhood of the event horizon with $\lVert\VV_{{error}}\rVert_{L^1(r)}\les \ep$.\end{enumerate}
\end{proposition}

\subsection{Control of the lower order term}
\lab{sec:controloflowerterm}
In this section we will address the second issue discussed in section \ref{section:maindifficulties} concerning the control of the term $\psi^2$ in the region $r \in [r_+, r_3]$. We fix it by borrowing\footnote{With the help of Lemma \ref{lem:localHardy}.} from the $|\psi|^2$ term in the region $r \in [r_3, \infty)$ to control $|\psi|^2$ in $[r_+,r_3]$ at the additional cost of an additional term in $|\partial_r \psi|^2$. To achieve this, we need the control of $|e_3\psi|^2$ near the event horizon (since $\AA_{\ep}|\pr_r\psi|^2\sim r^{-3}\De^2|\pr_r\psi|^2$ only controls $|(\De e_3-e_4)\psi|^2$ near $r=r_+$). This can be achieved by a partial redshift estimate.

\subsubsection{Non-degenerate version of Proposition \ref{prop:Choice-zhftemper}}
\lab{sec:partialredshift}
We combine Proposition \ref{prop:Choice-zhftemper} with the partial red shift estimates of section \ref{section:PRS} to derive the following.

\begin{proposition}\lab{prop:partialredshift}
Let $Y=d(r)e_3$ with $d(r)$ defined as in \eqref{eq:choicedr}. Then there exist small constants $\delta_{red}>0, c_0>0$ such that
\begin{align*}
|q|^2\EE[X_{\ep}+c_0Y, w_{\ep},0]\geq \AA_{\ep}|\pr_r\psi|^2+\UU_{\ep}^{\a\b}\D_\a\psi\D_\b\psi+\VV_{\ep}|\psi|^2+\WW|e_3\psi|^2
\end{align*}
The coefficients $\AA_{\ep}, \UU^{\a\b}_{\ep},\VV_{\ep}, \WW$ satisfy the following
\begin{enumerate}

\item The principal term $ P=\UU_{\ep}^{\a\b}\D_\a\psi \D_\b \psi$ is given by
\beq\lab{eq:principal-termtemper-summary}
P=\frac12u_{\ep}\left(  -\frac{2\TT_{-a}}{ r^6}   O^{\a\b}\D_\a\psi\D_\b \psi + \frac{4ar(r^2+a^2)}{r^6} \That \psi\c \Zhat \psi\right).
\eeq
where $\TT_{-a}=r(r^2-3mr+2a^2)$ and $u_{\ep}\geq 0$ for $r_+\leq r\leq r_2$ and $u_{\ep}\leq 0$ for $r\geq r_2$.

\item $\AA_{\ep}>0$ for $r>r_+$ for $0\le \frac{|a|}{m} < 1$ and $\AA_{\ep}\sim \epsilon^{-1}r^{-5}\De^2$ near $r=r_+$.
\item The coefficient $\WW=\frac{c_0}{8}(r-m)d(r)$ satisfies that $\WW\gtrsim m1_{r_+\leq r\leq r_+(1+\de_{red})}$.
\item The coefficient $\VV_{\ep}$ is given in \eqref{eq:VVepsilontemper}. Moreover, $\VV_{\ep}=\VV+\VV_{{error}}$ where $\VV_{{error}}$ is supported in a small neighborhood of the event horizon with $\lVert\VV_{{error}}\rVert_{L^1(r)}\les \ep$.\end{enumerate}
\end{proposition}

\subsubsection{Local Hardy estimate}
\lab{sec:localHardy}
 Since we finally have good control of $|\pr_r \psi|^2$, we may essentially apply the fundamental theorem of Calculus to ``borrow'' from the $|\psi|^2$ term in the region $r \in [r_3, \infty)$ to control $|\psi|^2$ in this region at the additional cost of some of the $|\partial_r \psi|^2$ term. This technique, which we refer to as the \textit{local Hardy estimate}, is made precise by the lemma below.

\begin{lemma}\lab{lem:localHardy}
\textit{Let $V_1$ be a non-negative smooth function of $r$ supported on the interval $[r_0, r_3]$ with $r_0\geq r_+$ and let $V_2$ be any non-negative smooth function supported on an interval $[r_3, R]$ (for any $R> r_3$) and satisfying $\lVert V_2\rVert_{L^1(r)} = 2\lVert V_1\rVert_{L^1(r)}$. Then there exists a function $A_{12}$ depending on $V_1$ and $V_2$ such that}
\[
    \int V_1 \psi^2 dr \leq \int A_{12}(\partial_r \psi)^2 + V_2 \psi^2 dr.
\]
Furthermore, $A_{12}$ is supported on the interval $[r_+, R]$ and satisfies
\[
    \|A_{12}\|_{L^\infty} \le 2(R - r_3) \|V_1\|_{L^1(r)}.
\]
\end{lemma}
\begin{proof}
See Lemma 4.2.18 in \cite{St}.
\end{proof}

\subsubsection{Summary of our results so far}
Combining Proposition \ref{prop:partialredshift} with Lemma \ref{lem:localHardy}, we obtain the following proposition.
\begin{proposition}\lab{prop:morawetzwithoutThat}
Let $\ep$ be a sufficiently small constant such that $\ep\ll \de_{red}, c_0$ where $\de_{red}, c_0$ are the constants chosen in proposition \ref{prop:partialredshift}. With the choices
\[
X_{\ep}=-zu_{\ep} \pr_r,  \quad
w_{\ep} =- z \pr_r u_{\ep} ,\quad Y=d(r)e_3
\]
and $z=\frac{\De}{r^4}$ and $u_{\ep}$ defined as in \eqref{eq:defofhftemper} and $d(r)$ defined as in \eqref{eq:choicedr}, then we have
\begin{align*}
\int \EE[X_{\ep}+c_0Y, w_{\ep},0]
&\quad\geq \int |q|^{-2}\Big(\AA_{\ep}|\pr_r\psi|^2+\UU_{\ep}^{\a\b}\D_\a\psi\D_\b\psi+\VV_{\ep}|\psi|^2+\WW|e_3\psi|^2\Big).
\end{align*}
The coefficients $\AA_{\ep}, \UU^{\a\b}_{\ep},\VV_{\ep}, \WW$ satisfy the following
\begin{enumerate}

\item The principal term $ P=\UU_{\ep}^{\a\b}\D_\a\psi \D_\b \psi$ is given by
\beq\lab{eq:principal-termtemper}
P=\frac12u_{\ep}\left(  -\frac{2\TT_{-a}}{ r^6}   O^{\a\b}\D_\a\psi\D_\b \psi + \frac{4ar(r^2+a^2)}{r^6} \That \psi\c \Zhat \psi\right).
\eeq
where $\TT_{-a}=r(r^2-3mr+2a^2)$ and $u_{\ep}\geq 0$ for $r_+\leq r\leq r_2$ and $u_{\ep}\leq 0$ for $r\geq r_2$.

\item $\AA_{\ep}>0$ for $r>r_+$ for $0\le \frac{|a|}{m} < 1$ and $\AA_{\ep}\sim \epsilon^{-1}r^{-5}\De^2$ near $r=r_+$.
\item The coefficient $\WW$ satisfies that $\WW\gtrsim m1_{r_+\leq r\leq r_+(1+\de_{red})}$.
\item The coefficient $\VV_{\ep}$ satisfies that $\VV_{\ep}>0$ for $r\geq r_+$ and $\VV_{\ep}\sim mr^{-2}$ when $r\to\infty$.
 \end{enumerate}
\end{proposition}

To finish the proof of Theorem \ref{Thm:Moraw2} it remains to control the $ \That\psi \Zhat\psi$ term in $P$.


\subsection{Proof of  Theorem  \ref{Thm:Moraw2}}
\lab{section:smallmodified}

In this section, we address the third issue discussed in section \ref{section:maindifficulties} concerning the control of $|\That\psi|^2$. This can be remedied with the help of a new
current of the form
\[
\PP^{\flat}_\mu=\frac 1 2  w^{\flat} \psi  \D_\mu \psi -\frac 1 4\psi^2   \pr_\mu w^{\flat}
\]
which corresponds to a current associated to a zero vector field $X=0$ and a scalar function $w^{\flat}$ to be chosen.
In view of Proposition \ref{proposition:Morawetz1}, for $X=0$, $w=w^{\flat}, w_{red}=-w^{\flat}$, we
derive
\[
|q|^2 \EE[0, w^{\flat}, 0] =\frac 1 2 \De  \, w^{\flat}     |\pr_r\psi|^2+\frac 1 2   w^{\flat} \frac 1 \De \GG^{\a\b} \, \D_\a \psi \D_\b \psi-\frac 1 4 |q|^2 \square_\g  w^{\flat}|\psi|^2.
\]
with
\[
\GG^{\a\b} \, \D_\a \psi  \D_\b \psi  =  -(r^2+a^2) ^2 |\That\psi|^2 + \De O^{\a\b}\D_\a \psi \D_\b \psi.
\]
Thus, we have
\begin{align*}
|q|^2 \EE[0, w^{\flat}, 0] =\frac 1 2 \De  \, w^{\flat}     |\pr_r\psi|^2-\frac{w^{\flat}(r^2+a^2)^2}{ 2 \De}|\That \psi|^2 +\frac 1 2  w^{\flat} O^{\a\b}\D_\a\psi \D_\b \psi  -\frac 1 4|q|^2 \square_\g  w^{\flat} |\psi|^2.
\end{align*}

\subsubsection{Small  modified trapping term $P_\delta$}
\lab{sec:Smalltrapmodified}
For some sufficiently small constant $\delta>0$, we choose
\beq\lab{eq:definitionofwdelta}
w^{\flat}=w_\delta=-\delta\frac{2m\De\TT_{-a}^2}{r^{10}}
\eeq
where $\TT_{-a}=r(r^2-3mr+2a^2)$. Then adding $|q|^2\EE[0,w^{\flat},0]$ to the current $|q|^2\EE[X_{\ep}+c_0Y, w_{\ep},0]$ in Proposition \ref{prop:morawetzwithoutThat} yields
\begin{align*}
\int \Big( \EE[X_{\ep}+c_0Y, w_{\ep},0]+  \EE[0,w^{\flat},0] \Big)
&\geq \int|q|^{-2}\Big(\AA_\delta|\pr_r\psi|^2 +\PP_\delta+ \VV_\delta|\psi|^2 +\WW|e_3\psi|^2\Big)
\end{align*}
where
\beq
\lab{eq:PdeltaAAdeltaVVdelta}
\bsplit
P_\delta&=\delta m\frac{\TT^2_{-a}}{r^6}\frac{(r^2+a^2)^2}{r^4}|\That\psi|^2+\frac{\TT_{-a}}{2r^6}\Big(-u_{\ep}-\delta\frac{2m\De\TT_{-a}}{r^4}\Big)O^{\a\b}\D_\a\psi\D_\b\psi\\
&\quad+ \frac{2ar(r^2+a^2)}{r^6}u_{\ep}  \That\psi \c \Zhat\psi\\
\AA_\delta&=\AA_{\ep}+\frac 1 2\De w_\delta=\AA_{\ep}+O(\delta m\De^2r^{-4})\\
\VV_\delta&=\VV_{\ep} - \frac 1 4|q|^2 \square_\g w_\delta= \VV_{\ep}+O(\delta mr^{-2}).
\end{split}
\eeq
By choosing $\delta$ sufficiently small, it is easy to check that
\[
\AA_\delta\gtrsim m\De^2r^{-4}, \quad  \VV_{\de}\gtrsim mr^{-2}.
\]
By the Cauchy-Schwarz inequality we derive
\[
P_\delta\gtrsim_\delta \frac{\TT_{-a}^2}{r^{7}}\Big(O^{\a\b}\D_\a\psi\D_\b\psi+mr|\That\psi|^2\Big)-\frac{a^2}{mr^2}|\Zhat\psi|^2.
\]

We summarize the result in the following
\begin{proposition}
\lab{proposition:smallmodifiedtrapping}
Consider the current $\EE_\delta=\EE+\EE^{\flat}$ where $\EE=\EE[X_{\ep}+c_0Y, w_{\ep},0]$ as in Proposition \ref{prop:morawetzwithoutThat} and $\EE^{\flat}=\EE[0,w_\delta, 0]$ with $w_\delta$ defined in \eqref{eq:definitionofwdelta}. Then,
\beq
\lab{eq:identity-forEEdelta}
\int \EE_\delta\geq \int |q|^{-2}\Big(\AA_\delta|\pr_r\psi|^2 +P_\delta+\VV_\delta|\psi|^2 + \WW|e_3\psi|^2\Big)
\eeq
with $P_\delta,\AA_\delta, \VV_\delta$ given by \eqref{eq:PdeltaAAdeltaVVdelta} and $\WW\gtrsim m1_{r_+\leq r\leq r_+(1+\de_{red})}$. Moreover, we have
\beq
\bsplit
\int \EE_\delta&\gtrsim_\delta \int \frac{m \De^2}{r^6} \big|\pr_r\psi|^2 +\frac{\TT_{-a}^2}{r^9}\Big(O^{\a\b}\D_\a\psi\D_\b\psi+mr|\That\psi|^2\Big)+m1_{r_+\leq r\leq r_+(1+\de_{red})}|e_3\psi|^2\\
&+ \int \frac{m}{r^4}|\psi|^2-\frac{a^2}{mr^4}|\Zhat\psi|^2
\end{split}.
\eeq
\end{proposition}
To finish the proof of Theorem \ref{Thm:Moraw2} we integrate $\EE_\delta$ and apply the divergence Lemma in $\DD(\tau_1,\tau_2 )$. This will generate the boundary terms which we discuss in the following.

\subsubsection{Control of the boundary terms}
\lab{section:boundary}
Recall that $\EE_\delta=\EE[X_{\ep}+c_0Y, w_{\ep},0]+\EE[0,w_\delta, 0]$, with $\EE, \PP$ defined in \eqref{definition-EE-gen1} and \eqref{eq:Gen-current}. By applying the divergence Lemma and using the partial red shift estimates\footnote{Note that the red shift estimate in \ref{section:PRS} generates positive boundary terms on $\pr^+\DD(\tau_1, \tau_2)$.} of section \ref{section:PRS}, we conclude that $\int_{\DD(\tau_1,\tau_2)}\EE_\delta$ can be estimated in terms of the boundary terms at $\pr\DD(\tau_1,\tau_2)$ generated by $X_{\ep}, w_{\ep}, w_{\de}$ and that at $\Si_{\tau_1}$ generated by $c_0Y$. Since
\[
X_{\ep}=O(1)\Rhat,\quad  w_{\ep}, w_{\de}=O(\frac{\De}{r^3}),\quad Y=O(1)e_3,
\]
and $\Rhat=\That$ at $\HH_+$, see \eqref{eq:formulasThat-Rhat},
the boundary terms can be controlled by
\beq
\EFdeg[\psi](\tau_1,\tau_2)+E[\psi](\tau_1).
\eeq
 This ends the proof of Theorem \ref{Thm:Moraw2}.


\subsection{Proof of Theorem \ref{Thm:Moraw3} }


In this section, we will prove Theorem \ref{Thm:Moraw3}. In order to achieve this, we make use of the generalized current
\[
|q|^2(\EE[X=-zhf\pr_r, w=-z\pr_r(hf), 0]+\EE[0, w^{\flat}, 0])
\]
with appropriate choices of $z, h, f, w^{\flat}$. More specifically, we again choose $z=\frac{\De}{(r^2+a^2+a\lz)^2}=\frac{\De}{r^4}$ with $\lz=-a$.

\begin{remark}\lab{rem:choice-a}
There are at least two reasons for making the choice $\lz=-a$.
\begin{enumerate}

\item The main reason for the choice $\lz=-a$ is that the trapping polynomial
\bea
\lab{eq:TT_{a}}
\TT_{-a} :=\TT +a^2 (r-m)
\eea
is closely related to the determination of the $r$-range of trapped null geodesics in section
\ref{section-range.trapped.ngeod}.

More specifically, consider the trapping term (see \eqref{eq:Pterm} in Lemma \ref{Lemma:compute-RRtp-z_lz}) in the generalized current $|q|^2\EE[-zhf\pr_r, -z\pr_r(hf), 0]$ with choice $z=\frac{\De}{(r^2+a^2+a\lz)^2}, f=\pr_rz$
\[
h\frac{\TT_{\lz}}{\Pi^3}\Big( \frac{2\TT_{\lz}}{ \Pi^3}   O^{\a\b}\D_\a\psi\D_\b \psi - \frac{4ar(r^2+a^2)}{\Pi^3} \That \psi\c Z_{\lz} \psi\Big).
\]
We will need to modify this by adding a multiple of the Lagrangian $\g^{\mu\nu}\pr_\mu\psi\pr_\nu \psi$ so that we also get
a $| \That\psi|^2$ term. In fact, the multiple has to be chosen as $-h\frac{\TT^2_{\lz}}{\Pi^6}$ such that we obtain
\[
\frac{P[\lz]}{h}:=\frac{\TT_{\lz}}{\Pi^3}\Big(\frac{\TT_{\lz}}{\Pi^3}O^{\a\b}\D_\a\psi\D_\b\psi- \frac{4ar(r^2+a^2)}{\Pi^3}\That\psi\c Z_{\lz}\psi+\frac{\TT_{\lz}}{\Pi^3}\frac{(r^2+a^2)^2}{\Delta}|\That\psi|^2\Big).
\]
Note that for $\lz=-a$,
\[
\frac{P[-a]}{h}=\frac{\TT^2_{-a}}{r^{12}}\Big(|q|^2 |\nab_{e_1} \psi|^2+ |q|^2|\nab_{e_2} \psi|^2\Big)- \frac{4ar(r^2+a^2)\TT_{-a}}{r^{12}}\That\psi\c \Zhat \psi+\frac{\TT^2_{-a}}{r^{12}}\frac{(r^2+a^2)^2}{\Delta}|\That\psi|^2.
\]
When\footnote{Recall that the same choice was used in \eqref{eq:PolynTheta}.} $\sin\th=1$, we have $\Zhat=|q| e_2$, see \eqref{eq:identityZhat}, and therefore
\[
\frac{P[-a]}{h}=\frac{\TT^2_{-a}}{r^{12}}|q|^2|\nab_{e_1}\psi|^2+\begin{pmatrix}
\That\psi& |q|\nab_{e_2}\psi
\end{pmatrix}
\begin{pmatrix}
\frac{\TT^2_{-a}}{r^{12}}\frac{(r^2+a^2)^2}{\Delta}&-\frac{2ar(r^2+a^2)\TT_{-a}}{r^{12}}\\
-\frac{2ar(r^2+a^2)\TT_{-a}}{r^{12}} &\frac{\TT^2_{-a}}{r^{12}}
\end{pmatrix}
\begin{pmatrix}
\That\psi \\
|q|\nab_{e_2}\psi
\end{pmatrix}.
\]
Then the determinant of the above matrix is given by $\frac{\TT_{-a}^2(r^2+a^2)^2}{r^{24}\Delta}\big(\TT_{-a}^2-4a^2r^2\Delta\big)$ where $\TT_{-a}^2-4a^2r^2\Delta$ is precisely the polynomial $\Th(r)$ in \eqref{eq:PolynTheta}. Therefore in the nontrapping region, $P[-a]$ gives control of $\That\psi, |q|\nab\psi$, consistent with the result of Theorem \ref{Thm:Moraw3}.
\item Compared to other choices of $\lz$, $\Zhat=Z_{-a}=Z+aT$ can be written as a special linear combination of $\That$ and $|q|e_2$, see \eqref{eq:identityZhat} in Lemma \ref{lemma:TtRtandZ}. This fact is very useful in analyzing the positivity of various quadratic forms in what follows.
\end{enumerate}
\end{remark}


\subsubsection{Large modified trapping term $\Pt$}
\lab{section:Largemodified}

We choose $z=\frac{\Delta}{r^4}, f=\frac{-2\TT_{-a}}{r^6}, h=r^5$ in $\EE[-zhf\pr_r, -z\pr_r(hf), 0]$ and $w^{\flat}$ in $\EE[0, w^{\flat}, 0]$ to be a smooth function satisfying 
\beq\lab{eq:definitionofw'}
w^{\flat}=\begin{cases}
w^{\flat}_1=-\frac{4m^2r_+^2}{(r_+-m)^2}\frac{\De\TT_{-a}^2}{r^{11}}\quad &r\leq r_++\delta_1\\
w^{\flat}_2=- \frac{2\TT_{-a}^2}{r^7}\quad &r_++2\delta_1\leq r\leq R\\
w^{\flat}_3=-\frac{2m}{r^2}\quad &r\geq R+1
\end{cases}
\eeq
where $\delta_1$ is a sufficiently small constant depending on $a/m$ and $R$ is a sufficiently large constant to be chosen later. Then we have
\begin{align*}
|q|^2\Big( \EE[-zhf\pr_r, -z\pr_r(hf),0]+  \EE[0, w^{\flat},0] \Big)
&=\AAt|\pr_r\psi|^2 +\PPt+ \widetilde{\VV}|\psi|^2 
\end{align*}
where
\beq
\lab{eq:PTAAtVVt}
\bsplit
\Pt&=\begin{cases}
&-\frac{w^\flat_1 (r^2+a^2)^2}{ 2 \De}|\That\psi |^2 +\left(\frac{2\TT_{-a}^2}{r^7} +\frac 1 2 w^\flat_1\right)O^{\a\b} \D_\a\psi \D_\b\psi+ \frac{4ar(r^2+a^2)\TT_{-a}}{r^7}  \That\psi \c \Zhat\psi    \quad r\leq r_++\delta_1\\
&\frac{\TT_{-a}^2 (r^2+a^2)^2}{ \Delta r^7}|\That\psi |^2 +\frac{\TT_{-a}^2}{r^7} O^{\a\b} \D_\a\psi\D_\b\psi -\frac{4ar(r^2+a^2)\TT_{-a}}{r^7}  \That\psi\c \Zhat\psi\quad  r_++2\delta_1\leq r\leq R\\
&\frac{2m (r^2+a^2)^2}{ 2r^2 \De}|\That\psi |^2 +\left(\frac{2\TT_{-a}^2}{r^7} -\frac{m}{r^2}\right)O^{\a\b} \D_\a\psi \D_\b\psi+ \frac{4ar(r^2+a^2)\TT_{-a}}{r^7}  \That\psi \c \Zhat\psi\quad r\geq R+1
\end{cases}\\
\AAt&
=O(mr^{-4}\Delta^2),\quad
\widetilde{\VV}
= O(mr^{-2}).
\end{split}
\eeq

\begin{remark}
Note that we do not need to derive positive lower bounds here for $\widetilde{\AA}, \widetilde{\VV}$.
The desired bounds will follow from the consideration in the next section in the context of the $\SS$-derivative estimates.
\end{remark}


We summarize the result in the following
\begin{proposition}
\lab{proposition:largemodifiedtrapping}
Consider the current $\widetilde{\EE}=\EE+\EE^{\flat}$ where $\EE=\EE[X, w, M]$ as in Proposition \ref{proposition:Morawetz1} (with $z=\frac{\Delta}{r^4}, f=\frac{-2\TT_{-a}}{r^6}, h=r^5, M_{\mu}=0$)
and $\EE^{\flat}=\EE[0,w^{\flat}, 0]$ with $w^{\flat} $ defined in \eqref{eq:definitionofw'}. Then,
\beq
\lab{eq:identity-forEEt}
|q|^2\widetilde{\EE}=\AAt|\pr_r\psi|^2 +\PPt+ \widetilde{\VV}|\psi|^2
\eeq
with $\Pt,\AAt, \widetilde{\VV}$ given by \eqref{eq:PTAAtVVt}.
\end{proposition}

\subsubsection{Coercivity of  $\Pt$ away from the trapping region}

\begin{proposition}\lab{proposition:coercivitynontrap}
The quadratic form $\Pt$ is positive in the non-trapping region $[r_+, \rhat_1)\cup(\rhat_2, \infty)$.
Moreover $\Pt$ controls the derivatives $|\That\psi|^2, |\nab \psi|^2$ outside a small neighborhood of the trapping region
$[\rhat_1, \rhat_2]$.
\end{proposition}
\begin{proof}
Before starting, we recall that the trapping region for the Kerr metric $\g_{a, m}$ has been identified in Proposition \ref{Prop:r-trapped.region} to be included in the interval $[\rhat_1, \rhat_2]$ where the two $r$-values are the exterior roots of the polynomial $\Th(r)$ introduced in \eqref{eq:PolynTheta}.

We analyze the new trapping term in \eqref{eq:PTAAtVVt}
\begin{align*}
\Pt&:=-\frac{w^{\flat} (r^2+a^2)^2}{ 2 \De}|\That\psi |^2 +\left(\frac{2\TT_{-a}^2}{r^7} +\frac 1 2 w^{\flat}\right)O^{\a\b} \D_\a\psi \D_\b\psi+ \frac{4ar(r^2+a^2)\TT_{-a}}{r^7}  \That\psi \c \Zhat\psi
\end{align*}
where $w^{\flat}$ is the smooth function introduced in \eqref{eq:definitionofw'}. Since according to \eqref{eq:identityZhat}
\[
\Zhat=Z+aT=\frac{(r^2+a^2)a\cos^2\th}{|q|^2}\That+\frac{r^2\sin\th}{|q|^2}|q|e_2,
\]
we rewrite
\[
\Pt=F_1|\That\psi |^2+2F_2\That\psi\c|q|e_2\psi+F_3\Big(|q|^2|e_2(\psi)|^2+|q|^2|e_1(\psi)|^2\Big)
\]
where
\begin{align*}
F_1&=-\frac{w^{\flat} (r^2+a^2)^2}{ 2 \De}-\frac{\TT_{-a} (r^2+a^2)^2}{r^7}\frac{4ra^2\cos^2\th}{|q|^2},\\
F_2&=-\frac{2a(r^2+a^2)\TT_{-a}\sin\th}{r^4|q|^2},\quad
F_3=\frac{2\TT_{-a}^2}{r^7} +\frac 1 2 w^{\flat}.
\end{align*}
Determining the region where the quadratic form $\Pt$ is positive definite is equivalent to identifying the region where
\beq\label{eq:F1anddeterminant}
F_1(\textrm{ or } F_3)>0\quad\textrm{and}\quad D=\det\begin{pmatrix}
F_1& F_2\\
F_2&F_3
\end{pmatrix}>0.
\eeq
According to the definition of $w^{\flat}$ in \eqref{eq:definitionofw'}, it suffices to prove the following statements
\begin{enumerate}
\item[(1)] When $w^{\flat}=w_3^{\flat}$, \eqref{eq:F1anddeterminant} holds true for $r\geq R'$ where $R'$ is a sufficiently large constant.
\item[(2)] When $w^{\flat}=w_1^{\flat}$, \eqref{eq:F1anddeterminant} holds true at $r=r_+$, and thus on $[r_+, r_++\delta']$ for some $\delta'>0$.
\item[(3)] When $w^{\flat}=w_2^{\flat}$, \eqref{eq:F1anddeterminant} holds true on $[r_+, \rhat_1)\cup (\rhat_2, \infty)$
\end{enumerate}
and then choose $\delta_1=\delta'/2, R=R'$.

For part (1), since $w_3^{\flat}=\frac{-2m}{r^2}$, we have for $r\gg1$
\[
F_1\gtrsim m,\quad F_2\sim ar^{-1},\quad F_3\gtrsim r^{-1}
\]
and thus $F_1>0$ and $D=F_1F_3-F_2^2\gtrsim mr^{-1}>0$.

For part (2), when $w^{\flat}=w^{\flat}_1=-\frac{4m^2r_+^2}{(r_+-m)^2}\frac{\Delta\TT_{-a}^2}{r^{11}}$, using the fact $\TT_{-a}|_{r=r_+}=(2r\De-r^2(r-m))|_{r=r_+}=-r_+^2(r_+-m)$ we compute
\begin{align*}
F_1|_{r=r_+}&\geq -\frac{w^{\flat} (r^2+a^2)^2}{ 2 \De}|_{r=r_+}=\frac{2m^2(r_+^2+a^2)^2}{r_+^5}>0\\
F_2|_{r=r_+}&=\frac{2a(r_+^2+a^2)(r_+-m)\sin\th}{r_+^2|q|^2},\quad F_3|_{r=r_+}=\frac{2(r_+-m)^2}{r_+^3}
\end{align*}
and thus
\begin{align*}
D|_{r=r_+}&=(F_1F_3-F_2^2)|_{r=r_+}=\frac{4(r_+^2+a^2)^2(r_+-m)^2}{r_+^4}\left(\frac{m^2}{r_+^4}-\frac{a^2\sin^2\th}{|q|^4}\right)\\
&\geq\frac{4(r_+^2+a^2)^2(r_+-m)^2}{r_+^4}\c  \frac{m^2-a^2}{r_+^4}>0.
\end{align*}
This proves part (2).

For part (3), when $w^{\flat}=w^{\flat}_2=\frac{-2\TT_{-a}}{r^7}$, we have
\[
F_1=\frac{\TT_{-a} (r^2+a^2)^2}{  r^7}\Big(\frac{\TT_{-a}}{\De}-\frac{4ra^2\cos^2\th}{|q|^2}\Big),\quad
F_2=-\frac{2a(r^2+a^2)\TT_{-a}\sin\th}{r^4|q|^2},\quad
F_3=\frac{\TT_{-a}^2}{r^7}
\]
Then the determinant is given by
\begin{align*}
D&=F_1F_3-F_2^2=\frac{\TT^3_{-a} (r^2+a^2)^2}{  r^{14}}\Big(\frac{\TT_{-a}}{\De}-\frac{4ra^2\cos^2\th}{|q|^2}\Big)-\frac{4a^2r^6(r^2+a^2)^2\TT_{-a}^2\sin^2\th}{r^{14}|q|^4}\\
&=\frac{\TT^2_{-a} (r^2+a^2)^2}{  \Delta r^{14}}\Big(\TT_{-a}^2-\frac{4a^2r\TT_{-a}\Delta\cos^2\th}{|q|^2}-\frac{4a^2r^6\Delta\sin^2\th}{|q|^4}\Big).
\end{align*}
As in Proposition \ref{Prop:r-trapped.region} let $r_2$ be the largest root of
\beaa
r^{-1}\TT_{-a}= r^{-1}\big( \TT +a^2 (r-m)\big)=\big( r^2 - 3mr +2a^2\big)= (r-r_1)(r-r_2)
\eeaa
so that \footnote{Because $r_1\le r_+$.} $\TT_{-a}(r)\le 0$ for $r\in [r_+, r_{2}]$ and $\TT_{-a}(r)\geq 0$ for $[r_2,\infty)$.

{\bf Case $r\in [r_+, r_2]$.} Since $\TT_{-a}(r)\leq 0$,
\begin{align*}
D&\geq \frac{\TT^2_{-a}(r^2+a^2)^2}{  \Delta r^{14}}\Big(\TT_{-a}^2-\frac{4a^2r^6\Delta\sin^2\th}{|q|^4}\Big)\geq \frac{\TT^2_{-a} (r^2+a^2)^2}{  \Delta r^{14}}\Big(\TT_{-a}^2-\frac{4a^2r^6\Delta}{r^4}\Big)\\
&=\frac{\TT^2_{-a} (r^2+a^2)^2}{  \Delta r^{14}}\Big(\TT_{-a}^2-4a^2r^2\Delta \Big)=\frac{\TT^2_{-a} (r^2+a^2)^2}{  \Delta r^{14}}\Theta(r)
\end{align*}
with $\Th(r)$ the polynomial defined in Remark \ref{remark-Theta}. According to Proposition \ref{Prop:r-trapped.region} we know that $\Theta(r)>0$ and $H_3(r)>0$ on $[r_+, \rhat_1)$ where $\rhat_1$ is the lower bound of the trapping region $[\rhat_1, \rhat_2]$.

{\bf Case $r\in [r_2, \infty)$.} Since $\TT_{-a}(r)\geq 0$, we have
\begin{align*}
D&=\frac{\TT^2_{-a} (r^2+a^2)^2}{  \Delta r^{14}}\Big(\TT_{-a}^2-\frac{4a^2r\TT_{-a}\Delta\cos^2\th}{|q|^2}-\frac{4a^2r^6\Delta\sin^2\th}{|q|^4}\Big)\\
&\geq\frac{\TT^2_{-a} (r^2+a^2)^2}{  \Delta r^{14}}\Big(\TT_{-a}^2-\frac{4a^2\TT_{-a}\Delta\cos^2\th}{r}-4a^2r^2\Delta\sin^2\th\Big)\\
&=\frac{\TT^2_{-a} (r^2+a^2)^2}{  \Delta r^{14}}\Big(\TT_{-a}^2-\frac{4a^2\TT_{-a}\Delta}{r}+\frac{4a^2\Delta\big(\TT_{-a}-r^3\big)}{r}\sin^2\th\Big).
\end{align*}
Since $\TT_{-a}-r^3=-3mr^2+2a^2r<0$ on $[r_+,\infty)$, we derive
\begin{align*}
D&\geq \frac{\TT^2_{-a} (r^2+a^2)^2}{  \Delta r^{14}}\Big(\TT_{-a}^2-\frac{4a^2\TT_{-a}\Delta}{r}+\frac{4a^2\Delta\big(\TT_{-a}-r^3\big)}{r}\sin^2\th\Big)\Big|_{\sin^2\th=1}\\
&=\frac{\TT^2_{-a} (r^2+a^2)^2}{  \Delta r^{14}}\Big(\TT_{-a}^2-4a^2r^2\Delta \Big)=\frac{\TT^2_{-a} (r^2+a^2)^2}{  \Delta r^{14}}\Theta(r).
\end{align*}
Appealing to Proposition \ref{Prop:r-trapped.region} again, we see that $\Theta(r)>0$ and $F_3(r)>0$ on $(\rhat_2, \infty)$.

Therefore, with the choice of $w^{\flat}$ in \eqref{eq:definitionofw'} we conclude that
\[
\widetilde{P}=F_1|\That\psi |^2+2F_2\That\psi\c|q|e_2\psi+F_3\big(|q|^2|e_1(\psi)|^2+|q|^2|e_2(\psi)|^2\big)
\]
is positive definite on $[r_+, \rhat_1)\cup(\rhat_2, \infty)$. Moreover, $\widetilde{P}$ controls the derivatives $|\That\psi|^2, |\nab\psi|^2$ outside the trapping region $[\rhat_1, \rhat_2]$.
\end{proof}

The proof of Theorem \ref{Thm:Moraw3} follows now by combining Propositions \ref{proposition:largemodifiedtrapping} and \ref{proposition:coercivitynontrap}, integrating over $\DD(\tau_1,\tau_2)$, and controlling the boundary terms as in section \ref{section:boundary}.

\section{ $\SS$-derivative Morawetz estimates}\lab{section-proof-mor-2}


In this section we prove $\SS$-derivative Morawetz estimates stated in Theorem
            \ref{Thm:Moraw4}.


\subsection{Basic spacetime  $\SS$-valued  identity}
\lab{section:SS-valuedidentity}


Recall, see \eqref{def:SS_aund}, the definition of the second order symmetry operators
\[
\SS_1\psi= T  T \psi,  \qquad \SS_2 \psi=  T Z \psi, \qquad \SS_3 \psi= ZZ  \psi, \qquad \SS_4\psi= \OO \psi
\]
commuting with $|q|^2 \square$.

\begin{definition}[Generalized Current]
\lab{def:generalizedcurrent} Let $\X$ be a double-indexed collection of vector fields $\X=\{ X^{\underline{a} \underline{b}}\}$, $\bold{w}$ be a double-indexed collection of functions $\bold{w}=\{ w^{\underline{ab}} \}$, and $\M=\{M^{\aund\bund} \}$ a double-indexed collection of $1$-forms, all symmetric in the indices $\aund, \bund$.

Consider a solution $\psi$ of the equation $\square\psi=N$ and let $\psi_\aund=\SSa\psi$ verify $\square \psia=|q|^{-2}\SS_\aund |q|^2N$. The generalized current
$\PP_\mu = \PP_\mu^{(\bold{X}, \bold{w}, \M)} [\psi] $ associated to $\psia, \psib$ is given by
\beq
\lab{eq:generalizedcurrent}
\begin{split}
\PP_\mu&= \QQ[\psi]_{\underline{ab} \mu \nu} X^{\underline{ab} \nu}+\frac 1 4  {w^{\underline{ab}} \, \big( \D_\mu  \psia\c  \psib+\D_\mu  \psib\c  \psia  \big)}\\
&-\frac 1 4 (\partial_\mu w^{\underline{ab}})\psia\c\psib  +\frac 1 4 M^{\aund\bund}_\mu\psi_\aund\c \psi_\bund.
\end{split}
\eeq
where $\QQ[\psi]_{\underline{ab} \mu \nu} =\QQ_{\mu\nu}(\psi_\aund, \psi_\bund)$ and
\begin{align*}
\QQ_{\mu\nu}(\psi_\aund, \psi_\bund)
&=  {\frac 1 2 \big( \D_\mu  \psi_\aund \c \D_\nu \psi _\bund+\D_\mu  \psi_\bund \c \D_\nu \psi _\aund\big)}    -\frac 12 \g_{\mu\nu}  \LL[\psi_\aund, \psi_\bund], \\
\LL[\psi_\aund, \psi_\bund]&=\g^{\a\b} \D_\a \psi_\aund\c\D_\b  \psi _\bund
.
\end{align*}
\end{definition}

We also define, in analogy with \eqref{definition-EE-gen1}, the modified divergence
\beq\lab{definition-EE-gen-SSvalued}
\begin{split}
\EE[\X, \w, \M] &:= \D^\mu  \PP_\mu[\X, \w, \M]- \NN[\X, \w], \\
\NN[\X, \w]&:={\frac12\Big( X^{\aund\bund}\psia+\frac 1 2   w^{\aund\bund} \psia\Big)\c \square\psib }+{\frac12\Big(X^{\aund\bund}\psib+\frac 1 2   w^{\aund\bund} \psib\Big)\c \square\psia }.
\end{split}
\eeq

Using the above notation we derive the following analogue of Proposition \ref{proposition:Morawetz1}.
\begin{proposition}
\lab{proposition:Morawetz3}
The following hold true
\begin{enumerate}
\item If we choose
\[
X^{\aund\bund} =\FF^{\aund\bund} \pr_r, \qquad w^{\aund\bund}=  |q|^2 \Div \big( |q|^{-2}  X ^{\aund\bund} \big)-w_{red}  ^{\aund\bund},
\]
then the generalized current defined in \eqref{definition-EE-gen-SSvalued} verifies the identity
\beq\label{generalized-current-operator}
\begin{split}
|q|^2\EE[\bold{X}, \bold{w}, \bold{M}]   &=\AA^{\aund\bund}  \pr_r\psia\c \pr_r\psib + \UU^{\a\b\aund\bund} \, \D_\a \psia \c \D_\b \psib  +\VV^{\aund\bund} \psia\c\psib \\
& +\frac 1 4 |q|^2 \D^\mu\Big(M^{\aund\bund}_\mu\psi_\aund\c \psi_\bund \Big)
\end{split}
\eeq
where, recalling the definition of $\GG^{\a\b} $ in \eqref{eq:expressionRR-O} (see also
 \eqref{eq:RR-Sa}, \eqref{components-RR-aund}),
\beq
\lab{eq:AAUUVV}
\bsplit
\AA^{\aund\bund}&= \De \pr_r \FF^{\aund\bund}- \frac 1 2 \FF^{\aund\bund}\pr_r \De-\frac 1 2 \De w^{\aund\bund}_{red},\\
\UU^{\a\b\aund\bund}&=  -\frac 1 2  \FF^{\aund\bund}\pr_r \left(\frac 1 \De\GG^{\a\b}\right)-\frac 1 2   w^{\aund\bund}_{red}\frac 1 \De \GG^{\a\b},\\
\VV^{\aund\bund}&= -\frac 1 4|q|^2 \square_\g  w^{\aund\bund}.
\end{split}
\eeq

\item If in addition we choose, for functions $z$ and $h$, and a double-indexed function $f^{\aund\bund}$
\beq\lab{choice-FF-w-operator}
\FF^{\aund\bund}=- z h f^{\aund\bund}, \qquad   w^{\aund\bund} =- z \pr_r \big( h  f^{\aund\bund}  \big), \qquad            w^{\aund\bund} _{red}=  \FF^{\aund\bund}  z^{-1}\partial_r z,
\eeq
then, recalling \eqref{eq:RR-Sa}, 

\beq
\lab{eq:AAUUVV-again}
\bsplit
\UU^{\a\b\aund\bund}&=   \frac{ 1}{2}  h f^{\aund\bund} \pr_r\left( \frac z \De\GG^{\a\b}\right),\\
\AA^{\aund\bund}&=-z^{1/2}\Delta^{3/2} \partial_r\left(h \frac{ z^{1/2}  f^{\aund\bund} }{\Delta^{1/2}}  \right),
\\
\VV^{\aund\bund}&=\frac 1 4 \pr_r \left(\De \pr_r \Big( z \pr_r \big( h  f^{\aund\bund} \big) \Big)\right).
\end{split}
\eeq

\item If $M^{\aund\bund} = v^{\aund\bund}(r,\th) \pr_r$, for a double-indexed function $v=v^{\aund\bund}(r,\th)$, we have
\beq\label{expression-Div-M}
\bsplit
&\quad\frac 1 4 |q|^2 \Div\big( (\psia\c\psib)M^{\aund\bund} \big)\\
&= \frac 1 4 |q|^2\left( 2 v^{\aund\bund}(r)\psia\c \pr_r \psib + \left(\pr_r v^{\aund\bund}+ \frac{2r}{|q|^2} v^{\aund\bund}\right) \psia\c\psib \right).
\end{split}
\eeq
\end{enumerate}
\end{proposition}

\begin{proof}
The proof follows exactly the same steps as in the proof of Proposition \ref{proposition:Morawetz1}. 
\end{proof}

\subsubsection{First choice $(\bold{X}, \bold{w},0)$}
\lab{sect:X,w.}
Following \eqref{choice-FF-w-operator}
\beq
\lab{Choice-X,w}
\bold{X}^{\aund\bund}=\FF^{\aund\bund}\pr_r,\quad\FF^{\aund\bund}=-zhf^{\aund\bund},\quad w^{\aund\bund}=-z\pr_r(hf^{\aund\bund}), \quad  w^{\aund\bund} _{red}=  \FF^{\aund\bund} z^{-1}\partial_r z.
\eeq
Then we have\footnote{Recall also the decomposition \eqref{eq:RR-Sa} and the definition of $ \RRtp^\cund[z]=\pr_r\Big( \frac{z}{\De} \GG^\cund\Big)$. }
\beq\lab{eq:EEboldXw}
|q|^2\EE[\bold{X}, \bold{w},0]=\frac12hf^{\aund\bund}\RRtp^{\cund}[z]S_{\cund}^{\a\b}\D_\a\psia\c\D_\b\psib+\AA^{\aund\bund}\pr_r\psia\c\pr_r\psib+\VV^{\aund\bund}\psia\c\psib
\eeq
where
\begin{align*}
\AA^{\aund\bund}&=-z^{1/2}\Delta^{3/2} \partial_r\left(h \frac{ z^{1/2}  f^{\aund\bund} }{\Delta^{1/2}}  \right),
\\
\VV^{\aund\bund}&=\frac 1 4 \pr_r \left(\De \pr_r \Big( z \pr_r \big( h  f^{\aund\bund} \big) \Big)\right).
\end{align*}

\subsubsection{Second choice       $(0, \bold{w}_{\flat}, 0)$    }
\lab{sect:w'}
We next consider a new general current of the form
\[
\PP^{\flat}_\mu[0,\bold{w}_{\flat},0]=\frac14w_{\flat}^{\aund\bund}(\D_\mu\psia\c\psib+\D_\mu\psib\c\psia)-\frac14(\pr_\mu w_{\flat}^{\aund\bund})\psia\c\psib.
\]
Then we have
\beq
\lab{eq:DivPP'}
|q|^2\EE[0,\bold{w}_{\flat},0]=\frac12w_{\flat}^{\aund\bund}\frac{1}{\De}\GG^{\cund}S_{\cund}^{\a\b}\D_\a\psia\c\D_\b\psib+\frac12\De w_{\flat}^{\aund\bund}\pr_r\psia\c\pr_r\psib-\frac14|q|^2(\square_{\g}w_{\flat}^{\aund\bund})\psia\c\psib.
\eeq
By summing the above to \eqref{eq:EEboldXw} we obtain
\beq\lab{eq:EEwt}
\bsplit
|q|^2\widetilde{\EE}&=|q|^2(\EE[\bold{X},\bold{w},0]+\EE[0,\bold{w}_{\flat},0])\\
&=\widetilde{\UU}^{\a\b\aund\bund}\D_\a\psia\c\D_{\b}\psib+\widetilde{\AA}^{\aund\bund}\pr_r\psia\c\pr_r\psib+\widetilde{\VV}^{\aund\bund}\psia\c\psib
\end{split}
\eeq
where
\beq
\lab{eq:coefficientsofEEwt}
\bsplit
\widetilde{\UU}^{\a\b\aund\bund}&=\frac12 hf^{\aund\bund}\GGtp^{\cund}[z]S_{\cund}^{\a\b}+\frac12\frac{w_{\flat}^{\aund\bund}}{\De}\GG^{\cund}S_{\cund}^{\a\b},\\
\widetilde{\AA}^{\aund\bund}&=\AA^{\aund\bund}+\frac12\De w_{\flat}^{\aund\bund}=-z^{1/2}\Delta^{3/2} \partial_r\left(h \frac{ z^{1/2}  f^{\aund\bund} }{\Delta^{1/2}}  \right)+\frac12\De w_{\flat}^{\aund\bund},\\
\widetilde{\VV}^{\aund\bund}&=\VV^{\aund\bund}-\frac14\pr_r \left(\De \pr_r w_{\flat}^{\aund\bund}\right)=\frac 1 4 \pr_r \left(\De \pr_r \Big( z \pr_r \big( h  f^{\aund\bund} \big)\Big)\right)-\frac14\pr_r \left(\De \pr_r w_{\flat}^{\aund\bund}\right).
\end{split}
\eeq

\subsubsection{Choice of $f^{\aund\bund}$ and $w^{\aund\bund}_{\flat}$.}

In what follows, we choose the double-indexed function $f^{\aund\bund}$ and $w_{\flat}^{\aund\bund}$
to ensure the positivity of the generalized current \eqref{eq:EEwt}.
More precisely\footnote{In order to obtain control of $\psi_1=TT\psi$ and $\psi_4=\OO\psi$, we define $\UUwtp^\aund$ as a combination of $\GGtp^\aund[z]$ and $\GG^\aund$ (because $\GGtp^1=0$ gives no information on $\psi_1$). Once we fix the choice of $\UUwtp^\aund$, the choice of $w_{\flat}^{\aund\bund}$ is naturally made to guarantee that $\widetilde{\UU}^{\a\b\aund\bund}\D_\a\psia\c\D_{\b}\psib$ is nonnegative.}
\beq\lab{choice-f-operator}
\bsplit
f^{\underline{a}\underline{b}}&= \UUwtp^{(\underline{a}} \LL^{\underline{b})}= \frac 1 2 \big(\UUwtp^{\aund}\LL^{\bund}+\UUwtp^{\bund}\LL^{\aund} \big),\quad \UUwtp^{\aund}=(1+\de_0 m^2z)\GGtp^{\aund}[z]-\delta_0m^2\GGtp^4[z]\frac{\GG^{\aund}}{r^4},\\
w_{\flat}^{\aund\bund}&=h\De(\de_0 m^2\frac{\GGtp^4[z]}{2r^4})^2
\GG^{(\aund}\LL^{\bund)}=h\De(\de_0 m^2\frac{\GGtp^4[z]}{2r^4})^2
\c\frac12\big(\GG^{\aund}\LL^{\bund}+\GG^{\bund}\LL^{\aund} \big)
\end{split}
\eeq
where $\de_0>0$ is a (not necessarily small\footnote{In fact $\de_0=10$ see \eqref{eq:de0value}.}) constant to be chosen later and $\LL^{\aund} $ are constant coefficients (to be chosen later) of a given 2-tensor of the form
\beq
\lab{eq:operato-LL2}
L^{\a\b} = \LL^{\underline{a}} S^{\a\b}_{\underline{a}}.
\eeq

We make the choices (as in section \ref{section:choicezhf} with $\lz=-a$)
\beq
\lab{eq:Pi-a}
z=\frac{\De}{r^4}.
\eeq
and\footnote{We note that the choice of $h$ here is inspired, but different from that in the scalar case in section \ref{section:choicezhf}, which was algebraically more involved. }
\beq
\lab{eq:defofhSS}
h=\frac12r^5(1-\frac{m^4}{r^4}\ln(\frac{\De}{m^2})).
\eeq

 Motivated by the choice of $z=\frac{\De}{(r^2+a^2+a\lz)^2}=\frac{\De}{r^4}$ with $\lz=-a$, we introduce another set of symmetric Killing $2$-tensors
\beq
\lab{eq:newStensor}
S_1^{\a\b}=T^\a T^\b,\quad S_2^{\a\b}=T^{(\a}\Zhat^{\b)},\quad S_3^{\a\b}=\Zhat^\a\Zhat^\b,\quad S_4^{\a\b}=O^{\a\b}
\eeq
and second order symmetry operators (see \eqref{eq:identityZhat} for the definition of $\Zhat$)
\beq
\lab{eq:newSS}
\SS_1\psi=TT\psi,\quad \SS_2\psi=T\Zhat\psi,\quad\SS_3\psi=\Zhat\Zhat\psi,\quad \SS_4\psi=\OO\psi
\eeq
commuting with $|q|^2\square$.

By abuse of notation, from now on $S_{\aund}^{\a\b}$ and $\SS_\aund$ will denote the above symmetric tensors and second order symmetry operators respectively. Then we rewrite
\beq
\lab{eq:rewriteGG}
\GG^{\a\b} =\GG^\aund  \Sa^{\a\b},
\eeq
with $\GG^\aund$, $\aund=1,2,3,4$, given by
\beq\lab{components-RR-aund_new}
\GG^1=-r^4, \quad \GG^2 = -2ar^2, \quad \GG^3 =-a^2, \quad \GG^4=\De.
\eeq



\subsection{Coefficients of generalized current $\widetilde{\EE}$}

\subsubsection{Main coefficients for $\EE[\X, \w,0]$}

In order to compute the coefficients $ \widetilde{\AA}^{\aund\bund}, \widetilde{\VV}^{\aund\bund}$ of $\widetilde{\EE}$, we first calculate the coefficients $\AA^{\aund\bund},\VV^{\aund\bund}$ of $\EE[\X, \w,0]$.
\begin{lemma}\label{lemma:coeff-z}
With the choice $z=\frac{\De}{r^4} $ (as in section \ref{section:lz=-a} with $\lz=-a$), 
we have the following.
\begin{enumerate}
\item The coefficients $\RRtp^\aund[z]= \pr_r\Big( \frac{z}{\De} \RRa\Big)=\pr_r\Big(\frac{1}{r^4}\GG^{\aund}\Big)$ are given by
\beq
\lab{eq:valuesforRRtpz}
\bsplit
\RRtp^1[z]&=0, \quad \RRtp^2[z]=  \frac{4a}{r^3} , \quad
\RRtp^3[z]=\frac{4a^2}{r^5},\\
\RRtp^4[z]&= \frac{-2\TT_{-a}}{r^6}= -\frac{2\big(\TT+a^2(r-m)\big)}{ r^6}=-\frac{2r(r^2-3mr+2a^2)}{r^6}.
\end{split}
\eeq
\item The coefficients $\UUwtp^\aund= (1+\de_0 m^2z)\GGtp^{\aund}[z]-\delta_0m^2\GGtp^4[z]\frac{\GG^{\cund}}{r^4}$ are given by
\beq
\lab{eq:valuesforUUwtpz}
\bsplit
\UUwtp^1&=\de_0 m^2\frac{-2\TT_{-a}}{r^6}, \qquad \UUwtp^2
=(1+\de_0 m^2\frac{mr-a^2}{r^4})\frac{4a}{r^3}, \\
\UUwtp^3&
=(1+\de_0 m^2\frac{r-m}{2r^3})\frac{4a^2}{r^5},\qquad
\UUwtp^4= \frac{-2\TT_{-a}}{r^6}.
\end{split}
\eeq
\item We write $\AA^{\aund\bund}=\AA^{(\aund}[z]\LL^{\bund)}$ where
\beq \label{eq:valuesforRRtppz}
 \AA^\aund[z]
 =-\frac{\De^2}{r^2}\pr_r \Big(  \frac{ h   \UUwtp^{\aund}  }{r^2} \Big).
 \eeq

\item We write $\VV^{\aund\bund}=\VV^{(\aund}[z]\LL^{\bund)}$ where the coefficients
\beq
\VV^{\aund}[z]=\frac14\pr_r(\De\pr_r(z\pr_r(h\UUwtp^{\aund}))).
\eeq
\end{enumerate}
\end{lemma}

\begin{proof}
The expressions for $\RRtp^\aund[z]$ follow from a straightforward calculation.
The expressions for $\UUwtp^\aund$ follow from \eqref{eq:valuesforRRtpz} and \eqref{components-RR-aund_new}. 
\end{proof}

\subsubsection{Computation for coefficients of $\widetilde{\EE}$}

We recall, see \eqref{eq:EEwt}, \eqref{eq:coefficientsofEEwt} and \eqref{choice-f-operator}, that
\begin{align*}
|q|^2\widetilde{\EE}&=|q|^2(\EE[\bold{X},\bold{w},0]+\EE[0,\bold{w}^{\flat},0])\\
&=\widetilde{\UU}^{\a\b\aund\bund}\D_\a\psia\c\D_{\b}\psib+\widetilde{\AA}^{\aund\bund}\pr_r\psia\c\pr_r\psib+\widetilde{\VV}^{\aund\bund}\psia\c\psib
\end{align*}
where
\beq
\lab{eq:detailcoeffofEEwt}
\bsplit
\widetilde{\UU}^{\a\b\aund\bund}&=\frac12 h\Big(\UUwtp^{(\aund}\LL^{\bund)}\GGtp^{\cund}[z]+(\de_0 m^2\frac{\RRtp^4[z]}{2r^4})^2\GG^{(\aund}\LL^{\bund)}\GG^{\cund}\Big)S_{\cund}^{\a\b},\\
 \UUwtp^\aund&=(1+\de_0m^2z)\RRtp^\aund[z]-\de_0 m^2\RRtp^4[z]\frac{\GG^\aund}{r^4},\\
\widetilde{\AA}^{\aund\bund}&=\AA^{\aund\bund}+\frac12\De w_{\flat}^{\aund\bund}=-z^{1/2}\Delta^{3/2} \partial_r\left(h \frac{ z^{1/2}  \UUwtp^{(\aund}\LL^{\bund)} }{\Delta^{1/2}}  \right)+\frac12h(\de_0 m^2\frac{\De\RRtp^4[z]}{2r^4})^2 \GG^{(\aund}\LL^{\bund)},\\
\widetilde{\VV}^{\aund\bund}&=\VV^{\aund\bund}-\frac14\pr_r \left(\De \pr_r w_{\flat}^{\aund\bund}\right)\\
&=\frac 1 4 \pr_r \left(\De \pr_r \Big( z \pr_r \big( h \UUwtp^{(\aund}\LL^{\bund)}\big)\Big)\right)-\frac14\pr_r \left(\De \pr_r \Big(h\De(\de m^2\frac{\RRtp^4[z]}{2r^4})^2\GG^{(\aund}\LL^{\bund)}\Big)\right),
\end{split}
\eeq
where, for convenience, we choose
\beq\lab{eq:de0value}
\de_0=10.
\eeq

\begin{proposition}
\lab{Prop:summary-6.3}
Equations \eqref{eq:EEwt} (including a double indexed $1$-form $\bold{M}^{\aund\bund}$) can be written in the form 
\beq\lab{eq:separation-EE-I-J-K}
\begin{split}
&    |q|^2\widetilde{\EE}[\bold{X}, \bold{w}+\bold{w'}, \bold{M}]   =\Pt+\It+\Jt+\Kt \\
&\Pt:= \widetilde{\UU}^{\a\b\aund\bund} \, \D_\a \psia \c \D_\b \psib,\\
&\It:= \widetilde{\AA}^{\aund\bund}  \pr_r\psia\c \pr_r\psib, \\
&\Jt:= \widetilde{\VV}^{\aund\bund} \psia\c\psib,\\
&\Kt:=\frac 1 4 |q|^2  \D^\mu ((\psia\c\psib) M^{\aund\bund}_\mu).
\end{split}
\eeq
\begin{enumerate}
\item The quadratic form $\It$ can be written as
\beq\label{eq:expr-I}
\It= \big(\widetilde{\AA}^{\aund}\pr_r \psia\big)\c \big( \LL^{\aund}\pr_r \psia\big) ,
\eeq
where
\beq\lab{eq:AAtz-identity}
\widetilde{\AA}^{\aund}=\AA^\aund[z]+\frac12h(\de_0 m^2\frac{\De\RRtp^4[z]}{2r^4})^2 \GG^{\aund}.
\eeq
In particular\footnote{The detailed expressions of $\AAt^\aund$ are not important. The important facts are that $\AAt^1, \AAt^3, \AAt^4$ are nonnegative and the behavior of $\AAt^1,\AAt^4$ is the same as that of $\AA$ in the scalar case, that is, $\AAt^1, \AAt^4\sim \De$ near $r=r_+$ and $\AAt^1, \AAt^4\sim 1$ near $r=\infty$.}, for $|a|/m\leq 0.9$ we have
\[\AAt^1\geq 0, \qquad \AA^3[z]\geq \AAt^3\geq 0, \qquad \AAt^4\geq \AA^4[z]\geq 0
\]
 for $r\geq r_+$. Moreover, $\AA^2[z]\geq \widetilde{\AA}^2\geq 0$ for $|a|/m\leq 0.9$ and $r\geq r_+$.
\item The quadratic form $\Jt$ can be written as
\beq\label{eq:expr-J}
\Jt= \big(\widetilde{\VV}^{\aund} \psia\big)\c \big( \LL^{\aund} \psia\big),
\eeq
where\footnote{Again, the detailed expressions are not important. The important facts are that $\VVt^1,\VVt^4\sim \ln\De$ near $r=r_+$ and $\VVt^1, \VVt^4\sim r^{-2}$ near $r=\infty$.}
\beq
 \lab{eq:VVtz-identity}
 \VVt^\aund=\VV^\aund[z]-\frac14\pr_r \left(\De \pr_r \Big(h\De(\de m^2\frac{\RRtp^4[z]}{2r^4})^2\GG^{\aund}\Big)\right).
  \eeq
\end{enumerate}
\end{proposition}
\begin{proof}
The proof of \eqref{eq:expr-I} and \eqref{eq:AAtz-identity} is a straightforward calculation from \eqref{eq:detailcoeffofEEwt}. Since $\GG^2=-2ar\leq0, \GG^3=-a^2\leq0$ and $\GG^4=\De\geq 0$, it follows that $\AA^2[z]\geq \AA^2,\ \AA^3[z]\geq \AAt^3$ and $\AAt^4\geq \AA^4[z]$. Finally, the positivity of $\AAt^1, \AAt^3,\AA^4[z], \AAt^2$ can be checked easily using Mathematica.

The proof of \eqref{eq:expr-J} and \eqref{eq:VVtz-identity} is a straightforward calculation from \eqref{eq:detailcoeffofEEwt}.
\end{proof}
\begin{remark}
 The strategy we pursue below, in the remainder of this section, is to replace the expressions $\Pt, \It$ with positive definite expressions $P, I$ plus divergence terms. The lower order terms $\Jt, \Kt$ are treated in section \ref{sec:proofofSSMor}.
\end{remark}

\subsection{Principal term $\Pt$}


We now consider the principal  term $\Pt= \widetilde{\UU}^{\a\b\aund\bund} \, \D_\a \psia \c \D_\b \psib$.

\begin{lemma}[Main Lemma 1]
\lab{lemma:IntegrationbypartsP}
Consider expressions of the form
\[
\CC^\cund  S^{\mu\nu}_\cund (\AA^\aund\pr_\mu \psia) (\BB^\bund\pr_\nu \psib)
\]
where $(\AA^\aund, \BB^\aund, \CC^\aund)_{\aund=1,\ldots,4}$ are $4$-tuples depending only on $r$.
We set, for $\YY\in\{\AA,\BB,\CC\} $
\[
\psi_\YY=\YY^\aund\psi_\aund, \qquad S^{\mu\nu}_\YY=\YY^\cund  S^{\mu\nu}_\cund
\]
With these notations we have
\beq
\CC^\cund  S^{\mu\nu}_\cund (\AA^\aund\pr_\mu \psia) (\BB^\bund\pr_\nu \psib)=   S^{\mu\nu}_\CC \pr_\mu  \psi_\AA  \pr_\nu \psi _\BB
\eeq
and
\beq
\bsplit
S^{\mu\nu}_\CC \pr_\mu  \psi_\AA  \pr_\nu \psi _\BB &=  S^{\mu\nu}_\BB  \pr_\nu \psi_\CC\pr_\mu   \psi_\AA+|q|^2\Div\big[\psi;\AA,\BB,\CC\big]\\
\Div\big[\psi;\AA,\BB,\CC\big]
&=\D_\mu \Big[  |q|^{-2} \psi_\AA  \Big( S_{\CC}^{\mu\nu}    \D_\nu \psi _\BB-  S_{\BB}^{\mu\nu}  \D_\nu \psi_\CC\Big)\Big]
\end{split}
\eeq

\end{lemma}
\begin{proof}
Recall that $\psia:=\SS_\aund\psi=|q|^2\D_\mu(|q|^{-2} S^{\mu\nu}_\aund\D_\nu \psi)$ with the tensors $S_{\aund}^{\mu\nu}$ introduced in Definition
\ref{definition:tensors-S}. Since $S^{r\mu}_\cund=0$ we have, for any function $f$ which depends only on $r$,
\[
\SS_\aund( f(r)\psi)= f(r) \SS_\aund \psi.
\]
Since $\AA^{\aund}$ and $\BB^{\bund}$ depend only on $r$, we can write
\[
S^{\mu\nu}_\cund
 (\AA^\aund\pr_\mu\psia)(\BB^\bund\pr_\nu\psib)
=S^{\mu\nu}_\cund\pr_\mu\psi_\AA\pr_\nu\psi_\BB.
\]
Consequently,
\[
\CC^\cund S^{\mu\nu}_\cund
 (\AA^\aund\pr_\mu\psia)(\BB^\bund\pr_\nu\psib)
=S_\CC^{\mu\nu}\pr_\mu\psi_\AA\pr_\nu\psi_\BB,
\]
which is the first identity.

Since the $\SS$ operators commute we easily check that
\[
\SS_\CC \psi_\BB=\CC^\cund \SS_\cund (\SS_\BB\psi)=
\CC^\cund \SS_\cund \big(\BB^\dund \SS_\dund \psi\big)=  \CC^\cund \BB^\dund  \SS_\cund\SS_\dund\psi=  \CC^\cund \BB^\dund  \SS_\dund\SS_\cund\psi=\SS_\BB \psi_\CC.
\]
Using these remarks we write
\begin{align*}
|q|^{-2} S_{\CC}^{\mu\nu} \pr_\mu  \psi_\AA  \pr_\nu \psi _\BB
&=-
\psi_\AA \D_\mu\Big(  |q|^{-2} S_{\CC}^{\mu\nu}   \D_\nu \psi _\BB\Big)+ \D_\mu \Big( |q|^{-2} S_{\CC}^{\mu\nu}   \psi_\AA  \D_\nu \psi _\BB\Big)\\
&=- |q|^{-2}\psi_\AA \SS_\CC\psi_\BB+  \D_\mu \Big( |q|^{-2} S_{\CC}^{\mu\nu}   \psi_\AA  \pr_\nu \psi _\BB\Big)\\
&=- |q|^{-2}\psi_\AA \SS_\BB\psi_\CC+  \D_\mu \Big( |q|^{-2} S_{\CC}^{\mu\nu}   \psi_\AA  \pr_\nu \psi _\BB\Big)
\end{align*}
Similarly,
\[
|q|^{-2}\psi_\AA\SS_\BB \psi_\CC=\psi_\AA  \D_\mu\Big(  |q|^{-2} S_{\BB}^{\mu\nu}   \D_\nu \psi_\CC \Big)=- |q|^{-2} S_{\BB}^{\mu\nu}   \pr_\nu \psi_\CC\pr_\mu   \psi_\AA+\D_\mu \Big(|q|^{-2} \psi_\AA S_{\BB}^{\mu\nu}   \pr_\nu \psi_\CC\Big).
\]
Therefore,
\[
|q|^{-2} S_{\CC}^{\mu\nu} \pr_\mu  \psi_\AA  \pr_\nu \psi _\BB= |q|^{-2} S_{\BB}^{\mu\nu}   \pr_\nu \psi_\CC\pr_\mu   \psi_\AA+\Div[\psi;\AA,\BB,\CC]
\]
where
\begin{align*}
\Div[\psi;\AA,\BB,\CC]&=  \D_\mu \Big( |q|^{-2} S_{\CC}^{\mu\nu}   \psi_\AA  \pr_\nu \psi _\BB\Big)- \D_\mu \Big(|q|^{-2} \psi_\AA S_{\BB}^{\mu\nu}   \pr_\nu \psi_\CC\Big)\\
&=\D_\mu \Big[  |q|^{-2}   \psi_\AA\Big( S_{\CC}^{\mu\nu}   \pr_\nu \psi _\BB-  S_{\BB}^{\mu\nu}   \pr_\nu \psi_\CC\Big)\Big]
\end{align*}
as stated.
\end{proof}

As a corollary of the above Lemma we deduce that our principal trapping term
\[
\Pt= \frac 12 h \RRtp^\cund[z]S^{\a\b}_\cund ( \UUwtp^\aund   \D_\a \psia)  \c  (\LL^\bund \D_\b  \psib)+\frac12h\c(\frac12\de_0 m^2\RRtp^4[z])^2\frac{\GG^{\cund}}{r^4}S_{\cund}^{\a\b}(\frac{\GG^{\aund}}{r^4}\D_\a\psia)\c(\LL^{\bund}\D_\b\psib)
\]
can be  written in the form $ \Pt=P+|q|^2 \Div\, \Bk_P$ where $P$ is nonnegative definite.
\begin{definition}
\lab{Def:acceptable}
A quadratic expression $\BB^\mu $ in $\pr^{\le 1} \psi_\aund$ is called an acceptable boundary term if $
\BB^\mu\les\sum|\ln(r-r_+)||(e_4,\De e_3, \nab)^{\leq 1}\psia|^2$ near $r=r_+$ and $
\BB^\mu\les\sum|(e_4, e_4,\nab)^{\leq 1}\psia|^2$ near $r=\infty$. 
 By abuse of language
we will also say that the spacetime divergence $\Div \BB$ is acceptable.
\end{definition}\begin{proposition}
\lab{Prop:positivityP}
The principal trapping term can be written in the form
\beq
\lab{eq:expression1-P}
\Pt=\frac 1 2h \LL^\cund S_\cund^{\a\b}\pr_\a  \Psi_1\pr_\b\Psi_1 +\frac 1 2\de_0 m^2 zh \LL^\cund S_\cund^{\a\b}\pr_\a  \Psi_2\pr_\b\Psi_2+|q|^2 \Div \Bk_P,
\eeq
where
\begin{enumerate}
\item The scalars $\Psi _1,\Psi_2$ are given by, see \eqref{components-RR-aund_new} and \eqref{eq:valuesforRRtpz},
\beq
\lab{eq:DefPsi-P}
\bsplit
\Psi_1&:=(\RRtp^\aund[z]-\frac12\delta_0 m^2\RRtp^4[z]\frac{\GG^{\aund}}{r^4})\psia\\
&\ =-\de_0 m^2\frac{\TT_{-a}}{r^6}TT\psi+\frac{4a}{r^3}(1-\de_0m^2\frac{\TT_{-a}}{2r^5})T\Zhat\psi+\frac{4a^2}{r^5}(1-\de_0m^2\frac{\TT_{-a}}{4r^5})\Zhat\Zhat\psi\\
&\ -(1-\frac12\de_0 m^2z)\frac{2\TT_{-a}}{r^6}\OO\psi,\\
\Psi_2&:=\RRtp^\aund[z]\psia=\frac{4a}{r^3}T\Zhat\psi+\frac{4a^2}{r^5}\Zhat\Zhat\psi-\frac{2\TT_{-a}}{r^6}\OO\psi.
\end{split}
\eeq
\item The boundary term $\Bk_P$ is acceptable in the sense of Definition \ref{Def:acceptable}.
\item Choosing,
\beq
\lab{eq:remark:choiceofLL-weak}
\LL^{1} =m^2, \qquad    \LL^{2}=0,  \qquad \LL^3=1, \qquad \LL^4=0,
\eeq
we  rewrite \eqref{eq:expression1-P} in the form  $\Pt=\Phat+|q|^2  \Div \,\Bk_P$ where
\beq
\lab{eq:expression2-PS}
\bsplit
\Phat&=\frac 1 2  h \Big( \big|  m T  \Psi_1 \big|^2 +|\Zhat\Psi_1|^2+\de_0m^2z\big|  m T  \Psi_2 \big|^2 +\de_0m^2z|\Zhat\Psi_2|^2\Big).
\end{split}
\eeq
\end{enumerate}
\end{proposition}

\begin{proof}
Recall that
\[
\Pt=\frac 12 h \RRtp^\cund\SS^{\a\b}_\cund ( \UUwtp^\aund   \D_\a \psia)  \c  (\LL^\bund \D_\b  \psib)+\frac12h\c(\frac12\de_0 m^2\RRtp^4[z])^2\frac{\GG^{\cund}}{r^4}S_{\cund}^{\a\b}(\frac{\GG^{\aund}}{r^4}\D_\a\psia)\c(\LL^{\bund}\D_\b\psib).
\]
Then applying Lemma \ref{lemma:IntegrationbypartsP} to $\AA^\aund=\UUwtp^\aund, \CC^\aund=\RRtp^\aund[z], \BB^\aund=\frac12h\LL^\aund$ and $\AA^\aund=\CC^\aund=\frac12\de_0 m^2\RRtp^4[z]\frac{\GG^\aund}{r^4}, \BB^\aund=\frac12h\LL^\aund$ for the first and the second term above respectively, we obtain
\[
\Pt=\frac12h\LL^\cund S_{\cund}^{\a\b}\D_\a\psi_{\RRtp}\D_\b\psi_{\UUwtp}+\frac12 h\LL^{\cund}S_{\cund}^{\a\b}\D_\a\psi_{(\frac{\de_0 m^2\RRtp^4\GG}{2r^4})}\D_\b\psi_{(\frac{\de_0 m^2\RRtp^4\GG}{2r^4})}+|q|^2\Div\Bk_P
\]
We recall that $\UUwtp^\aund= (1+\de_0 m^2z)\GGtp^{\aund}[z]-\delta_0m^2\GGtp^4[z]\frac{\GG^{\cund}}{r^4}$. Therefore
\begin{align*}
\Pt&=\frac12h\LL^\cund S_{\cund}^{\a\b}\D_\a\psi_{\RRtp}\D_\b\psi_{\RRtp}-\frac12h\c 2\LL^\cund S_{\cund}^{\a\b}\D_\a\psi_{\RRtp}\D_\b\psi_{(\frac{\de_0 m^2\RRtp^4\GG}{2r^4})}\\
&+\frac12 h\LL^{\cund}S_{\cund}^{\a\b}\D_\a\psi_{(\frac{\de_0 m^2\RRtp^4\GG}{2r^4})}\D_\b\psi_{(\frac{\de_0 m^2\RRtp^4\GG}{2r^4})}
+\frac12\de_0 m^2 zh\LL^\cund S_{\cund}^{\a\b}\D_\a\psi_{\RRtp}\D_\b\psi_{\RRtp}+|q|^2\Div\Bk_P\\
&=\frac12 h\LL^{\cund}S_{\cund}^{\a\b}\D_\a\psi_{(\RRtp-\frac{\de_0 m^2\RRtp^4\GG}{2r^4})}\D_\b\psi_{(\RRtp-\frac{\de_0 m^2\RRtp^4\GG}{2r^4})}
+\frac12\de_0 m^2 zh\LL^\cund S_{\cund}^{\a\b}\D_\a\psi_{\RRtp}\D_\b\psi_{\RRtp}+|q|^2\Div\Bk_P
\end{align*}
where $\Bk_P$ is acceptable in the sense of Definition \ref{Def:acceptable}.
This proves \eqref{eq:expression1-P} and \eqref{eq:DefPsi-P}.
\end{proof}

Later on, in the analysis of the lower order term, we will need the following more precise characterization of $\Phat$.
\begin{proposition}
\lab{Prop:refinedP}
Assume that\footnote{After subtracting the axisymmetric part of $\psi$, the non-axisymmetric part of $\psi$ satisfies this assumption. As discussed in section \ref{sec:intro-axisym} in the introduction, the Morawetz-Energy estimate can be derived for the axisymmetric part of $\psi$ for the full sub-extremal case.}
\begin{align*}
&\int_{\DD(\tau_1,\tau_2)}\frac{f(r)}{|q|^2}|Z\psi|^2\geq \int_{\DD(\tau_1,\tau_2)} \frac{f(r)}{|q|^2}|\psi|^2,\\ &\int_{\DD(\tau_1,\tau_2)}\frac{f(r)}{|q|^2}(|\pr_\th\psi|^2+\frac{1}{\sin^2\th}|Z\psi|^2)\geq \int_{\DD(\tau_1,\tau_2)} \frac{2f(r)}{|q|^2}|\psi|^2
\end{align*}
 where $f(r)$ is any nonnegative function depending only on $r$.
 \begin{enumerate}
\item For $|a|/m<1$ and for all $r\geq r_+$, we have
\beq\lab{eq:TZhatPsi}
\bsplit
\int_{\DD(\tau_1,\tau_2)}\frac{1}{|q|^2}(\widetilde{\VV}^2T\Zhat\psi\c\psi_{\LL}+\Phat)+\Div \Bk_{P_{c}}
\geq \int_{\DD(\tau_1,\tau_2)}\frac{1}{|q|^2}(\JtP_{c}^\aund\psia\c\psi_{\LL})
\end{split}
\eeq
where $\JtP_{c}^2=0$
\beq\lab{eq:JtPc}
\bsplit
\JtP_{c}^1&=\de_0m^2\frac{r^2-3mr+2a^2}{4ar^2}\frac{1}{1-\de_0m^2\frac{\TT_{-a}}{2r^5}}\widetilde{\VV}^2-\frac{(r^3\widetilde{\VV}^2)^2}{32(1-\de_0m^2\frac{\TT_{-a}}{2r^5})^2h},\\
\JtP_{c}^3&=-\frac{a}{r^2}\frac{1-\de_0m^2\frac{\TT_{-a}}{4r^5}}{1-\de_0m^2\frac{\TT_{-a}}{2r^5}}\widetilde{\VV}^2,\\
\JtP_{c}^4&=\frac{r^2-3mr+2a^2}{2ar^2}\frac{1-\frac12\de_0m^2z}{1-\de_0m^2\frac{\TT_{-a}}{2r^5}}\widetilde{\VV}^2-\frac{(r^3\widetilde{\VV}^2)^2}{32a^2(1-\de_0m^2\frac{\TT_{-a}}{2r^5})^2h}
\end{split}
\eeq
and $\Bk_{P_{c}}$ is acceptable in the sense of Definition \ref{Def:acceptable}.
\item For $|a|/m\leq 0.75$ and\footnote{Note that this region is outside the trapping set where we expect some positive coercivity for $P$.} $r\geq r_++3.5m$, we have
\beq\lab{eq:Psi1larger}
\bsplit
\int_{\DD(\tau_1,\tau_2)}\frac{1}{|q|^2}\Phat+\Div \Bk_{P_{f}}&\geq \int_{\DD(\tau_1,\tau_2)}\JtP_{f}^\aund\psia\c\psi_{\LL}
\end{split}
\eeq
where $\JtP_{f}^2=\JtP_{f}^3=0$ and
\beq\lab{eq:JtPf}\bsplit
\JtP^1_{f}&=\frac12h(\de_0-\frac{5ar^3}{m\TT_{-a}(1-\frac{5am}{r^2})})\Big((1-\frac12\de_0m^2z)-\frac{2ar^3}{5m\TT_{-a}}\Big)\frac{\TT^2_{-a}}{r^{12}},\\
\JtP_{f}^4&=\frac14h\big((1+\sqrt{7})(1-\frac12\de_0m^2z)-(3+\sqrt{7})\frac{2ar^3}{5m\TT_{-a}}\big)\big((1-\frac12\de_0m^2z)-\frac{2ar^3}{5m\TT_{-a}}\big)\frac{\TT^2_{-a}}{r^{12}}.\end{split}
\eeq
and $\Bk_{P_{f}}$ is acceptable in the sense of Definition \ref{Def:acceptable}.
\end{enumerate}
\end{proposition}

\begin{proof}
See appendix \ref{sec:refinedP}.
\end{proof}

\subsection{Comparison with the integral trapping norm}
\lab{subsec:comparison-integral-P}

\begin{remark}
\lab{remark-Psi-andP}
Let $\Pt$ denote the quantity defined in Proposition
\ref{Prop:positivityP}, with the choice $\LL=(m^2, 0,1,\sigma)$; thus in $\DD_{trap}$ 
\[
\Phat\sim |(T, Z, \nab)\Psi_1|^2+|(T, Z, \nab)\Psi_2|^2.
\]
 Recall that the
integral trapping norm introduced in section \ref{section:BEFquantities} is
\[
P[\psi](\tau_1,\tau_2)
=\int_{\DD_{trap}(\tau_1,\tau_2)}
\left|\Tfrak[\psi]\right|^2.
\]
We now analyze the connection between $\Phat$ and $P$.  Indeed, the definitions
\eqref{eq:defPsi-intro} and \eqref{eq:DefPsi-P} give
\begin{equation}\label{eq:Psi2-Psiintro-relation}
 \Psi_2=r^{-3}\Psi.
\end{equation}
Consequently, we rewrite
\beq 
\label{eq:Psi2-Tfrak-modification}
\Psi_2=-\frac{2}{r^4\De}
 \big((r^2+a^2)T+aZ\big)\Tfrak[\psi]
 -\frac{2\TT_{-a}}{r^6}|q|^2\square_\g\psi
 +\frac{2\TT_{-a}}{r^6}
   \partial_r(\De\partial_r\psi)
   \eeq 
Applying
$T, Z,\nab$ to
\eqref{eq:Psi2-Tfrak-modification}, we have the following
\[
\int_{\DD_{trap}(\tau_1,\tau_2)}|(\That,\nab)\That\Tfrak[\psi]|^2\les \int_{\DD_{trap}(\tau_1,\tau_2)}\frac{1}{|q|^2}\Phat+B^2[\psi](\tau_1,\tau_2)+\NN[N].
\]
Since
\[
\OO\Tfrak[\psi]=\frac{(r^2+a^2)^2}{\De}\That\That\psi-\pr_r(\De\pr_r\Tfrak[\psi])+|q|^2\square_{\g}\Tfrak[\psi]
\]
and
\[
\OO\psi=\frac{1}{\sin\th}\pr_\th(\sin\th\pr_\th\psi)+\frac{1}{\sin^2\th}\Big((1-\frac{a^2\sin^2\th}{r^2+a^2})Z+a\sin^2\th \That\Big)^2,
\]
by taking the square of $\OO\psi$ and integration by parts, we conclude that
\[
\int_{\DD_{trap}(\tau_1,\tau_2)}|\nab\nab\Tfrak[\psi]|^2\les \int_{\DD_{trap}(\tau_1,\tau_2)}|(\That,\nab)\That\Tfrak[\psi]|^2+ \int_{\DD_{trap}(\tau_1,\tau_2)}\frac{1}{|q|^2}\Phat+B^2[\psi](\tau_1,\tau_2)+\NN^1[N]
\]
and 
\beq
\lab{eq:integral-P-controlled-by-density}
P^2[\psi](\tau_1,\tau_2)=\int_{\DD_{trap}(\tau_1,\tau_2)}|\dk^{\leq 2}\Tfrak[\psi]|^2\les
\int_{\DD_{trap}(\tau_1,\tau_2)}\frac{1}{|q|^2}\Phat
+\BEF^2[\psi](\tau_1,\tau_2)+\NN^1[\psi](\tau_1,\tau_2).
\eeq
\end{remark}



\subsection{The quadratic form $\It$}
Recall that the term $\It$ in the decomposition \eqref{eq:separation-EE-I-J-K} has the form
\[
\It= \widetilde{\AA}^{\aund\bund}  \pr_r\psia\c \pr_r\psib=\big(\widetilde{\AA}^\aund\pr_r \psia\big)\c \big( \LL^{\aund}\pr_r \psia\big) .
\]

\begin{lemma}[Main Lemma 2]
\lab{lemma:IntegrationbypartsI}
Consider expression of the form
\[
\AA^\bund\psib\c\BB^\aund\psia
\]
where $(\AA^\aund, \BB^\aund)_{\aund=1,\ldots,4}$ are $4$-tuples depending only on $r$ and $\BB^4=0$. With the notations as in Lemma \ref{lemma:IntegrationbypartsP}, we have
\beq
\bsplit
\AA^\bund\psib\c\BB^\aund\psia&= S_\BB^{\a\b} S_\AA^{\mu\nu}\pr_\mu\pr_\a\psi\c \pr_\nu\pr_\b\psi+|q|^2\Div[\psi;\AA,\BB]\\
\Div[\psi;\AA,\BB]&=\D_\mu\left[|q|^{-2}\Big(\psi_\BB S_{\AA}^{\mu\nu}\pr_\nu\psi-\frac12S_{\BB}^{\mu\nu}\D_{\nu}\big(S_{\AA}^{\a\b}\pr_\a\psi\pr_\b\psi\big)\Big)\right].
\end{split}
\eeq
\end{lemma}
\begin{proof}
With the same notations as in Lemma \ref{lemma:IntegrationbypartsP}, we derive
\begin{align*}
|q|^{-2}    \AA^\bund\psib\c\BB^\aund\psia&=|q|^{-2}\psi_{\BB}\AA^\aund\psia=|q|^{-2}\psi_{\BB}\AA^\aund\SS_\aund\psi=\psi_{\BB}\AA^\aund\D_\mu(|q|^{-2} S^{\mu\nu}_\aund\D_\nu \psi)\\
&=\D_\mu(\psi_{\BB}|q|^{-2} S^{\mu\nu}_{\AA}\D_\nu \psi)-|q|^{-2}S_\aund^{\mu\nu}\pr_\mu(\psi_\BB\AA^\aund)\pr_{\nu}\psi\\
&=\D_\mu(\psi_{\BB}|q|^{-2} S^{\mu\nu}_{\AA}\pr_\nu \psi)-|q|^{-2}S_\AA^{\mu\nu}\pr_\mu\psi_\BB\pr_{\nu}\psi\
\end{align*}
where we use $S_\aund^{\mu\nu}\pr_\mu(f(r))=0$ in the last step. Since $\BB^4=0$, we have
\[
\psi_\BB=\BB^\aund\phia=\BB^{\aund}|q|^2\D_\mu(|q|^{-2}S_{\aund}^{\mu\nu}\D_\nu\psi)=\BB^\aund S_\aund^{\mu\nu}\pr_\mu\pr_\nu\psi=S^{\mu\nu}_{\BB}\pr_\mu\pr_\nu\psi
\]
and thus (using the fact that $S_\BB^{\a\b}$ only depend on $r$)
\[
S_{\AA}^{\mu\nu}\pr_\mu\psi_{\BB}=S_{\AA}^{\mu\nu}\pr_\mu(S_{\BB}^{\a\b}\pr_\a\pr_\b\psi)=S_{\AA}^{\mu\nu}S_{\BB}^{\a\b}\pr_\a\pr_\b\pr_\mu\psi.
\]
Similarly, using the facts that $S_\BB^{\th\b}=S_{\BB}^{r\b}=0$ and $S_{\AA}^{\a\b}, S_{\BB}^{\a\b}$ depend only on $r,\th$, we deduce
\begin{align*}
&    |q|^{-2}S_{\AA}^{\mu\nu}\pr_\mu\psi_{\BB}\pr_\nu\psi=|q|^{-2}S_{\AA}^{\mu\nu}S^{\a\b}_{\BB}\pr_\a\pr_\b\pr_\mu\psi\c\pr_\nu\psi\\
&=\pr_\a\Big((S_{\BB}^{\a\b}\pr_\b\pr_\mu\psi)|q|^{-2}S_{\AA}^{\mu\nu}\pr_\nu\psi\Big)-|q|^{-2}S_{\AA}^{\mu\nu}S_{\BB}^{\a\b}\pr_\mu\pr_\b\psi\c\pr_\nu\pr_\a\psi\\
&=\pr_\a\Big(\frac12|q|^{-2}S_{\BB}^{\a\b}\pr_\b\big(S_{\AA}^{\mu\nu}\pr_\nu\phi\pr_\mu\psi\big)\Big)-|q|^{-2}S_{\AA}^{\mu\nu}S_{\BB}^{\a\b}\pr_\mu\pr_\b\psi\c\pr_\nu\pr_\a\psi\\
&=\D_\a\Big(\frac12|q|^{-2}S_{\BB}^{\a\b}\D_\be\big(S_{\AA}^{\mu\nu}\pr_\nu\phi\pr_\mu\psi\big)\Big)-|q|^{-2}S_{\AA}^{\mu\nu}S_{\BB}^{\a\b}\pr_\mu\pr_\b\psi\c\pr_\nu\pr_\a\psi.
\end{align*}
Therefore
\[
|q|^{-2}\AA^\aund\psia\c\BB^\bund\psib=  |q|^{-2}S_\AA^{\mu\nu}S_\BB^{\a\b}\pr_\mu\pr_\a\psi\c \pr_\nu\pr_\b\psi+\Div[\psi;\AA,\BB]
\]
where
\[
\Div[\psi;\AA,\BB]=\D_\mu\left[|q|^{-2}\Big(\phi_\BB S_{\AA}^{\mu\nu}\pr_\nu\psi-S_{\BB}^{\mu\nu}S_{\AA}^{\a\b}\pr_{\nu}\pr_\a\phi\pr_\b\psi\Big)\right]
\]
as stated.
\end{proof}

Using Lemma \ref{lemma:IntegrationbypartsI}, we find (modulo spacetime divergences) a lower bound of the quadratic term $\It$. In order to obtain a sum of squares in the proof, we appeal to the following general Lemma.
\begin{lemma}\lab{Lem:sumofsquares}
Let $(\AA^\aund)_{\aund=1,\ldots,4}$ is a $4$-tuple depending only on $r$. Choosing $\LL^{\aund}=(m^2, 0, 1, 0)$, we have with $\mathfrak{X}=\{mT,\Zhat 
\}$
\beq
S_{\AA}^{\mu\nu}  S^{\a\b} _\LL\pr_\mu \pr_\a  \phi  \pr_\nu \pr_\b  \phi
=
\sum_{X\in \mathfrak{X}} S_{\AA}^{\mu\nu}\pr_\mu X\phi\c\pr_\nu X\phi.
\eeq
Moreover, if $S_{\AA}^{\mu\nu}=\AA^{\aund} S_{\aund}^{\mu\nu}$ is nonnegative, so is the above expression.
\end{lemma}
\begin{proof}
The proof follows immediately from the fact that $S_{\LL}^{\mu\nu}=m^2T^\mu T^\nu+\Zhat^\mu\Zhat^\nu$ for $\LL^{\aund}=(m^2, 0, 1, 0)$.
\end{proof}

\begin{proposition}\lab{Prop:positivityI}
Choosing $\LL^{\aund}=(m^2, 0, 1, 0)$, we have for $|a|/m\leq 0.75$
\beq\lab{eq:lowerboundofIts}
\It\geq I+|q|^2 \Div\, \Bk_{I}
\eeq
where
\begin{enumerate}
\item The term $I$ is given by
\beq\label{eq:expressionofIS}
I=\II^\aund(\pr_r\psi)_\aund\c\LL^\bund(\pr_r\psi)_\bund
\eeq
where\footnote{In fact, we will choose $f(r), g(r), h(r)$ such that $\II^1, \II^3, \II^4\geq0$ for $r\geq r_+$} 
\beq\label{eq:IIhatSlargea}
\bsplit
\II^1&=\AAt^1-mf(r)\AA^2[z],\quad
\II^2=0,\\
 \II^3&=\AA^3[z]-\frac{g(r)}{m}\AA^2[z],\quad \II^4=\AA^4[z]-\frac{h(r)}{m}\AA^2[z]
  \end{split}
\eeq
and\footnote{In fact, we will choose $f(r), g(r), h(r)$ to be step functions, see \eqref{eq:IIhatSlargeastep}.} for $r\geq r_+$
\beq
\frac{a^4}{4r^4}\leq f(r)\c\frac{am}{2r^2}\leq \frac{1}{4}, \quad g(r), h(r)\geq 0, \quad g(r)+h(r)\geq\frac{1}{4f(r)}.
\eeq
\item The boundary term $\Bk_{I}$ 
is acceptable in the sense of Definition \ref{Def:acceptable}.
\end{enumerate}
\end{proposition}

\begin{proof}
According to \eqref{eq:IIhatSlargea}, we have
\[
\It-I=[(\widetilde{\AA}^\aund-\II^{\aund})(\pr_r\psi)_\aund]\c\LL^\bund(\pr_r\psi)_\bund
\]
where
\beq\lab{eq:AAhat-IIhatSlargea}
\bsplit
\widetilde{\AA}^1-\II^1&=mf(r)\AA^2[z]\geq 0,\quad \widetilde{\AA}^2-\II^2=\widetilde{\AA}^2,\\    \widetilde{\AA}^3-\II^3&=\AAt^3-\AA^3[z]+\frac{g(r)}{m}\AA^2[z],\quad
\widetilde{\AA}^4-\II^4\geq \AA^4[z]-\II^4=\frac{h(r)}{m}\AA^2[z]\geq 0.
\end{split}
\eeq
By applying Lemma \ref{lemma:IntegrationbypartsI} to $ \AA^\aund=\widetilde{\AA}^\aund-\II^{\aund}, \BB^\aund=\LL^\aund,\pr_r\psi$, we obtain
\begin{align*}
\It-I-|q|^2\Div[\pr_r\psi;\widetilde{\AA}-\II,\LL]&=[(\widetilde{\AA}^\aund-\II^{\aund})(\pr_r\psi)_\aund]\c\LL^\bund(\pr_r\psi)_\bund\\
&=(\widetilde{\AA}^\aund-\II^{\aund})S_\aund^{\a\b}S_{\LL}^{\mu\nu}\pr_\al\pr_\mu\pr_r\psi\c\pr_\b\pr_{\nu}\pr_r\psi.
\end{align*}
According to Lemma \ref{Lem:sumofsquares}, it suffices to prove that $(\widetilde{\AA}^\aund-\II^{\aund})S_\aund^{\a\b}$ is nonnegative. Since $|q|e_2=\frac{\Zhat}{\sin\th}-\frac{a\cos^2\th}{\sin\th} T$, we write
\begin{align*}
&(\widetilde{\AA}^\aund-\II^{\aund})S_\aund^{\a\b}\geq mf(r)\AA^2[z]T^\mu T^\nu+\widetilde{\AA}^2T^{(\mu}\Zhat^{\nu)}+(\AAt^3-\AA^3[z]+\frac{g(r)}{m}\AA^2[z])\Zhat^\mu\Zhat^\nu+\frac{h(r)}{m}\AA^2[z]O^{\mu\nu}\\
&\quad=\begin{pmatrix}
mf(r)\AA^2[z]+\frac{a^2\cos^4\th}{\sin^2\th}\frac{h(r)}{m}\AA^2[z]\\
\frac{1}{2}\widetilde{\AA}^2-\frac{a\cos^2\th}{\sin^2\th}\frac{h(r)}{m}\AA^2[z]\\
\AAt^3-\AA^3[z]+\frac{g(r)}{m}\AA^2[z]+\frac{1}{\sin^2\th}\frac{h(r)}{m}\AA^2[z]\\
\frac{h(r)}{m}\AA^2[z]\end{pmatrix}^T\begin{pmatrix}
T^\mu T^\nu\\
2T^{(\mu}\Zhat^{\nu)}\\
\Zhat^\mu \Zhat^\nu\\
(|q|e_1)^\mu(|q|e_1)^\nu
\end{pmatrix}
\end{align*}
Therefore, it is equivalent to proving that the determinant of
\[
\begin{pmatrix}
mf(r)\AA^2[z]+\frac{a^2\cos^4\th}{\sin^2\th}\frac{h(r)}{m}\AA^2[z]&    \frac{1}{2}\widetilde{\AA}^2-\frac{a\cos^2\th}{\sin^2\th}\frac{h(r)}{m}\AA^2[z]\\
\frac{1}{2}\widetilde{\AA}^2-\frac{a\cos^2\th}{\sin^2\th}\frac{h(r)}{m}\AA^2[z]    &
\AAt^3-\AA^3[z]+\frac{g(r)}{m}\AA^2[z]+\frac{1}{\sin^2\th}\frac{h(r)}{m}\AA^2[z]\end{pmatrix}
\]
 is nonnegative.

 In fact, the determinant is given by
\begin{align*}
&mf(r)\AA^2[z]\big(\frac{g(r)+h(r)}{m}\AA^2[z]+\AAt^3-\AA^3[z]\big)-\frac14(\AAt^2)^2\\
&\quad+\frac{\cos^2\th}{\sin^2\th}\frac{h(r)}{m}\AA^2[z]\Big(mf(r)\AA^2[z]+a\AAt^2+a^2\cos^2\th(\AAt^3-\AA^3[z]+\frac{g(r)}{m}\AA^2[z])\Big)\\
&\quad\geq  \frac14(\AA^2[z]+\AAt^2)(\AA^2[z]-\AAt^2)-mf(r)\frac{a}{2r^2}\AA^2[z](\AA^2[z]-\AAt^2)\\
&\quad+\frac{\cos^2\th}{\sin^2\th}\frac{h(r)}{m}\AA^2[z]\Big(mf(r)\AA^2[z]+a\AAt^2-a^2\cos^2\th \frac{a}{2r^2}(\AA^2[z]-\AAt^2)+a^2\cos^2\th\frac{g(r)}{m}\AA^2[z])\Big)\\
&\quad\geq \frac14(\AA^2[z]+\AAt^2)(\AA^2[z]-\AAt^2)-mf(r)\frac{a}{2r^2}\AA^2[z](\AA^2[z]-\AAt^2)\geq \frac14\AAt^2(\AA^2[z]-\AAt^2)\geq 0
\end{align*}
where we use $\AA^2[z]\geq 0,\ g(r)+h(r)\geq \frac{1}{4f(r)},\ \AAt^3-\AA^3[z]=\frac{a}{2r^2}(\AAt^2-\AA^2[z])$ in the first step and $\AA^2[z]\geq \AAt^2\geq 0,\ mf(r)-\frac{a^3}{2r^2}\geq 0$ for $r\geq r_+$ in the second step, and $f(r)\c\frac{am}{2r^2}\leq \frac{1}{4}$ in the third step. This finishes the proof.
\end{proof}

\section{Proof of Theorem \ref{Thm:Moraw4}}
\lab{sec:proofofSSMor}
 We recall that $|q|^2\EEt=\Pt+\It+\Jt+\Kt$, see \eqref{eq:separation-EE-I-J-K}. Using Propositions \ref{Prop:positivityP} and \ref{Prop:positivityI}, we see that, modulo acceptable spacetime divergence, for $|a|/m\leq0.75$
\[
\Pt+\It\geq 0.
\]
with appropriate choices of $f(r), g(r), h(r)$ in Proposition \ref{Prop:positivityI} such that $\II^1,\II^3,\II^4\geq 0$. In order to prove Theorem \ref{Thm:Moraw4} we still need to control the lower order term $\Jt$.
We summarize what we need to do in the following results.
\begin{proposition}
\lab{prop:I+Kt-short}
Under the assumption $|a|/m\le 0.75$,
    for an appropriate, \textit{continuous}, choice of $\Kt=\frac 1 4 |q|^2 \Div ((\psia\c\psib) M^{\aund\bund})$, modulo acceptable boundary terms, 
    we have
        \beq
        \lab{eq:prop-I+Kt-short}
        \int_{\DD(\tau_1, \tau_2)}\EEt\gtrsim \sum_{\aund=1}^4\int_{\DD(\tau_1, \tau_2) }|q|^{-2}\Phat+ \De^{-1}  | \Rhat  \psia|^2  +(r^{-4} +r^{-5}|\ln(\De)|)|\psia|^2.
        \eeq
\end{proposition}
Combining Proposition \ref{prop:I+Kt-short} with Theorem \ref{Thm:Moraw3} applied to $\psia$, we immediately derive
\begin{corollary}
\lab{cor:I+Kt-Nontrap-short}
Under the assumption $|a|/m\le 0.75$,
    for an appropriate, \textit{continuous}, choice of $\Kt=\frac 1 4 |q|^2 \Div ((\psia\c\psib) M^{\aund\bund})$, modulo acceptable boundary terms, 
    we have
        \beq
        \lab{eq:prop-I+Kt-Nontrap-short}
        \bsplit
        &\sum_{\aund=1}^4\int_{\DD(\tau_1, \tau_2) } |q|^{-2}\Phat+\De^{-1}  | \Rhat  \psia|^2  +(r^{-4} +r^{-5}|\ln(\De)|)|\psia|^2\\
        &\quad+\sum_{\aund=1}^4\int_{\DD_{\ntrap}(\tau_1, \tau_2)}  \left(\frac{m}{r^2} |\That \psia|^2 + r^{-1}|\nab\psia|^2\right)\\
                \quad&\les\int_{\DD(\tau_1, \tau_2)}\EEt+\sum_{\aund=1}^4\Big(\EFdeg[\psia](\tau_1,\tau_2)+    \Edeg[\psia](\tau_1)\Big).
\end{split}
        \eeq
\end{corollary}

Corollary \ref{eq:prop-I+Kt-Nontrap-short} establishes the validity of the irregular version (due to the presence of the singular behavior of the left hand side of \eqref{eq:prop-I+Kt-Nontrap-short}) of Theorem
\ref{Thm:Moraw4}. It only remains to pass from the irregular to the regular version by using a method similar to that used in the scalar case in sections \ref{section:temperMorawetz} and \ref{sec:controloflowerterm}. The details are given in section \ref{sec:SStemper}.


\subsection{Proof of Proposition \ref{prop:I+Kt-short}}
\lab{section:Hardy}


\subsubsection{Strategy of the proof of Proposition \ref{prop:I+Kt-short}}
\lab{section:Proof-{prop:I+Kt-short}}

We will control the lower order term $\Jt$ by making use of a Hardy inequality. That amounts to finding an appropriate, \textit{continuous}, choice of $\Kt=\frac 1 4 |q|^2 \Div ((\psia\c\psib) M^{\aund\bund})$ such that
\bea
\lab{eq:positiv-EEt}
\int_{\DD(\tau_1,\tau_2)}\EEt=\int_{\DD(\tau_1,\tau_2)}\frac{1}{|q|^2}(\Pt+\It+\Jt+\Kt)\geq 0,
\eea
from which Proposition \ref{prop:I+Kt-short} follows.

The main steps in the proof of \eqref{eq:positiv-EEt} are as follows:

{\bf Step 1.} Complete the squares in the expression $I+\Kt$. More specifically, we introduce the corresponding corrections $\Kt=\frac 1 4 |q|^2 \Div ((\psia\c\psib) M^{\aund\bund})$ with $M^{\aund\bund}=v^{(\aund}\LL^{\bund)}\pr_r$ and then we complete the square to obtain the following
\beq\label{eq:K-Jt14short}
\bsplit
&\II^\aund\SS_\aund\pr_r\psi\c\LL^\bund(\pr_r\psi)_\bund+
\frac 1 4 |q|^2  \Div ((\psia\c\psib) M^{\aund\bund})-|q|^2\Div\Bk_{\widetilde{K}}\\
& \geq \Big(\frac14\pr_r\big(|q|^2v^\aund\big)-\frac{1}{16}\frac{(|q|^2v^\aund)^2}{\II^\aund}\Big)\SS_\aund\psi\c\psi_{\LL}
\end{split}
\eeq
where $\Bk_{\widetilde{K}}$ is an acceptable boundary term. We note that since $\II^\aund=(\II^1, 0, \II^3, \II^4)$ we choose $v^\aund=(v^1, 0, v^3, v^4)$. The details are given in section \ref{sec:I+Kt}.

{\bf Step 2.} From \eqref{eq:K-Jt14short}
   we set up the Riccati inequalities for $v^1, v^3, v^4 $
   \begin{align*}
&\frac14\pr_r\big(|q|^2v^1\big)-\frac{1}{16}\frac{(|q|^2v^1)^2}{\II^1} +\VVt^1+ \text{additional term}\ge 0\\
&\frac14\pr_r\big(|q|^2v^3\big)-\frac{1}{16}\frac{(|q|^2v^3)^2}{\II^1} +\VVt^3+ \text{additional term}\ge 0\\
&\frac14\pr_r\big(|q|^2v^4\big)-\frac{1}{16}\frac{(|q|^2v^4)^2}{\II^4} +\VVt^4 + \text{additional term} \ge 0
\end{align*}
where the additional term is mainly due
to the mixed term $\VVt^2$.
\begin{remark}
\lab{remark:restriction-am} It is in fact the presence of this mixed term $\VVt^2$ that requires us to restrict the range of our result to $|a| /m \le 0.75$.
\end{remark}
To deal with this mixed term $\VVt^2$ we proceed as follows:
\begin{itemize}
\item We make use of the trapping quantity $\Phat$ in two different ways (see Proposition \ref{Prop:refinedP}). In the trapping region, we
 rewrite $\VVt^2T\Zhat\psi$ in terms of $\Phat$ and the more manageable terms $ TT\psi, \Zhat\Zhat\psi, \OO\psi$. Away from the trapping region, we use $\Phat$ to gain more positivity of $TT\psi, \OO\psi$. See Propositions \ref{Prop:finalHardysmalla} and \ref{Prop:finalHardylargea} in section \ref{section:Riccati}.

 \item When $|a|/m$ is relativity small, the way we handle $\VVt^3\Zhat\Zhat\psi$ (with $a^2$ in $\VVt^3$) is to write $\Zhat$ in terms of $T, |q|e_2$ and then focus on the analysis for $TT\psi, \OO\psi$. When $|a|/m$ is relatively large, the term $\VVt^3\Zhat\Zhat\psi$ is not negligible and we need to treat it separately. See Propositions \ref{Prop:finalHardysmalla} and \ref{Prop:finalHardylargea} in section \ref{section:Riccati}.
  \end{itemize}

{\bf Step 3.} Finally, we solve the above Riccati inequalties for continuous $v^1, v^3, v^4$. The details are given in section \ref{sec:solveRiccati}.

{\bf Step 4.} Since we choose $\LL=(m^2,0,1,0)$ (see \eqref{eq:remark:choiceofLL-weak}), the control of $|\Rhat^{\leq 1}\OO\psi|^2$ is missing. To handle this issue, we choose instead $\LL=(m^2,0,1,\sigma)$ where $\sigma$ is a sufficiently small constant. The details are given in section \ref{sec:OOpsi}.

\subsection{Step 1. Proof of \eqref{eq:K-Jt14short}}
 \lab{sec:I+Kt}
\begin{proposition}\lab{prop:I+Kt}
 Let $M^{\aund\bund}= v^{\aund\bund}\pr_r=v^{(\aund}\LL^{\bund)}\pr_r$ with
  \[
\LL^{\aund}=(m^2, 0, 1, 0) \quad \text{and}\quad |q|^2v^\aund=|q|^2(v^1, 0, v^3, v^4)
  \]
depending only on $r$. Then we have
\beq\label{eq:K-Jt14}
\bsplit
&I+\Kt-|q|^2\Div\Bk_{\widetilde{K}}\\
&\quad=\II^\aund\SS_\aund\pr_r\psi\c\LL^\bund(\pr_r\psi)_\bund+
\frac 1 4 |q|^2  \Div ((\psia\c\psib) M^{\aund\bund})-|q|^2\Div\Bk_{\widetilde{K}}\\
& \quad\geq \Big(\frac14\pr_r\big(|q|^2v^\aund\big)-\frac{1}{16}\frac{(|q|^2v^\aund)^2}{\II^\aund}\Big)\SS_\aund\psi\c\psi_{\LL}
\end{split}
\eeq
where $\Bk_{\widetilde{K}}$ are acceptable boundary terms in the sense of Definition \ref{Def:acceptable}.
\end{proposition}

\begin{proof}
Since $M^{\aund\bund}= v^{(\aund}\LL^{\bund)}\pr_r$ with $\LL^\bund$ being constants, using $[\pr_r,\SS_\aund]=0$ we compute
\begin{equation}\lab{eq:divofM11}
\bsplit
&\frac 1 4 |q|^2  \Div ((\psia\c\psib) M^{\aund\bund})\\
&\quad=\frac14|q|^2\Div\left[\Big(\big(v^\aund\psia\big)\c\big(\LL^\bund\psib\big)\Big)\pr_r\right]=\frac14\pr_r\Big(\big(|q|^2v^\aund\psia\big)\c\big(\LL^\bund\psib\big)\Big)\\
&\quad=\frac14|q|^2\Big(\big(v^\aund\SS_\aund\pr_r\psi\big)\c\big(\LL^\bund\psib\big)+\big(v^\aund\psia\big)\c\big(\LL^\bund\SS_\bund\pr_r\psi\big)\Big)+\frac14\psi_{\LL}\pr_r(|q|^2v^\aund)\psia.
\end{split}
\end{equation}
Choosing $v^\aund=(v^1,0,v^3,v^4)$, we write
\beq
\lab{eq:ISplusDivM11}
\II^\aund\SS_\aund\pr_r\psi\c\LL^\bund\SS_\bund\pr_r\psi    +\frac 1 4 |q|^2  \Div ((\psia\c\psib) M^{\aund\bund})=\Ik+\Kk
\eeq
where
\beq\lab{eq:BkIK}
\bsplit
\Ik&:=\II^\aund\Big(\SS_\aund\pr_r\psi+\frac14\frac{|q|^2v^\aund}{\II^\aund}\SS_\aund\psi\Big)\c\Big((\pr_r\psi)_{\LL}+\frac14\frac{|q|^2v^\aund}{\II^\aund}\psi_{\LL}\Big),\\
\Kk&:=\Big(\frac14\pr_r(|q|^2v^\aund)-\frac{(|q|^2v^\aund)^2}{16\II^\aund}\Big)\SS_\aund\psi\c\psi_\LL.
\end{split}
\eeq
As for the term $\Ik$, by applying Lemma \ref{lemma:IntegrationbypartsI} to $\AA^\aund=\SS_i^\aund, \BB^\aund=\LL^\aund, \pr_r\psi+\frac14\frac{|q|^2v^i}{\II^i}\psi$ for $i=1,3,4$, 
we derive by using the fact that $\II^1,\II^3,\II^4\geq 0$
\beq\lab{eq:controlofIk1}
\bsplit
\Ik-|q|^2\Div\Bk_{\widetilde{K}}&=\sum_{i=1,3,4}\II^iS_i^{\mu\nu}S_{\LL}^{\a\b}\pr_\a\pr_\mu(\pr_r\psi+\frac{1}{4}\frac{|q|^2v^i}{\II^i}\psi)\pr_\b\pr_\nu(\pr_r\psi+\frac{1}{4}\frac{|q|^2v^i}{\II^i}\psi)\\
&=\sum_{i=1,3,4}\II^iS_i^{\mu\nu}S_{\LL}^{\a\b}\pr_\a\pr_\mu\Phi_i\pr_\b\pr_\nu\Phi_i\geq 0
\end{split}
\eeq
where $\Phi_i=\pr_r\psi+\frac{1}{4}\frac{|q|^2v^i}{\II^i}\psi$ 
and $\Bk_{\widetilde{K}}$ is an acceptable boundary term. Putting \eqref{eq:ISplusDivM11}, \eqref{eq:BkIK} and \eqref{eq:controlofIk1} together finishes the proof of \eqref{eq:K-Jt14}.
\end{proof}


\subsection{Step 2. Setup of Riccati inequalities}
\lab{section:Riccati}

Before deriving our Riccati inequalties we summarize the results obtained so far.

We recall that $|q|^2\EE=\Pt+\It+\Jt+\Kt$, see \eqref{eq:separation-EE-I-J-K}. Using Propositions \ref{Prop:positivityP} and \ref{Prop:positivityI}, we see that
\beq\lab{eq:defofH}
\int_{\DD(\tau_1,\tau_2)}\Big(\EE+\Div \Bk \Big)\geq\int_{\DD(\tau_1,\tau_2)}\frac{1}{|q|^2}(\Phat+H),\quad H:=I+\Kt+\Jt.
\eeq
where $\Bk$ is an acceptable boundary term in the sense of Definition \ref{Def:acceptable}.

Using Proposition \ref{prop:I+Kt}, we obtain that modulo acceptable spacetime divergences
\beq\lab{eq:controlofH}
H \geq \HH^\aund\SS_\aund\psi\c\psi_{\LL}
\eeq
where, with $v^\aund=(v^1, 0,v^3,v^4)$ defined in Proposition \ref{prop:I+Kt},
\beq\lab{eq:defofHHhat}
\bsplit
\HH^{1}&=\frac14\pr_r\big(|q|^2v^1\big)-\frac{1}{16}\frac{(|q|^2v^1)^2}{\II^1}+\widetilde{\VV}^{1}\\
\HH^{3}&=\frac14\pr_r\big(|q|^2v^3\big)-\frac{1}{16}\frac{(|q|^2v^3)^2}{\II^3}+\widetilde{\VV}^{3}\\\HH^{4}&=\frac14\pr_r\big(|q|^2v^4\big)-\frac{1}{16}\frac{(|q|^2v^4)^2}{\II^4}+\widetilde{\VV}^{4}\\\HH^{2}&=\widetilde{\VV}^2.
\end{split}
\eeq
Therefore, to show the positivity of $\int_{\DD(\tau_1,\tau_2)}\EEt$, it suffices to show the positivity of
$ \int_{\DD(\tau_1,\tau_2)}|q|^{-2}(\Phat+\HH^\aund\SS_\aund\psi\c\psi_{\LL})$. 
In the next two propositions we show that the positivity follows from the following Riccati inequalitites.

\begin{proposition}
\lab{Prop:finalHardysmalla}
Let $\HH^\aund$ be defined as above in \eqref{eq:defofHHhat}. For $|a|/m\leq 0.5$, let $v^\aund=(v^1, 0, 0, v^4)$ be continuous functions (i.e. $\HH^3=\VVt^3$). If for $r\in[r_+, r_++0.01m]\cup[r_++0.7m,r_++2.4m]$
\beq
\lab{eq:Prop-finalHardy0.5cpt}
\HH^1+\JtP_{c}^1+a^2(\HH^3+\JtP_{c}^3)_-\geq0 \quad \text{and}\quad \HH^4+\JtP_{c}^4+(\HH^3+\JtP_{c}^3)_-\geq0,
\eeq
where $x_-=\min\{x, 0\}$, and for $r\in[r_++0.01m, r_++0.7m]\cup [r_++2.4m, \infty)$
\beq
\lab{eq:Prop-finalHardy0.5large}
\HH^1-\frac{3m}{2}|\HH^2|+a^2\HH^3_-+\JtP^1_{f}\geq 0\quad \text{and}\quad \HH^4-\frac{1}{6m}|\HH^2|+\HH^3_-+\JtP^4_{f}\geq0 ,
\eeq
then we have
\[
\int_{\DD(\tau_1,\tau_2)}\frac{1}{|q|^2}(\Phat+\HH^\aund\SS_\aund\psi\c\psi_{\LL}-|q|^2\Div{\Bk_{\HH}})\geq 0
\]
where $\Bk_{\HH}$ is an acceptable boundary term in the sense of Definition \ref{Def:acceptable}.
\end{proposition}

\begin{proposition}
\lab{Prop:finalHardylargea}
Let $\HH^\aund$ be defined as above in \eqref{eq:defofHHhat}. For $0.5\leq |a|/m\leq 0.75$, let $v^\aund=(v^1, 0, v^3, v^4)$ be continuous functions. If for $r\in[r_+,r_++0.01m]\cup[r_++0.2m, r_++2.4m]$
\beq
\lab{eq:Prop-finalHardy0.8cpt}
\HH^1+\JtP_{c}^1\geq0,  \qquad  \HH^3+\JtP_{c}^3\geq0, \qquad \HH^4+\JtP_{c}^4\geq0,
 \eeq
 and for $r\in[r_++0.01m, r_++0.2m]$
 \beq\lab{eq:Prop-finalHardy0.8extra}
 \bsplit
 &\HH^1-mf(r)|\VVt^2|\geq 0,\qquad \HH^3-\frac{g(r)}{m}|\VVt^2|\geq 0,\qquad \HH^4-\frac{h(r)}{m}|\VVt^2|\geq0,
 \end{split}
 \eeq
where $f(r)\geq |a|/m, \ g(r), h(r)\geq 0$ and $f\c(g+h) \geq 1/4$, and for $r\in[r_++2.4m, \infty)$
\beq
\lab{eq:Prop-finalHardy0.8large}
\bsplit
&\HH^1+\JtP^1_{f}-2m|\HH^2|\geq 0,\\
&\HH^3-\VVt^3+\frac{3}{10}\big(\VVt^4-\frac{1}{8m}|\HH^2|+\VVt^3_-\big)\geq 0,\\
&\HH^4+\JtP_{f}^4-\VVt^4+\frac{7}{10}\big(\VVt^4-\frac{1}{8m}|\HH^2|+\VVt^3_-\big)\geq0 ,
\end{split}
\eeq
then we have
\[
\int_{\DD(\tau_1,\tau_2)}\frac{1}{|q|^2}(\Phat+\HH^\aund\SS_\aund\psi\c\psi_{\LL}-|q|^2\Div{\Bk_{\HH}})\geq0
\]
where $\Bk_{\HH}$ is an acceptable boundary term in the sense of Definition \ref{Def:acceptable}.
\end{proposition}

\subsubsection{Proof of Propositions \ref{Prop:finalHardysmalla} and \ref{Prop:finalHardylargea}}
{ \bf Proof under the condition \eqref{eq:Prop-finalHardy0.5cpt} (or \eqref{eq:Prop-finalHardy0.8cpt})
} We only prove for the condition \eqref{eq:Prop-finalHardy0.5cpt} as the other case can be proved in the same way. Using \eqref{eq:TZhatPsi} in Proposition \ref{Prop:refinedP}, we obtain (modulo acceptable spacetime divergence)
\begin{align*}
\int_{\DD(\tau_1,\tau_2)}\frac{1}{|q|^2}(\Phat+\HH^2\psi_2\c\psi_{\LL})=\int_{\DD(\tau_1,\tau_2)}\frac{1}{|q|^2}(\Phat+\VVt^2T\Zhat\psi\c\psi_{\LL})\geq \int_{\DD(\tau_1,\tau_2)}\frac{1}{|q|^2}(\JtP_{c}^\aund\psia\c\psi_{\LL})
\end{align*}
and thus
\begin{align*}
&\int_{\DD(\tau_1,\tau_2)}\frac{1}{|q|^2}(\Phat+\HH^\aund\psia\c\psi_{\LL})\\
&\quad\geq \int_{\DD(\tau_1,\tau_2)}\frac{1}{|q|^2}\Big((\HH^1+\JtP^1)TT\psi+(\HH^3+\JtP^3)\Zhat\Zhat\psi+(\HH^4+\JtP^4)\OO\psi\Big)\c\psi_{\LL}.
\end{align*}
By applying Lemma \ref{lemma:IntegrationbypartsI} to $\AA^\aund=(\HH^1+\JtP^1_c,0, \HH^3+\JtP^3_c, \HH^4+\JtP^4_c), \BB^\aund=\LL^\aund, \psi$, we have
\[
\Phat+\HH^\aund\psia\c\psi_{\LL}-|q|^2\Div{\Bk_{\HH}}\geq\sum_{\aund=1,3,4}(\HH^\aund+\JtP_c^\aund) S_{\aund}^{\a\b}S_{\LL}^{\mu\nu}\pr_\a\pr_{\mu}\psi\pr_{\b}\pr_{\nu}\psi.
\]
By Lemma \ref{Lem:sumofsquares}, we see that the positivity of $\sum_{\aund=1,3,4}(\HH^\aund+\JtP_c^\aund)S_{\aund}^{\a\b}S_{\LL}^{\mu\nu}\pr_\a\pr_{\mu}\psi\pr_{\b}\pr_{\nu}\psi$ is equivalent to that of $\sum_{\aund=1,3,4}(\HH^\aund+\JtP_c^\aund)S_{\aund}^{\a\b}$. Since $|q|e_2=\frac{\Zhat}{\sin\th}-\frac{a\cos^2\th}{\sin\th} T$, we write
\begin{align*}
\sum_{\aund=1,3,4}(\HH^\aund+\JtP_c^\aund) S_\aund^{\mu\nu}=\begin{pmatrix}
F_1
&F_2&F_3& F_4
\end{pmatrix}\begin{pmatrix}
T^\mu T^\nu\\
2T^{(\mu}\Zhat^{\nu)}\\
\Zhat^\mu\Zhat^\nu\\
(|q|e_1)^\mu(|q|e_1)^\nu
\end{pmatrix}
\end{align*}
where
\begin{align*}
F_1&=    \HH^1+\JtP_c^1+\frac{a^2\cos^4\th}{\sin^2\th}(\HH^4+\JtP_c^4),\quad F_2=-\frac{a\cos^2\th}{\sin^2\th}(\HH^4+\JtP_c^4),\\
 F_3&=\HH^3+\JtP^3_c+\frac{1}{\sin^2\th}(\HH^4+\JtP_c^4),\quad F_4=\HH^4+\JtP_c^4.
\end{align*}
It is clear that the second condition in \eqref{eq:Prop-finalHardy0.5cpt} gives $F_4\geq0$ and $F_3\geq0$. We now check that the two conditions in \eqref{eq:Prop-finalHardy0.5cpt} imply the following
\begin{align*}
F_1F_3-F_2^2&=(\HH^1+\JtP_c^1)(\HH^3+\JtP^3_c+\HH^4+\JtP_c^4)\\
&+\frac{\cos^2\th}{\sin^2\th}(\HH^4+\JtP_c^4)\big(\HH^1+\JtP_c^1+a^2\cos^2\th(\HH^3+\JtP^3_c)\big)\geq 0.
\end{align*}
Therefore, the nonnegativity of $\sum_{\aund=1,3,4}(\HH^\aund+\JtP_c^\aund)S_{\aund}^{\a\b}$ follows.

{ \bf Proof under the condition \eqref{eq:Prop-finalHardy0.8large} (or \eqref{eq:Prop-finalHardy0.8extra} or \eqref{eq:Prop-finalHardy0.5large})} We only prove for the condition \eqref{eq:Prop-finalHardy0.8large} as the other two cases can be proved in the same way. Using \eqref{eq:Psi1larger}
 in Proposition \ref{Prop:refinedP}, we obtain
\begin{align*}
\int_{\DD(\tau_1,\tau_2)}\frac{1}{|q|^2}(\Phat+\HH^\aund\psia\c\psi_{\LL})\geq \int_{\DD(\tau_1,\tau_2)}\frac{1}{|q|^2}((\HH^\aund+\JtP_f^\aund)\psia\c\psi_{\LL})\quad \text{with}\quad \JtP_f^2=\JtP_f^3=0.
\end{align*}
Again, it suffices to prove that $ (\HH^\aund+\JtP_f^\aund)S_\aund^{\mu\nu}$ is nonnegative. Since $|q|e_2=\frac{\Zhat}{\sin\th}-\frac{a\cos^2\th}{\sin\th} T$, we write
\begin{align*}
(\HH^\aund+\JtP_f^\aund)S_\aund^{\mu\nu}=\begin{pmatrix}
F_1
&F_2&F_3& F_4
\end{pmatrix}\begin{pmatrix}
T^\mu T^\nu\\
2T^{(\mu}\Zhat^{\nu)}\\
\Zhat^\mu\Zhat^\nu\\
(|q|e_1)^\mu(|q|e_1)^\nu
\end{pmatrix}
\end{align*}
where
\begin{align*}
F_1&=    \HH^1+\JtP^1_f+\frac{a^2\cos^4\th}{\sin^2\th}(\HH^4+\JtP_f^4),\quad F_2=\frac12\HH^2-\frac{a\cos^2\th}{\sin^2\th}(\HH^4+\JtP_f^4),\\
 F_3&=\HH^3+\frac{1}{\sin^2\th}(\HH^4+\JtP_f^4),\quad F_4=\HH^4+\JtP_f^4.
\end{align*}
It is clear that the first and the third conditions in \eqref{eq:Prop-finalHardy0.8large} (combined with $\VVt^4>0$ for $r\geq r_++2.4m$) give $F_4\geq0$ and $F_1\geq0$. We now check that the three conditions in \eqref{eq:Prop-finalHardy0.8large} imply the following
\begin{align*}
F_1F_3-F_2^2&=(\HH^1+\JtP_f^1)(\HH^3+\HH^4+\JtP_f^4)-\frac14(\HH^2)^2\\
&+\frac{\cos^2\th}{\sin^2\th}(\HH^4+\JtP_f^4)\big(\HH^1+\JtP_f^1+a^2\cos^2\th\HH^3+a\HH^2\big)\\
&\geq \frac{\cos^2\th}{\sin^2\th}(\HH^4+\JtP_f^4)\Big(2m|\HH^2|+a^2\cos^2\th(\frac{7}{10}\VVt^3_--\frac3{10}\VVt^4)+a\HH^2\Big)\\
&\geq\frac{\cos^2\th}{\sin^2\th}(\HH^4+\JtP_f^4)(m|\HH^2|+a^2(\frac{7}{10}\VVt^3_--\frac3{10}\VVt^4))\geq 0
\end{align*}
where we use the three conditions in \eqref{eq:Prop-finalHardy0.8large} in the second step and the facts $m|\HH^2|=m|\VVt^2|\geq 0.3a^2\VVt^4-0.7a^2\VVt^3_-$ (which can be verified easily using Mathematica) for $r\geq r_++2.4m$ in the last step. This proves the nonnegativity of $ (\HH^\aund+\JtP_f^\aund)S_\aund^{\mu\nu}$ and thus the nonnegativity of $\sum_{\aund=1,3,4}(\HH^\aund+\JtP_f^\aund)S_{\aund}^{\a\b}$ follows.
\subsection{Step 3. Solve the Riccati inequalities}\lab{sec:solveRiccati}
In this section, we will solve the Riccati inequalities derived in Propositions \ref{Prop:finalHardysmalla} and \ref{Prop:finalHardylargea} for \textit{continuous} $v^\aund=(v^1, 0, 0,v^4)$ resp. $v^\aund=(v^1, 0, v^3,v^4)$ for $|a|/m\leq 0.5$ resp. $0.5\leq |a|/m\leq 0.75$.

\subsubsection{Case $0\leq |a|/m\leq 0.5$}
We solve the following Riccati ineuqalities for continuous $\widetilde{v}^1=\frac14|q|^2v^1$ and $\widetilde{v}^4=\frac14|q|^2v^4$
\beq\lab{eq:Hardyforv1v4leq0.5}
\pr_r\widetilde{v}^1-\frac{(\widetilde{v}^1)^2}{\II^1}+\WW^1\geq0,\qquad \pr_r\widetilde{v}^4-\frac{(\widetilde{v}^1)^2}{\II^4}+\WW^4\geq0
\eeq
where (by choosing $f(r)=\frac53,\ g(r)=\frac{m\AAt^3}{\AA^2[z]},\ h(r)=\frac3{20}$ in Proposition \ref{Prop:positivityI})
\[
\II^1=\AAt^1[z]-\frac{5m}{3}\AA^2[z],\quad \II^4=\AA^4[z]-\frac{3}{20m}\AA^2[z]
\]
and (see \eqref{eq:Prop-finalHardy0.5cpt} and \eqref{eq:Prop-finalHardy0.5large})
\begin{align*}
\WW^1&=\begin{cases}
\VVt^1+\JtP^1_c+a^2(\VVt^3+\JtP^3_c)_-\quad &r\in[r_+, r_++0.01m]\cup[r_++0.7m, r_++2.4m]\\
\VVt^1-\frac{3m}{2}|\VVt^2|+a^2\VVt^3_-+\JtP^1_f\quad &r\in[r_++0.01m, r_++0.7m]\cup[r_++2.4m,\infty)\end{cases}\\
 \WW^4&=\begin{cases}
 \VVt^4+\JtP^4_c+(\VVt^3+\JtP^3_c)_-\quad &r\in[r_+, r_++0.01m]\cup[r_++0.7m, r_++2.4m]\\
\VVt^4-\frac{1}{6m}|\VVt^2|+\VVt^3_-+\JtP^4_f\quad &r\in[r_++0.01m, r_++0.7m]\cup[r_++2.4m,\infty)\end{cases}.
\end{align*}
Introducing the translated variable $x=r-r_+$, we write the coefficients $\II^1, \II^4, \WW^1, \WW^4$ in terms of $x$. Our strategy is as follows.

\begin{enumerate}
\item {\bf Solve the Riccati equations for $x\in[0,0.01m]$.} Since $\II^1,\II^4=O(\De)$ and $\WW^1, \WW^4=O(\ln(\frac{\De}{m^2}))$ near the event horizon $r=r_+$, we set $\widetilde{v}^1(x)=\widetilde{v}^1_{1}(x)$ and $\widetilde{v}^4(x)=\widetilde{v}^4_{1}(x)$ where
\beq\lab{eq:wtv11v41}
\bsplit
\widetilde{v}^1_1(x)&\geq \frac{11}{10}c_1(r-r_+)\ln(\frac{\De}{m^2}),\quad c_1=(\frac{-\WW^1}{\ln(\frac{\De}{m^2})})|_{r=r_+},\\
\widetilde{v}^4_1(x)&\geq \frac{11}{10}d_1(r-r_+)\ln(\frac{\De}{m^2}),\quad d_1=(\frac{-\WW^4}{\ln(\frac{\De}{m^2})})|_{r=r_+}.
\end{split}
\eeq
We can check by using Mathematica that for $r\in[r_+, r_1^1=r_++\frac{1}{100}]$
\beq
\pr_r\widetilde{v}_1^1-\frac{(\widetilde{v}_1^1)^2}{\II^1}+\WW^1>0,\qquad \pr_r\widetilde{v}_1^4-\frac{(\widetilde{v}_1^4)^2}{\II^4}+\WW^4>0.
\eeq

\item {\bf Solve the Riccati equations for $x\geq 0.01m$.} We then solve the inhomogeneous Riccati inequalitites below, with the initial data at $r=r_1^1=r_++\frac{1}{100}$.
\beq\lab{eq:RiccatiIVP}
\begin{cases}
&\pr_r\widetilde{v}^1\geq\frac{(\widetilde{v}^1)^2}{\II^1}-\WW^1\\
&\widetilde{v}^1(r_1^1)=\widetilde{v}^1_1(r_++\frac{1}{100})
\end{cases},\qquad
\begin{cases}
&\pr_r\widetilde{v}^4\geq\frac{(\widetilde{v}^4)^2}{\II^4}-\WW^4\\
&\widetilde{v}^4(r_1^4)=\widetilde{v}^4_1(r_++\frac{1}{100})
\end{cases}.
\eeq
Our strategy to find \textit{continuous} solutions to the above two initial value probelms is:
\begin{itemize}
\item First partition the interval of $x$ into a convenient set of subintervals.
 \item On each subinterval, we define $\widetilde{v}^1,\ \widetilde{v}^4$ as affine functions.
 \item We verify that the functions $\widetilde{v}^1,\ \widetilde{v}^4$, such defined, satisfy the Riccati inequalities \eqref{eq:RiccatiIVP}, by using interval arithmetic.
 \item We continue the process until we reach an interval where there is a point $r_2^1$ (resp. $r_2^4$) at which the constructed affine functions are zero.
\end{itemize}

\item {\bf Obtain a compactly supported continuous solution.} We finally define
\beq
\widetilde{v}^1(r)=\begin{cases}
\widetilde{v}_1^1(r)\quad &r_+\leq r\leq r_1^1=r_++\frac{1}{100}\\
\widetilde{v}^1_2(r)\quad &r_1^1\leq r\leq r^1_2\\
0\quad &r\geq r_2^1
\end{cases},
\eeq
and
\beq
\widetilde{v}^4(r)=\begin{cases}
\widetilde{v}_1^4(r)\quad &r_+\leq r\leq r_1^4=r_++\frac{1}{100}\\
\widetilde{v}^4_2(r)\quad &r_1^4\leq r\leq r^4_2\\
0\quad &r\geq r_2^4
\end{cases},
\eeq
where $\widetilde{v}^1_1(x),\ \widetilde{v}^4_1(x)$ are defined as in \eqref{eq:wtv11v41} and $\widetilde{v}^1_2(x),\ \widetilde{v}^4_2(x)$ are solutions of the initial value problems
\eqref{eq:RiccatiIVP}
with $\widetilde{v}^1_2(r_2^1)=\widetilde{v}^4_2(r_2^4)=0$.
\end{enumerate}

\subsubsection{Case $0.5\leq |a|/m\leq 0.75$}
We solve the following Riccati inequalities for continuous $\widetilde{v}^1=\frac14|q|^2v^1$, $\widetilde{v}^3=\frac14|q|^2v^3$ and $\widetilde{v}^4=\frac14|q|^2v^4$
\beq\lab{eq:Hardyforv1v4leq0.75}
\pr_r\widetilde{v}^1-\frac{(\widetilde{v}^1)^2}{\II^1}+\WW^1=0,\qquad \pr_r\widetilde{v}^3-\frac{(\widetilde{v}^3)^2}{\II^3}+\WW^3=0,\qquad \pr_r\widetilde{v}^4-\frac{(\widetilde{v}^1)^2}{\II^4}+\WW^4=0
\eeq
where (choosing $f(r)=\frac{25}{24}, g(r)=\frac{4a^2}{15m^2}, h(r)=\frac{6}{25}-\frac{4a^2}{15m^2}$ for $r_+\leq r\leq r_++2.4m$ and $f(r)=2, g(r)=0, h(r)=\frac18$ for $r\geq r_++2.4m$ in Proposition \ref{Prop:positivityI})
\beq\lab{eq:IIhatSlargeastep}
\bsplit
&\II^1=\begin{cases}
\AAt^1[z]-\frac{25m}{24}\AA^2[z] \qquad &r\in[r_+, r_++2.4m]\\
\AAt^1[z]-2m\AA^2[z]\qquad &r\in[r_++2.4m, \infty)
\end{cases},\\
& \II^3=\begin{cases}
\AA^3-\frac{4a^2}{15m^3}\AA^2[z]\qquad &r\in[r_+, r_++2.4m]\\
\AA^3[z]\qquad &r\in[r_++2.4m, \infty)
\end{cases} ,\\
&\II^4=\begin{cases}
\AA^4[z]-\frac{\frac{6}{25}-\frac{4a^2}{15m^2}}{m}\AA^2[z]\qquad &r\in[r_+, r_++2.4m]\\
\AA^4[z]-\frac{1}{8m}\AA^2[z]\qquad &r\in[r_++2.4m, \infty)
\end{cases},
\end{split}
\eeq
and (see \eqref{eq:Prop-finalHardy0.8cpt}, \eqref{eq:Prop-finalHardy0.8extra} and \eqref{eq:Prop-finalHardy0.8large} and we choose $f(r)=\frac{25m}{44a}, g(r)=\frac{a}{3m}, h(r)=\frac{8a}{75m}$ in \eqref{eq:Prop-finalHardy0.8extra})
\begin{align*}
\WW^1&=\begin{cases}
\VVt^1+\JtP^1_c\quad &r\in[r_+, r_++0.01m]\cup[r_++0.2m, r_++2.4m]\\
\VVt^1-\frac{25m}{44a}m|\VVt^2|\quad & r\in[r_++0.01m, r_++0.2m]\\
\VVt^1-2m|\VVt^2|+\JtP^1_f\quad &r\in[r_++2.4m, \infty)\end{cases},\\
\WW^3&=\begin{cases}
\VVt^3+\JtP^3_c\quad &r\in[r_+, r_++0.01m]\cup[r_++0.2m, r_++2.4m]\\
\VVt^3-\frac{a}{3m^2}|\VVt^2|\quad & r\in[r_++0.01m, r_++0.2m]\\
\frac{3}{10}(\VVt^4-\frac{1}{8m}|\VVt^2|+\VVt^3_-)\quad &r\in[r_++2.4m,\infty)\end{cases},\\
 \WW^4&=\begin{cases}
 \VVt^4+\JtP^4_c\quad &r\in[r_+, r_++0.01m]\cup[r_++0.2m, r_++2.4m]\\
 \VVt^4-\frac{8a}{75m^2}|\VVt^2|\quad & r\in[r_++0.01m, r_++0.2m]\\
\frac{7}{10}(\VVt^4-\frac{1}{8m}|\VVt^2|+\VVt^3_-)+\JtP^4_f\quad &r\in[r_++2.4m,\infty)\end{cases}.
\end{align*}Then we follows the same strategy as in the case $|a|/m\leq 0.5$.

This finishes the proof of \eqref{eq:prop-I+Kt-short} in Proposition \ref{prop:I+Kt-short} with $\aund$ summed over $1,2,3$.

\subsection{Step 4. Control of $\OO\psi$}
 \lab{sec:OOpsi}
To recover the control of $|\Rhat^{\leq 1}\OO\psi|^2$ in Proposition \ref{prop:I+Kt-short}, we choose instead $\LL=(m^2,0,1,\sigma)$ where $\sigma$ is a sufficiently small constant. Now we analyze the term $\EEt$ with the choice $\LL=(m^2,0,1,\sigma)$. Compared to the choice of $\LL=(m^2,0,1,0)$, the following additional terms (modulo acceptable spacetime divergence) are generated
\beq\lab{eq:additionalEEt}
\bsplit
&\Pt:\quad \frac 1 2  \sigma h \Big(O^{\a\b}\D_\a \Psi_1 \D_\b\Psi_1 +\de_0m^2zO^{\a\b}\D_\a \Psi_2 \D_\b\Psi_2\Big)\\
&\qquad\geq \frac 1 4  \sigma h \Big(|\pr_\th\Psi_1|^2+\frac{1}{\sin^2\th}|Z\Psi_1|^2\Big)-\frac 1 2  \sigma h|aT\Psi_1|^2,\\
&\It: \quad (\AAt^\aund\pr_r\psia)\c\sigma\pr_r\OO\psi=\sigma\AAt^4|\pr_r\OO\psi|^2+\sigma \sum_{\aund=1,2,3}O(r^{2}\De^{-1}|\Rhat\psia\c\Rhat\OO\psi|),\\
&\Jt: \quad (\VVt^\aund\psia)\c\sigma\OO\psi=\sigma\VVt^4|\OO\psi|^2+\sigma\VVt^2T\Zhat\psi\c\OO\psi+\sigma \sum_{\aund=1,3}O((r^{-2}+r^{-3}|\ln\De|)|\psia\c\OO\psi|),\\
&\Kt: \quad \frac14|q|^2\Big(\big(v^\aund\SS_\aund\pr_r\psi\big)\c\sigma\OO\psi+\big(v^\aund\psia\big)\c\big(\sigma\OO\pr_r\psi\big)\Big)+\frac14\sigma\OO\psi\pr_r(|q|^2v^\aund)\psia\\
&\qquad=\frac12\sigma|q|^2v^4\OO\pr_r\psi\c\OO\psi+\frac14\sigma\pr_r(|q|^2v^4)|\OO\psi|^2+\sigma \sum_{\aund=1,2,3}O(r^{-3}|\ln\De||\psia\c\OO\psi|)\\
&\qquad+\sigma \sum_{\aund=1,2,3}O\Big(r^{-3}|\ln\De|(|\Rhat\psia\c\OO\psi|+|\psia\c\Rhat\OO\psi|)\Big).
\end{split}
\eeq
Using
\begin{align*}
&\AAt^4|\pr_r\OO\psi|^2+\VVt^4|\OO\psi|^2+\frac12|q|^2v^4\OO\pr_r\psi\c\OO\psi+\frac14\pr_r(|q|^2v^4)|\OO\psi|^2\\
&=\AAt^4\big|\pr_r\OO\psi+\frac{|q|^2v^4}{4\AAt^4}\OO\psi\big|^2+\big(\frac14\pr_r(|q|^2v^4)-\frac{(|q|^2v^4)^2}{16\AAt^4}+\VVt^4\big)|\OO\psi|^2\\
&\geq \big(\frac14\pr_r(|q|^2v^4)-\frac{(|q|^2v^4)^2}{16\AAt^4}+\VVt^4\big)|\OO\psi|^2
\end{align*}
and applying Proposition \ref{Prop:refinedP} with $\psi_{\LL}$ there replaced by $\sigma\OO\psi$, we derive
\beq\lab{eq:controlofOOpsi}
\bsplit
&\int_{\DD(\tau_1,\tau_2)}\frac{\sigma}{|q|^2}\Big(\AAt^4|\pr_r\OO\psi|^2+\VVt^4|\OO\psi|^2+\frac12|q|^2v^4\OO\pr_r\psi\c\OO\psi+\frac14\pr_r(|q|^2v^4)|\OO\psi|^2\Big)\\
&\quad+\int_{\DD(\tau_1,\tau_2)}\frac{1}{|q|^2}\Big(\sigma\VVt^2T\Zhat\psi\c\OO\psi+\Pt\Big)\\
&\quad\geq \int_{\DD(\tau_1,\tau_2)}\frac{\sigma}{|q|^2} \big(\frac14\pr_r(|q|^2v^4)-\frac{(|q|^2v^4)^2}{16\AAt^4}+\WW^4\big)|\OO\psi|^2\\
&\quad-O(\sigma)\int_{\DD(\tau_1,\tau_2)}\frac{1}{|q|^2} \big(h|aT\Psi_1|^2+|T\Zhat\psi|\c|\OO\psi|\big)\\
&\quad\gtrsim \int_{\DD(\tau_1,\tau_2)}\frac{\sigma}{|q|^2} \big(r^{-2}\De|\Rhat\OO\psi|^2+(r^{-2}+r^{-3}|\ln\De|)|\OO\psi|^2)\big)\\
&\quad-O(\sigma)\int_{\DD(\tau_1,\tau_2)}\frac{1}{|q|^2} \big(h|aT\Psi_1|^2+|T\Zhat\psi|\c|\OO\psi|\big)\end{split}
\eeq
where we use $\AAt^2\geq \II^4$ in the last step. Putting \eqref{eq:additionalEEt} and \eqref{eq:controlofOOpsi} together, we conclude that the additional terms generated satisfy
\begin{align*}
&\int_{\DD(\tau_1,\tau_2)}\frac{1}{|q|^2}\c \text{additional terms}\\
&\quad\gtrsim \int_{\DD(\tau_1,\tau_2)}\frac{1}{|q|^2}\sigma\big(r^{-2}\De|\Rhat\OO\psi|^2+(r^{-2}+r^{-3}|\ln\De|)|\OO\psi|^2)\big)\\
&\quad-\int_{\DD(\tau_1,\tau_2)}\frac{1}{|q|^2}\Big(\sigma|aT\Psi_1|^2+\sum_{\aund=1,2,3}\sigma^2\big(r^{-2}\De|\Rhat\psia|^2+(r^{-2}+r^{-3}|\ln\De|)|\psia|^2)\big)\Big).\end{align*}
Combining this with \eqref{eq:prop-I+Kt-short} with $\aund$ summed over $1,2,3$, and choosing $\sigma$ sufficiently small, we finish the proof of Proposition \ref{prop:I+Kt-short}.

\subsection{End of the proof of Theorem \ref{Thm:Moraw4}}
\lab{sec:SStemper}

Proposition \ref{prop:I+Kt-short} and its Corollary \ref{cor:I+Kt-Nontrap-short} establishes the validity of the irregular version of Theorem
\ref{Thm:Moraw4}. It only remains to pass from the irregular to the regular version by using a method similar to that used in the scalar case.

We define $h_{\ep}$ by changing $h=\frac12r^5(1-\frac{m^4}{r^4}\ln(\frac{\De}{m^2}))$ so that it is bounded near the event horizon. Let $F: \mathbb{R} \to \mathbb{R}$ be a smooth function satisfying $F(x) = x$ for $x \leq 1$ and $F(x) = 2$ for $x \geq 3$. We use $F(x)$ to temper $h$ by defining
\beq\lab{eq:defofhtemperSS}
h_{\ep} = \begin{cases}
\ep^{-1} F(\ep h)\quad &r\leq r_++\ep\\
h\quad &r\geq r_++\ep
\end{cases}.
\eeq
where $\ep> 0$ is a small constant to be chosen later. Then $\X_{\ep}, \w_{\ep}, \w_{\flat,\ep}$ can be defined correspondingly. We note that $\X_{\ep}, \w_{\ep}, \w_{\flat,\ep}$ agree with $\X, \w, \w_{\flat}$ in the region where $h<\ep^{-1}$ or $r\geq r_++\ep$. So $\X_{\ep}, \w_{\ep}, \w_{\flat,\ep}$ only differ from $\X, \w, \w_{\flat}$ in the region $\{h>\ep^{-1}\}\cap\{r\leq r_++\ep\}$, which is a small neighborhood of the event horizon. Since $F(x) \le 2$, it follows that $h_{\ep}\leq 2\ep^{-1}$ in the region $r\leq r_++\ep$, which means that for all small $\ep> 0$, the vectorfield $\X_{\ep}$ and the functions $\w_{\ep}, \w_{\flat, \ep}$ are regular up to the event horizon.

Since $|q|^2v^{\aund} \sim -(r-r_+)\ln(\frac{\De}{m^2})$ for $r\in[r_+, r_++0.01m]$, we temper $v^\aund$ by defining
\beq
\lab{eq:defofvaundtemper}
v^\aund_{\ep} = \begin{cases}
\ep^{-1} \frac{r-r_+}{|q|^2}F(\ep \frac{|q|^2v^\aund}{r-r_+})\quad &r\leq r_++\ep\\
v^{\aund} \quad &r\geq r_++\ep
\end{cases}.\eeq

Recall that
\beq\lab{eq:separation-EE-I-J-K-ep}
    |q|^2\EEt_{\ep}=|q|^2\widetilde{\EE}[\bold{X}_{\ep}, \bold{w}_{\ep}+\w_{\flat,\ep}, \M_{\ep}]   =\Pt_{\ep}+\It_{\ep}+\Jt_{\ep}+\Kt_{\ep}
    \eeq
where $\Pt_{\ep}$ takes the same form as $\Pt$ with $h$ replaced by $h_{\ep}$ and
\begin{align*}
\It_{\ep}&= \big(\widetilde{\AA}_{\ep}^{\aund}\pr_r \psia\big)\c \big( \LL^{\aund}\pr_r \psia\big),\quad
\widetilde{\AA}_{\ep}^{\aund}=-\frac{\De^2}{r^2}\pr_r \Big(  \frac{ h_{\ep}   \UUwtp^{\aund}  }{r^2} \Big)+\frac12h_{\ep}(\de_0 m^2\frac{\De\RRtp^4[z]}{2r^4})^2 \GG^{\aund},\\
\Jt_{\ep}&= \big(\widetilde{\VV}_{\ep}^{\aund} \psia\big)\c \big( \LL^{\aund} \psia\big),\quad
 \VVt_{\ep}^\aund=\frac14\pr_r(\De\pr_r(z\pr_r(h_{\ep}\UUwtp^{\aund})))-\frac14\pr_r \left(\De \pr_r \Big(h_{\ep}\De(\de m^2\frac{\RRtp^4[z]}{2r^4})^2\GG^{\aund}\Big)\right),\\
 \Kt_{\ep}&=\frac 1 4 |q|^2  \Div ((\psia\c\psib) v_{\ep}^{(\aund}\LL^{\bund)}\pr_{r}).
   \end{align*}

First, $h_{\ep}$ also satisfies the property that $h_{\ep}\geq 0$ for $r\geq r_+$ and thus $\Pt_{\ep}$ is (modulo acceptable spacetime divergence\footnote{\lab{fn:acceptable}Since the multipliers are regular, the spacetime divergence is acceptable in the sense of Definition \ref{Def:acceptable} where $\BB^\mu\les\sum||(e_4,\De e_3, \nab)^{\leq 1}\psia|^2$ near $r=r_+$.}) nonnegative.

Second, for $h\leq \ep^{-1}$ or $ r\geq r_++\ep$, we have $ \AAt^{\aund}_{\ep}=\AAt^\aund$. In the region $\{h_\ep>\ep^{-1} \}\cap \{r\leq r_++\ep\}$, we compute
\begin{align*}
\AA_{\ep}^{\aund}&=h_{\ep}\frac{\De^2}{r^2}\Big(-\pr_r\big(  \frac{  \UUwtp^{\aund}  }{r^2} \big)+\frac18(\de_0 m^2\frac{\RRtp^4[z]}{r^3})^2 \GG^{\aund}\Big)+\frac{\De^2}{r^2}F'(\ep h)\c(-\pr_rh)\c\frac{\UUwtp^\aund}{r^2}\\
&:=h_{\ep}\frac{\De^2}{r^2}\BB^\aund+\frac{\De^2}{r^2}F'(\ep h)\c(-\pr_rh)\CC^\aund.
\end{align*}
Since $h_{\ep}>0, F'(\ep h)\c(-\pr_rh)\geq 0$, we now prove that $\BB^\aund S_{\aund}^{\mu\nu}$ and $\CC^{\aund} S_{\aund}^{\mu\nu}$ are positive definite in the region $\{h_\ep>\ep^{-1} \}\cap \{r\leq r_++\ep\}$, from which the positivity of $\It_{\ep}$ there follows. It is clear that the positivity of $\BB^\aund S_{\aund}^{\mu\nu}$ and $\CC^{\aund} S_{\aund}^{\mu\nu}$ follows from the following
\[
\BB^4, \BB^1>0,\quad \BB^1\BB^3-\frac14(\BB^2)^2>0;\qquad \CC^4, \CC^1>0,\quad \CC^1\CC^3-\frac14(\CC^2)^2>0.
\]
Since the region $\{h_\ep>\ep^{-1} \}\cap \{r\leq r_++\ep\}$ is a small neighborhood of the event horizon, it suffices to check the above inequalities at $r=r_+$ for $|a|/m\leq 0.75$. The verification of the above inequalities at $r+=r_+$ can be done easily using Mathematica.

Third, for $h\leq \ep^{-1}$ or $ r\geq r_++\ep$, we have $ \VVt^{\aund}_{\ep}=\VVt^\aund$. In the region $\{h_\ep>\ep^{-1} \}\cap \{r\leq r_++\ep\}$, a direct calculation implies that
\begin{align*}
|\VVt^\aund_{\ep}|&\les h_{\ep}+F'+|F''|\ep+|F'''|\ep^2\De^{-1}\\
&\les h_{\ep}+ \ep^2 (r-r_+)^{-1}\chi_{\ep}(r)
\end{align*}
where $\chi_{\ep}$ is some bounded function supported where $-\ln(\frac{r-r_+}{m})\in [c_1\ep^{-1}, c_2\ep^{-1}]$, that is, $(r-r_+)\in[me^{-\frac{c_2}{\ep}}, me^{-\frac{c_1}{\ep}}]$. Therefore
\[
   \int |\VVt^\aund_{\ep}|\,dr\les \epsilon(\ep^{-2}e^{-c_1/\epsilon} + c_2 - c_1)\les \ep.
\]
Finally, for $h\leq \ep^{-1}$ or $ r\geq r_++\ep$, we have $ \Kt_{\ep}=\Kt$. In the region $\{h_\ep>\ep^{-1} \}\cap \{r\leq r_++\ep\}$, a direct calculation implies that
\begin{align*}
\Kt_{\ep}&=\frac 1 4 |q|^2  \Div ((\psia\c\psib) v_{\ep}^{(\aund}\LL^{\bund)}\pr_{r})\\
&=\frac14\Big(\big(|q|^2v_{\ep}^\aund\SS_\aund\pr_r\psi\big)\c\big(\LL^\bund\psib\big)+\big(|q|^2v_{\ep}^\aund\psia\big)\c\big(\LL^\bund\SS_\bund\pr_r\psi\big)\Big)+\frac14\psi_{\LL}\pr_r(|q|^2v_{\ep}^\aund)\psia
\end{align*}
and thus
\begin{align*}
|\Kt_{\ep}|&\les \ep^{-1}F\sum_{\aund=1}^4(|\Rhat\psia\c\psi_{\LL}|+|\Rhat\psi_{\LL}\c\psia|)+(\ep^{-1}F+F')\sum_{\aund=1}^4|\psia\c\psi_{\LL}|\\
&\les\ep\It_{\ep}+(\ep^{-2}+\ep^{-1}F+F')\sum_{\aund=1}^4|\psia\c\psi_{\LL}|
\end{align*}
where
\[
\int_{\{h_\ep>\ep^{-1} \}\cap \{r\leq r_++\ep\}}(\ep^{-2}+\ep^{-1}F+F')\les \ep.
\]
To summarize, for $h\leq \ep^{-1}$ or $ r\geq r_++\ep$, we have
\[
\Pt_{\ep}+\It_{\ep}+\Jt_{\ep}+\Kt_{\ep}=\Pt+\It+\Jt+\Kt.
\]
In the region $\{h_\ep>\ep^{-1} \}\cap \{r\leq r_++\ep\}$
\[
\Pt_{\ep},\It_{\ep}\geq 0,\quad \frac12\It_{\ep}+\Jt_{\ep}+\Kt_{\ep}\gtrsim \sum_{\aund=1}^4Err^\aund|\psia\c\psi_{\LL}|
\]
where
\beq
\lab{eq:SSerror}
\int_{\{h_\ep>\ep^{-1} \}\cap \{r\leq r_++\ep\}}|Err^\aund|\les \ep.
\eeq
Combining the estimate of $\EEt_{\ep}$ with Theorem \ref{Thm:Moraw3} applied to $\psia$ (localized in the region $h\leq \ep^{-1}$ or $ r\geq r_++\ep$), we obtain the control of $|(\That,\nab)\psia|^2$ in the intersection of the nontrapping region and the region $\{h\leq \ep^{-1}\}\cup \{r\geq r_++\ep\}$. Then the term \eqref{eq:SSerror} can be handled in the same way as in the scalar case (see section \ref{sec:controloflowerterm}).

To finish the proof of Theorem \ref{Thm:Moraw4}, we integrate $\EEt_{\ep}$ and apply the divergence Lemma in $\DD(\tau_1,\tau_2)$. This will generate the boundary terms which are either acceptable, in the sense of definition \ref{Def:acceptable} (see also footnote \ref{fn:acceptable}), or satisfy the same property discussed in section \ref{section:boundary}. Therefore, all the boundary terms can be controlled by
\beq
\sum_{\aund=1}^4\Big(\EFdeg[\psia](\tau_1,\tau_2)+E[\psia](\tau_1)\Big).
\eeq
 This ends the proof of Theorem \ref{Thm:Moraw4}.

\section{Proof of Proposition  \ref{Prop.1.6-Intro}  and Theorem  \ref{Thm:Morawetz2}}
\lab{section:Proof of Thm-Morawetz2}


In this section we provide the proofs for  Proposition \ref{Prop.1.6-Intro} and the conditional Morawetz-energy estimate of
 Theorem \ref{Thm:Morawetz2}.


\subsection{Proof of Proposition \ref{Prop.1.6-Intro}}
\lab{subsection:Prop.1.6-intro}

\lab{section:Proof-Prop.1.6-Intro}
We prove the following more precise estimate
\beq
\lab{eq:Proof-Prop.1.6-Intro}
\bsplit
\EFdeg^s_p[\psi](\tau_0, \tau) &\les    \int_{\Si(\tau)\cap\{r\leq r_0\}}|\dk^{s} Z\psi|^2  + r_0^{1-p} \Bdeg^s_p[\psi](\tau_0, \tau) \\
&+ r_0^{3-p} \Bdeg_p^{s-1}[\T\psi] (\tau_0, \tau)+\Edeg^s[\psi](\tau_0)+ \NN_p^s[\psi,N](\tau_0, \tau).
\end{split}
\eeq
To carry fewer terms we prove below the result for the homogenous wave equation $\square_{a,m}\psi=0$, the
 general case requires no modifications.

We apply the standard energy estimate to $\square_{a,m}\psi=0$ with the vectorfield
\beq
 \lab{def:Tring}
\Tring=\chi(r)\T_+    +(1-\chi(r))T=T+\chi(r)\frac{a}{r_+^2+a^2}Z.
\eeq
 where $0\leq \chi(r)\leq 1$ is a smooth nonnegative cutoff function satisfying
 \beq
  \lab{def:Tring-chit}
  \chi(r) = \begin{cases}  1\qquad  &r_+ \le r\leq   r_0/ 2  \\
0\qquad &r\geq      r_0
  \end{cases}
 \eeq
  where $r_0\ge \rhat_2$ so that $\Tring=\T_{\om_+} =\T+\frac{a}{r_+^2+a^2} \Z$ is Killing on the trapping set.
  Moreover we have $\chi'(r)=O(r_0^{-1}) $. Note however that $\Tring $ could become spacelike for values of $r$ for which $\Tring=\T_{\om_+}$ is not timelike.

We deduce, with $\QQ=\QQ[\psi]$,
 \beq
 \lab{eq:EFE2-Tring}
\bsplit
\int_{\HH(\tau_0, \tau)}  \QQ(\Tring, N_{\HH} )+\int_{\Si(\tau)} \QQ(\Tring, N_\Si)=\int_{\Si(\tau_0)}  \QQ(\Tring, N_\Si)- \frac 12 \int_{\DD(\tau_0,\tau)} \QQ^{\a\b}\piring_{\a\b} .
\end{split}
\eeq
where $\piring$ is the deformation tensor of $\Tring$.
Note that $ \QQ[\psi](\Tring, N_{\HH} )= |\nab_4 \psi|^2$
thus the flux term has the form $ \int_{\HH} |\nab_4 \psi|^2 $.
On the other hand, since $\Tring$ can become spacelike in the region $r\leq r_0$,
\[
\int_{\Si(\tau)} \QQ(\Tring, N_\Si)\ge \Edeg[\psi](\tau) - \int_{\Si(\tau)\cap\{r\leq r_0\}}|Z\psi|^2.
\]
  Hence
  \[
   \int_{\HH(\tau_0,\tau)}  |\That  \psi|^2 + \Edeg[\psi](\tau) \les  \Edeg[\psi](\tau_0)+\int_{\Si(\tau)\cap\{r\leq r_0\}}|Z\psi|^2+|\int_{\DD(\tau_0,\tau)} \QQ^{\a\b}\piring_{\a\b} |
  \]
  We next consider bulk term $\int_{\DD(\tau_0,\tau)} \ \QQ^{\a\b}\piring_{\a\b}$ where, writing $\Tring=\T+h\Z$,
$h= \chi \frac{a}{r_+^2+a^2}$,
\[
\piring_{\a\b}=\Lie_{\Tring } \g_{\a\b}= \Lie_\T\g_{\a\b}+\Lie_{ (h\Z)}\g_{\a\b}= \D_\a h \Z_\b+\D_\b h\Z_\a
\]
  Thus, recalling that $e_4=\frac{r^2+a^2}{|q|^2} \pr_t +\frac{\De}{|q|^2} \pr_r +\frac{a}{|q|^2} \pr_\phi$,
 \begin{align*}
 \QQ\c \piring&=2\QQ^{\a\b}\D_\a h \Z_\b=2\Big(\D^\a\psi \D^\b \psi -\frac 12 \g^{\a\b} \LL[\psi]\Big)\D_\a h \Z_\b\\
 &=(\D^\a\psi \D_\a h )\Z\psi=   \frac{a}{r_+^2+a^2}   \frac{\De}{|q|^2}    \chi' \pr_r \psi \Z\psi\\
 &=    \frac{a}{r_+^2+a^2}   \chi'  \big( e_4\psi  - \frac{r^2+a^2}{|q|^2} \T\psi- \frac{a}{|q|^2} \Z\psi\big)  \Z\psi.
  \end{align*}
   Recalling the definition of the $r^p$ weighted norms, see \eqref{eq:Bnorm-p-deg}, we deduce
  \begin{align*}
  \int_{\DD(\tau_0,\tau)} \QQ^{\a\b}\piring_{\a\b} &\les   \int_{\DD_{\frac{r_0}{2} \le r \le r_0}(\tau_0,\tau)} r^{-1}  |\T\psi| |r\nab\psi| +    r^{-2}  |re_4 \psi | |r\nab\psi| +  r_0^{-3}   |r\nab\psi|^2\\
  &\les  r_0 ^{2-p}  \int_{\DD_{\frac{r_0}{2} \le r \le r_0}(\tau_0,\tau)} r^{p-3}  |\T\psi| |r\nab\psi|
  +r_0^{1-p} \int_{\DD_{\frac{r_0}{2} \le r \le r_0}(\tau_0,\tau)} r^{p-3}| \dk\psi|^2\\
  &\les  r_0^{2-p}\Big( \la^{-1} \Bdeg_p[\psi](\tau_0,\tau) +  \la  \int_{\DD_{\frac{r_0}{2} \le r \le r_0}(\tau_0,\tau)}r^{p-3}|\T\psi|^2 \Big)+r_0^{1-p} \Bdeg_p[\psi](\tau_0,\tau).
    \end{align*}

      Hence, for $1<p<2-\de$, choosing $\la=r_0$ sufficiently large,
      \begin{align*}
   \int_{\HH(\tau_0,\tau)}  |\That  \psi|^2 + \Edeg[\psi](\tau) &\les  \int_{\Si(\tau)\cap\{r\leq r_0\}}|Z\psi|^2
   +r_0^{3-p} \int_{\DD_{\frac{r_0}{2} \le r \le r_0}(\tau_0,\tau)}r^{p-3}|\T\psi|^2 \\
   &+  r_0^{1-p}    \Bdeg_p[\psi](\tau_0,\tau)+\Edeg[\psi](\tau_0).
  \end{align*}

 Commutating the equation with $\T, \Z, \OO$ and using Lemmas \ref{Lemma:calculuslemma} and \ref{lem:higherordercalculuslemma}, we deduce for $s\geq 2$
  \begin{align*}
  \EFdeg^s[\psi](\tau_0,\tau)  & \les  \int_{\Si(\tau)\cap\{r\leq r_0\}}|\dk^{s} Z\psi|^2+
   r_0^{3-p}  \Bdeg^{s-1}_p[\T\psi] (\tau_0,\tau)\\
   &+  r_0^{1-p}  \Bdeg^s_p[\psi](\tau_0,\tau)+\Edeg^s[\psi](\tau_0).
        \end{align*}
   The $r^p$ version of the estimate \eqref{eq:Proof-Prop.1.6-Intro} (the control of $\Edeg_p[\psi](\tau)$) follows then by a standard procedure, see
  for example section 10.3 of \cite{GKS}.

\subsection{Proof of  Theorem \ref{Thm:Morawetz2}}
The goal of this section is to prove Theorem \ref{Thm:Morawetz2}.
\begin{proof}
We sketch the proof as follows.

{\bf Step 1.}
     Start with the first conditional Morawetz estimate \eqref{eq:Thm-Morawetz1}. For $s\geq 2$
    \beq\lab{eq:oldMorawetz}
    \BPEF_p^s[\psi](\tau_0, \tau)\les \EFdeg[(e_3,e_4,r\nab)^{\leq s}\psi](\tau_0,\tau)+E_p^s[\psi](\tau_0)+\NN_p^s[\psi, N](\tau_0, \tau).
    \eeq
    {\bf Step 2.}
    Combining \eqref{eq:oldMorawetz} with the estimate \eqref{eq:Proof-Prop.1.6-Intro} in Proposition \ref{Prop.1.6-Intro} and taking $r_0$ sufficiently large, we derive for $s\geq 3$
    \beq
    \lab{eq:reducetoT}
    \bsplit
    \BPEF_p^s[\psi](\tau_0, \tau)&\les   r_0^{3-p} \Bdeg_p^{s-1}[\T\psi] (\tau_0, \tau)+ \int_{\Si(\tau)\cap\{r\leq r_0\}}|\dk^{s} Z\psi|^2 \\
    &+E_p^s[\psi](\tau_0)+\NN_p^s[\psi, N](\tau_0, \tau)\\
    &\les   r_0^{3-p} \EFdeg[(e_3, e_4, r\nab)^{s-1}T\psi] (\tau_0, \tau)+ \int_{\Si(\tau)\cap\{r\leq r_0\}}|\dk^{s} Z\psi|^2 \\
    &+E_p^s[\psi](\tau_0)+\NN_p^s[\psi, N](\tau_0, \tau).
    \end{split}
    \eeq
          {\bf Step 3.}   Using the standard  $\T$ energy  identity in  \eqref{eq:Festimates}   
	 and the definition  of  the degenerate energy  quantity  $\Edeg[\psi]$   in \eqref{eq:Bdegnorm}
   and  the inhomogenuous quantity $\NN[\psi, N]$ we derive
	\begin{align*}
	\int_{\Si(\tau)} \QQ[\psi](\T, N_\Si) &=\int_{\HH(\tau_0, \tau)} -\QQ(\T, N_\HH) +\int_{\Si(\tau_0)} \QQ[\psi](\T, N_\Si) -\int_{\DD(\tau_0, \tau)} \T\psi \c N\\
	&= -\int_{\HH(\tau_0, \tau)} \frac{r^2+a^2}{|q|^2}T\psi \c  T_+\psi + \int_{\Si(\tau_0)} \QQ[\psi](\T, N_\Si) +\int_{\DD(\tau_0, \tau)} \T\psi \c N\\
	&\les\int_{\HH(\tau_0,\tau)}|T\psi\c T_+\psi| +\Edeg[\psi](\tau_0) +\int_{\DD(\tau_0, \tau)} |\T\psi \c N|.
	\end{align*}
	On the other  hand, according to  \eqref{eq:EFestimates1} of Lemma   \ref{lemma:Energy1}, 
	we have
	\[
\QQ[\psi](\T, N_\Si)\ge c_0\Big( | e_4\psi|^2+\frac{|\De|}{r^2} |e_3 \psi|^2 + |\nab\psi|^2\Big) -C \frac{a^2m^2 }{r^6}   |\Z\psi|^2.
\]
Therefore, taking $r_0$ sufficiently large
\beq\lab{eq:controlofE}
\bsplit
\Edeg[\psi](\tau)&\les \int_{\HH(\tau_0,\tau)}|T\psi\c T_+\psi| +\Edeg[\psi](\tau_0)+ \int_{\Si(\tau)} \frac{a^2m^2 }{r^6}   |\Z\psi|^2  +\int_{\DD(\tau_0, \tau)} |\T\psi \c N|\\
&\les  \int_{\HH(\tau_0,\tau)}|T\psi\c T_+\psi| + \int_{\Si(\tau)\cap\{r\leq r_0\}}|\Z\psi|^2   +\Edeg[\psi](\tau_0)+\int_{\DD(\tau_0, \tau)} |\T\psi \c N|\\
&\les \de F[\psi](\tau_0,\tau) +\de^{-1}\Fdeg[\psi](\tau_0,\tau)+ \int_{\Si(\tau)\cap\{r\leq r_0\}}|\Z\psi|^2   \\
&+\Edeg[\psi](\tau_0)+\int_{\DD(\tau_0, \tau)} |\T\psi \c N|.
\end{split}
\eeq
where $\de$ is a sufficiently small number.

{\bf Step 4.} To control the flux at the event horizon, we need to use the estimate
\beq
\lab{eq:Flux}
\bsplit
 \Fdeg[ T^2\psi](\tau_0, \tau_1) & \les
\int_{\tau_1}^{\tau_1+1} E^2_{\leq r_0}[\psi](s) ds +\ep_0 \BPEF_p^2[\psi](\tau_0, \tau_1) \\
 &  +E^2_p[\psi](\tau_0)+\NN_p^2[\psi, N](\tau_0, \tau_1).
 \end{split}
\eeq
This is a highly non trivial   estimate,   which we  present in \cite{He-K2},  based on   new physical space  adaptations of the  
transformation technique used by Whiting in the   proof of  his    mode stability result  \cite{W}.

	{\bf Step 5.} Combining   \eqref{eq:reducetoT} with $s=3$ applied to $T\psi$, \eqref{eq:controlofE} and \eqref{eq:Flux} after commutation with $\T, \Z, \OO$, and using  Lemma \ref{Lemma:calculuslemma}, we derive
	\beq
	\lab{eq:estimateofTpsi}
	\bsplit
	\BPEF_p^3[T\psi](\tau_0,\tau_1)&\les  \int_{\Si(\tau_1)\cap\{r\le 	 r_0\}}|\dk^3T\Z\psi|^2+\int_{\tau_1}^{\tau_1+1}E^4_{\leq r_0}[\psi](s)\,ds\\
&+ E_p^4[\psi](\tau_0) + \NN^4_p[\psi, N](\tau_0, \tau_1).
	\end{split}
	\eeq
	Combining  \eqref{eq:estimateofTpsi}
and  the first inequality in \eqref{eq:reducetoT} with $s=4$ applied to $\psi$, we further derive
\begin{align*}
\BPEF_p^4[\psi](\tau_0,\tau_1)&\les  \int_{\Si(\tau_1)\cap\{r\le 	 r_0\}}|\dk^4\Z\psi|^2+\int_{\tau_1}^{\tau_1+1}E^4_{\leq r_0}[\psi](s)\,ds\\
&+ E_p^4[\psi](\tau_0) + \NN^4_p[\psi, N](\tau_0, \tau_1).
\end{align*}
	The estimates for $s\geq 5$
	\begin{align*}
	\BPEF_p^s[\psi](\tau_0,\tau_1)&\les \int_{\Si(\tau_1)\cap\{r\le 	 r_0\}}|\dk^s\Z\psi|^2+\int_{\tau_1}^{\tau_1+1} E^s_{\leq r_0}[\psi](s)\,ds\\
&+ E_p^s[\psi](\tau_0) + \NN^s_p[\psi, N](\tau_0, \tau_1)		
	\end{align*}
	can then be derived   by the standard procedure of commutation with $\T, \Z$ and  the Lemma \ref{lem:higherordercalculuslemma}.
	
	{\bf Step 6.}  Using the standard $T$-energy identity, we obtain
   \beq\lab{eq:Z}
   \bsplit
   \Edeg[\psi](\tau_1)+\int_{\HH(\tau_0, \tau_1)}|T_+\psi|^2&\les \int_{\HH(\tau_0, \tau_1)}|Z\psi|^2   + \int_{\Si(\tau_1)\cap\{r\leq r_0\}}|\Z\psi|^2   \\
&+\Edeg[\psi](\tau_0)+\int_{\DD(\tau_0, \tau_1)} |\T\psi \c N|
\end{split}
\eeq
Combining Lemmas \ref{Lemma:calculuslemma} and \ref{lem:higherordercalculuslemma} with \eqref{eq:Z}, we derive	
\begin{align*}
\sup_{\tau\in[\tau_1,\tau_1+1]}E^s_{\leq r_0}[\psi](\tau)&\les \sup_{\tau\in[\tau_1,\tau_1+1]}E^s_{\leq r_0}[\psi](\tau)\sim E^s[\psi](\tau_1)\\
&\les F^{s-1}[Z\psi](\tau_0, \tau_1)   + \int_{\Si(\tau_1)\cap\{r\leq r_0\}}|\dk^s\Z\psi|^2+E_p^s[\psi](\tau_0) + \NN^s_p[\psi, N](\tau_0, \tau_1)
\end{align*}
and thus
\begin{align*}
	\BPEF_p^s[\psi](\tau_0,\tau_1)&\les \int_{\Si(\tau_1)\cap\{r\le 	 r_0\}}|\dk^s\Z\psi|^2+F^{s-1}[Z\psi](\tau_0,\tau_1)\\
&+ E_p^s[\psi](\tau_0) + \NN^s_p[\psi, N](\tau_0, \tau_1)	\\
&\les \int_{\Si(\tau_1)\cap\{r\le 	 r_0\}}|\dk^s\Z\psi|^2+ \int_{\tau_1}^{\tau_1+1}\int_{\Si(s)\cap\{r\le 	 r_0\}}|\dk^s\Z\psi|^2\,ds\\
&+ E_p^s[\psi](\tau_0) + \NN^s_p[\psi, N](\tau_0, \tau_1).	
	\end{align*}
This finishes the  proof of the  conditional Morawetz-Energy  estimate \eqref{eq:Thm-Morawetz2}.
\end{proof}





\section{Proof of Theorem \ref{Thm:Morawetz3}}
\lab{section:Proof of Thm-Morawetz3}


The proof of Theorem \ref{Thm:Morawetz3} uses a
    continuity argument\footnote{A similar argument was used in \cite{DRS}, in the microlocalized setting of that paper.}, based on the expectation that the local energy quantity  on the right hand side of \eqref{eq:Thm-Morawetz2} converges to $0$ as $\tau\to \infty$. Since the continuity argument loses derivatives we need first to refine the estimate
     \eqref{eq:Thm-Morawetz2}. For simplicity of notation, we introduce the local (in time and space) energy quantity
     \beq
     \lab{eq:defofLE}
     \LE[\psi](\tau):=\int_{\Si(\tau)\cap\{r\ge      r_0\}}|\psi|^2+\int_{\tau}^{\tau+1} \int_{\Si(s)\cap\{r\leq r_0\}}|\psi|^2 ds     \eeq.

\begin{proposition}
    \lab{Prop:Main-quantitative-estim} For  solutions of   $\square_{a,m}\psi=N$  with  $|a|/m \le 0.75$ and $s\ge 4$, we have
    \beq
    \lab{eq:Morawetz-improved1}
    \bsplit
    \BPEF_{p}^{s}[\psi](\tau_0, \tau) &\les  \LE[\dk^4 \Z^{\leq s-3}\psi](\tau)   + E_p^{s}[\psi](\tau_0)+\NN_p^s[\psi, N](\tau_0,\tau).
\end{split}
\eeq
\end{proposition}
\begin{proof}
The proof is based on the estimate \eqref{eq:Thm-Morawetz2} of Theorem \ref{Thm:Morawetz2} and Lemmas \ref{Lemma:calculuslemma} and \ref{lem:higherordercalculuslemma}. Using Lemmas \ref{lem:higherordercalculuslemma}, we derive
\beq\lab{eq:BEFs-s+1}
    \BPEF_p^{s}[\psi](\tau_0,\tau)\les \BPEF^4_p[(T,Z)^{\leq s-4}\psi](\tau_0,\tau)+E_p^{s}[\psi](\tau_0)+\NN_p^{s}[\psi,N](\tau_0,\tau)
        \eeq
According to \eqref{eq:Thm-Morawetz2}
\[
 \BPEF^4_p[(T,Z)^{\leq s-4}\psi](\tau_0,\tau)\les  \LE[\dk^{\le 4} (\T, \Z)^{\leq s-4}    Z\psi](\tau) + E^{s}_{p}[\psi](\tau_0)+\NN_p^{s}[\psi,N](\tau_0, \tau).
 \]
Using finite in time energy, we note that
 \[
 \LE[\dk^{\le 4} (\T, \Z)^{\leq s-4}    Z\psi](\tau)\les E^4_p[(T,Z)^{\leq s-5}Z\psi]\les \BPEF^4_p[(T,Z)^{\leq s-5}Z\psi].
   \]
   Therefore
   \[
    \BPEF^4_p[(T,Z)^{\leq s-4}\psi](\tau_0,\tau)\les \BPEF^4_p[(T,Z)^{\leq s-5}Z\psi] + E^{s}_{p}[\psi](\tau_0)+\NN_p^{s}[\psi,N](\tau_0, \tau).
      \]
    Repeating the argument we deduce
\[
\BPEF^4_p[(\T, \Z)^{s-4}\psi](\tau_0, \tau) \les  \BPEF^4_p [\Z^{s-4}\psi](\tau_0, \tau)+ E^{s}_{p}[\psi](\tau_0)+\NN_p^{s}[\psi,N](\tau_0, \tau).
\]
Thus, in view of \eqref{eq:Thm-Morawetz2},
\[
\BPEF^{4}_p[(\T, \Z)^{s-4}\psi](\tau_0, \tau) \les   \LE[\dk^4Z^{s-3}\psi](\tau) + E^{s}_{p}[\psi](\tau_0)+\NN_p^{s}[\psi,N](\tau_0, \tau).
\]
Finally, using \eqref{eq:BEFs-s+1} we conclude that
\[
\BPEF^s_p[\psi](\tau_0, \tau) \les    \LE[\dk^4Z^{s-3}\psi](\tau) + E^{s}_{p}[\psi](\tau_0)+\NN_p^{s}[\psi,N](\tau_0, \tau).
\]
as stated.
\end{proof}

\begin{definition}
We say that a complex solution of $ \square_{a,m}\psi=N$ is $\Z$-equivariant\footnote{Note that this implies that $N$ is also $Z$-equivariant, $ZN=-i \ell N$.} i.e. if $Z\psi=-i \ell \psi $.
\end{definition}
\begin{corollary}
Under the same assumptions as in Theorem \ref{Thm:Morawetz2}, for any $\Z$-equivariant solution,
\beq
\lab{eq:Main-quantitative-estim2-inhom}
\BPEF_{p}^{s}[\psi](\tau_0, \tau)\les  \ell^{2s-6}\LE[\dk^4\psi](\tau)+E_p^{s}[\psi](\tau_0)+ \NN_p^{s} [\psi,  N](\tau_0, \tau).
\eeq
\end{corollary}
  We are now ready to prove the global Morawetz Theorem \ref{Thm:Morawetz3} which we rewrite below.
\begin{theorem}
    \lab{Thm:globalMorawetz}
    The following global Morawetz estimate\footnote{Recall that the inequality sign $\les$ is the same as $\le C$ with $C$ a universal constant independent of $\ell$.} holds true in the domain $\DD(\tau_0, \infty)$ for solutions of $\square_{a,m} \psi=N$, $|a|/m\le 0.75$,
     $1+\de \le p\le 2-\de $, $s\ge 8$, for which $E^8_p[\psi](\tau_0) + \NN_p^8[N] (\tau_0, \infty) <\infty$,
    \[
    \BPF^s_p[\psi](\tau_0, \infty)\les E_p^s[\psi](\tau_0)+ \NN_p^{s} [\psi,  N](\tau_0, \infty).
    \]
\end{theorem}
\begin{proof}
    By a simple orthogonality argument it suffices to prove the result for $Z$-equivariant solutions. We thus assume in what follows that $Z\psi=-i \ell \psi$ which allows us to use the estimate \eqref{eq:Main-quantitative-estim2-inhom}. We then proceed by a continuity argument: we assume that the statement of the theorem holds for all $Z$-equivariant solutions of $\square_{a,m}\psi=N$ with $a\le a_0$, and show that the same holds for $a>a_0$ close to $a_0$. The key
    point is that we can write
    $\square_{a,m} \psi=N$, $a>a_0$, in the form, see Lemma \ref{Lemma;RHS},
    \[\
    \square_{a_0, m}\psi=N+N_{a_0,a},\qquad
    N_{a_0,a}:=\big(\square_{a_0,m}-\square_{a,m} \big)\psi
    \approx  (a-a_0)r^{-2}   \dk^{\le 2}\psi,
    \]
    where the loss of derivatives on the right hand side can be handled using the gain of derivative estimates of \eqref{eq:Main-quantitative-estim2-inhom}.

    According to \eqref{eq:Main-quantitative-estim2-inhom}, we have, for $Z$ equivariant solution of $\square \psi=N$, with $s\ge 4$, $\tau\ge \tau_0$, for the entire range $|a|/ m \le 0.75$,
    \beq
    \lab{eq:continuity1'}
    \bsplit
    \BPEF_{p}^{s}[\psi](\tau_0, \tau)&\les
    \ell^{2s-6 } \LE[\dk^4\psi](\tau)   + E_p^{s}[\psi](\tau_0) + \NN_p^{s}[\psi, N](\tau_0, \tau)   .
\end{split}
\eeq
Based on \eqref{eq:continuity1'} we derive the following.
\begin{lemma}
    \lab{Lemma:*}
    Given a solution of $\square \psi=N$, with $Z\psi=-i\ell \psi$, which verifies the assumption
    \beq
    (*)  \qquad \qquad\qquad \qquad \qquad \qquad  \int_{\tau_0}^\infty  d\tau\int_{\Si_{\le   r_0}(\tau)}|\dk^4\psi|^2 <\infty \qquad \qquad \qquad \qquad\quad
    \eeq
    we have for $s\geq 4$
    \beq
    \lab{eq:continuity2-F}
    \bsplit
    \BPF^{s} _p[\psi](\tau_0, \infty) \les  E^{s}_p[\psi](\tau_0)+
    \NN_p^{s}[ \psi, N](\tau_0,\infty).
\end{split}
\eeq
where $\NN_p^{s}[\psi, N](\tau_0,\infty)=\limsup_{\tau\to \infty} \NN_p^{s}[\psi,N](\tau_0,\tau)$.
\end{lemma}
\begin{proof}
In view of (*) there exists a sequence $\tau_i\to \infty$ such that $ \LE[ \dk^4\psi](\tau_i) \to 0$. Plugging this information on the right hand side of \eqref{eq:continuity1'} we deduce,
\[
\BPEF^s_p[\psi](\tau_0, \tau_i) \les \ell^{2s-6 } \LE[ \dk^4\psi](\tau_i)+ E^{s}_p[\psi](\tau_0) +
\NN_p^{s}[ \psi, N](\tau_0,\tau_i).
\]
Hence,
passing to the limit, for all $s\ge 4$
\[
\BPF^s_p[\psi](\tau_0, \infty ) \les  E^{s}_p[\psi](\tau_0)+
\NN_p^{s}[ \psi, N](\tau_0,\infty)
\]
as stated.
\end{proof}
To prove Theorem \ref{Thm:globalMorawetz} we show by a continuity argument that $(*)$ holds true for $\ell$-equivariant solutions of $\square_{a,m}\psi=N$ for the entire range $|a|/m \le 0.75$. This implies, in view of Lemma \ref{Lemma:*}, that the estimate
$ B^s_p[\psi](\tau_0, \infty)\les E_p^s[\psi](\tau_0)
+\NN^s_p[\psi,N](\tau_0,\infty)$, for $s\ge 8$ also holds true.

We proceed in steps as follows,

{\bf Step 1.} We first check that the estimate $(*)$ holds true for $Z$-equivariant solutions of $\square_{a, m} \psi=N$, for small values of $|a|/m$. In view of the remark above this is a consequence of the Energy-Morawetz estimates derived in \cite{GKS}.

{\bf Step 2.} It is straightforward to check that the interval of $|a|/m$ for which $(*)$ holds true is closed.

{\bf Step 3.} We show that the interval of $|a|/m$ for which $(*)$ holds true is open.
More precisely, we assume that $(*)$ holds at $a=a_0$ and show that it also holds for $|a-a_0|$ sufficiently small. This is based on the following.
\begin{proposition}
\lab{prop:continuity-main}
Consider solutions $\psi $ of $\square_{a,m}\psi=N$ verifying $\Z^2\psi=-\ell^2 \psi$ for which $E^8_p[\psi](\tau_0) + \NN_p^8[N] (\tau_0, \infty) <\infty$.
Assume that $(*)$ holds true for $a\le a_0$. Then $ (*) $ holds true for all $a> a_0$ sufficiently close to $a_0$.
\end{proposition}
\begin{proof}

We write $\square_{a,m} \psi=N$, $a>a_0$ in the form
\beq
\square_{a_0, m}\psi=N+N_{a_0,a},\qquad N_{a_0, a}:=\big(\square_{a_0,m}-\square_{a,m} \big)\psi.
\eeq
Since $(*)$ is verified at $a=a_0$ we can apply
\eqref{eq:continuity2-F} to the equation
$\square_{a_0, m}\psi=N+N_{a_0, a}. $
We deduce\footnote{Note that the frame and derivatives are measured relative to the metric $g_{a_0, m}$.}, for all $s \ge 4$,
\[
\BPF^s_p[\psi](\tau_0, \infty) \les  E^s_p[\psi](\tau_0)+ \NN_p^{s}[ \psi,  N_{a_0, a}](\tau_0,\infty)+  \NN_p^{s}[ \psi,  N](\tau_0,\infty)
\]
To finish the proof of Proposition \ref{prop:continuity-main} we need the estimates
\beq
\lab{eq:continuity-main}
\NN_p^{s}[ \psi, N_{a_0,a}](\tau_0,\infty)\les|a-a_0| B_{p}^{s+2 }[\psi](\tau_0, \infty).
\eeq
whose proof we delay to Lemma \ref{Lemma;RHS}.
Assuming \eqref{eq:continuity-main} we deduce
\bea
\lab{eq:continuity5}
\BPF^s_p[\psi](\tau_0, \infty) \les  E^s_p[\psi](\tau_0)+|a-a_0| B_{p}^{s+2}[\psi](\tau_0, \infty)+  \NN_p^{s}[ \psi,  N](\tau_0,\infty).
\eea
On the other hand in view of the estimate
\eqref{eq:Main-quantitative-estim2-inhom},
\[
\BPEF_{p}^{s+2}[\psi](\tau_0, \tau)\les  \ell^{2s-2}\LE[ \dk^4\psi](\tau)+E_p^{s+2}[\psi](\tau_0)+ \NN_p^{s+2} [\psi,  N](\tau_0, \tau).
\]
Integrating by parts in $\tau$
\[
\int_{\Si(\tau)\cap\{r\leq r_0\}}|\dk^4\psi|^2\les \int_{\Si(\tau_0)\cap\{r\leq r_0\}}|\dk^4\psi|^2  + B^5_p[\psi](\tau_0, \tau)
\]
Hence
\[
\BPEF_{p}^{s+2}[\psi](\tau_0, \tau)\les  \ell^{2s-2} B_p^5[\psi](\tau_0,\tau)+E_p^{s+2}[\psi](\tau_0)+ \NN_p^{s+2} [\psi,  N](\tau_0, \tau).
\]
Recall from \eqref{eq:NN-weighted-psi-source} that
$\NN^{s+2}_p[\psi,N](\tau_0,\tau)$ contains the mixed contribution
 \[
 \sum_{k\leq s+2}\left|
 \int_{\Dtrap(\tau_0,\tau)}(\T,\Z)\dk^k\psi\c\dk^kN
 \right|.
 \]
  Integrating by parts and applying Cauchy--Schwarz,
  we easily derive the estimate
  \beq
  \sum_{k\leq s+2}\left|
  \int_{\Dtrap(\tau_0,\tau)}(\T,\Z)\dk^k\psi\c\dk^kN
  \right|
  \les \ep \BEF_p^{s+2} [\psi](\tau_0,\tau)
  +\ep^{-1} \NN_p^{s+3}[N](\tau_0,\tau).
  \eeq
  Thus
  \[
\BPEF_{p}^{s+2}[\psi](\tau_0, \tau)\les  \ell^{2s-2}B_p^5[\psi](\tau_0, \tau)+E_p^{s+2}[\psi](\tau_0)+ \NN_p^{s+3} [  N](\tau_0, \tau).
\]
Passing to the limit,
 \beq
  \lab{eq:continuity6}
\BPEF_{p}^{s+2}[\psi](\tau_0, \infty)\les  \ell^{2s-2}B_p^5[\psi](\tau_0, \infty)+E_p^{s+2}[\psi](\tau_0)+ \NN_p^{s+3} [  N](\tau_0, \infty).
\eeq
Therefore, back to \eqref{eq:continuity5},
\[
\BPF^s_p[\psi](\tau_0, \infty) \les  E^s_p[\psi](\tau_0)+|a-a_0|\Big( \ell^{2s-2}B_p^5[\psi](\tau_0, \infty)+E_p^{s+2}[\psi](\tau_0)+ \NN_p^{s+3} [ N](\tau_0, \infty)\Big).
\]
For fixed $\ell$ and $|a-a_0|\ell^{2s-2}$ sufficiently small and $s\ge 5 $ we deduce
\[
\BPF^s_p[\psi](\tau_0, \infty) \les  E^s_p[\psi](\tau_0) + |a-a_0|\Big( E_p^{s+2}[\psi](\tau_0)+ \NN_p^{s+3} [ N](\tau_0, \infty)\Big).
\]
Taking $s=5$, the assumption
$E_p^{7}[\psi](\tau_0)+ \NN_p^{8} [ N](\tau_0, \infty)<\infty$
therefore gives
\beq
 \int_{\tau_0}^\infty  d\tau\int_{\Si_{\le   r_0}(\tau)}|\dk^4\psi|^2 \le B^5_p[\psi](\tau_0, \infty)<\infty.
\eeq
This proves Proposition \ref{prop:continuity-main}.
\end{proof}
Thus all three steps of the continuity argument have been verified, which
ends the proof of Theorem \ref{Thm:globalMorawetz}.
\end{proof}


\subsection{Proof of Estimate \eqref{eq:continuity-main}}

The proof of the estimate is based on the definition of the quantity $\NN^s_p[ N]$, and the following lemma.
\begin{lemma}
\lab{Lemma;RHS}
The expression $ N_{a_0,a} =\big(\square_{a_0,m}-\square_{a,m} \big)\psi$ verifies the bound
\[
\big|N_{a_0,a} \big| \les    O(r^{-2})  |a-a_0|\,  |\dk^{\le 2}\psi|.
\]
Consequently,
\beaa
\NN^s_p[\psi, N_{a_0,a}](\tau_1, \tau_2)&\les & |a-a_0| B_p^{s+2}[\psi](\tau_1, \tau_2).
\eeaa
\end{lemma}
\begin{proof}
Since $N_{a_0,a}=-(\square_{a,m}-\square_{a_0,m})\psi$, it suffices to
estimate the latter difference. The estimate is immediate in a compact region
$r\le r_0$ from the expression \eqref{eq:decompose-square}.
For $r\ge r_0$, with $r_0$ sufficiently large we use the expressions given by Lemma 4.7.5 in \cite{GKS},
\[
\square_{a,m} \psi=-\nab_4 \nab_3 \psi  +O(r^{-1}) (\nab_4\psi+\nab_3\psi\big) +\lap  \psi+ O(r^{-1}) \nab  \psi +  O(r^{-2} )\psi.
\]
Denoting the frame of $\KK(a_0, m) $ with $\ezero_3, \ezero_4, \ezero_a$ and that of $\KK(a,m) $ with $(e_3, e_4, e_a)$ we easily check that, for $r\ge r_0$,
\begin{align*}
e_3-\ezero_3&\approx (a-a_0) r^{-2}\dkzero\\
e_4-\ezero_4&\approx ( a-a_0) r^{-2}\dkzero\\
e_a-\ezero_a &\approx ( a-a_0) r^{-2}\dkzero
\end{align*}
For example,
\begin{align*}
e_3-\ezero_3&= \frac{r^2+a^2}{\De} \pr_t -\pr_r +\frac{a}{\De} \pr_\phi- \big(\frac{r^2+a_0^2}{\De_0} \pr_t -\pr_r +\frac{a_0}{\De_0} \pr_\phi\big)\\
&=\frac{ (r^2+a^2) \De_0 - (r^2+a_0^2) \De}{\De\De_0} \pr_t + \frac{a\De_0-a_0\De}{\De\De_0} \pr_\vphi\\
&=\frac{2mr (a_0^2 -a^2) }{\De\De_0} \pr_t + \frac{(a-a_0) r^2 -2mr (a-a_0) + a a_0(a_0-a)}{\De\De_0}\pr_\vphi\\
&=\big(a-a_0\big)\Big(  O(r^{-3}) \T +  O(r^{-2}) \Z\big)
\end{align*}
On the other hand, according to
\eqref{eq:TZ-frame}, for $r \ge r_0$ sufficiently large,
\begin{align*}
\T &= \frac{1}{2}\left(e_4+\frac{\Delta}{|q|^2}e_3 -2a\Re(\Jk)^be_b\right) =  O(1) e_4+ O(1)  e_3 +O(r^{-1} )\nab  \\
\Z &= (r^2+a^2)\Re(\Jk)^be_b-\frac 1 2  a(\sin\th)^2 \big( e_4+\frac{\De}{|q|^2} e_3\big)= O(  r) \nab+  O(1) e_4+ O(1)  e_3
\end{align*}
Thus,
\[
e_3-\ezero_3= O(r^{-2}) \big( \ezero_3+\ezero_4 \big)+ O( r^{-1}) \nabzero= O(r^{-2}) \dkzero.
\]
We deduce
\begin{align*}
\big(\square_{a,m}-\square_{a_0, m}\big)\psi&= - \big(\nab_4 \nab_3 -    \nabzero_3 \ezero_4\big)\psi
+(\De-\Dezero )\psi
+ O( r^{-1})\big( \dk -\dkzero) \psi \\
&+(a-a_0) r^{-2} \psi= (a-a_0) r^{-2} \dkzero^{\le 2}\psi.
\end{align*}
Recalling the definition of $\NN_p^s, B_p^s $ quantities we deduce,
\[
\NN_p^{s}[ N_{a_0,a}](\tau_1,\infty)\les |a-a_0|^2  \int_{\DD(\tau_1, \infty)} r^{-4} r^{p+1} |\dk^{s+2} \psi|^2 \les  |a-a_0|^2 B_p^{s+2}[\psi] (\tau_1,\infty).
\]
Also,
\begin{align*}
\Big|\int_{\Dtrap(\tau_1, \infty)}(T,Z)\dk^{\le s}\psi \c\dk^{\le s}N_{a_0,a}\Big|&\les |a-a_0|  \int_{\Dtrap(\tau_1, \infty)}
 \big| (T,Z)\dk^{\le s}\psi \big| \big|  \dk^{\le s+2 }\psi\big|\\
 &\les  |a-a_0| B_p^{s+2}[\psi] (\tau_1,\infty)
\end{align*}
Thus,
\[
\NN_p^{s}[\psi, N_{a_0,a}](\tau_1,\infty) \les |a-a_0|  B_p^{s+2}[\psi] (\tau_1,\infty).
\]
as stated.
\end{proof}

\section{Energy estimates}
\lab{section:EnergyEstim}

The goal of this section is to prove the result stated in Theorem \ref{Thm:Morawetz4} which we restate as follows.
\begin{theorem}[Conditional Energy]
\lab{Thm:Morawetz4'}
Under the assumption $|a|/m <0.9$, $s\ge 0$, $\de<p< 2-\de$, we have
\beq\lab{eq:Thm-Morawetz4'}
E^s_{p}[\psi](\tau_2)\les  E^s_{p}[\psi](\tau_1)+ B_p^s[\psi] (\tau_1, \tau_2) +\NN_p^s[N](\tau_1, \tau_2).
\eeq
\end{theorem}

\subsection{Causal Properties  of the  Hawking vectorfield $T_+$}

We start with the following version of Lemma 3.2.3 in \cite{GKS}.
\begin{lemma}\label{lemma:inside-Hawking-vf}
    The Killing vectorfield $T_\la=\T+\la \Z$, for a scalar function $\la $, verifies
    \[
    |q|^2 \g(T_\la, T_\la) =-\De\big(1-\la  a \sin^2\th \big)^2 +\sin^2 \th   \big(a-\la (a^2+r^2)\big)^2.
    \]
\end{lemma}
\begin{proof}
    According to Lemma 3.2.3 in \cite{GKS} we have
    \begin{align*}
        \g(T_\la, T_\la) &=-\frac{\De}{|q|^2 } \Big( 1+ a^2\la^2(\sin\theta)^4 \Big) + \frac{\sin^2 \th}{|q|^2} E_\la(r),
        \\
        E_\la(r)&= a^2   - 4 am r\la         +\la ^2(r^2+a^2)^2.
    \end{align*}
    Note that
    \[
    E_\la(r)= \big(a-\la (a^2+r^2)\big)^2+ 2 \la  a (a^2+r^2)   - 4 am r\la =  \big(a-\la (a^2+r^2)\big)^2+ 2 \la a \De.
    \]
    Hence
    \begin{align*}
        |q|^2  \g(T_\la, T_\la) &=-\De\Big(  1+ a^2\la^2(\sin\theta)^4- 2 \la a \sin^2 \th \Big)+\sin^2 \th  \big(a-\la (a^2+r^2)\big)^2\\
        &= -\De(1-\la  a \sin^2\th )^2 +\sin^2 \th   \big(a-\la (a^2+r^2)\big)^2.
    \end{align*}
\end{proof}

\begin{lemma}
    \lab{Le:uniqunessofT_+}
    The function $M_+(\la):= \g(\T_\la, \T_\la)\big|_{r=r_+}$ verifies the following properties
    \[
    M_+(\om_+)=0, \quad M_+'(\om_+)=0,\quad  M_+''(\om)\geq0.
    \]
    Thus, among the family of vectorfields $T_\la$, the Hawking vectorfield $\Tplus=T_{\om_+}$ is the only one which is
    causal along the horizon.
\end{lemma}
\begin{proof}
    Denote $M(\la)=\g(\T_\la, \T_\la)$ and observe that
    \begin{align*}
        \frac 1 2 M'(\la)&=\g( Z, T_\la)=\g(Z, T)+\la \g(Z, Z)=  -    \frac{ 2 amr }{|q|^2}   \sin^2\theta+\la  \frac{ \Si^2}{|q|^2}\sin^2\theta\\
        &=\frac{\sin^2 \th}{|q|^2}\big(- 2 amr +\la \Si^2)= \frac{\sin^2 \th}{|q|^2} N(\la)
        \\
        M'' (\la)&= 2\frac{ \Si^2}{|q|^2}\sin^2\theta\ge 0.
    \end{align*}
    Along the horizon $r=r_+$ we have $M(\om_+)=0$ and $M'(\om_+)=0$.
    Indeed\footnote{Note that, $ \Si^2 - 4m^2 r^2=( r^2+a^2)^2 - 4m^2 r^2 -a ^2 sin^2\th\De =\De\big(|q|^2 + 2mr\big)$.}
    \begin{align*}
        N(\om_+) &=- 2 amr +\om_+  \Si^2= - 2 amr +   \frac{a}{2m r_+}\Si^2=  \frac{a}{ 2m r_+} \big(-4 m^2 r r_++ \Si^2\big)\\
        &=   \frac{a}{ 2m r_+}\Big(  4m^2 r(r-r_+)- 4m^2 r^2   +\Si^2\Big) =  \frac{a}{ 2m r_+}\Big(  4m^2 r(r-r_+)+ \De \big(|q|^2 + 2m r\big)\Big)
    \end{align*}
    We deduce that, along the horizon, $M(\la)$ reaches the minimum value $0$ at $\la=\om_+$.
    Thus any change in $\la$ leads to a vectorfield $\T_\la$ which is spacelike along the horizon.
\end{proof}


\subsection{Timelike property of   $T_\la$}


In this section we look for values of $\la$, different\footnote{In view of Lemma \ref{Le:uniqunessofT_+} we do not ask that $T_\la$ be causal on the horizon.} from $\om_+$, for which the vectorfield $T_\la =T+\la Z$
is timelike on the trapping set.
We start by proving the following.
\begin{lemma}
    \lab{Le:sufficientcond}
    A sufficient condition for the vectorfield $T_{\la=\frac{a}{a^2+ r_0^2}}$ to be timelike on the trapping set $[\rhat_1, \rhat_2]$ is that the quadratic polynomial
    \beq
    \lab{eq:P(r_0,r)}
    P_{r_0}(r):= 6 a^2 m r^2  -  \big( r_0^4+2 r_0^2 a^2+ 9 a^2 m^2 \big) r + 2mr_0^4+ 4 ma^4
    \eeq
    is negative on the trapping set.
\end{lemma}
\begin{proof}
    Recall that
    \[
    M_\la(r):=|q|^2 \g(T_\la, T_\la)=-\De\big(1-\la  a \sin^2\th \big)^2 +\sin^2 \th   \big(a-\la (a^2+r^2)\big)^2.
    \]
    Setting $x=\sin^2 \th$ we deduce
    \begin{align*}
        \frac{d}{dx}M_\la&=\frac{d}{dx}\Big(-\De\big(1-\la  a x  \big)^2 +x   \big(a-\la (a^2+r^2)\big)^2\Big)\\
        &=  2 \la a \De\big(1-\la  a x \big) +\big(a-\la (a^2+r^2)\big)^2.
    \end{align*}
    Note that for $\la=\la_0=\frac{a}{a^2+r_0^2}$, since $0\le x\le 1$,
    \begin{align*}
        \frac{d}{dx}M_{\la_0}&=  2 \frac{a^2}{a^2+r_0^2}  \De\big(1- \frac{a^2}{a^2+r_0^2}  x \big)  +\big(a-\frac{a}{a^2+r_0^2} (a^2+r^2)\big)^2\\
        &= 2 \frac{a^2}{a^2+r_0^2}  \De\frac{a^2+r_0^2 - a^2 x}{a^2+r_0^2}    +\big(a-\frac{a}{a^2+r_0^2} (a^2+r^2)\big)^2>0.
    \end{align*}
    We deduce that $M_{\la_0}$ is increasing in $x=\sin^2\th$ and therefore it is maximal for $\sin^2\th=1$.
    We infer that to have $M_{\la_0}<0 $ on the trapping set it suffices to have $M_{\la_0, x=1}(r)$ verify that property.
    Note that
    \begin{align*}
        M_{\la_0, x=1}(r)&= -\De\big(1-\frac{a^2}{r_0^2 +a^2}    \big)^2 +   \big(a-\frac{a}{r_0^2 +a^2} (a^2+r^2)\big)^2
        \\
        &=\frac{1}{( r_0^2 +a^2)^2}\Big(-\De r_0^4 +a^2 (r_0^2-r^2)^2\Big)= \frac{1}{( r_0^2 +a^2)^2} H_{r_0}(r)
    \end{align*}
    where
    \beq
    H_{r_0}(r) :=-\De r_0^4 +a^2 (r^2-r_0^2)^2.
    \eeq
    Thus $ M_{\la_0, x=1}<0$ on the trapping set if and only if $H_{r_0}$ verifies the same property.
    Note that
    \begin{align*}
        H_{r_0}(r)&=-(r^2+a^2- 2mr)  r_0^4+a^2 r^4 - 2 a^2 r^2 r_0^2+ a^2r_0^4\\
        &=-(r^2-2mr) r_0^4  +a^2 r^2(r^2  - 2  r_0^2)\\
        &=a^2 r^4- r^2( r_0^4+2 r_0^2 a^2)+ 2mr r_0^4
    \end{align*}
    On the other hand, recalling the definition of the polynomial $ \Tht(r)$ introduced in proposition \ref{Prop:r-trapped.region},
    \[
    \Tht(r)=r^{-2}\Th(r)=r\Ththree(r)
    = r^4 -6 m r^3 +9m^2 r^2 - 4ma^2 r
    \]
    we deduce
    \begin{align*}
        H_{r_0}(r)- a^2 \Tht(r)&= - r^2( r_0^4+2 r_0^2 a^2)+ 2mr r_0^4 -a^2\big( -6 m r^3 +9m^2 r^2 - 4ma^2 r\big)\\
        &= 6 a^2 m r^3 - \big( r_0^4+2 r_0^2 a^2+ 9 a^2 m^2 \big) r^2+ \big( 2mr_0^4+ 4 ma^4 \big) r\\
        &= r\Big( 6 a^2 m r^2  -  \big( r_0^4+2 r_0^2 a^2+ 9 a^2 m^2 \big) r + 2mr_0^4+ 4 ma^4\Big)\\
        &=r P_{r_0}(r)
    \end{align*}
    as defined in \eqref{eq:P(r_0,r)}.
    Therefore
    \bea
    H_{r_0}(r)= a^2 \Tht(r)+ r P_{r_0}(r).
    \eea
    Since $\Tht(r)$ vanishes at $\rhat_1, \rhat_2$ and is strictly negative in between we deduce that
    to have $H_{r_0}$ negative on the trapping set $[\rhat_1, \rhat_2]$, the roots $\mathring{r}_1<\mathring{r}_2$ of the quadratic polynomial
    $P_{r_0}(r)$ should satisfy $\mathring{r}_1<\rhat_1<\rhat_2<\mathring{r}_2$.
\end{proof}
Introducing
\[
\rdot= r/ m, \quad \mu= a^2/m^2, \quad  \rdot_0=r_0/m
\]
we rewrite
\[
m^{-5} P_{r_0}= 6 \mu \rdot^2   -\big( \rdot^4_0+2\mu \rdot_0^2+9 \mu\big)\rdot+2 \rdot_0^4+ 4 \mu^2
\]
Replacing, for convenience, $\rdot$ with $y$, we need to pick $\rdot_0$ so that the roots $y_1<y_2$ of the reduced polynomial
\beq
\lab{eq:reducedP}
P_{(\mu, \rdot_0)}(y):=  6 \mu y^2-   \big( \rdot^4_0+2\mu \rdot_0^2+ 9 \mu\big) y+\big( 2 \rdot_0^4+ 4\mu^2\big)
\eeq
satisfy
\[
y_1<\yhat_1< \yhat_2<y_2.
\]
where, see Lemma \ref{Lemma:rangeofrfortrappednullgeodesics},
\beq
\lab{eq:yhat1yhat2}
\yhat_1=2\left(1+\cos\left(\frac{2}{3}\arccos\left(-\mu^{1/2}\right)\right)\right),\quad
\yhat_2=2\left(1+\cos\left(\frac{2}{3}\arccos\left(\mu^{1/2}\right)\right)\right).
\eeq

\subsubsection{Relations between $\mu=\frac{a^2}{m^2} $ and the trapping set}
\lab{section:mu-trapping}


\begin{remark}
\lab{rem:monotonicity-rhat}
    In view of Lemma \ref{Lemma:Deriv.yhat} we have
    \begin{itemize}
        \item $\yhat_1$ is decreasing as a function of $\mu$ from the maximum value $3$ when $\mu=0$ to
        the minimum value $1$ when $\mu=1$.
        \item $\yhat_2$ is increasing as a function of $\mu$ from the minimum value $3$ when $\mu=0$ to
        the maximum value $4$ when $\mu=1$.
    \end{itemize}
\end{remark}
Take $\frac{|a|}{m}=0.9$, $\mu=0.81, r_0=1.92m, \rdot_0=1.92$. Then
\[
P_{\mu, \rdot_0}(y)=P_{0.81, 1.92}(y)=4.86y^2-26.85151296y+29.80348992.
\]
So $P_{0.81,1.9}(y)$ is decreasing on $[-\infty, \frac{517969}{187500}]$ and increasing on $[\frac{517969}{187500},\infty)$.
For $\mu=0.81$, we have
\[
\yhat_1\in[1.55,1.56],\quad \yhat_2\in[3.91. 3.92]
\]
Then
\begin{align*}
P_{0.81,1.92}(\yhat_1)&<P_{0.81,1.92}(1.55)=-0.140205168<0,\\
P_{0.81,1.92}(\yhat_2)&<P_{0.81, 1.92}(3.92)=-0.7737368832<0.
\end{align*}

We summarize the results in the following statement.
\begin{lemma}
    \lab{lemma:0.81}
    For $\mu =0.81$ the vectorfield $T_\la$ with $\la =\frac{a}{ a^2+ r_0^2}$ and $r_0=1.92m$ is timelike on the corresponding trapping set.
\end{lemma}
We will show next that the statement remains valid, with the same $r_0=1.92m$, for the full range $0 \le \mu \le 0.81$.
This will follow from the following Lemma.
\begin{lemma}
    \lab{Lemma:ontrappingset}
    Let $P_\mu(y)= 4\mu^2+\big(6y^2-(2\rdot_0^2+9)y\big)\mu+(2-y)\rdot_0^4$ be the polynomial
    \eqref{eq:reducedP} at given $\mu$ and $\rdot_0=1.92$. We claim that $P_\mu(\yhat_1), P_\mu(\yhat_2)<0$
    for all $\mu\in[0, 0.81]$.
\end{lemma}
Together with Lemma \ref{Lemma:ontrappingset} and Lemma \ref{Le:sufficientcond} we establish the following important result.
\begin{proposition}
    \lab{Prop:Tring}
    For all $\mu \in [0, 0.81]$, the vectorfield $\Ta:=T+\frac{a}{a^2+ (1.92m)^2} Z$ is timelike on the trapping set.
\end{proposition}

\begin{proof}[Proof of  Lemma \ref{Lemma:ontrappingset}.]
     We compute
    \begin{align*}
        \frac{d}{d\mu} P(y(\mu))&= 8 \mu+\big(6y^2-(2\rdot_0^2+9) y\big) +\mu( 12 y y'-(2\rdot_0^2+9)y')-\rdot_0^4 y'\\
        &=8 \mu+\big(6y^2-(2\rdot_0^2+9) y\big) +  y'\big( 12 \mu y-(2\rdot_0^2+9) \mu -\rdot_0^4\big),\\
        \frac{d}{d\mu}\big( 12 \mu y-(2\rdot_0^2+9) \mu -\rdot_0^4\big)&=12y-(2\rdot_0^2+9)+12\mu y'.
    \end{align*}
    Recall that, see Lemma \ref{Lemma:Deriv.yhat}
    \begin{align*}
        \frac{d}{d\mu}\yhat_1&=-  \frac 1 3 \frac{1}{\sqrt{(\mu(1-\mu)}}  \sqrt{\yhat_1(4-\yhat_1)}<0, \qquad
        \frac{d}{d\mu}\yhat_2= \frac  1 3   \frac{1}{\sqrt{(\mu(1-\mu)}} \sqrt{\yhat_2(4-\yhat_2)}>0.
    \end{align*}
    \begin{align*}
        \frac{d}{d\mu}\big(\sqrt{(\mu(1-\mu)} \yhat_1' \big) &=\frac{1}{3}  \frac{\yhat_1-2}{\sqrt{\mu(1-\mu) }},\qquad\,
        \frac{d}{d\mu}\big(\sqrt{(\mu(1-\mu)} \yhat_2' \big) =-\frac{1}{3}  \frac{\yhat_2-2}{\sqrt{\mu(1-\mu) }}<0.
    \end{align*}
    \textbf{ Show that $P(\yhat_2(\mu))<0$}: First, we calculate
    \begin{align*}
        \frac{d}{d\mu}\big( 12 \mu \yhat_2-(2\rdot_0^2+9) \mu -\rdot_0^4\big)&=12\yhat_2-(2\rdot_0^2+9)+12\mu \yhat_2'\\
        &\geq 12\times3-(2(1.92)^2+9)>0,\\
        \big( 12 \mu \yhat_2-(2\rdot_0^2+9) \mu -\rdot_0^4\big)|_{\mu=0}&=-\rdot_0^4<0,\\
        \big( 12 \mu \yhat_2-(2\rdot_0^2+9) \mu -\rdot_0^4\big)|_{\mu=0.5}&>(24+12\sqrt{3}-2(1.92)^2-9)\times 0.5-(1.92)^4>0.
    \end{align*}
    Then we conclude that $ 12 \mu \yhat_2-(2\rdot_0^2+9) \mu -\rdot_0^4$ is negative for $\mu\in[0, \mu_0)$ and positive for $\mu\in(\mu_0, 0.81]$ for some $0<\mu_0<0.5$. Therefore, $\frac{d}{d\mu}P(\yhat_2(\mu))>0$ and thus $P(\yhat_2(\mu))$ is increasing for $\mu\in[\mu_0, 0.81]$. We further compute
    \begin{align*}
        \frac{d^2}{d\mu^2} P(\yhat_2(\mu))&= \underbrace{8+\yhat_2'(24\yhat_2-(4\rhat_0^2+18)+12\mu\yhat_2')}_{>0}+\yhat_2'' \big( 12 \mu \yhat_2-(2\rdot_0^2+9) \mu -\rdot_0^4\big).
    \end{align*}
    Since $\yhat_2''<0$ for $\mu\in[0,0.5]$ and $\mu_0<0.5$, we infer that $ \frac{d^2}{d\mu^2} P(\yhat_2(\mu))>0$ and thus $ \frac{d}{d\mu} P(\yhat_2(\mu))$ is increasing for $\mu\in[0, \mu_0]$. Since $\frac{d}{d\mu} P(\yhat_2(\mu))|_{\mu=0}<0$ and $\frac{d}{d\mu} P(\yhat_2(\mu))|_{\mu=\mu_0}>0$, we find that $\frac{d}{d\mu} P(\yhat_2(\mu))$ is negative for $\mu\in[0, \mu_1)$ and positive for $\mu\in(\mu_1, \mu_0]$ for some $0<\mu_1<\mu_0$. Therefore, $ P(\yhat_2(\mu))$ is decreasing for $\mu\in[0, \mu_1)$ and increasing for $\mu\in(\mu_1, 0.81]$ for some $0<\mu_1<\mu_0$. Combining the monotonicity of $ P(\yhat_2(\mu))$ with the facts that
    \[
    P(\yhat_2(0))=(2-\yhat(0))\rdot_0^4=-\rdot_0^4<0,\quad  P(\yhat_2(0.81))<0,
    \]
    we obtain $ P(\yhat_2(\mu))<0$ for $\mu\in[0,0.81]$.

    \textbf{ Show that $P(\yhat_1(\mu))<0$}: We recall that
    \[
        \frac{d}{d\mu} P(\yhat_1(\mu))=8 \mu+\big(6\yhat_1^2-(2\rdot_0^2+9) \yhat_1\big) +  \yhat_1'\big( 12 \mu \yhat_1-(2\rdot_0^2+9) \mu -\rdot_0^4\big).
    \]
    We further compute
    \[
        \frac{d^2}{d\mu^2} P(\yhat_1(\mu))=8 +\yhat_1'\big(24\yhat_1-(4\rdot_0^2+18)+12\mu\yhat_1'\big) +  \yhat_1''\big( 12 \mu \yhat_1-(2\rdot_0^2+9) \mu -\rdot_0^4\big).
    \]
    \begin{enumerate}
        \item $\mu\in[0,0.5]$: Since $\mu\in[0,0.5]$, we have $\yhat_1\in[2,3]$ and
     \[
    \mu\yhat_1'=-\frac{1}{3}\sqrt{\frac{\mu}{1-\mu}}\sqrt{\yhat_1(4-\yhat_1)}\geq-\frac{2}{3}.
    \]
    Then we have
    \[
    \yhat_1'\big(24\yhat_1-(4\rdot_0^2+18)+12\mu\yhat_1'\big)\leq -\frac{4}{3}\big(24\c2-(4\rdot_0^2+18)+12\c(-\frac23)\big)<-8
    \]
    and thus
    \[
    8 +\yhat_1'\big(24\yhat_1-(4\rdot_0^2+18)+12\mu\yhat_1'\big)<0.
    \]
    Since $\yhat_1''\geq 0$ and
    \[
    12 \mu \yhat_1-(2\rdot_0^2+9) \mu -\rdot_0^4\leq(12\times3-(2\rdot_0^2+9))\mu-\rdot_0^4\leq 20\mu-13\leq 10-13<0,
    \]
    we infer that
    \[
    \yhat_1''\big( 12 \mu \yhat_1-(2\rdot_0^2+9) \mu -\rdot_0^4\big)<0.
    \]
    Therefore, we obtain that $ \frac{d^2}{d\mu^2} P(\yhat_1(\mu))<0$ and thus $ \frac{d}{d\mu} P(\yhat_1(\mu))$ is decreasing for $\mu\in[0,0.5]$. Since $\frac{d}{d\mu}P(\yhat_1(\mu))|_{\mu=0.5}>0$, we conclude that $\frac{d}{d\mu}P(\yhat_1(\mu))>0$ and thus $P(\yhat_1(\mu))$ is increasing for $\mu\in[0,0.5]$.

    \item $\mu\in[0.5,0.81]$:
    Since $6\yhat_1^2-(2\rdot_0^2+9)\yhat_1$ is increasing for $\yhat_1\in[\frac{2r_0^2+9}{12}=1.3644, 2]$ and $\yhat_1\in(1.5, 2]$ for $\mu\in[0.5,0.81]$, we infer that
     \[
    6\yhat_1^2-(2\rdot_0^2+9) \yhat_1\geq (6\yhat_1^2-(2\rdot_0^2+9) \yhat_1)|_{\yhat_1=1.5}> -12.
    \]
We also have
    \[
    12 \mu \yhat_1-(2\rdot_0^2+9) \mu -\rdot_0^4\leq(12\c2-(2\rdot_0^2+9))\mu-\rdot_0^4\leq 8\mu-13\leq 8\c0.81-13=-6.52.
    \]
    and thus
    \[
     \yhat_1'\big( 12 \mu \yhat_1-(2\rdot_0^2+9) \mu -\rdot_0^4\big)\geq (-\frac43)\c(-6.52)\geq 8.69.
    \]
    Therefore
    \begin{align*}
    \frac{d}{d\mu} P(\yhat_1(\mu))&=8 \mu+\big(6\yhat_1^2-(2\rdot_0^2+9) \yhat_1\big) +  \yhat_1'\big( 12 \mu \yhat_1-(2\rdot_0^2+9) \mu -\rdot_0^4\big)\\
    &> 8\mu-12+8.69\geq 8\c0.5-12+8.69>0.
    \end{align*}
    This proves that $P(\yhat_1(\mu))$ is increasing for $\mu\in[0.5,0.81]$.
        \end{enumerate}
        As a consequence, $P(\yhat_1(\mu))$ is increasing for $\mu\in[0,0.81]$, and then the negativity of $P(\yhat_1(\mu))$ for $\mu\in [0, 0.81]$ follows from the negativity of $P(\yhat_1(\mu))|_{\mu=0.81}$.
\end{proof}




\subsection{Proof of  Theorem \ref{Thm:Morawetz4}}
\lab{section-energy-flux}


The goal of this section is to prove the estimate \eqref{eq:Thm-Morawetz4} which we rewrite as follows.
\begin{proposition}
\lab{Prop:Condit-EnergyFlux}
Given a solution of the equation $\square_{a,m}\psi=N$ with $|a|/m<0.9$, we have
\beq
\lab{eq:Thm-Morawetz4-energy-flux}
E^s_{p}[\psi](\tau_2)\les  E^s_{p}[\psi](\tau_1)+ B_p^s[\psi] (\tau_1, \tau_2) +\NN_p^s[\psi,N](\tau_1, \tau_2).
\eeq
\end{proposition}
\begin{proof}
Consider the vectorfield $\Ta:=T+\frac{a}{a^2+ (1.92m)^2} Z$ of Proposition \ref{Prop:Tring} which is both Killing
 and timelike on the trapping set $[\rhat_1, \rhat_2]$. We choose a smooth vectorfield $\Tring $ as follows
 \beq
 \Tring=\begin{cases}
 \That, \qquad &r_+\le r \le \rhat_1-\de\\
 \Ta, \qquad &\rhat_1\le r \le \rhat_2\\
 \T,\qquad   &\rhat_2+\de \le r
 \end{cases},
 \eeq
 which is everywhere timelike and Killing in the trapping set.
 We deduce
 \begin{align*}
&\int_{\Sigma_{\tau_2}}\QQ[\psi](\mathring{T}, N_\Sigma)\,+\int_{\HH(\tau_1, \tau_2)}\QQ[\psi](\Tplus, N_\HH)\\
&\quad\le \int_{\DD(\tau_1, \tau_2)}\Big(\frac12\big|\QQ_{\mu\nu}\ ^{(\mathring{T})}\pi^{\mu\nu}\big|+|\mathring{T}\psi\c N|\Big)+\int_{\Sigma_{\tau_1}}\QQ[\psi](\mathring{T}, N_\Sigma)
\end{align*}
which implies
\[
\EFdeg[\psi](\tau_1,\tau_2) \les  \int_{\DD(\tau_1, \tau_2)} |\QQ\c \,^{(\Tring)}\pi|+ \Edeg[\psi](\tau_1)+\NN_1[\psi,N](\tau_1,\tau_2).
\]
Since $\Tring$ is Killing in the trapping set and for $r\ge \rhat_2+\de$ and of the form $\Tring= T+ f(r) \Z$ everywhere else
 we deduce, wherever $\Tring$ is not Killing, with $c>0$ small,
 \[
\big|  \QQ\c \,^{(\Tring)}\pi \big| \les  \de^{-1}  \big|\Rhat \psi \c Z\psi\big|.
 \]
 Therefore,
 \[
 \EFdeg[\psi](\tau_1,\tau_2) \les  \Edeg[\psi](\tau_1) + \delta^{-1}  \Bdeg[\psi](\tau_1, \tau_2)+\NN_1[\psi,N](\tau_1,\tau_2).
 \]
 Combining this estimate with the partial red shift estimate in Proposition \ref{prop:partialredshiftestimate} and $r^p$ estimates, we deduce
 \[
 E_p[\psi](\tau_2)\les  E_p[\psi](\tau_1)+ B_p[\psi] (\tau_1, \tau_2) +\NN_p[\psi,N](\tau_1, \tau_2).
 \] Also, by  the standard procedure of taking higher derivative,
 \[
 E_p^s[\psi](\tau_1,\tau_2) \les  E_p^s[\psi](\tau_1) + \delta^{-1} B_p^s[\psi](\tau_1, \tau_2)+\NN_p^s[N](\tau_1,\tau_2).
 \]
\end{proof}


\subsection{A new unrestricted proof of the energy estimate}
\lab{subsec:new-unrestricted-energy}


We provide a simpler energy estimate, valid for subextremal
Kerr spacetimes, by choosing a multiplier adapted to the trapping term $P=\int_{\DD_{trap}}|\Tfrak[\psi]|^2$.  Consider
\[
Y=f(r)\That=fT+\frac{af}{r^2+a^2}Z,\quad f(r)=\frac{r^2+a^2}{\sqrt{\De}}.
\]
Then
\beq
\lab{eq:deformation-fThat}
\piY_{\a\b}
=2\frac{r^2+a^2}{\De^{3/2}}\Big((dr)_{(\a}(\TT T_{\b)}-a(r-m)Z_{\b)})\Big).
\eeq
 Hence 
\beq
\lab{eq:Qpi-fThat-general}
\frac12\QQ^{\a\b}\piY_{\a\b}
=\frac{r^2+a^2}{\De^{3/2}}\frac{\De}{|q|^2}\partial_r\psi\c
 \left(\TT T\psi-a(r-m)Z\psi\right)=\frac{r^2+a^2}{\De^{3/2}}=\frac{r^2+a^2}{|q|^2\De^{3/2}}\Rhat\psi\c\Tfrak[\psi]
  \eeq
where we use $\Tfrak[\psi]=\TT_{-a}T\psi-a(r-m)\Zhat\psi=\TT T\psi-a(r-m)Z\psi$.

We now choose a smooth causal multiplier $Y$ such that $Y=\That$ on $[r_+, \rhat_1-2\de_{trap}]$ and $Y=\frac{r^2+a^2}{\sqrt{\De}}\That$ on $\DD_{trap}=[\rhat_1-\de_{trap}, \rhat_2+\de_{trap}]$ and $Y=T$ on $[\rhat_2+2\de_{trap}, \infty)$. Then we have
\beq
\lab{eq:Qpi-fThat-trapping-bound}
\left|\int_{\DD(\tau_1,\tau_2)}
 \frac12\QQ^{\a\b}\piY_{\a\b}\right|
\les_{a,m}P[\psi](\tau_1,\tau_2)
+B[\psi](\tau_1,\tau_2)
=\BP[\psi](\tau_1,\tau_2).
\eeq
Applying the $Y$-energy estimate and then using the
partial red-shift and $r^p$ estimates gives, after commuting $T,Z, \OO$ in the standard
way,
\beq
\lab{eq:unrestricted-energy-with-P}
E_p^s[\psi](\tau_2)
\les E_p^s[\psi](\tau_1)
+\BP_p^s[\psi](\tau_1,\tau_2)
+\NN_p^s[\psi,N](\tau_1,\tau_2).
\eeq

\appendix



\section{Proof of results from section \ref{section:preliminaries}}

\subsection{Proof of  Proposition \ref{prop:decompose-square}}

\lab{appendix:decompose-square}
In view of \eqref{eq:inversemetric-Kerr}
\begin{align*}
\square \psi &=\frac{1}{\sqrt{ |\g|}  } \pr_\b\Big(\sqrt{|\g|} \g^{\a\b} \pr_\a \psi\Big)=\pr_t \Big( \g^{tt}  \pr_t \psi +\g^{t\phi}\pr_\phi\psi \Big)+ \pr_\phi \Big( \g^{\phi\phi}  \pr_\phi  \psi +\g^{t\phi}\pr_t \psi \Big)\\
&\quad+ \frac{1}{|q|^2 \sin \th} \pr_r \big(  |q|^2 \sin \th \g^{rr} \pr_r \psi\big) + \frac{1}{|q|^2 \sin \th} \pr_\th \big(  |q|^2 \sin \th \g^{\th\th} \pr_\th \psi\big)\\
&=\Big( \g^{tt}  \pr^2_t   + 2 \g^{t\phi}\pr_t\pr_\phi  +\g^{\phi\phi}    \pr^2_\phi    \Big)  \psi +  \frac{1}{|q|^2 } \Big(\pr_r\big(\De \pr_r \psi\big)+ \frac{1}{\sin \th} \pr_\th \big(\sin\th \pr_\th \psi\big)\Big).
\end{align*}
Hence
\begin{align*}
|q|^2 \square \psi &=\Big(-\frac{(r^2+a^2)^2- a^2 \sin^2\th \De}{\De} \pr_t^2  -\frac{4 a mr }{\De}\pr_t\pr_ \phi +\frac{\De- a^2 \sin^2\th}{ \De \sin^2 \th}\pr_\phi^2\Big)\psi\\
&\quad+ \Big(\pr_r\big(\De \pr_r \psi\big) +\frac{1}{\sin\th} \pr_\th \big(\sin\th \pr_\th \psi\big)\Big)\\
&= \Big(-\frac{(r^2+a^2)^2}{\De} \pr_t^2  +\frac{2a(\De- r^2 - a^2) }{\De}\pr_t\pr_ \phi -\frac{a^2}{\De} \pr_\phi^2 \Big)\psi+\Big(a^2 \sin^2 \th \pr_t^2+\frac{1}{\sin^2\th} \pr_\phi^2\Big)\psi\\
&\quad+ \Big(\pr_r\big(\De \pr_r \psi\big) + \frac{1}{\sin\th}  \pr_\th \big(\sin\th \pr_\th \psi\big)\Big)\\
&= \Big(-\frac{(r^2+a^2)^2}{\De} \pr_t^2  - \frac{2a(r^2+a^2) }{\De}\pr_t\pr_ \phi -\frac{a^2}{\De} \pr_\phi^2 \Big)\psi+ \Big(a^2 \sin^2 \th \pr_t^2+ 2 a \pr_t\pr_\phi+\frac{1}{\sin^2\th} \pr_\phi^2\Big)\psi\\
&\quad+ \Big(\pr_r\big(\De \pr_r \psi\big) + \frac{1}{\sin\th}  \pr_\th \big(\sin\th \pr_\th \psi\big)\Big).
\end{align*}
We rewrite in the form
\begin{align*}
|q|^2 \square \psi &= \pr_r\big(\De \pr_r \psi\big) +  \Big(a^2 \sin^2 \th \pr_t^2+ 2 a \pr_t \pr_\phi+\frac{1}{\sin^2\th} \pr_\phi^2\Big)\psi + \frac{1}{\sin\th}  \pr_\th \big(\sin\th \pr_\th \psi\big)\\
&+\frac{1}{\De}\Big(-(r^2+a^2)^2 \pr_t^2 - 2  a(r^2+a^2)  \pr_t\pr_\phi- a^2 \pr_\phi^2\Big)\psi
\end{align*}
or, introducing the operators
\begin{align*}
\OO&=  a^2 \sin^2 \th \pr_t^2+ 2 a \pr_t\pr _\phi +\frac{1}{\sin^2\th} \pr_\phi^2 + \frac{1}{\sin\th}  \pr_\th \big(\sin\th \pr_\th \big),\\
\RR&=
\pr_r\big(\De \pr_r \big) - \frac{1}{\De}\Big((r^2+a^2)^2 \pr_t^2 +2  a(r^2+a^2)  \pr_t\pr_\phi+a^2 \pr_\phi^2\Big)\\
&= \pr_r\big(\De \pr_r \big) -\frac{(r^2+a^2)^2}{\De} \That \That ,
\end{align*}
we deduce $|q|^2\square= \RR +\OO$
as stated.

To check \eqref{eq:OO-tensorial} it is convenient to re-express $|q|^2 \square\psi$ as follows:
\begin{align*}
|q|^2 \square \psi &=|q|^2\D_\a\big(\g^{\a\b} \D_\b\psi\big)=\D_\a\big(|q|^2\g^{\a\b} \D_\b\psi\big)- \D_\a(|q|^2)  \g^{\a\b} \D_\b\psi\\
&= \D_\a\big(|q|^2\g^{\a\b} \D_\b\psi\big)- \D_\a(|q|^2)  \g^{\a\b} \D_\b\psi\\
&=  \D_\a\big(|q|^2\g^{\a\b} \D_\b\psi\big)-\frac{\De}{|q|^2} \pr_r(|q|^2)\pr_r \psi-\frac{1}{|q|^2}  \pr_\th (|q|^2)\pr_\th .
\end{align*}
In view of \eqref{inverse-metric-vfs}
\[
\D_\a\big(|q|^2\g^{\a\b} \D_\b\psi\big) = \D_\a\big(O^{\a\b}\D_\b \psi\big)+ \D_\a \Big(  \frac{(r^2+a^2)^2}{\De} \big( -\That^\a\That^\b +      \Rhat^\a \Rhat^\b\big)\D_\b \psi\Big)
\]
Hence
\begin{align*}
|q|^2 \square \psi &= \D_\a\big(O^{\a\b}\D_\b \psi\big) -\frac{1}{|q|^2}  \pr_\th (|q|^2)\pr_\th \psi+\RR\psi\\
\RR\psi&=\D_\a \Big(  \frac{(r^2+a^2)^2}{\De} \big( -\That^\a\That^\b +      \Rhat^\a \Rhat^\b\big) \D_\b\psi\Big)- \frac{\De}{|q|^2} \pr_r(|q|^2)\pr_r \psi
\end{align*}
Thus
\[
\OO\psi=\D_\a\big(O^{\a\b}\D_\b \psi\big) -\frac{1}{|q|^2}  \pr_\th (|q|^2)\pr_\th \psi=|q|^2 \D_\a\big(|q|^{-2} O^{\a\b}\D_\b \psi\big)
\]
as stated.


\subsection{Proof of Lemma \ref{lemma:N_Si}}
\lab{Appendix:Lemma:N_Si}

\begin{proof}
The first statement follows immediately from the above definitions.
To check the second we write
\begin{align*}
\g(\gradtau, \gradtau)&= -\frac{\Si^2}{|q|^2 \De}+\frac{\De}{|q|^2} f'(r) f'(r)\\
&=-\frac{1}{|q|^2}\Big(\frac{(r^2+a^2)^2}{\Delta}-a^2\sin^2\th- \De  \big(\frac{r^2+a^2}{\Delta}-\frac{m^2}{r^2}\big)^2\chi^2(r)\Big)\\
&= -\frac{1}{|q|^2}\Bigg( \frac{(r^2+a^2)^2}{\Delta}-a^2\sin^2\th- \De\chi^2 \big( \frac{(r^2+a^2)^2}{\Delta^2}- 2\frac{r^2+a^2}{\Delta} \frac{m^2}{r^2}+
\frac{m^4}{r^4}\big)\Bigg)\\
&= -\frac{1}{|q|^2}\Bigg(  \frac{(r^2+a^2)^2}{\Delta}(1-\chi^2)- a^2\sin^2\th+ \Big(2(r^2+a^2) \frac{m^2}{r^2}  -\De \frac{m^4}{r^4}\Big) \chi^2\Bigg).
\end{align*}
Note that in the region where $\chi=1$, for all values of $r$,
\begin{align*}
-|q|^2 \g(\gradtau, \gradtau)&= - a^2 \sin^2\th + 2(r^2+a^2) \frac{m^2}{r^2}  -\De \frac{m^4}{r^4}\\
&=  - a^2 \sin^2\th+ (r^2+a^2) \frac{m^2}{r^2} +\frac{m^2}{r^4} \big(r^2+a^2)r^2 -\De m^2\big)\\
&=   \frac{   m^2(r^2+a^2) - a^2 r^2 \sin^2\th} {r^2}    +  \frac{m^2}{r^4} \big( (r^2-m^2)(r^2+a^2) + 2m^3 r\big)\\
&> 0.
\end{align*}
In the region where $\chi=0$ we have
\begin{align*}
-|q|^2 \g(\gradtau,\gradtau) &= \frac{(r^2+a^2)^2}{\Delta} - a^2 \sin^2\th = \frac{ (r^2+a^2)^2- a^2 \sin^2\th\De}{\De}\\
&=\frac{(r^2+a^2) |q|^2 + 2 a^2 mr}{\De}>0
\end{align*}
which proves \eqref{eq:g(N_Si,N_Si)}.

Finally, relative to the null frame $\gradtau=-\frac 12 ( e_4(\tau) e_3 + e_3(\tau) e_4)+ e_2(\tau) e_2$.
We calculate
\begin{align*}
e_2(\tau)&=\frac{a\sin\th}{|q|} \chi \\
e_3(\tau)&=\frac{r^2+a^2}{\De}- f'(r) =\frac{r^2+a^2}{\De}-\big(\frac{r^2+a^2}{\Delta}-\frac{m^2}{r^2} \big) \chi\\
&=\big(1-\chi\big)\frac{r^2+a^2}{\De}+\chi\frac{m^2}{r^2}\\
e_4(\tau)&=\frac{r^2+a^2}{|q|^2} +\frac{\De}{|q|^2}f'(r)= \frac{r^2+a^2}{|q|^2}+\frac{\De}{|q|^2} \big(\frac{r^2+a^2}{\Delta}-\frac{m^2}{r^2} \big) \chi\\
&= \frac{r^2+a^2}{|q|^2}(1+\chi)-\frac{\De m^2}{r^2|q|^2} \chi=O_+(1)
\end{align*}
In the region $r\le r_0$, for which $\chi=1$,
we have
\[
e_3(\tau)=\frac{m^2}{r^2}, \qquad
e_4(\tau)=\frac{(2r^2-m^2)(r^2+a^2) + 2m^3 r}{r^2|q|^2}\ges  1.
\]
Thus, for $r\le r_0 $,
\[
\gradtau=-\frac 12\big( \frac{m^2}{r^2}  e_4+ O_+(1)  e_3\big)+  \frac{a\sin\th}{|q|} e_2
\]
while for $r \ge 2r_0$,
\[
\gradtau=\gradt = - \frac 12\big(  e_3(t) e_4+ e_4(t) e_3 \big)+ e_a(t) e_a=- O_+(1)  \big(e_3+e_4\big) + O(ar^{-1}) e_2.
\]
\end{proof}

\subsection{Proof of the partial red shift  Lemma  \ref{lemma:PRS}}

\lab{appendix:Lemma-PRS}
For convenience we repeat the content of the Lemma below.
\begin{lemma}
\lab{lemma:PRS'}
 Let $Y=d(r)e_3=d(r)(-\pr_r+\frac{r^2+a^2}{\De}T+\frac{a}{\De}Z)$ with $d(r)=\chi(\frac{r-r_+}{r_+\de_{red}})\Big(1+\frac{r-r_+}{r_+}\Big)$ and $\chi(s)=1$ for $s\leq 1$ and $\chi(s)=0$ for $s\geq 2$.
Then if $\de_{red}$ is sufficiently small, we have
\beq
|q|^2\EE[Y,0,0]\geq\frac18(r-m)d(r)|e_3\psi|^2+\chi'\frac{1}{r_+\de_{red}}\big||q|\nab\psi\big|^2-\frac{40d(r)}{r^2(r-m)}|\De\pr_r\psi|^2.
\eeq
Also
\beq
\begin{split}
\QQ(Y, N_{\Si})&\gtrsim d(r)(|e_3\psi|^2+|\nab\psi|^2 )  \\
\QQ(Y, N_{\HH})&=\QQ(e_3, \frac12e_4)=\frac12|\nab\psi|^2.\end{split}
\eeq
where $\QQ$ is the energy momentum tensor of $\square_{a,m}$ and, with the notation of
 section \eqref{sect:Basicscalarident}, $\EE[Y,0,0]=\Div\big(\QQ \c Y)- Y(\psi) \square \psi$.
\end{lemma}
\begin{proof}
A straightforward calculation yields, see \eqref{eq:Gen-identityXwM} and \eqref{definition-EE-gen1}
\[
|q|^2\EE[Y,0,0]=\frac12|q|^2\QQ\c\piY=\frac{1}{2}|q|^2d(r)\QQ\c\pithree+\De d'(r)\QQ_{\a\b}\pr_r^{(\a}e_3^{\b)}.
\]
Using
 \[
\De \pr_r=\frac12(|q|^2e_4-\De e_3)
\]
and (see Lemma 9.2.18 in \cite{GKS})
\begin{align*}
\pithree_{33}&=\pithree_{3a}=\pithree_{34}=0,\quad a=1,2\\
\pithree_{44}&=4\pr_r(\frac{\De}{|q|^2}),\quad \pithree_{4a}=-\frac{2ar\sin\th}{|q|^3}\de_{a2},\quad \pithree_{ab}=-\frac{2r}{|q|^2}\de_{ab},\quad a,b=1,2,
\end{align*}
we further compute
\begin{align*}
|q|^2\EE[Y, 0, 0]&=\frac12\Big(|q|^2d(r)\pr_r(\frac{\De}{|q|^2})-\De d'(r)\Big)|e_3\psi|^2+\frac12|q|^2d'(r)|\nab\psi|^2\\
&-2rd(r)e_3\psi\c e_4\psi+\frac{ar\sin\th}{|q|}d(r)e_4\psi\c e_2\psi.
\end{align*}
 Recall that
 \[
d(r)=\chi(\frac{r-r_+}{r_+\de_{red}})\Big(1+\frac{r-r_+}{r_+}\Big).
\]
Then by choosing $\de_{red}$ sufficiently small we have
\begin{align*}
\frac12\big(|q|^2d(r)\pr_r(\frac{\De}{|q|^2})-\De d'(r)\big)&=d(r)\big((r-m)-\frac{r\De}{|q|^2}\big)-\frac{\De}{2r_+}\chi-\frac12\frac{\De}{r_+\de_{red}}(1+\frac{r-r_+}{r_+})\chi'\\
&\geq d(r)\big((r-m)-\frac{r\De}{|q|^2}\big)-\frac{\De}{2r_+}\frac{d(r)}{1+\frac{r-r_+}{r_+}}\geq \frac12 d(r)(r-m).
\end{align*}
 Therefore, for $|q|^2\EE[Y]=|q|^2\EE[Y,0,0]$ we derive
\beq\lab{eq:partialredshift}\bsplit
|q|^2\EE[Y]&\geq \frac12(r-m)d(r)|e_3\psi|^2+\frac{1}{2r_+}\chi(\frac{r-r_+}{r_+\de_{red}})\big||q|\nab\psi\big|^2+\chi'(\frac{r-r_+}{r_+\de_{red}})\frac{1}{r_+\de_{red}}\big||q|\nab\psi\big|^2\\
&-2rd(r)e_3\psi\c e_4\psi+\frac{ar\sin\th}{|q|}d(r)e_4\psi\c e_2\psi\\
&\geq \frac12(r-m)d(r)|e_3\psi|^2+\frac{1}{2r_+}\chi(\frac{r-r_+}{r_+\de_{red}})\big||q|\nab\psi\big|^2+\chi'(\frac{r-r_+}{r_+\de_{red}})\frac{1}{r_+\de_{red}}\big||q|\nab\psi\big|^2\\
&-d(r)\big(\frac14(r-m)|e_3\psi|^2+\frac{4r^2}{r-m}|e_4\psi|^2\big)-d(r)\big(\frac{2a^2}{r_+}|e_4\psi|^2+\frac{1}{8r_+}||q|e_2\psi|^2\big)\\
&\geq \frac14(r-m)d(r)|e_3\psi|^2+\chi'\frac{1}{r_+\de_{red}}\big||q|\nab\psi\big|^2-\frac{6r^2}{r-m}d(r)|e_4\psi|^2\\
&\geq\frac18(r-m)d(r)|e_3\psi|^2+\chi'\frac{1}{r_+\de_{red}}\big||q|\nab\psi\big|^2-\frac{48\De^2}{r^2(r-m)}d(r)|\pr_r\psi|^2
\end{split}
\eeq
where we use $2\De\pr_r=|q|^2e_4-\De e_3$ in the last step.
\end{proof}

\subsection{Proof  of Lemma \ref{Lemma:calculuslemmadeg}}

\lab{appendix:Lemma-calculuslemmadeg}
\begin{proof}
We recall (see \eqref{eq:square-taucoords} in Lemma \ref{lemma:square-taucoords}), that, by writing $\lapSS\psi=\frac{1}{ \sin \th} \pr_\th \big( \sin \th \pr_\th \psi\big)
+\frac{1}{\sin^2\th} \pr_{\phi}^2\psi $
\beq\lab{eq:square-N}
\bsplit
&\frac{1}{|q|^2} \Big(\pr_{\tilde r}\big(\De \pr_{\tilde r}  \psi\big)+2(r^2+a^2)\That\pr_{\tilde r}+\lapSS\psi \Big) \\
&\quad=O(1)\big| (T, r^{-2}\De\pr_{\tilde r}, r^{-1}Z)(\T, r^{-1} \Z)   \psi \big|
+O( r^{-1})\big| (\T, r^{-1} \Z)\psi \big|+N.
\end{split}
\eeq

{\bf Step 1.} We derive the control of the degenerate energy $\Edeg$ of at most two degenerate derivatives of $\psi$. By writing $d\gamma_{\mathbb{S}^2}=\sin\th d\th d\phit$ and denoting $\nabSS$ the covariant derivative on the standard sphere $\mathbb{S}^2$ in the direction of $\eSS_1=\pr_\th, \eSS_2=\frac{1}{\sin\th}\pr_{\phit}$, we derive
\beq
\lab{eq:lapSSsquaredeg}
\bsplit
\int_{\Sigma(\tau)}\frac{f(r)}{|q|^2}|\lapSS\psi|^2&=\int_{r_+
}^\infty\int_{\mathbb{S}^2}f(r)|\lapSS\psi|^2\,d\gamma_{\mathbb{S}^2}d\tilde r\\
&=\int_{r_+
}^\infty\int_{\mathbb{S}^2}f(r)\Big(|\nabSS\nabSS\psi|^2+|\nabSS\psi|^2\Big)\,d\gamma_{\mathbb{S}^2}d\tilde r\\
&= \int_{\Sigma(\tau)}\frac{f(r)}{|q|^2}\Big(|\nabSS\nabSS\psi|^2+|\nabSS\psi|^2\Big).
\end{split}
\eeq
Using the definition of $\OO$ and the calculation in \eqref{eq:lapSSsquaredeg}, we deduce
\begin{align*}
&\Edeg[\nabSS\nabSS\psi](\tau)+\Edeg[\nabSS\psi](\tau)\\
&\quad\quad\les \Edeg[\lapSS\psi](\tau)
\les \Edeg[\OO\psi](\tau)+\Edeg[(T,Z)^{\leq 2}\psi](\tau).
\end{align*}
Therefore
\beq
\lab{eq:secondorderenergynabdeg}
\bsplit
\Edeg[(r\nab)^{\leq 2}\psi](\tau)\les \Edeg[\OO\psi](\tau)+\Edeg[(T,Z)^{\leq 2}\psi](\tau).
\end{split}
\eeq
Using \eqref{eq:square-N} and \eqref{eq:secondorderenergynabdeg}, we infer
\begin{align*}
\int_{\Sigma(\tau)}r^{-2}|\De||\pr_{\tilde r}(r^{-2}\De\pr_{\tilde r})\psi|^2&\les \Edeg[(T, \nab)^{\leq 1}\psi](\tau)+\int_{\Sigma(\tau)}r^{-2}|\De||N|^2\\
&\les \Edeg[\OO \psi](\tau)+\Edeg[(T,Z)^{\leq 2}\psi](\tau)+\int_{\Sigma(\tau)}r^{-2}|\De||N|^2
\end{align*}
and thus
\beq \lab{eq:firstorderenergydeg}
\bsplit
&\Edeg[(e_4, r^{-2}\De e_3, r\nab)^{\leq 1}\psi](\tau)\\
&\quad\quad\les \Edeg[\OO\psi](\tau)+\Edeg[(T,Z)^{\leq 2}\psi](\tau)+\int_{\Sigma(\tau)}r^{-2}|\De||N|^2.
\end{split}
\eeq
Using \eqref{eq:square-N} (with $\psi$ replaced by $T\psi,Z\psi$), we derive
\beq\lab{eq:secondorderenergyrrtaudeg}
\bsplit
&\int_{\Sigma(\tau)}r^{-2}|\De||\pr_{\tilde r}(r^{-2}\De\pr_{\tilde r})(T,Z)\psi|^2\\
&\quad\les \Edeg[\OO\psi](\tau)+\Edeg[(T, Z)^{\leq 2}](\tau)+\int_{\Sigma(\tau)}r^{-2}|\De||(T,Z)N|^2.
\end{split}
\eeq
Using \eqref{eq:square-N}, \eqref{eq:firstorderenergydeg} and \eqref{eq:secondorderenergyrrtaudeg}, we deduce
\beq\lab{eq:secondorderenergyrrrdeg}
\bsplit
&\int_{\Sigma(\tau)}r^{-2}|\De||\pr_{\tilde r}(r^{-2}\De\pr_{\tilde r})^2\psi|^2\\
&\quad\les \Edeg[\OO\psi](\tau)+\Edeg[(T,Z)^{\leq 2}\psi](\tau)+\Edeg[r^{-2}\De\pr_{\tilde r}\psi](\tau)\\
&\quad+\int_{\Sigma(\tau)}r^{-2}|\De||\pr_{\tilde r}(r^{-2}\De\pr_{\tilde r})(T,Z)\psi|^2+\int_{\Sigma(\tau)}r^{-2}|\De||(r^{-2}\De\pr_{\tilde r})^{\leq 1}N|^2\\
&\quad\les \Edeg[\OO\psi](\tau)+\Edeg[(T,Z)^{\leq 2}\psi](\tau)+\int_{\Sigma(\tau)}r^{-2}|\De||(T, r^{-2}\De\pr_{\tilde r},Z)^{\leq 1}N|^2.
\end{split}
\eeq
Applying integration by parts to $\int_{\Sigma(\tau)
}|\eSS_i(r^{-2}\De\pr_{\tilde r})T\psi|^2$ and $\int_{\Sigma(\tau)
}r^{-2}\De|\eSS_i\pr_{\tilde r}(r^{-2}\De\pr_{\tilde r})\psi|^2$, and using \eqref{eq:secondorderenergyrrtaudeg}, \eqref{eq:secondorderenergyrrrdeg} and \eqref{eq:firstorderenergydeg}, we infer
\begin{align*}
\int_{\Sigma(\tau)
}|(r\nab)(r^{-2}\De\pr_{\tilde r})T\psi|^2&\les \Edeg[TT\psi](\tau)+\int_{\Sigma_\tau}|\lapSS T\psi|^2+|(r^{-2}\De\pr_{\tilde r})^{ 2}T\psi|^2\\
&\les \Edeg[\OO\psi](\tau)+\Edeg[(T, Z)^{\leq 2}\psi](\tau)+\int_{\Sigma(\tau)}r^{-2}|\De||TN|^2\\
\int_{\Sigma(\tau)}r^{-2}|\De||(r\nab)\pr_{\tilde r}(r^{-2}\De\pr_{\tilde r})\psi|^2&\les\int_{\Sigma(\tau)}r^{-2}|\De|\Big(|\pr_{\tilde r}(r^{-2}\De\pr_{\tilde r})T\psi|^2+\pr_{\tilde r}(r^{-2}\De\pr_{\tilde r})^2\psi|^2+|\pr_{\tilde r}\lapSS \psi|^2\Big)\\
&\les \Edeg[\OO\psi](\tau)+\Edeg[(T, Z)^{\leq 2}\psi](\tau)\\
&+\int_{\Sigma(\tau)}r^{-2}|\De||(T, r^{-2}\De\pr_{\tilde r},Z)^{\leq 1}N|^2
\end{align*}
and thus
\begin{align*}
\int_{\Sigma(\tau)
}|(r\nab)TT\psi|^2
&\les \Edeg[TT\psi](\tau)+\int_{\Sigma(\tau)\cap\{r\geq 10m\}}|(r\nab)TT\psi|^2\\
&\les \Edeg[TT\psi](\tau)+\int_{\Sigma(\tau)\cap\{r\geq 10m\}}|(r\nab)\pr^2_{\tilde r}\psi|^2+|(r\nab)(\pr_\tau, \nab)^{\leq 1}\nab\psi|^2\\
&+\int_{\Sigma(\tau)\cap\{r\geq 10m\}}|r^{-1}(r\nab) \pr_{\tilde r}\psi|^2+\int_{\Sigma(\tau)\cap\{r\geq 10m\}}|r\nab N|^2\\
&\les \Edeg[\OO\psi](\tau)+\Edeg[(T, Z)^{\leq 2}\psi](\tau)+\int_{\Sigma(\tau)}r^{-2}|\De||(T, r^{-2}\De\pr_{\tilde r}, r\nab)^{\leq 1}N|^2.
\end{align*}
Putting \eqref{eq:secondorderenergynabdeg}, \eqref{eq:firstorderenergydeg}, \eqref{eq:secondorderenergyrrtaudeg}, \eqref{eq:secondorderenergyrrrdeg} and the above estimates together yields
\beq\lab{eq:secondorderenergydeg}
\bsplit
&\Edeg[(r^{-2}\De e_3, e_4, r\nab)^{\leq 2}\psi](\tau)\sim\Edeg[(\pr_\tau, r^{-2}\De\pr_{\tilde r}, r\nab)^{\leq 2}\psi](\tau)\\
&\quad\quad\les
\Edeg[(r\nab)^{\leq 2}\psi](\tau)+\Edeg[TT\psi](\tau)+\int_{\Sigma(\tau)
}\Big(|r\nab TT\psi|^2+r^{-2}|\De||\pr_{\tilde r}(r\nab) \T\psi|^2\Big)\\
&\quad\quad+\int_{\Si(\tau)}r^{-2}|\De||\pr_{\tilde r}(T, r\nab)(r^{-2}\De\pr_{\tilde r})\psi|^2 +\Edeg[(r^{-2}\De\pr_{\tilde r})^2\psi](\tau)\\
&\quad\quad+\Edeg[(\pr_\tau, r^{-2}\De\pr_{\tilde r}, r\nab)^{\leq 1}\psi](\tau)\\
&\quad\quad\les \Edeg[\OO\psi](\tau)+\Edeg[(T, Z)^{\leq 2}\psi](\tau)+\int_{\Sigma(\tau)}r^{-2}|\De||( r^{-2}\De e_3, e_4, r\nab)^{\leq 1}N|^2.
\end{split}
\eeq
To control the inhomogeneous term $N$, we integrate by parts in $\tau$ and thus finish the proof of the energy part in \eqref{eq:seconddegenergy}.

{\bf Step 2.} We derive the control of the degenerate bulk term $\Bdeg$ of at most two degenerate derivatives of $\psi$. Proceeding as in the proof of estimate \eqref{eq:secondorderenergynabdeg} in step 1, we deduce
\[
\Bdeg[(\eSS_1, \eSS_2)^{\leq 2}\psi](\tau_1, \tau_2)\leq \Bdeg[\OO\psi](\tau_1,\tau_2)+\Bdeg[(T,Z)^{\leq 2}\psi](\tau_1,\tau_2).
\]
Applying integration by parts to $\int_{\DD(\tau_1,\tau_2)\cap I}r^2|q|^{-2}h(r)\pr_\tau^2f\c\lapSS f$ (where $f=(\pr_\tau, \Rhat, \eSS_i)^{\leq 1}\psi$) yields
\[
\Bdeg[(\eSS_1, \eSS_2)\pr_\tau\psi]\les \BEdeg[\OO\psi](\tau_1,\tau_2)+\BEdeg[(T,Z)^{\leq 2}\psi](\tau_1,\tau_2).
\]
Therefore
\beq\lab{eq:secondorderbulktaunabdeg}
\Bdeg[(\pr_\tau, r\nab)^{\leq 2}\psi]\les \BEdeg[\OO\psi](\tau_1,\tau_2)+\BEdeg[(T,Z)^{\leq 2}\psi](\tau_1,\tau_2).
\eeq
Using \eqref{eq:square-taucoords2} and \eqref{eq:secondorderbulktaunabdeg} (as in the proof of estimate \eqref{eq:firstorderenergydeg} in step 1), we infer
\begin{align*}
\int_{\DD(\tau_1,\tau_2)}r^{-2}|\Rhat^2\psi|^2&\les \Bdeg[(\pr_\tau, r\nab)^{\leq 2}\psi](\tau_1,\tau_2)+\int_{\DD(\tau_1,\tau_2)}r^{-6}\De^2|N|^2\\
&\les\BEdeg[\OO\psi](\tau_1,\tau_2)+\BEdeg[(T,Z)^{\leq 2}\psi](\tau_1,\tau_2)+\int_{\DD(\tau_1,\tau_2)}r^{-6}\De^2|N|^2
\end{align*}
and thus
\beq \lab{eq:firstorderbulkdeg}
\bsplit
&\Bdeg[(\pr_\tau, \Rhat, r\nab)^{\leq 1}\psi](\tau_1, \tau_2)\\
&\quad\quad\les\BEdeg[\OO\psi](\tau_1,\tau_2)+\BEdeg[(T,Z)^{\leq 2}\psi](\tau_1,\tau_2)+\int_{\DD(\tau_1,\tau_2)}r^{-6}\De^2|N|^2.
\end{split}
\eeq
Using \eqref{eq:square-taucoords2}, \eqref{eq:secondorderbulktaunabdeg} and \eqref{eq:firstorderbulkdeg} (as in the proof of estimate \eqref{eq:secondorderenergyrrrdeg} in step 1), we derive
\beq\lab{eq:secondorderbulkrrrdeg}
\bsplit
&\int_{\DD(\tau_1,\tau_2)}r^{-2}|\Rhat^3\psi|^2\\
&\quad\les \Bdeg[(\pr_\tau, r\nab)^{\leq 2}\psi](\tau_1,\tau_2)+\Bdeg[\Rhat\psi](\tau_1,\tau_2)+\int_{\DD(\tau_1,\tau_2)}r^{-6}\De^2|\Rhat^{\leq 1}N|^2\\
&\quad\les\BEdeg[\OO\psi](\tau_1,\tau_2)+\BEdeg[(T,Z)^{\leq 2}\psi](\tau_1,\tau_2)+\int_{\DD(\tau_1,\tau_2)}r^{-6}\De^2|\Rhat^{\leq 1}N|^2.
\end{split}
\eeq
Applying integration by parts to $\int_{\DD(\tau_1, \tau_2)
}|q|^{-2}|(\pr_\tau, \eSS_i)\Rhat^2\psi|^2$, and using \eqref{eq:secondorderbulktaunabdeg} and \eqref{eq:secondorderbulkrrrdeg}, we obtain
\beq\lab{eq:secondorderbulkrmixeddeg}
\bsplit
&\int_{\DD(\tau_1, \tau_2)
}r^{-2}|(\pr_\tau, r\nab)\Rhat^2\psi|^2\\
&\quad\les \Bdeg[(\pr_\tau, r\nab)^{\leq 2}\psi](\tau_1,\tau_2)+\int_{\DD(\tau_1, \tau_2)}r^{-2}|\Rhat^{3}\psi|^2\\
&\quad+\Edeg[(r^{-2}\De e_3,e_4,r\nab)^{\leq 2}\psi](\tau_1, \tau_2)\\
&\quad\les \Bdeg[\OO\psi](\tau_1,\tau_2)+\Bdeg[(T,Z)^{\leq 2}\psi](\tau_1,\tau_2)\\
&\quad+\Edeg[(r^{-2}\De e_3,e_4,r\nab)^{\leq 2}\psi](\tau_1, \tau_2)+\int_{\DD(\tau_1,\tau_2)}r^{-6}\De^2|\Rhat^{\leq 1}N|^2.
\end{split}
\eeq
Putting \eqref{eq:secondorderbulktaunabdeg}, \eqref{eq:firstorderbulkdeg}, \eqref{eq:secondorderbulkrrrdeg} and \eqref{eq:secondorderbulkrmixeddeg} together yields
\beq\lab{eq:secondorderbulkdeg}
\bsplit
&\Bdeg[(r^{-2}\De e_3, e_4, r\nab)^{\leq 2}\psi](\tau_1,\tau_2)\sim\Bdeg[(T, \Rhat, r\nab)^{\leq 2}\psi](\tau_1,\tau_2)\\
&\quad\quad\les
\Bdeg[(\pr_\tau, r\nab)^{\leq 2}\psi](\tau_1,\tau_2)+\int_{\DD(\tau_1,\tau_2)}r^{-2}|(\pr_\tau, \Rhat, r\nab)\Rhat^2\psi|\\
&\quad\quad+\Bdeg[(\pr_\tau, \Rhat, r\nab)^{\leq 1}\psi](\tau_1,\tau_2)\\
&\quad\quad\les \Bdeg[\OO\psi](\tau_1,\tau_2)+\Bdeg[(T, Z)^{\leq 2}\psi](\tau_1,\tau_2)\\
&\quad\quad+\Edeg[(r^{-2}\De e_3,e_4,r\nab)^{\leq 2}\psi](\tau_1, \tau_2)+\int_{\DD(\tau_1,\tau_2)}r^{-6}\De^2|\Rhat^{\leq 1}N|^2.
\end{split}
\eeq
This finishes the proof of the bulk part in \eqref{eq:seconddegbulk}.
\end{proof}


\subsection{Proof of Lemma  \ref{Lemma:calculuslemma}}
\lab{appendix:Lemma-calculuslemma}

\begin{proof} We proceed in steps as follows.

{\bf Step 1.} We derive the control of the bulk term $B$ of at most two derivatives of $\psi$. Proceeding as in the proof of estimate \eqref{eq:secondorderbulktaunabdeg}, we deduce
\beq\lab{eq:secondorderbulktaunab}
B[(\pr_\tau, r\nab)^{\leq 2}\psi](\tau_1,\tau_2)\les \BE[\OO\psi](\tau_1,\tau_2)+\BE[(T,Z)^{\leq 2}\psi](\tau_1,\tau_2).
\eeq
We apply partial
red shift estimate to $e_3\psi$ to obtain
\beq\lab{eq:redshift1}
\bsplit
&\int_{\DD_{< r_+(1+\de_{red})}(\tau_1, \tau_2)} |e_3 ^2\psi|^2+\sup_{\tau\in[\tau_1,\tau_2]}\int_{\Si (\tau){< r_+(1+\de_{red})} }|(e_3,\nab)e_3\psi|^2+\int_{\HH(\tau_1,\tau_2)}|\nab e_3\psi|^2\\
&\quad\leq B[(\pr_\tau,r\nab)^{\leq1}\psi](\tau_1, \tau_2)+E_{< r_+(1+2\de_{red})}^1[\psi](\tau_1)+\int_{\DD_{< r_+(1+2\de_{red})}(\tau_1, \tau_2)}|\dk^{\leq1} N|^2.
\end{split}
\eeq
Therefore
\beq\lab{eq:firstorderbulk}
\bsplit
&B[(e_3, e_4, r\nab)^{\leq 1}\psi](\tau_1,\tau_2)\\
&\quad\les B[(\pr_\tau,r\nab)^{\leq1}\psi](\tau_1, \tau_2)+\Bdeg[(r^{-2}\De e_3, e_4,r\nab)^{\leq1}\psi](\tau_1,\tau_2)\\
&\quad+\int_{\DD_{< r_+(1+\de_{red})}(\tau_1, \tau_2)} |e_3 ^2\psi|^2\\
&\quad \les \BE[\OO\psi](\tau_1,\tau_2)+\BE[(T,Z)^{\leq 2}\psi](\tau_1,\tau_2)+\Bdeg[(r^{-2}\De e_3, e_4,r\nab)^{\leq1}\psi](\tau_1,\tau_2)\\
&\quad+E_{< r_+(1+2\de_{red})}^1(\tau_1)+\int_{\DD_{< r_+(1+2\de_{red})}(\tau_1, \tau_2)}|\dk^{\leq1} N|^2.
\end{split}
\eeq
Apply partial
red shift estimate to $e_3(T,Z)\psi, e_3\nab\psi, e_3e_3\psi$ gives
\beq\lab{eq:redshift2}
\bsplit
&\int_{\DD_{< r_+(1+\de_{red})}(\tau_1, \tau_2)} |e_3 ^2(T,\nab,e_3)\psi|^2+\sup_{\tau\in[\tau_1,\tau_2]}\int_{\Si (\tau){< r_+(1+\de_{red})} }|(e_3,\nab)e_3(T,\nab,e_3)\psi|^2\\
&\quad+\int_{\HH(\tau_1,\tau_2)}|\nab e_3(T,\nab,e_3)\psi|^2\\
&\quad\leq B[(\pr_\tau,r\nab)^{\leq2}\psi](\tau_1, \tau_2)+E_{< r_+(1+2\de_{red})}^2(\tau_1)+\int_{\DD_{< r_+(1+2\de_{red})}(\tau_1, \tau_2)}|\dk^{\leq2} N|^2.
\end{split}
\eeq
Combining \eqref{eq:secondorderbulktaunab}, \eqref{eq:firstorderbulk} and \eqref{eq:redshift2}, we infer
\beq\lab{eq:secondorderbulk}
\bsplit
&B[(e_3, e_4, r\nab)^{\leq 2}\psi](\tau_1,\tau_2) \\
&\quad\les B[(e_3,e_4, r\nab)^{\leq1}\psi](\tau_1,\tau_2)+B[(\pr_\tau, r\nab)^{\leq 2}\psi](\tau)\\
&\quad+\int_{\DD_{< r_+(1+\de_{red})}(\tau_1, \tau_2)} |e_3 ^2(T,\nab,e_3)\psi|^2+\Bdeg[(r^{-2}\De e_3, e_4,r\nab)^{\leq2}\psi](\tau_1,\tau_2)\\
&\quad\les B[\OO\psi](\tau_1,\tau_2)+B[(T,Z)^{\leq 2}\psi](\tau_1,\tau_2)+\Bdeg[(r^{-2}\De e_3, e_4,r\nab)^{\leq1}\psi](\tau_1,\tau_2)\\
&\quad+E_{< r_+(1+2\de_{red})}^2(\tau_1)+\int_{\DD_{< r_+(1+2\de_{red})}(\tau_1, \tau_2)}|\dk^{\leq2} N|^2.
\end{split}
\eeq
Combining \eqref{eq:secondorderbulk} with the degenerate bulk estimate \eqref{eq:seconddegbulk}, we finish the proof of the bulk part in the statement.

{\bf Step 2.}
 Proceeding as in the proof of \eqref{eq:secondorderenergynabdeg}, 
we derive
\beq
\lab{eq:secondorderenergynab}
\bsplit
E[(r\nab)^{\leq 2}\psi](\tau)\les E[\OO\psi](\tau)+E[(T,Z)^{\leq 2}\psi](\tau).
\end{split}
\eeq
Combining \eqref{eq:secondorderenergynab} with partial redshift estimate \eqref{eq:redshift1}, we derive
\beq\lab{eq:firstorderenergy}
\bsplit
&E[(e_3, e_4, r\nab)^{\leq 1}\psi](\tau)\\
&\quad\les E[(\pr_\tau,r\nab)^{\leq1}\psi](\tau)+\Edeg[(r^{-2}\De e_3, e_4, r\nab)^{\leq 1}\psi](\tau)+ \int_{\Si (\tau){< r_+(1+\de_{red})} }|e_3^2\psi|^2\\
&\quad \les E[\OO\psi](\tau_1,\tau_2)+E[(T,Z)^{\leq 2}\psi](\tau_1,\tau_2)+\Edeg[(r^{-2}\De e_3, e_4, r\nab)^{\leq 1}\psi](\tau)\\
&\quad +B[(\pr_\tau,r\nab)^{\leq1}\psi](\tau_1, \tau_2)+E_{< r_+(1+2\de_{red})}^1(\tau_1)+\int_{\DD_{< r_+(1+2\de_{red})}(\tau_1, \tau_2)}|\dk^{\leq1} N|^2.
\end{split}
\eeq
Combining \eqref{eq:secondorderenergynab}, \eqref{eq:firstorderenergy} and \eqref{eq:redshift2}, we infer
\beq\lab{eq:secondorderenergy}
\bsplit
&E[(e_3, e_4, r\nab)^{\leq 2}\psi](\tau) \\
&\quad\les E[(e_3, e_4, r\nab)^{\leq1}\psi](\tau)+E[(r\nab)^{\leq 2}\psi](\tau)+E[TT\psi](\tau)\\
&\quad+\int_{\Si (\tau){< r_+(1+\de_{red})} }|(e_3,\nab)e_3(T,\nab,e_3)\psi|^2+\Edeg[(r^{-2}\De e_3, e_4, r\nab)^{\leq 2}\psi](\tau)\\
&\quad \les E[\OO\psi](\tau)+E[(T, Z)^{\leq 2}\psi](\tau)+\Edeg[(r^{-2}\De e_3, e_4, r\nab)^{\leq 1}\psi](\tau_1,\tau_2)\\
&\quad+B[(\pr_\tau,r\nab)^{\leq2}\psi](\tau_1, \tau_2)+E_{< r_+(1+2\de_{red})}^2[\psi](\tau_1)+\int_{\DD_{< r_+(1+2\de_{red})}(\tau_1, \tau_2)}|\dk^{\leq2} N|^2\end{split}
\eeq
Combining \eqref{eq:secondorderenergy}, the degenerate energy estimate \eqref{eq:seconddegenergy} and the estimate for the bulk term in step 1, we finish the proof of the energy part in the statement.

{\bf Step 3.} Finally, we derive the control of the flux $F$ of at most two derivatives of $\psi$. We proceed as in the proof of \eqref{eq:secondorderenergynabdeg} and obtain
\beq
F[(r\nab)^{\leq 2}\psi](\tau_1,\tau_2)\les F[\OO\psi](\tau_1,\tau_2)+F[(T,Z)^{\leq2}](\tau_1, \tau_2).
\eeq
and thus
\beq\lab{eq:secondorderfluxnabT}
F[(r\nab, T)^{\leq 2}\psi](\tau_1,\tau_2)\les F[\OO\psi](\tau_1,\tau_2)+F[(T,Z)^{\leq2}](\tau_1, \tau_2).
\eeq
Using \eqref{eq:square-N}, we derive 
\beq\lab{eq:firstorderfluxe3}
\bsplit
\int_{\HH(\tau_1,\tau_2)}|e_4e_3\psi|^2&\les F[\OO\psi](\tau_1,\tau_2)+F[(T,Z)^{\leq2}](\tau_1, \tau_2)+\int_{\HH(\tau_1,\tau_2)}|e_3\psi|^2\\
&\les F[\OO\psi](\tau_1,\tau_2)+F[(T,Z)^{\leq2}](\tau_1, \tau_2)+B[(e_3,e_4,r\nab)^{\leq1}](\tau_1,\tau_2).
\end{split}
\eeq
and
\beq\lab{eq:firstorderfluxe3e3}
\bsplit
&\int_{\HH(\tau_1,\tau_2)}|e_4e_3(e_3,T,Z)\psi|^2\\
&\quad\les F[\OO\psi](\tau_1,\tau_2)+F[(T,Z)^{\leq2}](\tau_1, \tau_2)+\int_{\HH(\tau_1,\tau_2)}|\dk^2\psi|^2\\
&\quad\les F[\OO\psi](\tau_1,\tau_2)+F[(T,Z)^{\leq2}](\tau_1, \tau_2)+B[(e_3,e_4,r\nab)^{\leq2}](\tau_1,\tau_2).
\end{split}
\eeq
Applying partial red shift estimates \eqref{eq:redshift1} and \eqref{eq:redshift2}, we obtain
\beq
\lab{eq:redshiftforflux}
\bsplit
    &\int_{\HH(\tau_1,\tau_2)}|\nab e_3(T,\nab,e_3)^{\leq1}\psi|^2\\
        &\quad\les B[(\pr_\tau,r\nab)^{\leq2}\psi](\tau_1, \tau_2)+E_{< r_+(1+2\de_{red})}^2[\psi](\tau_1)+\int_{\DD_{< r_+(1+2\de_{red})}(\tau_1, \tau_2)}|\dk^{\leq2} N|^2.
    \end{split}
\eeq
By combining \eqref{eq:secondorderfluxnabT}, \eqref{eq:firstorderfluxe3}, \eqref{eq:firstorderfluxe3e3} and \eqref{eq:redshiftforflux}
we derive
\beq\lab{eq:flux2order}
\bsplit
&F[(e_3,e_4,r\nab)^{\leq 2}\psi](\tau_1,\tau_2)\\
&\quad\les F[(r\nab,T)^{\leq 2}\psi](\tau_1,\tau_2)+\int_{\HH(\tau_1,\tau_2)}|e_4e_3(e_3,T,Z)^{\leq1}\psi|^2+\int_{\HH(\tau_1,\tau_2)}|\nab e_3(T,\nab,e_3)^{\leq1}\psi|^2\\
&\quad\les F[\OO\psi](\tau_1,\tau_2)+F[(T,Z)^{\le 2}\psi](\tau_1,\tau_2)+B[(e_3,e_4,r\nab)^{\leq2}\psi](\tau_1,\tau_2)\\
&\quad+E_{< r_+(1+2\de_{red})}^2[\psi](\tau_1)+\int_{\DD_{< r_+(1+2\de_{red})}(\tau_1, \tau_2)}|\dk^{\leq2} N|^2.
\end{split}
\eeq
Combining \eqref{eq:flux2order} with the estimate for the bulk part in step 1 finishes the proof of the flux part in the statement.
\end{proof}

\subsection{Proof of Lemma \ref{lem:higherordercalculuslemma}}
\lab{appendix:higherordercalculuslemma}
\begin{proof}
We prove below the case $s=1$ and note that the cases $s>1$ can be handled in a completely analogous manner. We proceed in steps as follows.

{\bf Step 1.} We derive the control of the bulk term $B$. Using \eqref{eq:square-N}, the calculation in \eqref{eq:lapSSsquaredeg} and applying integration by parts, we derive
\beq\lab{eq:rrnabnab}
\bsplit
&\int_{\DD(\tau_1,\tau_2)}r^{-2}|\pr_\rt(r^{-2}\De e_3,\nab)^{\leq2}\psi|^2+r^{-4}|(r^{-2}\De e_3,\nab)^{\leq2}\psi|^2\\
&\quad+\int_{\DD_{\geq 10m}(\tau_1,\tau_2)}r^{-2}|T(r^{-2}\De e_3,\nab)^{\leq2}\psi|^2+r^{-1}|\nab(r^{-2}\De e_3,\nab)^{\leq2}\psi|^2\\
&\quad\les B[(e_3,e_4, \nab)^{\leq1}(T,Z)^{\leq1}\psi]+\NN[(e_3, e_4, \nab)^{\leq1}N](\tau_1,\tau_2)
\end{split}
\eeq
and
\beq
\bsplit
&\int_{\DD_{\ntrap}(\tau_1,\tau_2)\cap\{r\leq 10m\}}|T\nab^{\leq2}\psi|^2+|\nab\nab^{\leq2}\psi|^2\\
&\quad\les  B[(e_3,e_4, \nab)^{\leq1}(T,Z)^{\leq1}\psi]+\NN[(e_3, e_4, \nab)^{\leq1}N](\tau_1,\tau_2).
\end{split}
\eeq
Therefore
\beq\lab{eq:degbulkstos+1}
B[(r^{-2}\De e_3,e_4, \nab)^{\leq2}\psi]\les  B[(e_3,e_4, \nab)^{\leq1}(T,Z)^{\leq1}\psi]+\NN[(e_3, e_4, \nab)^{\leq1}N](\tau_1,\tau_2).
\eeq
Combining \eqref{eq:degbulkstos+1} with \eqref{eq:redshift2}, we obtain
\beq
\bsplit
&B[(e_3,e_4,\nab)^{\leq 2}\psi](\tau_1,\tau_2)\\
&\quad\les B[(r^{-2}\De e_3,e_4, \nab)^{\leq2}\psi]+\int_{\DD_{< r_+(1+\de_{red})}(\tau_1, \tau_2)} |e_3 ^2(T,\nab,e_3)\psi|^2\\
&\quad\les B[(e_3,e_4, \nab)^{\leq1}(T,Z)^{\leq1}\psi]+E[(e_3, e_4,\nab)^{\leq2}\psi](\tau_1)+\NN[(e_3, e_4, \nab)^{\leq2}N](\tau_1,\tau_2)
.
\end{split}
\eeq
This finishes the proof of the bulk part in the statement.

{\bf Step 2.} Proceeding as in the proof of \eqref{eq:rrnabnab}, we derive
\beq\lab{eq:degenergystos+1}
\bsplit
&E[(r^{-2}\De e_3,e_4, \nab)^{\leq2}\psi](\tau)\\
&\quad\les E[(e_3,e_4, \nab)^{\leq1}(T,Z)^{\leq1}\psi](\tau)+\int_{\Si(\tau)}|(e_3, e_4, \nab)^{\leq1}N|^2\\
&\quad\les E[(e_3,e_4, \nab)^{\leq1}(T,Z)^{\leq1}\psi](\tau)+\NN[(e_3,e_4,\nab)^{\leq2}N](\tau_1,\tau_2).
\end{split}
\eeq
Combining \eqref{eq:degenergystos+1}, \eqref{eq:redshift2} with \eqref{eq:degbulkstos+1}, we obtain
\beq
\bsplit
&E[(e_3,e_4,\nab)^{\leq 2}\psi](\tau_1,\tau_2)\\
&\quad\les E[(r^{-2}\De e_3,e_4, \nab)^{\leq2}\psi]+\sup_{\tau\in[\tau_1,\tau_2]}\int_{\Si (\tau){< r_+(1+\de_{red})} }|(e_3,\nab)e_3(T,\nab,e_3)\psi|^2\\
&\quad\les \BE[(e_3,e_4, \nab)^{\leq1}(T,Z)^{\leq1}\psi]+E[(e_3,e_4,\nab)^{\leq2}\psi](\tau_1)+\NN[(e_3,e_4,\nab)^{\leq2}N](\tau_1,\tau_2).
\end{split}
\eeq
This finishes the proof of the energy part in the statement.

{\bf Step 3.} Using \eqref{eq:firstorderfluxe3e3}, we obtain
\beq\lab{eq:degfluxstos+1}
\bsplit
\int_{\HH(\tau_1,\tau_2)}|e_4e_3e_3\psi|^2&\les \int_{\HH(\tau_1,\tau_2)}|e_3\nab^2\psi|^2+F[(e_3, e_4,\nab)^{\leq 1}(T,Z)^{\leq1}\psi](\tau_1,\tau_2)\\
&+B[(e_3,e_4,\nab)^{\leq 2}\psi](\tau_1,\tau_2).
\end{split}
\eeq
Combining \eqref{eq:degfluxstos+1}, \eqref{eq:redshift2} and the control of the bulk term in step 1, we conclude
\beq
\bsplit
&F[(e_3,e_4,\nab)^{\leq 2}\psi](\tau_1,\tau_2)\\
&\quad\les \int_{\HH(\tau_1,\tau_2)}|e_4e_3e_3\psi|^2+\int_{\HH(\tau_1,\tau_2)}|\nab e_3(T,\nab,e_3)\psi|^2\\
&\quad+F[(e_3, e_4,\nab)^{\leq 1}(T,Z)^{\leq1}\psi](\tau_1,\tau_2)\\
&\quad\les \BEF[(e_3,e_4, \nab)^{\leq1}(T,Z)^{\leq1}\psi]+E[(e_3,e_4,\nab)^{\leq2}\psi](\tau_1)+\NN[(e_3,e_4,\nab)^{\leq2}N](\tau_1,\tau_2).
\end{split}
\eeq
This finishes the proof of the flux part in the statement.
\end{proof}
\subsection{Proof of Lemma   \ref{lemma:Energy1}}
\lab{appendix:identities-QQ}
We write a more complete version of Lemma \ref{lemma:Energy1}
\begin{lemma}
    \lab{lemma:Energy1-App}
    The following statements hold true.
    \begin{itemize}
        \item Along the horizon $r=r_{+}$, the normal $N_{\HH}=-\Rt=\frac 1 2 e_4= \frac{r^2+a^2}{|q|^2}T_+$ and
        \beq
        \lab{eq:Festimates-App}
        \bsplit
        \QQ(T, N_\HH)&=   \frac{r^2+a^2}{|q|^2}T\psi \c  T_+\psi.
        \\
        \QQ(T_+, N_\HH)&=  \frac{r^2+a^2}{|q|^2}|T_+\psi|^2.
    \end{split}
    \eeq

    \item Along $\Si(\tau)$, for $r\ge r_+$ where $N_\Si=c_{\Si} \Nt$ (with $\Nt=- \gradtau $),
    we have, for some small $ c_0>0 $ and a large $C>0$, relative to the standard frame.
    \beq
    \lab{eq:EFestimates0-App}
\QQ(\Tt,  N_\Si)\ge  c_0\big(  | e_4  \psi|^2 +\frac{|\De|}{r^2}|e_3\psi|^2 +|\nab\psi|^2 \big).
\eeq
and
\beq
\lab{eq:EFestimates1-App}
\QQ(\T, N_\Si)\ge c_0\Big( | e_4\psi|^2+\frac{|\De|}{r^2} |e_3 \psi|^2 + |\nab\psi|^2\Big) - C\frac{a^2m^2 }{r^6 }  |\Z\psi|^2
\eeq

\end{itemize}
\end{lemma}
\begin{proof}
The identities in \eqref{eq:Festimates-App} are straightforward. To check \eqref{eq:EFestimates0-App},
recall, see \eqref{eq:gtadt-gradr},
\begin{align*}
\Rt&=- \g^{\mu\nu}\pr_\mu  r  \pr_\nu  =\frac 12  e_3(r ) e_4 +\frac 12 e_4(r)  e_3=-\frac 1 2 \big( e_4 - \frac{\De}{|q|^2} e_ 3 \big),\\
\Tt&=-\frac{|q|^2\De}{\Si^2}\gradt= T+\frac{2amr}{\Si^2} Z.
\end{align*}

\textit{$1^{st}$ lower bound in
\eqref{eq:EFestimates0-App}.}
In view of Lemma \eqref{lemma:TtRtandZ}, using also the identity\footnote{ Indeed $
\g(\Tt, \widetilde{N}) =\g\big(\Tt, \frac{\Si^2}{|q|^2 \De}\Tt+ f'(r)\Rt\big)=\frac{\Si^2}{|q|^2 \De}\g(\Tt, \Tt)= -\frac{\Si^2}{|q|^2 \De} \frac{\Delta|q|^2}{\Sigma^2}=-1$.} $\g(\Tt, \widetilde{N})=-1$ we deduce
\begin{align*}
\QQ[\psi](\Tt, \widetilde{N})&= \Tt \psi \widetilde{N}(\psi)- \frac 12 \g(\Tt, \widetilde{N}) \g^{\mu\nu} \pr_\mu \psi\pr_\nu \psi=  \Tt \psi \widetilde{N}(\psi)+ \frac 12 \g^{\mu\nu} \pr_\mu \psi\pr_\nu \psi\\
&= \Tt \psi \big(\frac{\Si^2}{|q|^2 \De}\Tt\psi+f'(r)\Rt\psi)\\
&\quad +\frac 12 \Big(-\frac{\Sigma^2}{\Delta |q|^2}|\widetilde{\T}\psi|^2+\frac{|q|^2}{\Delta}|\widetilde{R}\psi|^2+|e_1\psi|^2+\frac{|q|^2}{\Sigma^2\sin^2\th}|Z\psi|^2\Big)\\
&= \frac 12 \frac{\Si^2}{|q|^2 \De} I + \frac 12 \big(|e_1\psi|^2+\frac{|q|^2}{\Sigma^2\sin^2\th}|Z\psi|^2\big)
\end{align*}
where
\[
I  = |\Tt\psi|^2 + \frac{|q|^4}{\Si^2 } |\Rt\psi|^2 +2 \frac{|q|^2 \De}{\Si^2}  f'(r) \Tt \psi \Rt \psi.
\]
We rewrite $I$ in the form
\begin{align*}
I&=\big|\widetilde{T}\psi+ f'(r) \frac{|q|^2\De}{\Si^2}  \widetilde{R}\psi\big|^2+| \widetilde{R}\psi|^2\big(  - f'(r)^2  \frac{|q|^4\De^2}{\Si^4} + \frac{|q|^4}{\Si^2 } \big)\\
&=\big|\widetilde{T}\psi+f'(r) \frac{|q|^2\De}{\Si^2}  \widetilde{R}\psi\big|^2+| \widetilde{R}\psi|^2\frac{|q|^4}{\Si^4} \big(\Si^2- f'(r)^2 \De^2\big)\\
&=\big|\frac{\De |q|^2}{\Si^2}(\gradtau)\psi\big|^2+| \widetilde{R}\psi|^2\frac{|q|^4}{\Si^4} \big(\Si^2- f'(r)^2 \De^2\big)
\end{align*}
where we use \ref{eq:N_Si} in the last step. On the other hand, see \eqref{eq:g(N_Si,N_Si)},
\[
-1\sim  \g(\widetilde{N}, \widetilde{N})=-\frac{\Si^2}{|q|^2 \De}+\frac{\De}{|q|^2} f'(r) f'(r)= \frac{1}{|q|^2 \De}\big(-\Si^2+ f'(r)^2 \De^2\big)
\]
Hence, for some $c_0>0$
\begin{align*}
I&=\big|\frac{\De |q|^2}{\Si^2}(\gradtau)\psi\big|^2-\frac{\De|q|^2}{\Si^2}\g(\widetilde{N}, \widetilde{N})|\Rt\psi|^2 \ge  c_0
\frac{\De|q|^2}{\Si^2}\big(\frac{\De|q|^2}{\Si^2}\big|(\gradtau)\psi\big|^2+|\Rt \psi|^2 \big)
\end{align*}
Therefore, with some other $c_0>0$, $r\ge r_+$,
\[
\QQ(\Tt, \widetilde{N})\ge c_0\Big(\frac{|\De||q|^2}{\Si^2}\big|(\gradtau)\psi\big|^2 +|\Rt\psi|^2 +|e_1\psi|^2 +\frac{1}{r^2\sin^2 \th} |\Z\psi|^2\Big).
\]
The desired estimate follows in view of the fact that $N_\Si\sim \Nt$. In view of Lemma \ref{lemma:equivalence} below we
can write, for a somewhat smaller $c_0>0$,
\[
\QQ(\Tt,  N_\Si)\ge  c_0\big(  | e_4  \psi|^2 +\frac{|\De|}{r^2}|e_3\psi|^2 +|\nab\psi|^2 \big).
\]

\noindent \textit{$2^{nd}$ lower bound in
\eqref{eq:EFestimates0-App}.}

We use the relation $ \T=\Tt-\frac{2a mr }{\Si^2} \Z$.
Since $\g(\Z, N_\Si)=0$
\beaa
\QQ(\T, N_\Si)= \QQ(\Tt, N_\Si)- \frac{2a mr }{\Si^2} \QQ(\Z, N_\Si)=  \QQ(\Tt, N_\Si) -\frac{2a mr }{\Si^2} \Z\psi \c N_\Si \psi.
\eeaa
and, using \eqref{eq:TZ-frame},
\beaa
\frac{2a mr }{\Si^2}  \Z\psi \c N_\Si \psi&\les&  \de    \big(  | e_4  \psi|^2 +\frac{|\De|}{ r^2}  |e_3\psi|^2 +|\nab\psi|^2 \big) + O(\de^{-1})\frac{a^2 m^2}{r^6} |Z\psi|^2
\eeaa
Therefore, for $r\ge r_+$, with a slightly smaller $c_0$ and a large $C=C(\de)>0$,
\[
\QQ(\T, N_\Si)\ge  c_0\Big( | e_4  \psi|^2 +\frac{|\De|}{r^2}|e_3\psi|^2 +|\nab\psi|^2\Big)- C\frac{a^2m^2 }{r^6 }  |\Z\psi|^2.
\]
\end{proof}

\begin{lemma}
\lab{lemma:equivalence}
We have the following inequality:
\beq
| e_4  \psi|^2 +\frac{|\De|}{r^2}|e_3\psi|^2 +|\nab\psi|^2 \les \frac{|\De|}{r^2}|(\gradtau)\psi|^2+  |\Rt \psi|^2 +|e_1\psi|^2  +\frac{1}{r^2\sin^2\th}|Z\psi|^2 .
\eeq
\end{lemma}

\begin{proof}
Observe that $\frac{f'(r)\De}{r^2+a^2}-1=-m^2r^{-2}\De$ near $r=r_+$ and
\begin{align*}
e_4  \psi&=-\frac{\De}{(r^2+a^2)} (\gradtau) \psi  -(\frac{f'(r)\De}{r^2+a^2}+1)\Rt \psi + \frac{a\sin\th \De}{(r^2+a^2)|q|}  e_2\psi,\\
e_3 \psi &=-\frac{|q|^2}{(r^2+a^2)} (\gradtau) \psi  -\frac{|q|^2}{\De}(\frac{f'(r)\De}{r^2+a^2}-1)\Rt \psi + \frac{a\sin\th |q|}{(r^2+a^2)}  e_2\psi,\\
e_2 \psi&= \frac{a\sin\th}{|q|}\T\psi +\frac{1}{|q| \sin\th } \Z\psi =\frac{a\sin\th}{|q|} \big(\Tt -\frac{2amr}{\Si^2} \Z\big)\psi  +\frac{1}{|q| \sin\th } \Z\psi\\
&=-\frac{a\sin\th |q|\De}{\Si^2}(\gradtau) \psi -\frac{a\sin\th |q|\De f'(r)}{\Si^2}(\gradtau) \psi +\big(\frac{1}{|q| \sin\th } - \frac{2a^2mr\sin\th }{|q|\Si^2} \big)\Z\psi.
\end{align*}
Therefore,
\begin{align*}
|e_2\psi|^2&\les \frac{a^2\De^2}{r^6}|(\gradtau) \psi|^2 +\frac{1}{r^2\sin^2\th} |\Z\psi|^2, \\
| e_4  \psi|^2&\les\frac{\De^2}{r^4}|(\gradtau) \psi|^2 +|\Rt\psi|^2 +  a^2 \frac{\De^2}{r^3} |e_2\psi|^2,\\
|e_3\psi|^2 &\les |(\gradtau)\psi|^2 +|\Rt\psi|^2 +  a^2 \frac{\De^2}{r^3} |e_2\psi|^2.
\end{align*}
Hence
\[
| e_4  \psi|^2 +\frac{|\De|}{r^2}|e_3\psi|^2 +|\nab\psi|^2 \les \frac{|\De|}{r^2}|(\gradtau)\psi|^2+  |\Rt \psi|^2 +|e_1\psi|^2   +\frac{1}{r^2\sin^2\th}|Z\psi|^2
\]
as stated.
\end{proof}


\section{Proof of results from section \ref{section:geodesics-Kerr}}


\subsection{Proof of Lemma \ref{Lemma:rangeofrfortrappednullgeodesics}}

\lab{sec:rangeofrfortrappednullgeodesics}
We have to check that the values
\begin{align*}
\rhat_1 &:=2m\left(1+\cos\left(\frac{2}{3}\arccos\left(-\frac{|a|}{m}\right)\right)\right),\\
\rhat_2 &:=2m\left(1+\cos\left(\frac{2}{3}\arccos\left(\frac{|a|}{m}\right)\right)\right),
\end{align*}
verify the equation $\Ththree(r)=0$ or, equivalently (since $r>r_+$),
$\Th(r)=0$.

{\bf The case of $\rhat_2$.}
Let $x=\frac{|a|}{m}$ and assume $a>0$. Set $y= \cos \big(\frac{1}{3}\arccos x) $.
To calculate $\cos\big(\frac{2}{3}\arccos x\big) $
we make use of $\cos 2\th= 2 \cos^2 \th-1$ to deduce,
\[
\cos\big(\frac{2}{3}\arccos x) =2 \cos^2 \big(\frac{1}{3}\arccos x) -1 = 2y^2-1.
\]
Since $\cos(3 \th)= 4 \cos^3 \th - 3 \cos \th$
we deduce
\[
x=\cos \big(3  (\frac{1}{3}\arccos x)\big)= 4 \cos^3   (\frac{1}{3}\arccos x) -3 \cos  (\frac{1}{3}\arccos x) = 4 y^3 - 3 y.
\]
Therefore,
\[
x= 4 y^3 - 3 y, \quad \cos\big(\frac{2}{3}\arccos x\big) = 2y^2-1.
\]
i.e.,
\[
\frac am = 4 y^3 - 3 y,\quad  \cos\big(\frac{2}{3}\arccos \frac am \big)= 2y^2-1.
\]
We deduce
\[
\rhat_2=2m\left(1+\cos\left(\frac{2}{3}\arccos\left(\frac am \right)\right)\right)= 2m( 1+2y^2-1)=  4m y^2.
\]
Finally
\[
\rhat_2^3-6m\rhat_2^2+9m^2\rhat_2-4ma^2=(4my^2)^3- 6m( 4my^2)^2 +9m^2(4 m y^2)- 4ma^2:=E(y)
\]
We write
\[
E(y)
=64m^3(y^3)^2- 96 m^3 y^3 y+ 36 m^3 y^2- 4ma^2.
\]
Since $4 y^3 - 3 y=\frac a m$, i.e. $y^{3}=\frac{1}{4}\big(3 y +\frac am \big)$, we deduce
\begin{align*}
E(y)&= 64 m^3( y^3)^2 - 24 m^3 y^3 y + 36 m^3 y^2- 4ma^2\\
&=4 m^3\big( 3  y +\frac am \big)^2 -24 m^3\big( 3  y +\frac am \big)y + 36 m^3 y^2- 4ma^2\\
&=4 m^3\big( 9 y^2 + 6\frac am  y+\frac{a^2}{m^2}\big) - 72 m^3 y^2 - 24 am^2 y + 36 m^3 y^2- 4ma^2\\
&=36 m^3 y^2 + 24am^2 y - 72 m^3 y^2 - 24 am^2 y + 36 m^3 y^2\\
&=0.
\end{align*}
Thus $\rhat_2$ verifies our equation.

{\bf The case of $\rhat_1$.} Set $z= \cos \big(\frac{1}{3}\arccos (-\frac a m) \big)$.
Then, exactly as before,
\[
-\frac am = 4 z^3 - 3 z,\quad  \cos\big(\frac{2}{3}\arccos(- \frac am)\big)= 2z^2-1.
\]
Hence
\[
\rhat_1 = 2m\left(1+\cos\left(\frac{2}{3}\arccos\left(-\frac am \right)\right)\right) = 4mz^2
\]
verifies the same equation.

\subsection{Proof of Lemma \ref{Lemma:Deriv.yhat}}
\lab{sec:Deriv.yhat}
It suffices to check the identity for $\yhat_2=2\left(1+ \cos\left(\frac{2}{3}\arccos\left(\mu^{1/2}\right)\right)\right)$, the one for $\yhat_1$ is identical.
\beq
\lab{eq:1^stbis.Deriv-yhat}
\bsplit
\frac{d}{d\mu}\yhat_2&= -2\sin\left(\frac{2}{3}\arccos\left(\mu^{1/2}\right)\right)\frac 2 3  \frac{d}{d\mu}\arccos (\mu^{1/2})\\
&=-2\sin\left(\frac{2}{3}\arccos\left(\mu^{1/2}\right)\right)\Big(-\frac 2 3 \frac{1}{\sqrt{1-\mu}}\frac 12 \mu^{-1/2}\Big)\\
&=\frac{2}{3} \frac{1}{\sqrt{(1-\mu)\mu}}\sin\left(\frac{2}{3}\arccos\left(\mu^{1/2}\right)\right).
\end{split}
\eeq
Continuing
\beaa
\frac{d}{d\mu}\yhat_2  &=&\frac{2}{3} \frac{1}{\sqrt{(1-\mu)\mu}}\sqrt{1-\cos^2\left(\frac{2}{3}\arccos\left(\mu^{1/2}\right)\right)}\\
&=&\frac 2 3  \frac{1}{\sqrt{\mu(1-\mu)}}  \sqrt{ 1-\big(\frac{\yhat_2}{2} -1\big)^2}\\
&=& \frac 1 3   \frac{1}{\sqrt{\mu(1-\mu)}}\sqrt{\yhat_2(4-\yhat_2)}
\eeaa
as stated. To calculate the second derivative of $\yhat_2$ we start from \eqref{eq:1^stbis.Deriv-yhat}
from which we derive
\beaa
\frac{d}{d\mu}\Big(\sqrt{\mu(1-\mu)} \yhat_2' \Big)&=& \frac{2}{3}  \frac{d}{d\mu} \sin\left(\frac{2}{3}\arccos\left(\mu^{1/2}\right)\right)\\
&=&\cos \left(\frac{2}{3}\arccos\left(\mu^{1/2}\right)\right) \Big(-\frac 2 3 \frac{1}{\sqrt{\mu(1-\mu) }}\Big)=- \frac 2 3 \frac{1}{\sqrt{\mu(1-\mu) }} \big(\frac{\yhat_2}{2} -1\big)\\
&=&-\frac{1}{3}  \frac{1}{\sqrt{\mu(1-\mu) }} \big(\yhat_2-2\big)
\eeaa
as stated. The identity for $\yhat_1$ follows in the same manner.

\section{Proof of Proposition \ref{Prop:refinedP}}
\lab{sec:refinedP}
Using \eqref{eq:DefPsi-P}, we write
\begin{align*}
\widetilde{\VV}^2T\Zhat\psi\c\psi_{\LL}&=\frac{r^3\widetilde{\VV}^2}{ 4a(1-\de_0m^2\frac{\TT_{-a}}{2r^5})}\Psi_1\c\psi_{\LL}+\de_0m^2\frac{r^2-3mr+2a^2}{4ar^2}\frac{\widetilde{\VV}^2}{1-\de_0m^2\frac{\TT_{-a}}{2r^5}}TT\psi\c\psi_{\LL}\\
&-\frac{a}{r^2}\frac{1-\de_0m^2\frac{\TT_{-a}}{4r^5}}{1-\de_0m^2\frac{\TT_{-a}}{2r^5}}\widetilde{\VV}^2\Zhat\Zhat\psi\c\psi_{\LL}+\frac{r^2-3mr+2a^2}{2ar^2}\frac{1-\frac12\de_0m^2z}{1-\de_0m^2\frac{\TT_{-a}}{2r^5}}\widetilde{\VV}^2\OO\psi\c\psi_{\LL}.
\end{align*}
We further compute (modulo acceptable spacetime divergence)
\begin{align*}
\frac{r^3\widetilde{\VV}^2}{ 4a(1-\de_0m^2\frac{\TT_{-a}}{2r^5})}\Psi_1\c\psi_{\LL}&=\frac{r^3\widetilde{\VV}^2}{ 4a(1-\de_0m^2\frac{\TT_{-a}}{2r^5})}\Psi_1\c(m^2TT\psi+\Zhat\Zhat\psi)\\
&=\frac{r^3\widetilde{\VV}^2}{ 4a(1-\de_0m^2\frac{\TT_{-a}}{2r^5})}(mT\Psi_1\c mT\psi+\Zhat\Psi_1\c\Zhat\psi)\\
&\geq -\frac12h(|mT\Psi_1|^2+|\Zhat\Psi_1|^2)-\frac{(r^3\widetilde{\VV}^2)^2}{32a^2(1-\de_0m^2\frac{\TT_{-a}}{2r^5})^2h}(|mT\psi|^2+|\Zhat\psi|^2)\\
&\geq -\Phat-\frac{(r^3\widetilde{\VV}^2)^2}{32a^2(1-\de_0m^2\frac{\TT_{-a}}{2r^5})^2h}(|mT\psi|^2+|\Zhat\psi|^2).
\end{align*}
Using $\int_{\DD(\tau_1,\tau_2)}\frac{f(r)}{|q|^2}(|\pr_\th\psi|^2+\frac{1}{\sin^2\th}|Z\psi|)\geq \int_{\DD(\tau_1,\tau_2)} \frac{2f(r)}{|q|^2}|\psi|
$ and $\pr_\th=|q|e_1,\ \frac{Z}{\sin\th}=|q|e_2-a\sin\th T$, we further derive (modulo acceptable boundary terms)\begin{align*}
&\int_{\DD(\tau_1,\tau_2)}\frac{1}{|q|^2}\frac{(r^3\widetilde{\VV}^2)^2}{32a^2(1-\de_0m^2\frac{\TT_{-a}}{2r^5})^2h}|mT\psi|^2\\
&\quad\leq \frac12\int_{\DD(\tau_1,\tau_2)}\frac{1}{|q|^2}\frac{(r^3\widetilde{\VV}^2)^2}{32a^2(1-\de_0m^2\frac{\TT_{-a}}{2r^5})^2h}\Big(\big||q|e_1mT\psi\big|^2+\big|(|q|e_2-a\sin\th T)mT\psi\big|^2\Big)\\
&\quad \leq\int_{\DD(\tau_1,\tau_2)}\frac{1}{|q|^2}\frac{(r^3\widetilde{\VV}^2)^2}{32a^2(1-\de_0m^2\frac{\TT_{-a}}{2r^5})^2h}\Big(O^{\a\b}\D_\a mT\psi\D_\b mT\psi+a^2TT\psi\c m^2TT\psi\Big)\\
&\quad=\int_{\DD(\tau_1,\tau_2)}\frac{1}{|q|^2}\frac{(r^3\widetilde{\VV}^2)^2}{32a^2(1-\de_0m^2\frac{\TT_{-a}}{2r^5})^2h}\Big(a^2TT\psi+\OO\psi\Big)\c m^2TT\psi.
\end{align*}
Similarly, we have
\begin{align*}
&\int_{\DD(\tau_1,\tau_2)}\frac{1}{|q|^2}\frac{(r^3\widetilde{\VV}^2)^2}{32a^2(1-\de_0m^2\frac{\TT_{-a}}{2r^5})^2h}|\Zhat\psi|^2\\
&\leq\int_{\DD(\tau_1,\tau_2)}\frac{1}{|q|^2}\frac{(r^3\widetilde{\VV}^2)^2}{32a^2(1-\de_0m^2\frac{\TT_{-a}}{2r^5})^2h}\Big(a^2TT\psi+\OO\psi\Big)\c \Zhat\Zhat\psi.
\end{align*}
This finishes the proof of \eqref{eq:TZhatPsi}.

It remains to prove \eqref{eq:Psi1larger}. Recall that
\begin{align*}
\Psi_1&=-\de_0 m^2\frac{\TT_{-a}}{r^6}TT\psi+\frac{4a}{r^3}(1-\de_0m^2\frac{\TT_{-a}}{2r^5})T\Zhat\psi+\frac{4a^2}{r^5}(1-\frac12\de_0m^2\frac{\TT_{-a}}{4r^5})\Zhat\Zhat\psi\\
&-(1-\frac12\de_0 m^2z)\frac{2\TT_{-a}}{r^6}\OO\psi
\end{align*}
and we rewrite
\begin{align*}
\Psi_1&=-(\de_0-\frac{5ar^3}{m\TT_{-a}(1-\frac{5am}{r^2})})\frac{\TT_{-a}}{r^6}m^2TT\psi-\Big((1-\frac12\de_0m^2z)\frac{\TT_{-a}}{r^6}-\frac{2a}{5mr^3}\Big)\Zhat\Zhat\psi\\
&-\Big((1-\frac12\de_0m^2z)\frac{\TT_{-a}}{r^6}-\frac{2a}{5mr^3}\Big)\OO\psi+\Psi_3
\end{align*}
where
\begin{align*}
\Psi_3&=\Psi_3^\aund\psia=-\frac{5am}{r^3(1-\frac{5am}{r^2})}TT\psi+\frac{4a}{r^3}(1-\de_0m^2\frac{\TT_{-a}}{2r^5})T\Zhat\psi-\Big((1-\frac12\de_0m^2z)\frac{\TT_{-a}}{r^6}+\frac{2a}{5mr^3}\Big)\OO\psi\\
&+\Big(\frac{4a^2}{r^5}(1-\de_0m^2\frac{\TT_{-a}}{4r^5})+\big((1-\frac12\de_0m^2z)\frac{\TT_{-a}}{r^6}-\frac{2a}{5mr^3}\big)\Big)\Zhat\Zhat\psi.
\end{align*}
Then we compute
\beq\lab{eq:Psi1^2}
\bsplit
|\Psi_1|^2&=\Romanupper{1}+\Romanupper{2}+ (\de_0-\frac{5ar^3}{m\TT_{-a}(1-\frac{5am}{r^2})})^2\frac{m^2\TT^2_{-a}}{r^{12}}TT\psi\c m^2TT\psi\\
&+2(\de_0-\frac{4ar^3}{m\TT_{-a}(1-\frac{5am}{r^2})})\Big(1-\frac12\de_0m^2z-\frac{2ar^3}{5m\TT_{-a}}\Big)\frac{m^2\TT^2_{-a}}{r^{12}}TT\psi\c\Zhat\Zhat\psi\\
&+2(\de_0-\frac{4ar^3}{m\TT_{-a}(1-\frac{5am}{r^2})})\Big(1-\frac12\de_0m^2z-\frac{2ar^3}{5m\TT_{-a}}\Big)\frac{\TT^2_{-a}}{r^{12}}\OO\psi\c m^2TT\psi
\end{split}
\eeq
where
\begin{align*}
\Romanupper{1}&=2\Big((\de_0-\frac{5ar^3}{m\TT_{-a}(1-\frac{5am}{r^2})})\frac{\TT_{-a}}{r^6}m^2TT\psi+\big(1-\frac12\de_0m^2z-\frac{2ar^3}{5m\TT_{-a}}\big)\frac{\TT_{-a}}{r^6}\Zhat\Zhat\psi\Big)\c(-\Psi_3),\\
\Romanupper{2}&=\big(1-\frac12\de_0m^2z-\frac{2ar^3}{5m\TT_{-a}}\big)^2\frac{\TT_{-a}^2}{r^{12}}\big(\Zhat\Zhat\psi+\OO\psi\big)^2-2\big(1-\frac12\de_0m^2z-\frac{2ar^3}{5m\TT_{-a}}\big)\frac{\TT_{-a}}{r^6}\OO\psi\c\Psi_3+|\Psi_3|^2.
\end{align*}
{\bf $\Romanupper{1}$ is nonnegative.} Since (with $\de_0=10$) for $|a|/m\leq 0.75$ and $r\geq r_++3.5m>5m$
\begin{align*}
&\de_0-\frac{5ar^3}{m\TT_{-a}(1-\frac{5am}{r^2})}\geq10-\frac{5a}{m(1-\frac{3m}{r}+\frac{2a^2}{r^2})(1-\frac{a}{5m})}>10-\frac{5a}{m(\frac25+\frac{2a^2}{25m^2})(1-\frac{a}{5m})}>0,\\
&1-\frac12\de_0m^2z-\frac{2ar^3}{5m\TT_{-a}}>0,
\end{align*}
according to Lemmas \ref{lemma:IntegrationbypartsI} and \ref{Lem:sumofsquares}, it suffices to prove that $-\Psi_3^\aund S_\aund^{\mu\nu}$ is nonnegative. We write by using $|q|e_2=\frac{\Zhat}{\sin\th}-\frac{a\cos^2\th}{\sin\th}T$
\begin{align*}
-\Psi_3^\aund S_\aund^{\mu\nu}=\begin{pmatrix}
F_1
&F_2&F_3& F_4
\end{pmatrix}\begin{pmatrix}
T^\mu T^\nu\\
2T^{(\mu}\Zhat^{\nu)}\\
\Zhat^\mu \Zhat^\nu\\
(|q|e_1)^\mu(|q|e_1)^\nu
\end{pmatrix}
\end{align*}
where
\begin{align*}
F_1&=\frac{5am}{r^3(1-\frac{5am}{r^2})}+\frac{a^2\cos^4\th}{\sin^2\th}\Big((1-\frac12\de_0m^2z)\frac{\TT_{-a}}{r^6}+\frac{2a}{5mr^3}\Big),\quad F_4=(1-\frac12\de_0m^2z)\frac{\TT_{-a}}{r^6}+\frac{2a}{5mr^3},\\
F_2&=-\frac{2a}{r^3}(1-\de_0m^2\frac{\TT_{-a}}{2r^5})-\frac{a\cos^2\th}{\sin^2\th}\Big((1-\frac12\de_0m^2z)\frac{\TT_{-a}}{r^6}+\frac{2a}{5mr^3}\Big),\\
 F_3&=\frac{1}{\sin^2\th}\Big((1-\frac12\de_0m^2z)\frac{\TT_{-a}}{r^6}+\frac{2a}{5mr^3}\Big)-\Big(\frac{4a^2}{r^5}(1-\de_0m^2\frac{\TT_{-a}}{4r^5})+(1-\frac12\de_0m^2z)\frac{\TT_{-a}}{r^6}-\frac{2a}{5mr^3}\Big)\\
 &\geq \frac{1}{\sin^2\th}\frac{\frac{4a}{5m}r^3-4a^2r}{r^6}+\frac{\cos^2\th}{\sin^2\th}\Big((1-\frac12\de_0m^2z)\frac{\TT_{-a}}{r^6}-\frac{2a}{5mr^3}\Big)
 \end{align*}
Since $F_1, F_4>0$ for $r\geq r_++3.5m$, it suffices to prove that $F_1F_3-F_2^2>0$ for $r\geq r_++3.5m$. We compute
\begin{align*}
F_1F_3-F_2^2&> \frac{a^2}{\sin^2\th}\frac{4}{r^6}-\frac{4a^2}{r^6}(1-\de_0m^2\frac{\TT_{-a}}{2r^5})^2+\frac{\cos^2\th}{\sin^2\th}\frac{5am}{r^3}\Big((1-\frac12\de_0m^2z)\frac{\TT_{-a}}{r^6}-\frac{2a}{5mr^3}\Big)\\
&-\frac{\cos^2\th}{\sin^2\th}\frac{4a^2}{r^3}\Big((1-\frac12\de_0m^2z)\frac{\TT_{-a}}{r^6}-\frac{2a}{5mr^3}+\frac{4a}{5mr^3}\Big)\\
&-\frac{a^2\cos^4\th}{\sin^2\th}(\frac{4a^2}{r^5}+\frac{\TT_{-a}}{r^6}-\frac{2a}{5mr^3})\Big((1-\frac12\de_0m^2z)\frac{\TT_{-a}}{r^6}-\frac{2a}{5mr^3}+\frac{4a}{5mr^3}\Big)\\
&> \frac{a^2\cos^2\th}{\sin^2\th}\Big(\frac{4}{r^6}-\frac{4a}{5mr^3}(\frac{4a^2}{r^3}+\frac{\TT_{-a}}{r^6}+\frac{4a^2}{r^5}-\frac{2a}{5mr^3})\Big)\\
&+\frac{a^2\cos^2\th}{\sin^2\th}\Big(\frac{5m}{ar^3}-\frac{4}{r^3}-\frac{\TT_{-a}}{r^6}-\frac{4a^2}{r^5}+\frac{2a}{5mr^3}\Big)\Big((1-\frac12\de_0m^2z)\frac{\TT_{-a}}{r^6}-\frac{2a}{5mr^3}\Big)\geq 0
\end{align*}
where we use $\frac{\TT_{-a}}{r^6}\leq\frac{1}{r^3}$ and $\frac{4a^2}{r^5}-\frac{2a}{5mr^3}<0$ for $r\geq r_++3.5m$ in the last step.

{\bf Analysis of $\Romanupper{2}$.} We further compute
\begin{align*}
\Romanupper{2}&=\big((1-\frac12\de_0m^2z)\frac{\TT_{-a}}{r^6}-\frac{2a}{5mr^3}\big)^2\big(\Zhat\Zhat\psi+\OO\psi\big)^2\\
&-2\big((1-\frac12\de_0m^2z)\frac{\TT_{-a}}{r^6}-\frac{2a}{5mr^3}\big)\OO\psi\c\Psi_3+|\Psi_3|^2\\
&\geq\big((1-\frac12\de_0m^2z)\frac{\TT_{-a}}{r^6}-\frac{2a}{5mr^3}\big)^2\big(\frac34|\OO\psi|^2+|\Zhat\Zhat\psi|^2+2\OO\psi\c\Zhat\Zhat\psi\big)\\
&-\big((1-\frac12\de_0m^2z)\frac{\TT_{-a}}{r^6}-\frac{2a}{5mr^3}\big)\OO\psi\c\Psi_3.
\end{align*}
We rewrite
\begin{align*}
\Psi_3=-\Big((1-\frac12\de_0m^2z)\frac{\TT_{-a}}{r^6}+\frac{2a}{5mr^3}\Big)(\OO\psi-\Zhat\Zhat\psi)+\Psi_4
\end{align*}
where
\begin{align*}
\Psi_4&=\Psi_4^\aund\psia=-\frac{5am}{r^3(1-\frac{5am}{r^2})}TT\psi+\frac{4a}{r^3}(1-\de_0m^2\frac{\TT_{-a}}{2r^5})T\Zhat\psi\\
&+\Big(\frac{4a^2}{r^5}(1-\de_0m^2\frac{\TT_{-a}}{4r^5})-\frac{4a}{5mr^3}\Big)\Zhat\Zhat\psi\end{align*}
Proceeding as in the proof of the nonnegativity of $-\Psi_3^\aund S_{\aund}^{\mu\nu}$, we can also prove the nonnegativity of $-\Psi_4^\aund S_{\aund}^{\mu\nu}$ and thus modulo acceptable spacetime divergence
\begin{align*}
&-\big((1-\frac12\de_0m^2z)\frac{\TT_{-a}}{r^6}-\frac{2a}{5mr^3}\big)\OO\psi\c\Psi_3\\
&\quad\geq\big((1-\frac12\de_0m^2z)\frac{\TT_{-a}}{r^6}-\frac{2a}{5mr^3}\big)\big((1-\frac12\de_0m^2z)\frac{\TT_{-a}}{r^6}+\frac{2a}{5mr^3}\big)\Big(|\OO\psi|^2-\OO\psi\c\Zhat\Zhat\psi\Big).
\end{align*}
Therefore
\beq\lab{eq:controlofRomanupper2}
\bsplit
\Romanupper{2}&\geq \big((1-\frac12\de_0m^2z)\frac{\TT_{-a}}{r^6}-\frac{2a}{5mr^3}\big)\big((1-\frac12\de_0m^2z)\frac{\TT_{-a}}{r^6}-\frac{6a}{5mr^3}\big)\OO\psi\c\Zhat\Zhat\psi\\
&+\big((1-\frac12\de_0m^2z)\frac{\TT_{-a}}{r^6}-\frac{2a}{5mr^3}\big)\big((1-\frac12\de_0m^2z)\frac{7\TT_{-a}}{4r^6}+\frac14\frac{2a}{5mr^3}\big)|\OO\psi|^2\\
&+\big((1-\de_0m^2z)\frac{\TT_{-a}}{r^6}-\frac{2a}{5mr^3}\big)^2|\Zhat\Zhat\psi|^2\\
&\geq \big((1-\de_0m^2z)\frac{\TT_{-a}}{r^6}-\frac{a}{2mr^3}\big)\big((1-\frac12\de_0m^2z)\frac{\TT_{-a}}{r^6}-\frac{6a}{5mr^3}\big)\OO\psi\c\Zhat\Zhat\psi\\
&+\frac{7}{4}\big((1-\frac12\de_0m^2z)\frac{\TT_{-a}}{r^6}-\frac{2a}{5mr^3}\big)^2|\OO\psi|^2+\big((1-\de_0m^2z)\frac{\TT_{-a}}{r^6}-\frac{2a}{5mr^3}\big)^2|\Zhat\Zhat\psi|^2\\
&\geq \big((1+\sqrt{7})(1-\frac12\de_0m^2z)\frac{\TT_{-a}}{r^6}-(3+\sqrt{7})\frac{2a}{5mr^3}\big)\big((1-\frac12\de_0m^2z)\frac{\TT_{-a}}{r^6}-\frac{2a}{5mr^3}\big)\OO\psi\c\Zhat\Zhat\psi.
\end{split}
\eeq
Putting \eqref{eq:Psi1^2}, \eqref{eq:controlofRomanupper2} and the nonnegativity of $\Romanupper{1}$ together,
we obtain
\begin{align*}
|\Psi_1|^2&\geq(\de_0-\frac{5ar^3}{m\TT_{-a}(1-\frac{5am}{r^2})})^2\frac{m^2\TT^2_{-a}}{r^{12}}TT\psi\c m^2TT\psi\\
&+2(\de_0-\frac{5ar^3}{m\TT_{-a}(1-\frac{5am}{r^2})})\Big((1-\frac12\de_0m^2z)-\frac{2ar^3}{5m\TT_{-a}}\Big)\frac{m^2\TT^2_{-a}}{r^{12}}(TT\psi\c\Zhat\Zhat\psi+\OO\psi\c m^2TT\psi)\\
&+\big((1+\sqrt{7})(1-\frac12\de_0m^2z)-(3+\sqrt{7})\frac{2ar^2}{5m\TT_{-a}}\big)\big((1-\frac12\de_0m^2z)-\frac{2ar^3}{5m\TT_{-a}}\big)\frac{\TT^2_{-a}}{r^{12}} \OO\psi\c\Zhat\Zhat\psi\\
&\geq 2(\de_0-\frac{5ar^3}{m\TT_{-a}(1-\frac{5am}{r^2})})\Big((1-\frac12\de_0m^2z)-\frac{2ar^3}{5m\TT_{-a}}\Big)\frac{m^2\TT^2_{-a}}{r^{12}}TT\psi\c\psi_{\LL}\\
&+\big((1+\sqrt{7})(1-\frac12\de_0m^2z)-(3+\sqrt{7})\frac{2ar^2}{5m\TT_{-a}}\big)\big((1-\frac12\de_0m^2z)-\frac{2ar^3}{5m\TT_{-a}}\big)\frac{\TT^2_{-a}}{r^{12}}\OO\psi\c\psi_{\LL}.
\end{align*}
where the last inequality can be easily checked using Mathematica. Finally, using the facts that $\Phat\geq \frac12h(|mT\Psi_1|^2+|\Zhat\Psi_1|^2)\geq \frac14h|Z\Psi_1|^2$ and $
\int_{\DD(\tau_1,\tau_2)}\frac{h}{4|q|^2}|Z\psi|^2\geq \int_{\DD(\tau_1,\tau_2)}\frac{h}{4|q|^2}|\psi|^2$, we finish the proof of \eqref{eq:Psi1larger}.


\begin{thebibliography}{99}
\raggedright

\bibitem{A}
S. Alinhac,
\textit{Energy multipliers for perturbations of the Schwarzschild metric},
Comm. Math. Phys. \textbf{288} (2009), no. 1, 199--224.

\bibitem{AB}
L. Andersson and P. Blue,
\textit{Hidden symmetries and decay for the wave equation on the Kerr spacetime},
Ann. of Math. (2) \textbf{182} (2015), no. 3, 787--853.

\bibitem{Bieri}
L. Bieri,
\textit{An extension of the stability theorem of the Minkowski space in general relativity},
Ph.D. thesis, ETH Z\"urich, 2007.

\bibitem{BPT}
J. M. Bardeen, W. H. Press, and S. A. Teukolsky,
\textit{Rotating black holes: Locally nonrotating frames, energy extraction, and scalar synchrotron radiation},
Astrophys. J. \textbf{178} (1972), 347--369.

\bibitem{Carter1}
B. Carter,
\textit{Global structure of the Kerr family of gravitational fields},
Phys. Rev. \textbf{174} (1968), 1559--1571.

\bibitem{Carter2}
B. Carter,
\textit{Hamilton--Jacobi and Schr\"odinger separable solutions of Einstein's equations},
Comm. Math. Phys. \textbf{10} (1968), 280--310.

\bibitem{CeJa}
C. Cederbaum and S. Jahns,
\textit{Geometry and topology of the Kerr photon region in the phase space},
Gen. Relativity Gravitation \textbf{51} (2019), no. 6, Paper No. 79, 29 pp.

\bibitem{Chr-Survey}
D. Christodoulou,
\textit{The nonlinear stability of rotating black holes},
in \textit{Surveys in Differential Geometry, Vol. XXVII: Geometry and Physics},
International Press, Somerville, MA, 2024, pp. 1--123.

\bibitem{Ch-memory}
D. Christodoulou,
\textit{Nonlinear nature of gravitation and gravitational-wave experiments},
Phys. Rev. Lett. \textbf{67} (1991), 1486--1489.

\bibitem{Ch02}
D. Christodoulou,
\textit{The global initial value problem in general relativity},
in V. G. Gurzadyan, R. T. Jantzen, and R. Ruffini (eds.),
\textit{The Ninth Marcel Grossmann Meeting},
World Scientific, Singapore, 2002, pp. 44--54.

\bibitem{DR1}
M. Dafermos and I. Rodnianski,
\textit{The red-shift effect and radiation decay on black hole spacetimes},
Comm. Pure Appl. Math. \textbf{62} (2009), no. 7, 859--919.

\bibitem{DR2}
M. Dafermos and I. Rodnianski,
\textit{A new physical-space approach to decay for the wave equation with applications to black hole spacetimes},
in \textit{XVIth International Congress on Mathematical Physics},
World Sci. Publ., Hackensack, NJ, 2010, pp. 421--432.

\bibitem{DRS}
M. Dafermos, I. Rodnianski, and Y. Shlapentokh-Rothman,
\textit{Decay for solutions of the wave equation on Kerr exterior spacetimes III: The full subextremal case $|a|<M$},
Ann. of Math. (2) \textbf{183} (2016), no. 3, 787--913.

\bibitem{DHRT2}
M. Dafermos, G. Holzegel, I. Rodnianski, and M. Taylor,
\textit{Quasilinear wave equations on Kerr black holes in the full subextremal range $|a|<M$},
arXiv:2410.03639.

\bibitem{BD86}
L. Blanchet and T. Damour,
\textit{Radiative gravitational fields in general relativity. I. General structure of the field outside the source},
Philos. Trans. Roy. Soc. London Ser. A \textbf{320} (1986), 379--430.

\bibitem{Dyatlov}
S. Dyatlov,
\textit{Asymptotics of linear waves and resonances with applications to black holes},
Comm. Math. Phys. \textbf{335} (2015), no. 3, 1445--1485.

\bibitem{FKSY}
F. Finster, N. Kamran, J. Smoller, and S.-T. Yau,
\textit{Decay of solutions of the wave equation in the Kerr geometry},
Comm. Math. Phys. \textbf{264} (2006), no. 2, 465--503.

\bibitem{GKS}
E. Giorgi, S. Klainerman, and J. Szeftel,
\textit{Wave equations estimates and the nonlinear stability of slowly rotating Kerr black holes},
arXiv:2205.14808.

\bibitem{He-K2}
L. He and S. Klainerman,
\textit{Flux   Estimates   for  Wave Equations in  Sub-Extremal Kerr 
		via a  physical spacetime version of the Whiting transform},
in preparation.

\bibitem{Hintz1}
P. Hintz,
\textit{Nonlinear stability of subextremal Kerr black holes},
arXiv:2606.28253.

\bibitem{Hormander}
L. H\"ormander,
\textit{Lectures on Nonlinear Hyperbolic Differential Equations},
Math\'ematiques \& Applications, vol. 26,
Springer-Verlag, Berlin, 1997.

\bibitem{Kehr1}
L. M. A. Kehrberger,
\textit{The case against smooth null infinity I: Heuristics and counter-examples},
Ann. Henri Poincar\'e \textbf{23} (2022), no. 3, 829--921.

\bibitem{KS}
S. Klainerman and J. Szeftel,
\textit{Global non-linear stability of Schwarzschild spacetime under polarized perturbations},
Annals of Mathematics Studies, vol. 210,
Princeton University Press, Princeton, NJ, 2020.

\bibitem{KS-GCM1}
S. Klainerman and J. Szeftel,
\textit{Construction of GCM spheres in perturbations of Kerr},
Ann. PDE \textbf{8} (2022), no. 2, Paper No. 17, 153 pp.

\bibitem{KS-GCM2}
S. Klainerman and J. Szeftel,
\textit{Effective results in uniformization and intrinsic GCM spheres in perturbations of Kerr},
Ann. PDE \textbf{8} (2022), no. 2, Paper No. 18, 89 pp.

\bibitem{KS:Kerr}
S. Klainerman and J. Szeftel,
\textit{Kerr stability for small angular momentum},
Pure Appl. Math. Q. \textbf{19} (2023), no. 3, 791--1678.

\bibitem{KS:Survey}
S. Klainerman and J. Szeftel,
\textit{Brief introduction to the nonlinear stability of Kerr},
Pure Appl. Math. Q. \textbf{20} (2024), no. 4, 1721--1761.

\bibitem{LT}
H. Lindblad and M. Tohaneanu,
\textit{A local energy estimate for wave equations on metrics asymptotically close to Kerr},
Ann. Henri Poincar\'e \textbf{21} (2020), no. 11, 3659--3726.

\bibitem{Ma}
S. Ma,
\textit{Uniform energy bound and Morawetz estimate for extreme components of spin fields in the exterior of a slowly rotating Kerr black hole II: Linearized gravity},
Comm. Math. Phys. \textbf{377} (2020), no. 3, 2489--2551.

\bibitem{MMTT}
J. Marzuola, J. Metcalfe, D. Tataru, and M. Tohaneanu,
\textit{Strichartz estimates on Schwarzschild black hole backgrounds},
Comm. Math. Phys. \textbf{293} (2010), no. 1, 37--83.

\bibitem{MaS}
S. Ma and J. Szeftel,
\textit{Energy-Morawetz estimates for the wave equation in perturbations of Kerr},
arXiv:2410.02341.

\bibitem{MaS2}
S. Ma and J. Szeftel,
\textit{Energy-Morawetz estimates for Teukolsky equations in perturbations of Kerr},
arXiv:2603.23437.

\bibitem{Millet}
P. Millet,
\textit{Optimal decay for solutions of the Teukolsky equation on the Kerr metric for the full sub-extremal range},
arXiv:2302.06946.

\bibitem{Mor}
C. S. Morawetz,
\textit{Decay for solutions of the exterior problem for the wave equation},
Comm. Pure Appl. Math. \textbf{28} (1975), 229--264.

\bibitem{Yacov}
Y. Shlapentokh-Rothman,
\textit{Quantitative mode stability for the wave equation on the Kerr spacetime},
Ann. Henri Poincar\'e \textbf{16} (2015), no. 1, 289--345.

\bibitem{S-Rita1}
Y. Shlapentokh-Rothman and R. Teixeira da Costa,
\textit{Boundedness and decay for the Teukolsky equation on Kerr in the full subextremal range $|a|<M$: Frequency space analysis},
arXiv:2007.07211.

\bibitem{S-Rita2}
Y. Shlapentokh-Rothman and R. Teixeira da Costa,
\textit{Boundedness and decay for the Teukolsky equation on Kerr in the full subextremal range $|a|<M$: Physical space analysis},
arXiv:2302.08916.

\bibitem{Shen}
D. Shen,
\textit{Construction of GCM hypersurfaces in perturbations of Kerr},
Ann. PDE \textbf{9} (2023), no. 2, Paper No. 11, 142 pp.

\bibitem{Shen23}
D. Shen,
\textit{Stability of the Minkowski space with minimal decay},
arXiv:2310.07483.

\bibitem{Shen24}
D. Shen,
\textit{Exterior stability of Minkowski spacetime with borderline decay},
arXiv:2405.00735, accepted in Ann. Sci. \'Ec. Norm. Sup\'er.

\bibitem{St}
J. Stogin,
\textit{Nonlinear wave dynamics in black hole spacetimes},
arXiv:2603.14638.

\bibitem{TT}
D. Tataru and M. Tohaneanu,
\textit{A local energy estimate on Kerr black hole backgrounds},
Int. Math. Res. Not. IMRN (2011), no. 2, 248--292.

\bibitem{Teo}
E. Teo,
\textit{Spherical photon orbits around a Kerr black hole},
Gen. Relativity Gravitation \textbf{35} (2003), no. 11, 1909--1926.

\bibitem{W}
B. F. Whiting,
\textit{Mode stability of the Kerr black hole},
J. Math. Phys. \textbf{30} (1989), no. 6, 1301--1305.

\end{thebibliography}
\end{document}